\pdfoutput=1
\documentclass[12pt]{caltech_thesis}
\usepackage[T1]{fontenc} 
\usepackage[utf8]{inputenc}
\usepackage{csquotes} 
\usepackage[hyphens]{url}
\usepackage{lipsum}
\usepackage{amsmath,amsthm}
\usepackage{bbm}
\usepackage{graphicx}
\usepackage{hyperref}
\hypersetup{
    breaklinks = true,
    colorlinks = true,
    urlcolor = blue,
    citecolor = magenta,
    linkcolor = purple
}

\usepackage{tabularx}
\usepackage{braket}
\usepackage[disable]{todonotes}
\usepackage{enumitem}

\usepackage{newtxtext,newtxmath}

\usepackage{algorithm}
\usepackage[noend]{algpseudocode}
\theoremstyle{plain}
\newtheorem{theorem}{Theorem}

\newtheorem{lemma}{Lemma}
\newtheorem{corollary}{Corollary}
\newcommand{\tr}[1]{\operatorname{\textnormal{Tr}}\left( {#1} \right)} 

\usepackage{wrapfig2}
\usepackage{subcaption}
\graphicspath{{figures/}{figures-temp/}}

\usepackage[
    backend=biber,
    natbib=true,
    style=ieee,
    citestyle=numeric-comp,
    sortcites=true,
    sorting=noney
]{biblatex}

\usepackage[dvipsnames,svgnames,table]{xcolor}
\usepackage{cleveref}
\crefname{figure}{Fig.}{Figs.}

\newcommand{\psistate}[0]{$\ket{\psi}$}

\newcommand\extrafootertext[1]{%
    \bgroup
    \renewcommand\thefootnote{\fnsymbol{footnote}}%
    \renewcommand\thempfootnote{\fnsymbol{mpfootnote}}%
    \footnotetext[0]{#1}%
    \egroup
}

\begin{document}

\frontmatter

\title{Resource Management in Heterogeneous Quantum\\Repeater Networks}
\author{Naphan Benchasattabuse}

\degreeaward{Doctor of Philosophy}
\university{Keio University}
\address{Fujisawa, Kanagawa}
\unilogo{keio-color.png}
\graduateschool{Graduate School of Media and Governance}
\academicyear{2025}
\copyyear{2025}
\defenddate{25 June 2025}

\orcid{0000-0003-4475-3015}

\rightsstatement{All rights reserved}

\newgeometry{
  top=1in,
  bottom=1in,
  left=1.25in,
  right=1.25in
}
\maketitle
\restoregeometry

\begin{abstract}
  In this thesis, I explore whether it is possible to build a unified Quantum Internet architecture that supports different types of quantum repeaters---especially the two most distinct and seemingly incompatible ones: memory-based quantum repeaters and all-photonic, memoryless repeaters. These technologies have traditionally been developed with the aim of becoming the single dominant solution, but I ask: Can they work together in the same network? What kind of architecture would support both? And how can simulation help us understand what is needed to manage such a network at scale? To address these questions, I propose an architecture based on an existing recursive network design and programmable RuleSet-based protocols that can coordinate diverse hardware components. I introduce a new emitter-photon building block to bridge memory-based and all-photonic segments, and show how classical networking abstractions can be extended to manage quantum operations. While I have developed a simulation tool grounded in these architectural principles and validated it against existing simulators and analytical models, a full-scale investigation of the resource trade-offs and performance implications remains future work. Nevertheless, the results so far suggest that a unified, heterogeneous quantum network is not only possible but increasingly practical with current technologies---though ongoing experimental progress will be essential to fully realize this vision.
  
\end{abstract}

\begin{acknowledgements}
I am very fortunate to have Rodney Van Meter and Michal Hajdu\v{s}ek as my advisors.
Rod always reminded me not just to focus on the cool and fun ideas, but to always think of the impact and solve the real problems that would advance the field.
He taught me, and everyone else in our group, to think beyond the short term, to look ahead and identify the challenges that truly matter.
Rod has changed how I view research, how I look at problems, how I think, and how I write.
Michal was always helpful in answering my sometimes dumb, basic questions about quantum information; he helped me refine my thinking, make my ideas concrete, and polish rough concepts from scratch, and was invaluable in preparing my papers and figures.

I would like to thank my thesis committee, Rodney Van Meter, Michal Hajdu\v{s}ek, Takahiko Satoh, Shota Nagayama, Hiroyuki Kusumoto, and Hideyuki Kawashima, for their valuable advice and comments in shaping this thesis.

I would also like to express my gratitude to Andreas B\"{a}rtschi and Stephan Eidenbenz for being good company, wonderful hosts, and great mentors during my internship at Los Alamos National Laboratory.
Andreas and Stephan helped me learn to think with the rigor and formality of a computer scientist, a perspective I greatly needed.
I would also like to thank Luis Pedro Garc\'{i}a-Pintos, Nathan Lemons, John Golden, and Yi\u{g}it Suba\c{s}i for the fruitful discussions that led to a paper included in this thesis.

My PhD would not have started without the IBM Quantum Challenge 2019, and for that, I thank the organizers for an event that made me realize I wanted to pursue a PhD and a research career in the field of quantum computation.
I would also like to thank the MEXT Super Global scholarship program for supporting the second and third year of my graduate school journey, which enabled me to study without financial burden.

I would also like to express my appreciation to Ryosuke Satoh, Yasuhiro Okura, Sitong Liu, Takaaki Matsuo, Kentaro Teramoto, Shigetora Miyashita, Kento Samuel Soon, Poramet Pathumsoot, Monet Tokuyama, and all the members of AQUA for being good friends.
Thank you for the good food, for exploring the city, and for all the enjoyable times in Japan during my PhD.
I also thank my high school friends in the same journey of PhD here in Japan, Sopak Supakul and Panchanit Ubolsaard, for our occasional meetups, fun times, and mutual encouragement.

Finally, I am forever indebted and grateful to my parents, Somrak and Narongchai Benchasattabuse, and to my girlfriend, Thamolwan Pathomkasikul, for their unwavering emotional support.
I could not have navigated all the obstacles and struggles to arrive at where I am today without you.
\end{acknowledgements}

\begin{publishedcontent}[iknowwhattodo]
\nocite{naphan-engineering-challenges-in-rgs,naphan-half-rgs,cocori-quisp,rdv-qi-architecture,naphan-rgs-tqe,naphan-integrating-purification-rgs,joaquin-michal-cross-validation-quantum-network-simulators,naphan-spo,naphan-lower-bound-qaoa}
\end{publishedcontent}

\tableofcontents
\listoffigures
\listoftables

\printnomenclature

\mainmatter
\makeatletter
\def\clearforchapter{\cleardoublepage}
\makeatother

\clearpage
\chapter{Introduction}
\label{chapter:introduction}

\section{Motivations}

Quantum communication is an indispensable element of quantum information science.
It enables the exchange of quantum states and, more importantly, the creation of entangled states shared across distant parties.
Such entanglement is essential for scaling up quantum computers, enhancing quantum sensing technologies, strengthening privacy in communications and securing remote or cloud resources, and may ultimately reshape our relationship with the Internet.
While significant progress has been made in the physics of entanglement distribution, building quantum networks in practice also requires tight coordination between quantum and classical components, as well as mechanisms to accommodate diverse and evolving technologies.
In particular, how the classical side of the network should coordinate, control, and manage resources while remaining flexible enough to support future hardware presents significant engineering challenges.

In this thesis, I will describe our progress in developing architectures, communication protocols, and resource management strategies for entanglement-based quantum networks.
By approaching these systems from an engineering perspective, we begin to see which aspects need to be revisited and worked out to move from theory to practice.
Drawing on classical network design, emphasizing the role of classical communication exchanges, and considering the realities of engineering practical systems, I have identified and made progress in addressing the following questions.

\subsection{How do we design a unified network architecture that supports heterogeneous quantum repeaters?}
Quantum repeaters are the backbone of long-distance quantum communication~\cite{azuma-rmp-repeater-review}, yet they come in diverse forms, ranging from memory-based repeaters with quantum memories to emerging all-photonic designs that rely entirely on flying qubits. These differences pose a challenge: how can we design a unified, end-to-end network architecture that supports such fundamentally distinct technologies? Unlike classical networks, where early flexibility gave way to the rigidity of the TCP/IP stack---especially the narrow waist of the IP layer, which has become nearly impossible to change without breaking the Internet---we have the opportunity to avoid locking ourselves into a premature standard. To accommodate heterogeneity and future evolution, we need abstractions and interfaces that support interoperability, modularity, and adaptability across the stack. What architectural principles can enable this? What abstractions should we define to bridge fundamentally different hardware assumptions? How can we design a quantum network stack that remains flexible as the technology landscape continues to evolve?

\subsection{How can simulation guide the design and resource management of scalable quantum networks?}
Simulation has long been a powerful tool for evaluating complex systems, especially when analytical models fall short or when experimentation is costly or infeasible. In quantum networking, where physical testbeds are still limited in scale and flexibility, simulation becomes a key enabler for exploring network behavior, validating architectural choices, and refining protocol designs. Yet quantum systems introduce unique challenges---such as non-determinism, error correlations, and the need to track entangled states---which make classical simulation nontrivial. A good simulation platform must strike a balance between physical fidelity and computational efficiency, enabling meaningful predictions while remaining usable at large scale. What models and abstractions are necessary to capture the key behaviors of quantum networks? How do we validate simulations against real experiments? What kind of metrics, outputs, or insights should simulations deliver to inform the next generation of quantum network design?

\subsection{How can all-photonic quantum repeaters become practical and interoperable with memory-based designs?}
All-photonic quantum repeaters~\cite{azuma-rgs} offer a compelling vision: high-rate entanglement distribution without relying on fragile or expensive quantum memories. While traditionally considered futuristic, advances in deterministic photon sources and photonic entanglement generation have brought them closer to experimental reality. However, engineering challenges remain---particularly when moving beyond idealized settings to realistic regimes that include loss, decoherence, and imperfect hardware. Moreover, most existing designs treat all-photonic repeaters as self-contained chains, often disconnected from broader network architectures based on memory-equipped nodes. To harness their advantages in practice, we need strategies to improve their robustness and integrate them with existing memory-based infrastructures. Can all-photonic repeaters operate effectively in high-error environments? How do we connect them to end nodes with memory-based capabilities, and what protocols or abstractions enable seamless interoperability? What new hybrid architectures become possible when both paradigms work together?

\section{Overview}
\label{sec:introduction_outline}

This thesis begins by reviewing the mathematical foundations and concepts of quantum state representation and error modeling used in quantum networking (\cref{chapter:preliminaries}), followed by an overview of quantum networks and the engineering challenges in their realization (\cref{chapter:quantum-networks}).
The work then presents a unified network architecture for memory-based quantum repeaters, building on the principles of quantum recursive network architecture and demonstrating its effectiveness through simulation (\cref{chapter:architecture-memory}).
This framework is subsequently extended to support all-photonic quantum repeaters---previously thought to be incompatible with memory-based systems---by detailing their integration into a single heterogeneous network.
This extension provides new protocol designs, addresses engineering challenges lacking in prior conceptual models, and examines resource and error management to move all-photonic designs toward practical implementation (\cref{chapter:architecture-all-photonic}).
With these architectural foundations established, the thesis shifts its focus to applications, using quantum optimization algorithms as a case study to examine how their runtime behavior translates into demands on network resources (\cref{chapter:qaoa-and-network-resources}).
The thesis concludes with a discussion of open problems and future directions (\cref{chapter:discussion-and-conclusion}).

In the following, I provide a chapter-by-chapter summary that details the background, methodology, and main results of this thesis. Each chapter's introduction will provide more in-depth motivation and context for the topics discussed.

\textit{\Cref{chapter:preliminaries}: \nameref*{chapter:preliminaries}}

This chapter provides essential background knowledge.
It will briefly review key concepts from quantum information theory, entanglement and Bell pairs, quantum channels, Clifford groups, and fidelity of quantum states.
Furthermore, it will introduce fundamental building blocks relevant to quantum networks, including an overview of stabilizer formalism, graph states, and the basics of quantum error correction codes that underpin advanced repeater designs.

\textit{\Cref{chapter:quantum-networks}: \nameref*{chapter:quantum-networks}}

This chapter delves into the foundational principles of quantum repeater networks, establishing the state of the art and the core motivations for this work.
The promise of a global Quantum Internet is driven by revolutionary applications—from quantum key distribution (QKD)~\cite{bb84-conf,ekertQuantumCryptographyBased1991} and distributed sensing~\cite{degen-sensing,proctor-quantum-sensing} to private delegated computation~\cite{broadbent-bfk-protocol,morimaeBlindQuantumComputation2013,fitzsimonsPrivateQuantumComputation2017} and large-scale quantum computing—all of which rely on distributing entanglement as their key enabling task.
However, the fundamental obstacle of photon loss limits the distance of direct communication, making quantum repeaters indispensable.
While repeaters solve the distance problem using core operations like entanglement swapping, this process in turn introduces accumulated operational errors that must be managed.
To provide context for the solutions developed in this thesis, this chapter surveys these challenges and provides a classification of different quantum repeater generations (1G, 2G, 3G)~\cite{muralidharan-repeater-generation}, outlining their distinct approaches to error management and setting the stage for the architectural work that follows.

\textit{\Cref{chapter:architecture-memory}: \nameref*{chapter:architecture-memory}}

While experimental progress in quantum networking has been advancing rapidly~\cite{valivarthiTeleportationSystemsQuantum2020,davisEntanglementSwappingSystems2025,davisTeleportationEntanglementFermilab2024,ruskucMultiplexedEntanglementMultiemitter2025,pompiliRealizationMultinodeQuantum2021,hermansQubitTeleportationNonneighbouring2022,vanleentEntanglingSingleAtoms2022,krutyanskiy-northup-ion-trap-entanglement,senaRobustHighFidelityQuantum2025,valivarthiQuantumTeleportationMetropolitan2016a,sunQuantumTeleportationIndependent2016,chungDesignImplementationIllinois2022,alshowkanReconfigurableQuantumLocal2021,craddockAutomatedDistributionPolarizationEntangled2024,liuCreationMemoryMemory2024,knautEntanglementNanophotonicQuantum2024}, the architectural work required to build scalable, interoperable networks has lagged behind, with most frameworks focusing on single-network scenarios rather than a complete, end-to-end design~\cite{bacciottiniLeveragingInternetPrinciples2025a,beauchampModularQuantumNetwork2025a,dahlbergLinkLayerProtocol2019,gauthierArchitectureControlEntanglement2023,gauthierControlArchitectureEntanglement2023a,illianoQuantumInternetProtocol2022,kumarQuantumInternetTechnologies2025a,liSurveyQuantumInternet2024,dahlberg-netqasm}.
To address this research gap, this chapter details a comprehensive architecture for memory-based networks.
The framework builds upon the foundational principles of the Quantum Recursive Network Architecture (QRNA) by Van Meter, Touch, and Horsman~\cite{vanmeterRecursiveQuantumRepeater2011}, which adapts the recursive design of the classical Internet~\cite{touchDynamicRecursiveUnified2011,touchRNAMetaprotocol2008,touchRecursiveNetworkArchitecture2006a} to the quantum domain.
The core idea of this recursion is to use layers of abstraction to manage complexity; for instance, an entire network at a lower layer can be represented as a single virtual node at a higher layer.
This allows network details to be hidden when crossing administrative boundaries and greatly simplifies protocol reasoning.
The chapter details the collaborative work that established these principles in a cohesive framework, presented in two complementary papers at the \textit{2022 IEEE International Conference on Quantum Computing and Engineering (QCE)}.

The first of these papers, led by Rodney Van Meter, introduced the overarching \emph{A Quantum Internet Architecture}~\cite{rdv-qi-architecture}.
The second, led by Ryosuke Satoh, detailed the software realization of this architecture---the Quantum Routing Software Architecture (QRSA)---and its implementation in the Quantum Internet Simulation Package (QuISP)~\cite{cocori-quisp}.
Built on the OMNeT++ discrete-event simulator, QuISP implements the principles of QRNA and QRSA using modular, task-oriented software components that differ from a classical OSI-like layered model often utilized in other models of quantum network protocols~\cite{dahlbergLinkLayerProtocol2019}.
My specific contribution within this effort was to concretize the framework by formalizing the RuleSet protocol's inter-node messages and execution logic, as well as improving the quantum state and error models for the QuISP simulator.

To validate this simulation-based approach, the final part of the chapter presents a rigorous cross-validation study of QuISP against another leading simulator, SeQUeNCe~\cite{wuSeQUeNCeCustomizableDiscreteevent2021}, developed by the University of Chicago and Argonne National Laboratory.
In this collaborative work, led by Joaquin Chung and Michal Hajdu\v{s}ek and published at the \textit{IEEE INFOCOM 2025 - IEEE Conference on Computer Communications, Workshop on Quantum Networked Applications and Protocols}, we show that despite different protocol implementations, both simulators agree on network behavior where analytical models are available~\cite{joaquin-michal-cross-validation-quantum-network-simulators}.
This result establishes the credibility of our architectural model and the simulation tools used to evaluate it.

\textit{\Cref{chapter:architecture-all-photonic}: \nameref*{chapter:architecture-all-photonic}}

While most repeater designs rely on quantum memories, the distinct paradigm of memoryless, all-photonic quantum repeaters~\cite{azuma-rgs,borregaard-3g-repeater} offers a compelling alternative, particularly for applications like QKD where end nodes can process incoming photons without storage.
However, a key research issue has been that these schemes have been studied in isolation, with significant engineering challenges and a lack of a clear path for integration with memory-based networks.
This chapter presents my novel contributions, detailed in a series of lead-author publications~\cite{naphan-engineering-challenges-in-rgs,naphan-half-rgs,naphan-rgs-tqe,naphan-integrating-purification-rgs} with my collaborators Michal Hajdu\v{s}ek and Rodney Van Meter, that aim to make all-photonic repeaters practical and interoperable, thereby realizing the vision of a truly heterogeneous repeater network that is a central theme of this thesis.

The investigation begins with our work published in \textit{IEEE Network}~\cite{naphan-engineering-challenges-in-rgs}, which provides a tutorial on all-photonic repeaters and a systematic analysis of the engineering challenges that are often overlooked.
This work identified key issues in timing synchronization, integration with memory-equipped end nodes, and the classical communication overhead required for each Bell pair.
A key result was showing that the classical post-processing could be made dramatically more efficient, reducing communication costs by orders of magnitude.

Building on this analysis, the chapter introduces a novel building block, the half-Repeater Graph State (half-RGS), which is constructed from deterministic quantum emitter interfaces~\cite{buterakos-graph-generation,hilaire-rgs-optimizing-gen-time}.
This modular design, published at the \textit{2024 IEEE International Conference on Quantum Computing and Engineering (QCE)}~\cite{naphan-half-rgs}, solves nearly all of the challenges raised in our \textit{IEEE Network} paper.
It simplifies photonic state generation, resolves the issue of tight global timing, and provides a natural interface to memory-based systems, allowing an all-photonic chain to be abstracted as a single virtual link within the QRNA model.
Finally, the chapter details how this building block can be used to solve a decade-long open problem raised in the original all-photonic repeater paper by Azuma, Tamaki, and Lo~\cite{azuma-rgs}: how to incorporate entanglement purification.
In a paper to be published at the \textit{2025 IEEE International Conference on Quantum Computing and Engineering (QCE)}~\cite{naphan-integrating-purification-rgs}, we provide a proof-of-concept analysis showing that by using optimistic purification~\cite{razavi-optimistic-purification-pumping,hartmann-optimistic-purification,mobayenjarihani-optimistic-purification-qce}, this primitive creates a purification-enhanced RGS scheme that improves fidelity without sacrificing the high generation rates that make all-photonic repeaters attractive.

\textit{\Cref{chapter:qaoa-and-network-resources}: \nameref*{chapter:qaoa-and-network-resources}}

With the architectural framework established, a critical research issue remains: understanding the workloads that a quantum network must support.
It is widely believed that distributed quantum computation will be necessary to scale quantum computing beyond a single device~\cite{barral-review-distributed-qc,marcello-angela-sara-distributed-qc-survey}, but there is a lack of concrete traffic models for such scenarios.
Focusing on optimization as a key application area where quantum advantage is anticipated with a few hundred logical qubits~\cite{abbasChallengesOpportunitiesQuantum2024}, this chapter bridges the gap between network architecture and application requirements.

The chapter first details a novel optimization algorithm based on a subdivided phase oracle, a work with my collaborators Takahiko Satoh, Michal Hajdu\v{s}ek, and Rodney Van Meter, published at the \textit{2022 IEEE International Conference on Quantum Computing and Engineering (QCE)}~\cite{naphan-spo}.
This algorithm utilizes knowledge of the problem's objective value distribution to modify the standard Grover search, and can be understood as a specific instance of the broader QAOA framework.

The main contribution, however, is a rigorous analysis of QAOA itself.
In a work with my collaborators Andreas Bärtschi, Luis Pedro García-Pintos, John Golden, Nathan Lemons, and Stephan Eidenbenz, published in \textit{Physical Review A}~\cite{naphan-lower-bound-qaoa}, we utilize speed limit results from quantum annealing~\cite{luis-pedro-annealing-lower-bounds} to derive some of the first analytical lower bounds on QAOA's runtime.
The key result of this work is showing that the required runtime is not monolithic; it can range from constant-time for easy problems to polynomial or even square-root scaling for unstructured search.
Using these findings, I argue about the workload implications for a distributed datacenter model.
This chapter discusses how factors like problem partitioning, the choice of fault-tolerant computational model, and trust assumptions fundamentally alter the nature of the network traffic and the type of entangled resources required, providing crucial insights that inform the design of realistic traffic models.

\textit{\Cref{chapter:discussion-and-conclusion}: \nameref*{chapter:discussion-and-conclusion}}

Finally, I summarize the primary contributions of this thesis, revisits the main research questions in light of the presented results, and discusses the broader implications of this work for the development of the Quantum Internet.
It will also outline promising directions for future research and address remaining open challenges.
This structured approach is intended to guide the reader from fundamental principles to the specific contributions of this thesis and their significance for the field.

\subsection{How to Read This Thesis}
\label{subsec:introduction_how_to_read}

This thesis is written with the assumption that the reader is already familiar with the basics of quantum information and quantum computing, including concepts such as manipulating quantum states, quantum circuit notation, mixed states, and state fidelities. Prior extensive experience with quantum networking, classical networking theory, or specific repeater architectures is not a strict requirement, as this work aims to briefly cover the necessary foundational concepts for these areas.

Each chapter, particularly those detailing our contributions, endeavors to begin with a high-level intuition and motivation before delving into the technical specifics. This approach is intended to allow readers interested primarily in the overarching ideas and outcomes to access them readily, while still providing the necessary depth for those wishing to understand the detailed mechanisms. Readers with significant prior experience in quantum networks, particularly those familiar with repeater generations and basic protocols, may find it productive to start with \cref{chapter:architecture-memory} where our primary architectural contributions begin, after reviewing the preliminaries in \cref{chapter:preliminaries} and \cref{chapter:quantum-networks} as needed.

\clearpage
\chapter{Preliminaries on stabilizer formalism and graph states}
\label{chapter:preliminaries}

This thesis is written with the goal of being accessible to a broad audience in STEM, especially those with backgrounds in engineering, computer science, physics, or related fields at the senior undergraduate level.
I aim to assume as little prior knowledge as possible, while still moving efficiently through the material.
Readers are expected to be comfortable with basic quantum information concepts---such as representing quantum states with state vectors and density matrices, interpreting quantum circuits, and understanding the action of common quantum gates.
In terms of mathematics, I assume familiarity with foundational linear algebra concepts like linear independence, orthogonality, bases, eigenvalues, and eigenvectors, as well as a working knowledge of probability theory and set notation.
Some basic understanding of graph theory will also be helpful---knowing what a graph is, how adjacency matrices or lists work, and concepts like paths and spanning trees.
For any essential ideas that are easy to forget or less commonly used outside of specialized work, I will provide a brief reminder in the relevant sections, so that readers do not need to pause to look them up.

\section{Useful mathematical and quantum information concepts}

In this section, I will give a brief review of some relevant concepts and definitions we will use to derive my results in this thesis.

\subsection{Entanglement and Bell pairs}

Assuming a basic familiarity with writing quantum states in bra-ket notation~\cite{mike-ike-book}, let us briefly review the concept of entanglement.
Entanglement is a property of two or more qubits in which measurements in certain bases will always give correlated results.
That is, knowing the results of a measurement on a subset of the qubits can determine the measurement outcomes for the remaining unmeasured qubits.
The smallest unit of entanglement, and one that will appear repeatedly in this thesis, is the \emph{Bell pair}---an entangled quantum state of two qubits.
And the most common form of Bell pair is the $\ket{\Phi^+}$ which is given by
\begin{equation}
    \ket{\Phi^+} = \frac{1}{2}\left(\ket{00} + \ket{11}\right).
\end{equation}
As we can see, measuring any of the two qubits of $\ket{\Phi^+}$ in the Z basis will determine the results of the other qubit.

A two-qubit system has four basis states.
While quantum states are typically written using the computational basis, sometimes it is more useful to use a different, fully entangled basis known as the \emph{Bell basis}.
The four basis states of the Bell basis are given by
\begin{align}
    \ket{\Phi^+} &= \frac{1}{\sqrt{2}}\left(\ket{00} + \ket{11}\right), \label{eq:bellpair-1}\\
    \ket{\Phi^-} &= \frac{1}{\sqrt{2}}\left(\ket{00} - \ket{11}\right), \label{eq:bellpair-2}\\
    \ket{\Psi^+} &= \frac{1}{\sqrt{2}}\left(\ket{01} + \ket{10}\right), \label{eq:bellpair-3}\\
    \ket{\Psi^-} &= \frac{1}{\sqrt{2}}\left(\ket{01} - \ket{10}\right). \label{eq:bellpair-4}
\end{align}
This basis is useful because the four states can be transformed into one another (ignoring the global phase) by applying a single-qubit Pauli gate to just one of the two qubits.
For example, the relationships between the states include
\begin{equation}
    (Z \otimes I)\ket{\Phi^+} = \ket{\Phi^-}, \quad (X \otimes I)\ket{\Phi^+} = \ket{\Psi^+}, \quad \text{and} \quad (Y \otimes I)\ket{\Phi^+} = i\ket{\Psi^-}.
\label{eq:bellbasis-relationship}
\end{equation}

\subsection{Quantum channels}

When dealing with real-world quantum systems, errors from noise and environmental interactions will inevitably alter the quantum states.
A \emph{quantum channel} is the formalism used to characterize and describe these transformations.
Because a channel can change a pure state into a statistical mixture of possible outcomes, it is necessary to use the density matrix representation, $\rho$, to describe the state of the system.

Formally, a quantum channel is a completely positive and trace-preserving (CPTP) map, $\mathcal{E}$.
The most common way to represent this map is via the operator-sum, or Kraus, representation, given by
\begin{equation}
    \mathcal{E}(\rho)=\sum_k A_k \rho A_k^{\dagger},
\end{equation}
where the $\{A_k\}$ are the set Kraus operators for the channel.
To ensure that the total probability is conserved, these operators must satisfy the trace-preserving condition $\sum_k A_k^{\dagger} A_k=I$.

Conceptually, it is often simpler to think of each Kraus operator $A_k$ as representing a possible error that can occur, with the probability of that specific outcome being $\tr{A_k \rho A_k^{\dagger}}$.
It is important to note that the Kraus representation for a given channel is not unique, thus the probability of any single event $k$ is not a fixed property of the channel, but is dependent on both the chosen set of operators and the input state $\rho$.

In this thesis, we will limit our discussion primarily to \emph{Pauli channels}.
For this type of channel, the errors can be described as combinations of single-qubit Pauli gates acting on the state with certain probabilities.
An example of a single-qubit Pauli channel acting on a state $\rho$ is
\begin{equation}
    \mathcal{E}(\rho) = (1 - p_x - p_y - p_z) \rho + p_x X \rho X + p_y Y \rho Y + p_z Z \rho Z,
\end{equation}
where $p_x$, $p_y$, and $p_z$ denote the probability that an $X$, $Y$, or $Z$ error occurs.
This model can be extended to multi-qubit states; for a general two-qubit state, there are 15 possible Pauli error terms.

\subsection{Fidelity of quantum states}

To quantify how close an experimentally prepared quantum state is to an ideal target state, we use a metric called \emph{fidelity}.
It provides a single number, between 0 and 1, that captures how well a state has been prepared or preserved, with a fidelity of 1 indicating a perfect match.

For the comparison between a potentially mixed quantum state $\rho$ and a pure target state $\ket{\psi}$, the fidelity is formally defined as the overlap between the two states.
This can be expressed in two equivalent ways:
\begin{equation}
    F(\rho, \ket{\psi}) = \bra{\psi}\rho\ket{\psi} = \mathrm{Tr}(\rho \ket{\psi}\bra{\psi}).
\end{equation}
Note that other definitions of fidelity exist in the literature, such as the square root of this value or a more general form for comparing two mixed states~\cite{mike-ike-book}.
However, for the purposes of this thesis, we will limit ourselves to the definition above.

As the focus of this thesis is the distribution of Bell pairs, it is useful to consider this definition in that context.
Given that the four Bell states in \cref{eq:bellpair-1,eq:bellpair-2,eq:bellpair-3,eq:bellpair-4} form a complete and orthogonal basis, any single-qubit Pauli error acting on a Bell pair simply transforms it into another Bell state.
Consequently, the fidelity of a noisy Bell pair with respect to a target state (e.g., $\ket{\Phi^+}$) is equivalent to the probability of measuring it in the Bell basis and obtaining the expected outcome.

\subsection{Clifford groups and Clifford gates}
\label{subsec:preliminaries:clifford}

Assuming that we know the basic quantum gates used in simple quantum programs.
Many of the familiar gates, Hadamard (H), not (X), phase (S), and controlled-not (CX/CNOT) gates, are classified as Clifford gates.
So what are Clifford gates and what is a Clifford group?

A \emph{group} (in the Clifford group) is a set $G$ with an operation that takes two elements of $G$ and produce a new element inside $G$.
The group operation is a function $G \times G \rightarrow G$ that needs to also satisfy certain properties; (1) associativity, (2) identity, and (3) inverse.
Some examples of groups with their accompanied operation are integers with addition, non-zero real numbers with multiplication, and a set of bijective function (or permutation) over $n$ items.

For the purpose of this thesis, we only need to understand the key consequence of Clifford operators forming a group---that is, circuits composed entirely of Clifford gates can often be simplified or ``compressed.''
Specifically, a sequence of Clifford operations applied to the initial state $\ket{0}^{\otimes n}$ can only produce quantum states that belong to a small, finite subset of the vast $n$-qubit Hilbert space.
This special property gives rise to two powerful descriptive tools that are used extensively in this thesis: the stabilizer formalism and the equivalent graph state formalism, which are covered next.

\section{Stabilizer formalism}
\label{sec:stabilizer-formalism}

A very useful class of quantum states that we will see repeatedly in this thesis is the class of \emph{stabilizer states}~\cite{gottesman-phd-thesis}.
These states possess structure that allows for efficient classical description and simulation under certain operations~\cite{aaronson-chp-sim}.
Stabilizer states are states that can be prepared by applying a quantum circuit composed entirely of Clifford gates to the all-zero initial state $\ket{0}^{\otimes n}$.
The circuits that prepare stabilizer states are called stabilizer circuits.
Unlike arbitrary quantum states---which generally require classical resources scaling exponentially with the number of qubits $n$ to store their amplitudes---$n$-qubit stabilizer states admit an efficient classical representation.
Specifically, an $n$-qubit stabilizer state \psistate{} is uniquely defined as the simultaneous +1 eigenstate of a set of $n$ independent, commuting Pauli operators $\{g_1, g_2, \ldots, g_n \}$, known as the generating set or the \emph{stabilizer generators} (or simply the generators).
\nomenclature{Stabilizer formalism}{A compact way to describe and efficiently simulate a class of quantum states, known as stabilizer states, which are simultaneously +1 eigenstates of a set of commuting Pauli operators}

Formally, the set of all Pauli operators that \emph{stabilize} a given state---that is, operators that map the state back to itself---form what is called the stabilizer group $\mathcal{S}$ of \psistate.
This group, $\mathcal{S}$, is an abelian subgroup of the $n$-qubit Pauli group $\mathcal{P}_n$.
For a pure $n$-qubit stabilizer state, $\mathcal{S}$ contains exactly $2^n$ elements.
But rather than listing all of them, we can describe the group efficiently using its generating set (the generators) given by
\begin{equation}
    \mathcal{S} = \left\langle g_1, g_2, \ldots, g_n \right\rangle.
\end{equation}
When we want to be explicit about which state the stabilizer group stabilizes, we will use a subscript to denote the state for which the stabilizers stabilize, e.g., $\mathcal{S}_{\ket{\psi}}$.

Note that, in many contexts where strict mathematical precision is not required, people may refer to the generating set simply as the stabilizer set, or simply the stabilizers.
Stabilizer formalism plays a crucial part in various areas of quantum information science, including quantum error correction~\cite{gottesman-phd-thesis}, efficient simulation of Clifford circuits~\cite{gottesman-heisenberg-representation,aaronson-chp-sim}, and quantum communication and networking---the primary focus of this thesis.

In a nutshell, the stabilizer formalism is a compact way to describe stabilizer states by writing down their stabilizer generators.

Stabilizer states and circuits are utilized (sometimes in conjunction with oracle) in elementary topics of quantum information science such as quantum teleportation~\cite{bennett-teleportation}, the Bernstein-Vazirani~\cite{bernstein-vazirani-algorithm}, Simon~\cite{simon-algorithm} and Deutsch-Jozsa~\cite{deutsch-jozsa-algorithm} algorithms or quantum key distribution~\cite{bb84-conf,ekertQuantumCryptographyBased1991}.
For example, the Bell state and 5-qubit GHZ state,
\begin{align}
    \ket{\Phi^+} &= \frac{1}{\sqrt{2}} (\ket{00} + \ket{11}), \\
    \ket{\text{GHZ}}_5 &= \frac{1}{\sqrt{2}} (\ket{00000} + \ket{11111}),
\end{align}
can be described with the stabilizer groups defined by
\begin{gather}
        \mathcal{S}_{\ket{\Phi^+}} = \left\langle XX, ZZ \right\rangle, \\
        \mathcal{S}_{\ket{\text{GHZ}}_5} =
            \left\langle\begin{array}{c}
                X X X X X, \\
                Z Z I I I, \\
                I Z Z I I, \\
                I I Z Z I, \\
                I I I Z Z
            \end{array}\right\rangle
\end{gather}
respectively.
(Note that I omitted the tensor product symbols.)
The two stabilizer groups can also be written with a more compact way, dropping the $I$ terms and denoting the index of the qubit that the operator acts on, as
\begin{gather}
    \mathcal{S}_{\ket{\Phi^+}} = \left\langle X_0 X_1, Z_0 Z_1 \right\rangle, \\
    \mathcal{S}_{\ket{\text{GHZ}}_5} = \left\langle X_0 X_1 X_2 X_3 X_4, Z_0 Z_1, Z_1 Z_2, Z_2 Z_3, Z_3 Z_4 \right\rangle,
\end{gather}
respectively.
This short notation will be mainly used throughout this thesis.

\subsection{Representing subspaces with the stabilizer formalism}
\label{subsec:preliminaries:stabilizer-subspace}

Beyond describing unique quantum states, the true power of the stabilizer formalism lies in its ability to describe \emph{subspaces} of the total Hilbert space. This is the foundation of its use in quantum error correction~\cite{gottesman-phd-thesis}.
If an $n$-qubit system is described by a stabilizer group with only $n-k$ independent generators, this group does not stabilize a single vector but rather a $2^k$-dimensional subspace.
This subspace is large enough to encode $k$ logical qubits, providing a framework for storing and protecting quantum information.

The set of all operators that commute with the stabilizer group $\mathcal{S}$ is called its \emph{normalizer}.
Operators that are in this normalizer but not in $\mathcal{S}$ itself are known as \emph{logical operators}.
For each logical qubit, there is a pair of such logical Pauli operators, typically denoted ($X_L, Z_L$), that act on the encoded information without disturbing the stabilizing conditions.

A simple but illustrative example is the three-qubit bit-flip code, which uses $n=3$ physical qubits to encode $k=1$ logical qubit.
The logical basis states are defined as $\ket{0}_L = \ket{000}$ and $\ket{1}_L = \ket{111}$.
The stabilizer group for this code's subspace must stabilize both of these states.
Two independent generators for this group are $g_1 = Z_1 Z_2$ and $g_2 = Z_2 Z_3$.
The stabilizer group is $\mathcal{S} = \langle Z_1 Z_2, Z_2 Z_3 \rangle$.
The logical operators are operators that transform the logical states; for instance, the logical X operator, $X_L = X_1 X_2 X_3$, flips $\ket{0}_L \leftrightarrow \ket{1}_L$.
Notice that $X_L$ commutes with both $g_1$ and $g_2$.
This relationship is often summarized using the \emph{$[[n, k, d]]$ notation}, where $n$ is the number of physical qubits, $k$ is the number of logical qubits, and $d$ is the code distance (the weight of the smallest logical operator).
The three-qubit bit-flip code is therefore an $[[3, 1, 1]]$ code.
This is because while the code can detect any single-qubit X (bit-flip) error, a single-qubit Z (phase-flip) error, such as $Z_1$, acts as a logical Z operator and is therefore undetectable by the code's stabilizers.

\subsection{Measurements on stabilizer states}
\label{subsec:preliminaries:stabilizer-measurement}

A powerful feature of the stabilizer formalism is that the effect of single-qubit Pauli measurements can be tracked efficiently.
The state resulting from such a measurement is also a stabilizer state (or a state within a stabilizer-defined subspace), and we can determine its new set of generators by applying a simple update rule.
There are two main cases to consider, depending on whether the measurement operator $P$ commutes or anti-commutes with every generator in the stabilizer group $\mathcal{S}$.

\textbf{Case 1: $P$ commutes with all generators of $\mathcal{S}$.}
When the measurement operator commutes with the entire stabilizer group, it belongs to the normalizer of $\mathcal{S}$. This leads to two distinct sub-cases.

\emph{Sub-case 1a: Measuring a stabilizer.}
If $P$ is an element of the stabilizer group itself (up to a phase), the measurement outcome is deterministic.
Applying $P$ to the state \psistate{} will always yield the eigenvalue associated with that stabilizer (typically +1 by convention).
The state after measurement is unchanged, as are its stabilizer generators.
For example, measuring the operator $X_1 X_2$ on the Bell state $\ket{\Phi^+}$, which is stabilized by $\langle X_1 X_2, Z_1 Z_2 \rangle$, will always yield the outcome +1, leaving the state as $\ket{\Phi^+}$.

\emph{Sub-case 1b: Measuring a logical operator.}
If the state resides in a code space (defined by $n-k$ generators for $k>0$) and $P$ is a logical operator (i.e., it commutes with $\mathcal{S}$ but is not an element of $\mathcal{S}$), the measurement corresponds to reading out a logical qubit.
In this scenario, the measurement outcome is dependent on the actual state.
The measurement projects the logical qubit into one of its basis states and reduces the dimension of the code space by half.
The new stabilizer group is formed by adding the measured operator, $\pm P$, to the generating set, increasing the number of generators by one.
For example, consider the [[3, 1, 1]] code space stabilized by $\langle Z_1 Z_2, Z_2 Z_3 \rangle$. Measuring the logical Z operator, $P = Z_1$, on the logical state $\ket{+}_L = (\ket{000} + \ket{111})/\sqrt{2}$ will yield +1 or -1 with equal probability.
If the outcome is +1, the state collapses to $\ket{000}$, and the new stabilizer set becomes $\langle Z_1 Z_2, Z_2 Z_3, +Z_1 \rangle$, which is equivalent to $\langle Z_1, Z_2, Z_3 \rangle$.

\textbf{Case 2: $P$ anti-commutes with some generators of $\mathcal{S}$.}
In this case, the measurement outcome is always random, with a 50\% probability of being +1 and 50\% for -1.
The measurement projects the state onto an eigenstate of $P$, so the post-measurement state is different, and its stabilizer group must be updated.
The update rule is as follows:
\begin{enumerate}
    \item Find any generator $g_i$ in the generating set that anti-commutes with the measurement operator $P$.
    \item Replace $g_i$ in the generating set with $\pm P$, where the sign is determined by the classical measurement outcome (+1 or -1).
    \item For every other generator $g_j$ ($j \neq i$) that also anti-commutes with $P$, replace it with the product $g_j \cdot g_i$. This new generator will now commute with $P$.
    \item All generators that already commuted with $P$ are left unchanged.
\end{enumerate}

As an example, let us measure the first qubit of a Bell state $\ket{\Phi^+}$ in the Z-basis, so $P=Z_1$.
The initial generators are $g_1 = X_1 X_2$ and $g_2 = Z_1 Z_2$.
The measurement operator $Z_1$ anti-commutes with $g_1$ but commutes with $g_2$.
Following the rule, we replace $g_1$ with $\pm Z_1$ (depending on the outcome).
If the outcome is +1, the new generating set is $\langle +Z_1, Z_1 Z_2 \rangle$.
This set can be simplified by multiplying the second generator by the first, giving the equivalent set $\langle Z_1, Z_2 \rangle$, which stabilizes the state $\ket{00}$.
This ability to efficiently calculate the post-measurement state is what makes the classical simulation of Clifford circuits possible.

\section{Graph states}
\label{sec:preliminaries:graph-states}

Another useful representation of quantum states is the graph state formalism~\cite{hein-graph-state-arxiv,hein-graph-state-pra}.
They are quantum states represented as mathematical graphs $G(V, E)$, where $V$ is the set of vertices and $E$ is the set of edges.
In the graph state formalism, a vertex represents a qubit initialized in $\ket{+}$ while an edge represents an entanglement relation between two qubits given by the controlled-phase gate (CZ).
\nomenclature{Graph states}{Quantum states represented as mathematical graphs where vertices symbolize qubits, and edges represent entanglement relations between qubits. Graph states are a subset of stabilizer states}
Formally, a graph state $\ket{G(V, E)}$ represented by a graph $G$ is given by
\begin{equation}
    \ket{G(V, E)} = \prod_{\{u, v\} \in E} CZ(u, v) \ket{+}^{\otimes |V|}.
\label{eq:vanila-graph-state-definition}
\end{equation}

Interestingly, graph states are a subset of stabilizer states.
They can be described via the stabilizer formalism with stabilizer generators given by
\begin{equation}
    g_j = X_j \bigotimes_{v \in N_j} Z_v
\label{eq:graph-state-stabilizer-generators}
\end{equation}
associated with each vertex $j \in V$ in the graph $G$, where $N_j$ denotes the set of neighbors of vertex $j$.
Furthermore, any stabilizer state is equivalent to (possibly non-unique) graph states up to local Clifford gates---that is, applications of single-qubit Clifford gates take the graph state to the stabilizer state.
In other words, we can extend the graph state definition to include a single layer of local Clifford operators applied on qubits of the graph state $\ket{G; \underline{C}}$ with
\begin{equation}
    \ket{\psi} = \ket{G; \underline{C}} = \ket{G; C_1, C_2, \ldots, C_n} = \bigotimes_{i=1}^n C_i \ket{G},
\label{eq:graph-state-with-side-effects}
\end{equation}
where $C_i \in \mathcal{C}_1$ is a single qubit Clifford gate acting on qubit $i$, with $\mathcal{C}_1$ being the local Clifford group (single qubit Clifford gates).
We call the Clifford operator acting on each vertex of the graph the ``Clifford side effect'' or just ``side effect'' when the context is clear.
These side effects also known as vertex operators in some literature~\cite{anders-briegel-graph-state-simulator}.
An example of a stabilizer state with two graph representations is shown in \cref{fig:graph-state-example}.

\begin{figure}[htb]
    \centering
    \includegraphics[width=\textwidth]{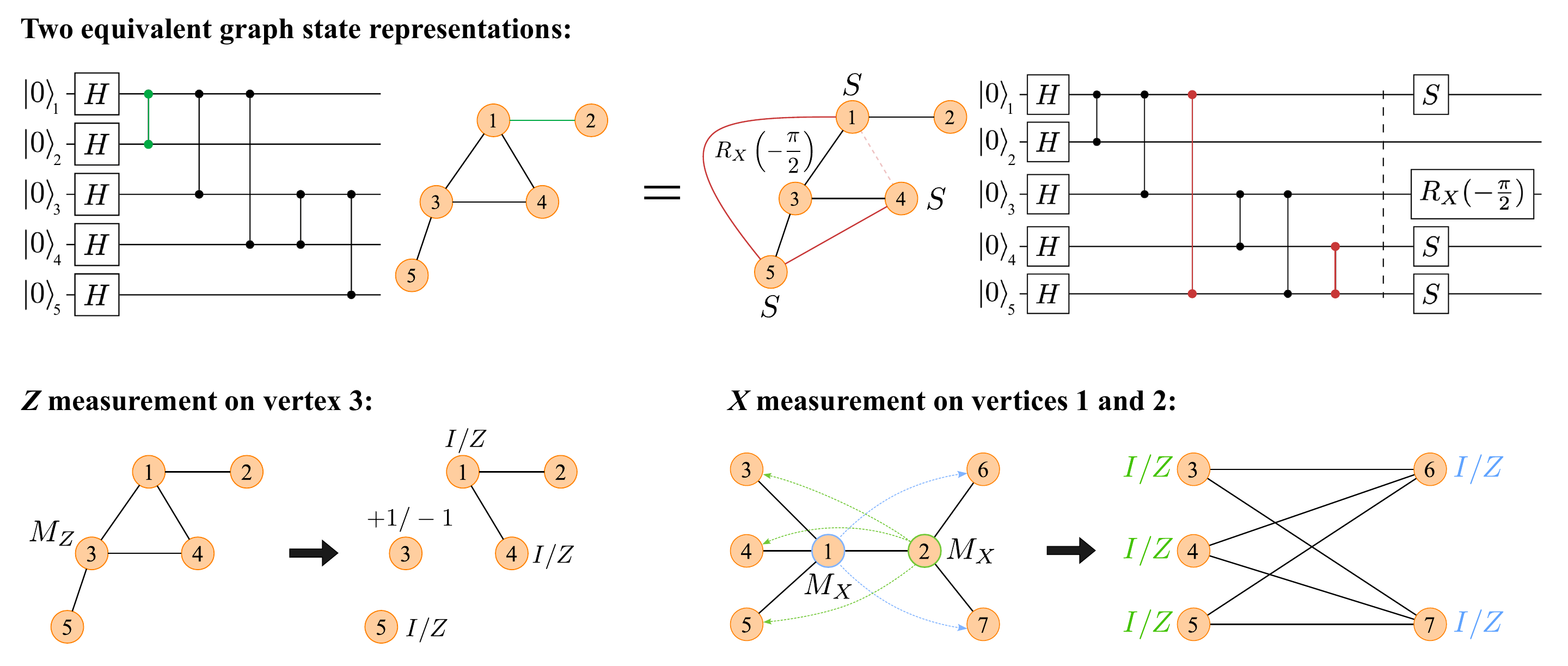}
    \caption[Examples of graph states]{Two graph state representations of the same quantum state are depicted at the top. Here, vertices represent qubits in the quantum circuit, and the application of a controlled-phase gate between two qubits corresponds to an edge between the corresponding vertices. An example of such an edge is highlighted in green in the top left. The graph state on the top right is obtained by altering the description through local complementation on vertex 3. This process deletes any existing edges between neighboring vertices of vertex 3, such as the edge $(1,4)$ in this case, and introduces new edges that were previously absent, highlighted in red as $(1,5)$ and $(4,5)$. The application of the depicted Clifford operations ensures that despite the differing graph representations, the quantum states remain identical in both cases. Visualization of a Z measurement with side effects is presented in the bottom left, while the impact of two X measurements (XX measurement) on a different graph is illustrated in the lower right. The Clifford side effects ($I/Z$) of qubits 3, 4, and 5 depend on the measurement outcome of qubit 2, while those of qubits 6 and 7 hinge on the outcome of qubit 1, as indicated by the blue and green arrows. (From~\cite{naphan-engineering-challenges-in-rgs}.)}
    \label{fig:graph-state-example}
\end{figure}

In many works that use graph states as a tool for analysis, the focus is on the graph's structure itself, as this is often sufficient for proving key properties of the state.
In this context, it is common to work under the principle of local Clifford (LC) equivalence, where two graph states are considered equivalent if they can be transformed into one another by applying only single-qubit Clifford gates.
This effectively means we can disregard the Clifford side effects and work with the simpler state definition in \cref{eq:vanila-graph-state-definition}, treating the side effects as details that can be corrected later.
By extension, the related concept of local unitary (LU) equivalence is also used in certain scenarios.

In this thesis, when I talk about graph states, I will be mainly using the extended definition in \cref{eq:graph-state-with-side-effects} since working on tracking an exact state will be easier to reason about.
A few times when we only want to focus on the structure of the graph---when concerning only whether or not the state are still entangled or certain paths between pair of vertices in the graph---I will drop the consideration of side effects and focus on the equivalent properties of the graph state.

\subsection{Local complementations}

As we have seen briefly earlier, a stabilizer state may have multiple graph state representations.
One useful operation that can be performed on a graph is the \emph{local complementation}~\cite{bouchet-local-complementation}.
By definition, local complementation operation $\tau_a(G)$ is an operation on a graph $G$ specified by a vertex $a$ which transform the graph by complementing a subgraph induced by the neighborhood $N_a$ of $a$---that is, the edge between each pair of neighbors of $a$ is toggled.
In simple words, for each pair of neighbors ($u$, $v$) of $a$, if there exists an edge $\{u, v\} \in E$, we remove it, while if $u$ and $v$ are not previously connect by an edge, we add edge $\{u, v\}$ to the edge set.
An example of local complementation is shown in \cref{fig:graph-state-example}.

Performing local complementation on graph states also changes the Clifford side effects on qubits not just the graph connectivity, as given by (recall gate definitions in \cref{subsec:preliminaries:clifford})
\begin{equation}
    \ket{\tau_a(G)} = \sqrt{-iX} \bigotimes_{v \in N_a} \sqrt{iZ} \ket{G}.
\label{eq:local-complementation}
\end{equation}
And for its action on the graph states with side effects, we get
\begin{equation}
    \ket{\tau_a(G); \underline{C}} = \left( \bigotimes_{j \in V} C_j \right) \left( \sqrt{-iX} \bigotimes_{v \in N_a} \sqrt{iZ} \right) \ket{G}.
\end{equation}
Notice the order of operations in the above equation, that the Clifford side effects are applied \emph{after} the operations introduced via local complementation.
Although we use the term ``apply'' or ``application'' with local complementation, in this context (and in most context of quantum information science), it is a means to change the classical state description and no actual physical operations are performed on the state.
In this view, local complementation is a tool that we utilize to change our mental description of the state (or the tracked state), and only when we want to use the state after, i.e., then we need to account for these side effects.

Local complementation is a very powerful tool in changing the graph representation of the state and is crucial in optimizing graph state preparations, which has wide range of applications in quantum information science, including quantum computations~\cite{raussendorf-mbqc-pra,raussendorf-mbqc-prl}, compilations of quantum programs~\cite{sitong-substrate-scheduler,krishnan-vijayan-jabalizer}, and quantum networking applications~\cite{azuma-rgs,borregaard-3g-repeater,sen-goodenough-towsley-subgraph-complementation-in-star-network,ghanbari-hoi-kwong-rgs-optimization}.

\subsection{Measurements on graph states}

Performing Pauli basis measurements on graph states preserve the graph states representation.
This seems like a trivial assertion, since performing Pauli basis measurements on a stabilizer state gives us another stabilizer state and we have already established that any stabilizer state can be represented with graph states.
But conveniently, the action of Pauli measurements on graph states---when considering at the level of local Clifford equivalent---can be explained simply with operations on graphs (vertex deletion and local complementation).

First, let us look at the action of Z basis measurement on graph states without any side effects.
The resulting graph after the measurement is the same graph but with the measured vertex removed.
(Or if we consider non-destructive projective measurement, we delete all the edges incident to the measured vertex.)
To see this, we can rewrite the graph state as 
\begin{align}
    \ket{G} 
        &= \prod_{\{u, v\} \in E} CZ(u, v) \ket{+}^{\otimes |V|} \\
        &= \left(\prod_{v \in N_a} CZ(a, v)\right) \left(\prod_{\substack{\{u, v\} \in E \\ u, v \neq a}} CZ(u, v)\right) \ket{+}^{\otimes |V|} \\
        &= \left(\prod_{v \in N_a} CZ(a, v)\right) \ket{G - a} \\
        &= \ket{0}_a \ket{G - a} + \ket{1}_a \prod_{v \in N_a} Z_v \ket{G - a}.
\end{align}
We can now see that the effect of measuring $a$ in the Z basis has the same effect as deleting vertex $a$ from the graph while the side effects, depending on the measurement outcome, are possible $Z$ on the neighbor vertices.

We can also derive the resulting state after Z measurement $M_{a, Z}$ on qubit $a$ with stabilizer formalism.
To see this, recall the generator set defined from each vertex (\cref{eq:graph-state-stabilizer-generators}).
Notice that every generator $g_i$, where $i \neq a$, commutes with the measurement operator $P^a_{Z, \pm}$.
Only $g_a$ anticommutes with the measurement, thus, we can replace $g_a$ in the generator set with $\pm Z_a$ where the sign depends on the measurement outcome.
For each $g_v$ with $v \in N_a$ in the original graph, we can then update them to remove the $Z_a$ term by multiplying it with the newly added $\pm Z_a$ term.
Notice that the updated generators $g \in \Set{\pm Z_a g_v | v \in N_a}$ will have a minus sign in front of them if we had obtain -1 as the measurement result.
This is the same as if we define the generator set of the graph $G - a$ with \cref{eq:graph-state-stabilizer-generators} but with $Z_v$ acting on $g_v$ for $v \in N_a$ of the original graph.

Now that we have covered Z measurement, let us start with the intuition for Y and X measurements.
Recall that we can use local complementation to change the representation of the graph, and more importantly, change the side effects of vertices.
The idea is, say we want to perform X measurement on qubit $a$, to transform the graph state with series of local complementations such that the side effect $C_a$ of qubit $a$ satisfies
\begin{equation}
    P^a_{X, \pm} C_a = C'_a P^a_{Z, \pm},
\end{equation}
where $C'_a$ is a Clifford operator to fix the resulting projection from Z axis back to X axis.
A straightforward example of $C_a$ would be $H$.
The Clifford fix $C'_a$ can usually be omitted when we are only interested in the outcome or when we use destructive measurements.

Using the above intuition, we can define the actions of projective measurements in the Z, Y, and X basis as 
\begin{align}
    P^{a}_{Z, \pm} \ket{G} &= \ket{Z, \pm} \otimes U^a_{Z, \pm} \ket{G - a}, \\
    P^{a}_{Y, \pm} \ket{G} &= \ket{Y, \pm} \otimes U^a_{Y, \pm} \ket{\tau_a(G) - a}, \\
    P^{a}_{X, \pm} \ket{G} &= \ket{X, \pm} \otimes U^a_{X, \pm} \ket{\tau_b(\tau_a \circ \tau_b(G) - a)},
\end{align}
where $P^a_{B, \pm}$ denotes measurement on qubit $a$ with basis $B$ and with eigenvalue result $\pm1$, $b$ denotes a vertex chosen from neighborhood of $a$ assuming $a$ is not an isolated vertex and $U^a_{B, \pm}$ denotes added side effects introduced by the measurements.
These side effects are given by~\cite{hein-graph-state-pra,hein-graph-state-arxiv}
\begin{align}
    U^a_{Z, +} &= I, &U^a_{Z, -} &= \bigotimes_{v \in N_a} Z_v,\\
    U^a_{Y, +} &= \bigotimes_{v \in N_a} \sqrt{-iZ}_v, &U^a_{Y, -} &= \bigotimes_{v \in N_a} \sqrt{+iZ}_v,\\
    U^a_{X, +} &= \sqrt{+iY}_b \bigotimes_{v \in N_a \setminus (N_b \cup b)} Z_v , &U^a_{X, -} &= \sqrt{-iY}_b \bigotimes_{v \in N_b \setminus (N_a \cup a)} Z_v.
\end{align}

The choices of side effects to change the basis of intended measurement into Z are chosen by noticing that
\begin{align}
    P^a_{Y, \pm} \sqrt{iX} &= \sqrt{-iX} P^a_{Z, \pm}, \\
    P^a_{X, \pm} \sqrt{iX} \sqrt{iZ} &= \sqrt{-iZ} \sqrt{-iX} P^a_{Z, \pm}.
\end{align}
Note that although the action of local complementation in \cref{eq:local-complementation} suggests that we need to perform local complementation three times to realize $\sqrt{iX}$, we need to perform it only once, since the resulting states differ only by a global phase.
Also, notice that to perform an X measurement, we actually only need two local complementation sequences, but the way we defined it above requires us to perform it three times.
The reason for this is most likely due to~\cite{hein-graph-state-pra,hein-graph-state-arxiv} wanting to constrains the side effects to only rotations in the Y and Z axes.

\clearpage
\chapter{Quantum networks}
\label{chapter:quantum-networks}

Quantum networking holds the promise to revolutionize the way information is processed and communicated across distributed systems.
By harnessing quantum mechanical phenomena such as superposition and entanglement, these networks are envisioned to enable new modes of communication, sensing, and computation that are fundamentally beyond the reach of classical systems.
The possibility of interconnecting quantum devices in a similar manner to today's classical Internet has sparked intense interest, with scientists envisioning the transformative capabilities that a future Quantum Internet could unlock.
In this chapter, I will give an overview of the motivations behind quantum networking and briefly summarize key foundational concepts in quantum networks.
\nomenclature{Quantum network}{A collection of interconnected quantum devices that can establish entangled states between distant parties, enabling tasks that are impossible to achieve with classical communication alone}
\nomenclature{Quantum repeater}{A device essential for extending the reach of entanglement distribution beyond direct transmission limits, serving as an intermediate node in the network topology and implementing error management strategies}
\nomenclature{Entanglement purification}{An error management technique that improves the fidelity of noisy entangled pairs by consuming multiple lower-fidelity copies to produce fewer, higher-fidelity ones}

To achieve this, the chapter is structured as follows.
We first define what a quantum network is and the key tasks it must perform.
We then delve into the essential hardware and protocol building blocks, including different link architectures and the role of quantum repeaters.
Finally, we will survey the critical challenge of error management and conclude with an overview of the major open problems that define the frontier of this exciting field.

\section{Introduction}

Let us begin with the pressing question, ``What is a quantum network?''
For the purposes of this thesis, a \emph{quantum network}~\cite{rdv-quantum-networking-book,sasaki-quantum-networks-where-should-we-be-heading,simon-towards-global-quantum-network,munroDesigningTomorrowsQuantum2022} refers to a way of connecting quantum devices such that they can accomplish tasks that would be impossible if they were limited to communicating only through classical information.

At its core, a quantum network---or more broadly, the field of \emph{quantum communication}---enables today's communication networks to be supplemented with quantum mechanical phenomena.
This includes not only the transmission and reception of quantum information, typically encoded in qubits, but also the ability to establish entangled states between multiple parties across the network.

The most well-known application of quantum communication is quantum key distribution (QKD)~\cite{bb84-conf,bb84-reprint,ekertQuantumCryptographyBased1991,xu-pan-qkd-review}, which enables two parties in a network to create a shared string of secret classical bits that can later be used as a key for encryption.
This task of establishing a shared secret key cannot be achieved through classical communication alone.
Today, there are deployed networks that support QKD~\cite{sasakiFieldTestQuantum2011,peevSECOQCQuantumKey2009,chenIntegratedSpacetogroundQuantum2021,pittaluga600kmRepeaterlikeQuantum2021,pittalugaLongdistanceCoherentQuantum2025a} (often referred to as QKD networks), along with commercially available quantum sources and detectors that enable qubit transmission over optical fibers or free-space links.
It is important to note, however, that many QKD implementations operate without requiring entanglement between the communicating parties as certain QKD protocols can be realized by transmitting unentangled quantum states over the channel.

Beyond QKD, however, a rich landscape of applications emerges that fundamentally relies on the distribution of \emph{entanglement}.
Each application imposes unique demands on the underlying network architecture.
For instance, quantum sensing and metrology promise enhancements such as improved telescope resolution~\cite{gottesman-longer-baseline-telescopes,khabiboullineOpticalInterferometryQuantum2019}, ultra-precise clock synchronization~\cite{jozsaQuantumClockSynchronization2000a,nicholElementaryQuantumNetwork2022,komarQuantumNetworkClocks2014}, vibration detection~\cite{chenQuantumKeyDistribution2022}, and physical location verification~\cite{kanneworffExperimentalDemonstrationQuantum2025}.
These capabilities require distributed entangled states with high fidelity, low timing jitter, and robustness to loss~\cite{degen-sensing,proctorMultiparameterEstimationNetworked2018,proctor-quantum-sensing,giovannetti-advances-in-metrology,ge-distributed-metrology}.

In cryptography and secure communication, entanglement unlocks new functionalities such as quantum secret sharing~\cite{hilleryQuantumSecretSharing1999,cleveHowShareQuantum1999}, certified data deletion~\cite{broadbentQuantumEncryptionCertified2020}, quantum commitments~\cite{gunnCommitmentsQuantumStates2023,morimaeQuantumCommitmentsSignatures2022}, and multi-party primitives like conference key agreement, leader election, and secure voting~\cite{taniExactQuantumAlgorithms2005,taniExactQuantumAlgorithms2012,murtaQuantumConferenceKey2020,augusiakMultipartiteSecretKey2009,chenMultipartiteQuantumCryptographic2007,ben-orFastQuantumByzantine2005}.
These protocols not only require entanglement but often rely on multipartite or high-dimensional entangled states, placing further demands on network capabilities and trust models.

Finally, distributed quantum computing represents a natural culmination of these trends.
By connecting quantum processors over a network, one can realize composite systems that surpass the capabilities of any single device~\cite{cleve-buhrman-entanglement-for-communication,cirac-distributed-qc,barral-review-distributed-qc,gottesman-chuang-universality-teleportation}.
Entanglement also enables blind quantum computation~\cite{broadbent-bfk-protocol,fitzsimonsPrivateQuantumComputation2017,morimaeBlindQuantumComputation2013}, homomorphic encryption~\cite{mahadevClassicalHomomorphicEncryption2023,dulekQuantumHomomorphicEncryption2018}, and quantum interactive proofs~\cite{aharonovInteractiveProofsQuantum2017a}, where a client can securely delegate quantum computations while preserving the privacy of their data and algorithm.
These applications demand not only robust entanglement distribution but also reliable error correction, secure authentication, and sophisticated coordination across the network stack.

The diversity of these applications makes it clear that a quantum network must be more than just a channel for sending qubits.
It must be a sophisticated system for creating, distributing, and managing entanglement.
The following section breaks down the fundamental tasks required to build such a system.

\section{Tasks of quantum networks}

To fulfill the promise of the applications described above, a quantum network must perform a series of fundamental tasks.
These tasks form a logical hierarchy, starting from the definition of the network itself, to generating entanglement over a single link, extending that entanglement across the network, and managing the inevitable errors that arise.
This section will outline each of these core operational requirements.

Unlike in classical networks, where communication involves sending and receiving messages, directly transmitting quantum information from one node to another is often not the best approach.
Quantum states are inherently fragile, and due to the no-cloning theorem~\cite{dieks-no-cloning,wootters-zurek-no-cloning,peresHowNocloningTheorem2003}, if a quantum message is lost or corrupted during transmission, it cannot be recovered.
This makes quantum communication particularly susceptible to noise and loss, especially for states that are difficult or costly to prepare.

For this reason, the most fundamental task of a quantum network is to establish generic shared entangled states between distant parties.
These shared entangled states serve as fundamental resources for various quantum communication tasks---most notably, quantum teleportation~\cite{bennett-teleportation}, which enables the transmission of quantum data (quantum channel) without the need to physically send the qubits themselves.
The smallest such entangled resource that enables quantum communication is the \emph{Bell pair}, which can also serve as a building block for creating more complex multipartite entangled states.

\subsection{Definition of a network}

When I refer to a quantum network in this thesis, I mean a collection of interconnected quantum devices---referred to as network nodes---that are capable of generating, processing, or storing quantum states.
At an abstract level, a quantum network can be modeled as a graph where vertices represent network nodes and edges denote physical quantum links between them.
Depending on the level of abstraction, some nodes (e.g., those without quantum memories or only used for routing photons through) may be omitted (as depicted in \cref{fig:quantum-network-example}).
In such cases, the topology of the graph reflects only the essential components for the protocol under discussion.

\begin{figure}
    \centering
    \includegraphics[width=\textwidth]{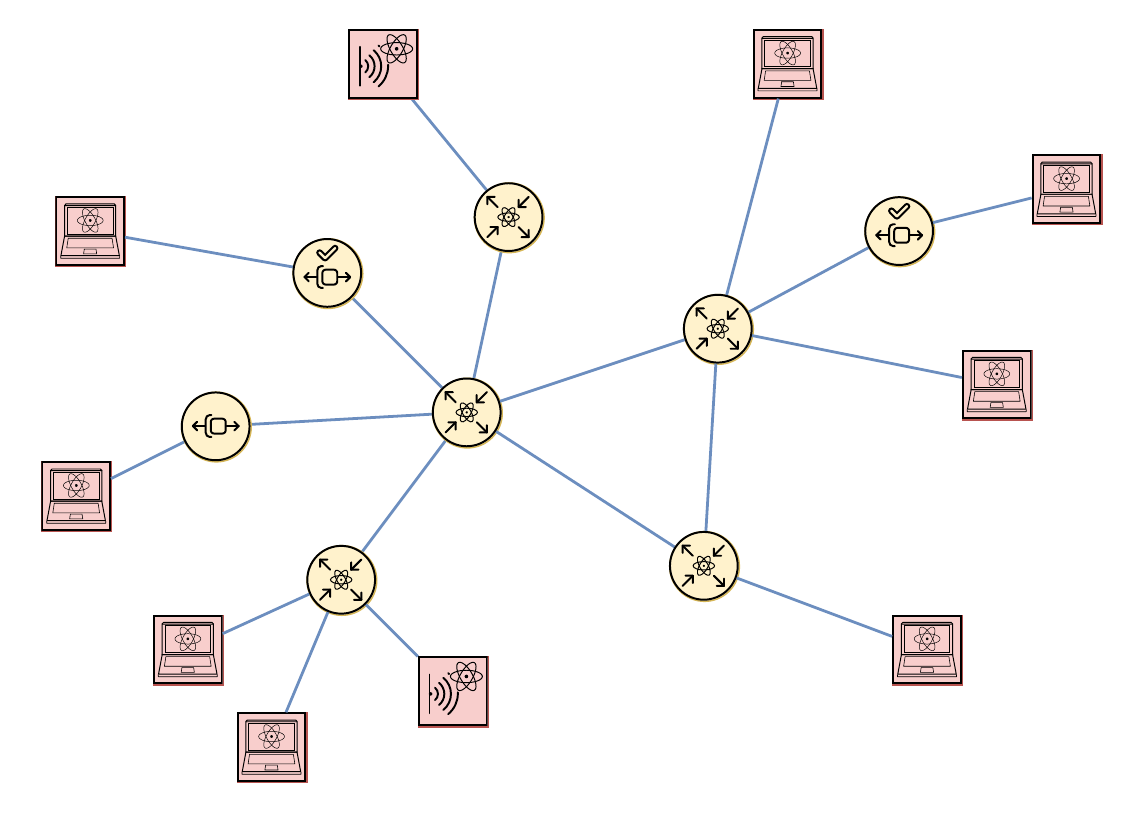}
    \caption[Illustrative topology of a quantum network]{
    An illustrative topology of a quantum network.
    The network consists of end nodes, depicted as red squares (including both computational and sensor types), and intermediate nodes, shown as yellow circles (representing quantum repeaters and routers).
    A complete legend and detailed descriptions of each node's function are discussed in \cref{chapter:architecture-memory}.}
    \label{fig:quantum-network-example}
\end{figure}
Throughout this thesis, the network graphs I present generally correspond to a level of abstraction in which only memory-equipped nodes are included with exceptions when we need emphasis on link types and in \cref{chapter:architecture-all-photonic} where I discuss all-photonic quantum repeaters.
An edge between two nodes implies that they are capable of generating Bell pairs---i.e., there exists a physical quantum channel between them---and such an edge corresponds to what I will refer to as an elementary link~\cite{cody-msm,simon-sneakernet,khartri-policies-elementary-links} that creates \emph{link-level entanglement}.

Having defined the network in terms of nodes and links, we now turn to the first operational task: generating entanglement across one of these elementary links.

\subsection{Link-level entanglement generation}

\nomenclature{Quantum memory}{A physical system capable of storing quantum information for a period of time that preserves the coherence of quantum state and any entanglement the state shared with other system}

As is common in introductory treatments of quantum communication~\cite{mike-ike-book}, foundational protocols such as quantum teleportation typically assume that two parties, Alice and Bob, begin with a pre-shared Bell pair.
In a practical network, this entanglement is not pre-existing but must be actively generated.
This can be achieved in several ways, for instance, by having one party create a Bell pair between a stationary quantum memory and a flying qubit, which is then sent to the other party through a direct physical channel.
The entanglement established between adjacent nodes through such a process is referred to as \emph{link-level entanglement}.

\nomenclature{Link-level entanglement}{An entangled quantum state shared between neighbor nodes in the network through physical quantum link without involving intermediate quantum repeater nodes}
Formally, link-level entanglement is an entangled state created between nodes that are neighbors in the network graph, meaning they are connected by a direct physical quantum channel such as an optical fiber~\cite{cody-msm,simon-sneakernet}.
One of the most common methods for this is heralded entanglement generation, which uses a central measurement station to create entanglement between two distant memories without the direct transmission of a memory-entangled qubit between them, as illustrated in \cref{fig:link-example}.
\begin{figure}[h]
    \centering
    \includegraphics[width=\textwidth]{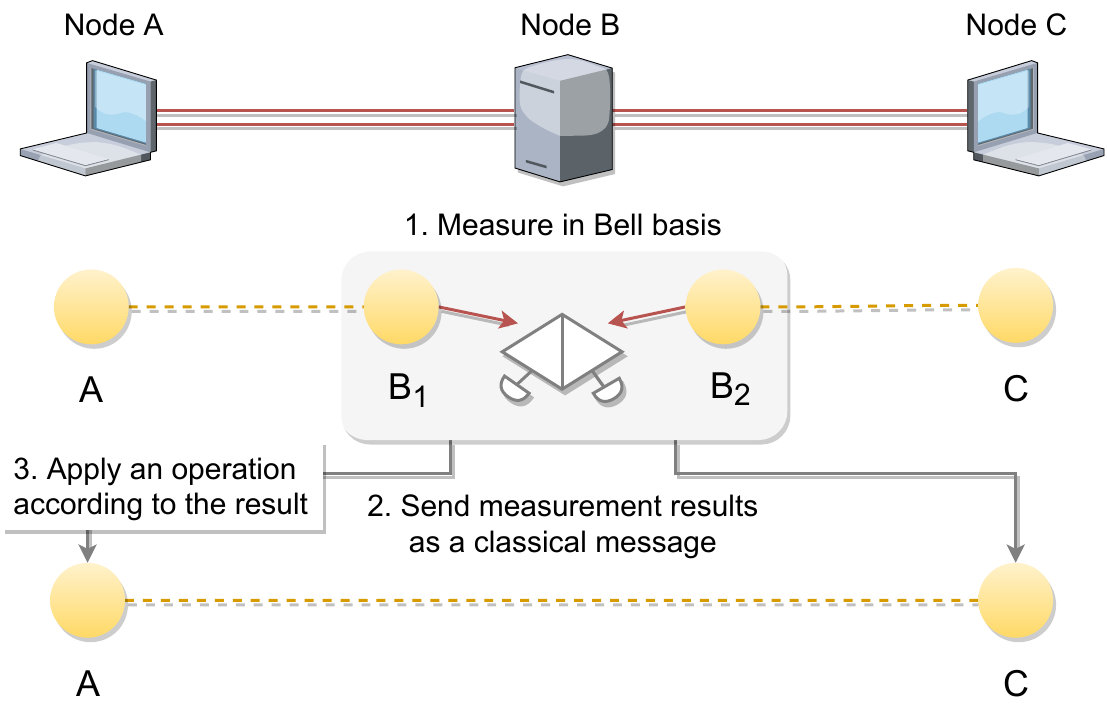}
    \caption[Heralded link-level entanglement generation]{
    A heralded scheme for generating link-level entanglement, characteristic of a Memory-Interference-Memory (MIM) link architecture.
    Two adjacent nodes, Alice (A) and Charlie (C), each entangle a local quantum memory with a photon.
    These photons are sent to a Bell-state analyzer (BSA) node (B).
    A successful measurement at the BSA projects the two distant memories into an entangled state, which is announced, or heralded, by a classical signal.
    (From~\cite{cocori-quisp}.)
    }
    \label{fig:link-example}
\end{figure}

While generating entanglement over a single link is a foundational capability, its reach is severely limited by physical constraints.

\subsection{Entanglement extension}

Although photons are excellent carriers of quantum data for transmission due to their weak interaction with the environment, which helps maintain quantum coherence, there is a fundamental limit to how far entanglement can be distributed over optical fiber because of photon loss.
Unlike in free-space communication, where photons can travel astronomical distances, the transmission probability in fiber-optic channels decays exponentially with distance.
For instance, a typical low-loss optical fiber has an attenuation of about 0.2 dB per kilometer.
This translates to a photon arrival probability of roughly 39\% at 20 km, 10\% at 50 km, and only 1\% at 100 km.
These figures, which consider only the intrinsic loss of the fiber, already severely limit the distance over which Bell pairs can be reliably distributed.
In practice, the total loss is even greater due to system inefficiencies, such as imperfect coupling between fibers and other components.

These limitations are formalized by bounds on point-to-point quantum communication---most notably the PLOB~\cite{plob-bound} and TGW~\cite{tgw-bound} bounds---which are often referred to as the \emph{repeaterless bound}.
They characterize the maximum rate of extractable secret key or entanglement distribution achievable without intermediate quantum nodes, and make clear that direct transmission alone cannot scale to enable a global quantum Internet.

To overcome this exponential attenuation and extend the reach of entanglement distribution, \emph{quantum repeaters} are essential~\cite{dur-repeater,guha-efficient-repeater-above-bound,azuma-rmp-repeater-review}.
These devices break up long communication paths into shorter segments, distribute entanglement over each, and then utilize \emph{entanglement swapping}~\cite{zukowski-entanglement-swap} to effectively ``stitch together'' these segments---allowing entanglement to be extended over much longer distances.

Quantum repeaters are cornerstones of quantum networks and have remained a central focus of research in the field.
We will examine them in more detail later in \cref{sec:quantum-networks:repeaters}.

Successfully extending entanglement across multiple hops, however, introduces a new critical challenge: the accumulation of errors.

\subsection{Error management}

Besides the entanglement distribution rate, an even more critical metric is the fidelity of the distributed states.
High-fidelity entangled state is essential for many quantum networking applications that I have previously mentioned, which require strict fidelity thresholds.
If the fidelity of link-level entanglement is not high enough, performing entanglement swapping will only further degrade the final state's fidelity due to error accumulation.
To address this, we can apply either \emph{entanglement purification}~\cite{bennett-purification, deutsch-purification, briegel-repeater, dur-briegel-purification-review}---which consumes multiple lower fidelity copies to produce fewer, higher-fidelity ones---or perform \emph{quantum error-correction}~\cite{gottesmanSurvivingQuantumComputer,simon-qec-beginners,roffe-qec-intro}, which encode quantum state into a logical state to protect it from future errors, provided the base fidelity is above a certain threshold.

We will discuss both techniques further in \cref{sec:quantum-networks:repeaters,sec:error-management}, as error management is one of the core responsibilities of quantum repeaters.

While the quantum tasks of generation, extension, and error management are paramount, they must be orchestrated by a robust classical control layer responsible for overall network operations.

\subsection{Network operations}

On the classical side, the network also needs to remain operational.
It must collectively handle application requests from end nodes (users) and make progress toward completing them based on certain objectives---such as minimizing the waiting time of end nodes or prioritizing requests by assigning them ranks.
Each node must also monitor its links to ensure they are still operational, while participating in routing, multiplexing, and maintaining operational security.

Most of these tasks are carried out by quantum repeaters.
As I will outline later (in \cref{sec:quantum-networks:repeaters}), the responsibilities of quantum repeaters go beyond simply extending the range of distributed entangled states.
In certain network models~\cite{cicconetti-request-scheduling,dharaSubexponentialRateDistance2021,pant-routing-entanglement,patil-distance-ind-rate}, tasks like request scheduling, routing, and multiplexing are delegated to a central controller.
We will discuss various network models and architectures later in \cref{chapter:architecture-memory}.

In summary, the operation of a quantum network relies on a stack of fundamental tasks, from defining the network topology to managing classical control operations.
Having outlined these essential functions, we now turn our attention to the specific physical and protocol-level implementations that realize them, beginning with a survey of different quantum link architectures.

\section{Quantum link architectures}
\label{sec:quantum-link-architecture}

Establishing link-level entanglement between adjacent quantum nodes is the first crucial step in any repeater-based quantum network.
Various physical link implementations, often termed \emph{quantum link architectures}, have been proposed~\cite{cody-msm,fengEntanglingDistantAtoms2003,duanEfficientEngineeringMultiatom2003,cody-jones-msm-preprint,simon-sneakernet,simonRobustLongDistanceEntanglement2003,huangVacuumBeamGuide2024}.
This section will review a few of these key architectures.
They differ in hardware requirements, specific operational procedures, and detailed performance characteristics, such as entanglement generation rate and achievable fidelity.
Understanding these differences is crucial for guiding the selection of an appropriate architecture based on the components available at hand.

\subsection{Memory-interference-memory link}

\begin{figure}[htbp]
    \centering
    \includegraphics[width=0.85\textwidth]{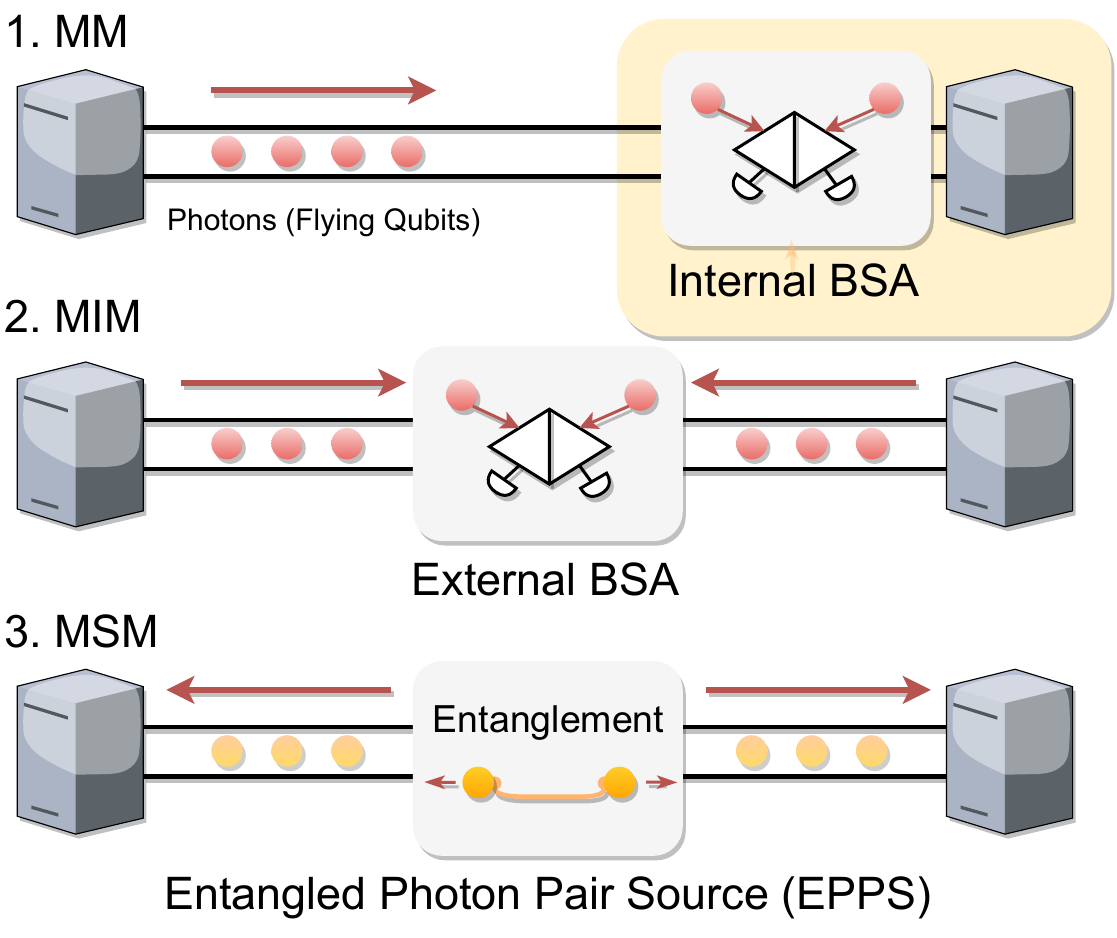}
    \caption[Comparison of primary quantum link architectures]{
    A comparison of three primary quantum link architectures for heralded entanglement generation. 
    \textbf{(1) Memory-Memory (MM):} An asymmetric architecture where the Bell-state analyzer (BSA) is co-located at the receiver node, which can simplify network topology. 
    \textbf{(2) Memory-Interference-Memory (MIM):} A symmetric architecture where two nodes each entangle a local memory with a photon and send the photons to a central BSA for a joint measurement. 
    \textbf{(3) Memory-Source-Memory (MSM):} A central entangled photon pair source (EPPS) distributes entangled photons to each node, where they are subsequently stored in quantum memories.
    (From~\cite{cocori-quisp}.)}
    \label{fig:quantum-link-architectures}
\end{figure}
The memory-interference-memory (MIM) link is also referred to as the midpoint interference link~\cite{simonRobustLongDistanceEntanglement2003,fengEntanglingDistantAtoms2003,duanEfficientEngineeringMultiatom2003}.
In this scheme, two nodes first generate entanglement between their respective quantum memory and a photon.
These photons are then sent through an optical fiber to a central Bell-state analyzer (BSA).
The photons meet at the BSA and undergo a Bell-state measurement (BSM).
Upon a successful BSM, the entanglement is effectively transferred, creating an entangled Bell pair between the two remote quantum memories.
An example of MIM link is shown in \cref{fig:quantum-link-architectures}
A notable drawback is that BSM outcomes are inherently probabilistic.
When implemented with linear optics, the BSM can only distinguish two of the four Bell states, which fundamentally limits the theoretical success probability to 50\%, even with ideal hardware.
Methods to surpass this 50\% limit exist, but they require more complex protocols involving more demanding quantum states, specialized hardware, or the use of non-linear interactions, which will be discussed in \cref{sec:discussion-and-variants}.

\subsection{Memory-memory link}

The memory-memory (MM) link~\cite{munroQuantumMultiplexingHighperformance2010} is sometimes referred to as the sender-receiver model.
It operates similarly to the MIM link, but with a key architectural change.
The BSA is not a standalone station but is instead located inside one of the nodes.
This modification simplifies both network communication and its topology.
This arrangement allows the receiver node to potentially operate with fewer memories because they can be recycled more quickly for subsequent entanglement generation rounds due to near-instantaneous feedback.
This architecture is shown in \cref{fig:quantum-link-architectures}.

\subsection{Memory-source-memory link}

The memory-source-memory (MSM) link, or midpoint source architecture (shown in \cref{fig:quantum-link-architectures}), takes a different approach by starting with a source that directly creates entangled photon pairs~\cite{cody-msm,cody-jones-msm-preprint}.
An entangled photon pair source (EPPS) is placed between the two end nodes.
The photons generated by the EPPS are sent to the two nodes at the ends of the link.
At each node, which must house a BSA, a measurement is performed to entangle the incoming photon with a local quantum memory.
This can be achieved by having memories emit their own photons to interfere with the incoming ones or by using absorptive memories to directly store the photons~\cite{allen-zang-simulation-absorptive-memories}.
While requiring two BSAs reduces the theoretical maximum success rate to 25\% (compared to MIM's 50\%), its key advantage is that memories that fail to entangle can be quickly recycled, significantly reducing their idle time.
This can boost the overall generation rate, especially when the number of available memories is limited.
However, the protocol requires careful design, as poorly tuned parameters for the EPPS generation rate and memory recycling time can severely degrade or even halt the generation process~\cite{soon-optimizing-msm-protocol-qcnc-2024}.

\subsection{Sneakernet link}

A drastically different approach is the sneakernet link, which is analogous to how shipping physical hard drives can be faster for transferring massive datasets than sending them over the internet~\cite{simon-sneakernet}.
This architecture proposes establishing a quantum link by physically transporting quantum memories.
Entanglement is first created between two memories locally, and then one memory is shipped to a distant location.
This is considered a viable option for distributing large volumes of high-fidelity entanglement, particularly given the potentially low entanglement generation rates in early-stage quantum networks.

\subsection{Vacuum beamguide link}

A novel approach for long-distance communication is the vacuum beamguide (VBG) link~\cite{huangVacuumBeamGuide2024}.
Unlike optical fibers, a VBG guides light using an array of precisely aligned lenses housed within an evacuated tube.
This design overcomes the fundamental limit of material absorption in fiber by having photons travel primarily through vacuum.
The VBG promises ultrahigh transparency and an exceptionally low attenuation rate, potentially on the order of $10^{-4}$ dB/km.
Such low loss could enable ground-based, repeater-less quantum communication over continental scales ($>1000$ km).
It offers a high-bandwidth, robust channel that is isolated from environmental perturbations.
While constructing a VBG requires significant investment and engineering to maintain vacuum and lens alignment---similar in scale to the LIGO project~\cite{abbottLIGOLaserInterferometer2009}---it represents a feasible pathway to transform and realize quantum links and quantum networks using current technology.

The choice of link architecture is thus a foundational design decision, balancing hardware complexity against performance. We have surveyed protocols that strategically place sources and measurement stations (MIM, MM, MSM) and alternative concepts that circumvent channel loss entirely (Sneakernet, VBG). Each of these methods establishes a single quantum link with its own characteristic rate and fidelity. Having surveyed how these individual links can be established, we now zoom out to consider the pivotal devices that connect these links together: the quantum repeaters.

\section{Quantum repeaters}
\label{sec:quantum-networks:repeaters}

As previously discussed, point-to-point entanglement distribution via optical fiber cannot scale to a global quantum Internet.
Even within a metropolitan area network, establishing a direct point-to-point connection between every pair of end nodes is impractical.
Quantum repeaters~\cite{dur-repeater,briegel-repeater,azuma-rmp-repeater-review} help alleviate both of these problems.
Their main task is to store link-level entanglement and perform entanglement swapping to create end-to-end entangled states, while also serving as the intermediate nodes that form the network topology.
\Cref{fig:link-and-swap-for-repeater} illustrates these two fundamental operations.
Part~(a) shows the first step, where entanglement is generated over an elementary link using a scheme like the memory-interference-memory (MIM) architecture.
Once two adjacent links have succeeded, the repeater performs the second task, entanglement swapping, which connects the two links into one longer end-to-end entanglement, as depicted in part~(b).
\begin{figure}[t]
    \centering
    \includegraphics[width=\linewidth]{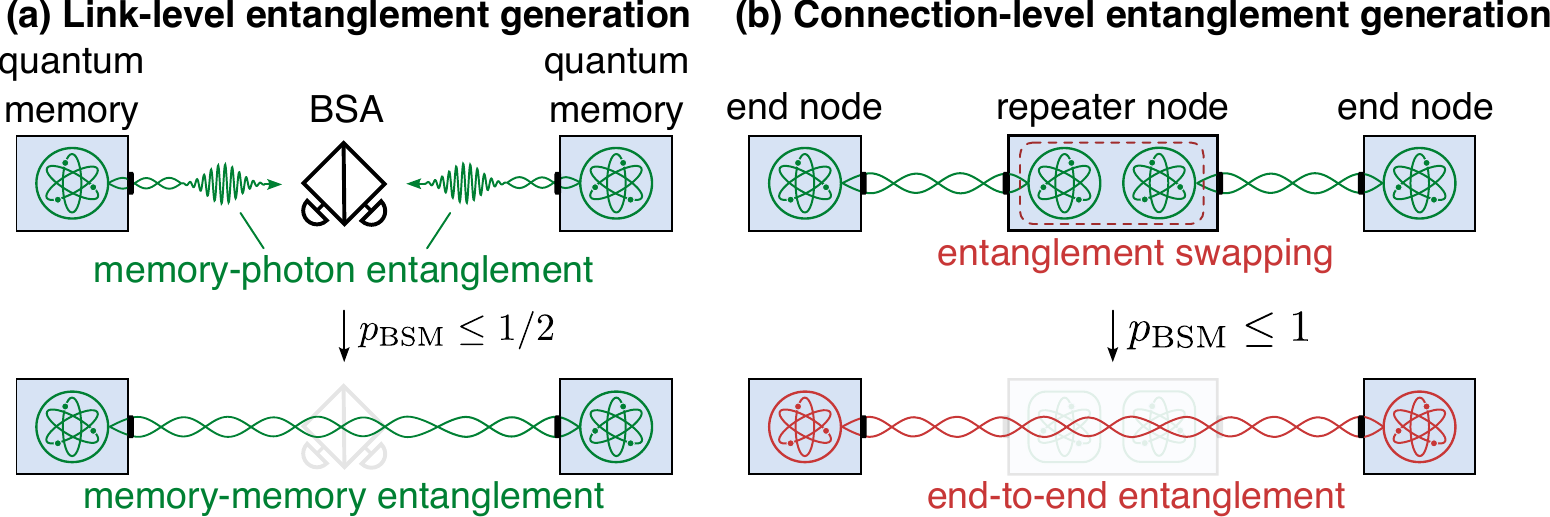}
    \caption[Illustrations of entanglement generation at the link level and inside quantum repeater between memories]{(a) Entanglement generation using a photonic \textit{Bell-state measurement} (BSM) at the \textit{Bell-state analyzer} (BSA) in the \textit{memory-interference-memory} (MIM) link architecture. Success probability is denoted by $p_{\text{BSM}}$. (b) After successful generation of two neighboring links, the repeater performs entanglement swapping on its quantum memories to create an end-to-end entangled connection.
    (From~\cite{joaquin-michal-cross-validation-quantum-network-simulators}.)}
    \label{fig:link-and-swap-for-repeater}
\end{figure}

In this thesis, I will make a clear distinction in terminology to precisely define the roles of network nodes. I define a quantum repeater as a node that connects to exactly two neighboring nodes, its primary function being to extend entanglement along a line. I will term a node with more than two connections a \emph{quantum router}, as it is responsible for directing entanglement across different paths in a network. Finally, while other literature sometimes uses the term ``quantum switch,'' I will reserve the term switch for a purely optical device (without quantum memories) that physically routes photons from one link to another without performing measurements.

\subsection{Classifications of repeaters}

Muralidharan et al.~\cite{muralidharan-repeater-generation} categorize quantum repeaters into three generations based on how they manage photon loss and operational errors, as summarized in Table~\ref{tab:repeater_generations}.

\begin{table}[htbp]
    \centering
    \caption[Generations of quantum repeaters and their error management strategies]{Generations of quantum repeaters and their error management strategies.}
    \label{tab:repeater_generations}
    \begin{tabular}{lll}
        \toprule
        \textbf{Gen} & \textbf{Photon Loss Management} & \textbf{Error Management} \\
        \midrule
        1G & Heralded Entanglement Gen. & Heralded Entanglement Purification \\
        2G & Heralded Entanglement Gen. & Quantum Error Correction (QEC) \\
        3G & Quantum Error Correction (QEC) & Quantum Error Correction (QEC) \\
        \bottomrule
    \end{tabular}
\end{table}

\textbf{First-generation (1G)} repeaters manage both error types using heralded, probabilistic mechanisms~\cite{briegel-repeater,dur-repeater}.
Link-level Bell pairs are heralded by successful Bell-state measurements (BSMs), and operational errors are typically addressed via heralded entanglement purification protocols.

\textbf{Second-generation (2G)} repeaters still rely on heralded generation to manage photon loss but use quantum error-correcting codes (QEC) to handle operational errors~\cite{jiang-2g-repeater,zwergerHybridArchitectureEncoded2014,sangouardQuantumRepeatersBased2009,kimiaeeasadiProtocolsLongdistanceQuantum2020,kimiaeeasadiQuantumRepeatersIndividual2018,mazurekLongdistanceQuantumCommunication2014,liLongRangeFailuretolerant2013,munroQuantumMultiplexingHighperformance2010}.
This involves encoding physical link-level Bell pairs into logical Bell pairs.
2G networks place stringent demands on the fidelity of local gates and memory coherence times, as QEC codes only suppress errors if the physical error rates are below a certain threshold.

\textbf{Third-generation (3G)} repeaters use QEC to handle both photon loss and operational errors~\cite{fowlerSurfaceCodeQuantum2010,muralidharanOvercomingErasureErrors2017,muralidharanUltrafastFaultTolerantQuantum2014,munroQuantumCommunicationNecessity2012a,borregaard-3g-repeater,rozpedekQuantumRepeatersBased2021,diadamoPacketSwitchingQuantum2022}.
In contrast to 1G and 2G networks, 3G repeaters enable a store-and-forward transmission model for quantum data, akin to packet switching~\cite{borregaard-3g-repeater,munroQuantumCommunicationNecessity2012a,diadamoPacketSwitchingQuantum2022}.
This requires decoding and re-encoding photonic quantum states at each repeater station, demanding extremely high-fidelity components and potentially denser repeater placement.

Note that this generational framework is not the only classification in the literature.
Repeaters can also be classified by their communication model (e.g., entanglement swapping vs. packet forwarding) or by their underlying physical platform~\cite{razaviIntroductionQuantumCommunications2018}.

\subsection{Choosing the right generation}

The classification of repeaters into generations was introduced to reflect a progression in technological capability~\cite{muralidharan-repeater-generation,munro-inside-quantum-repeaters}.
Higher generations correspond to more demanding hardware requirements but, in return, enable different error handling strategies that can unlock significantly higher entanglement distribution rates.

However, a higher generation number does not inherently imply better performance in all contexts.
The choice of generation is a design decision aimed at matching the error management strategy to the characteristics of the available hardware.
Each generation operates optimally within a specific regime.
For example, if a network has high-fidelity gates and long-lived memories but suffers from high fiber loss, 2G repeaters may outperform 3G designs.
Similarly, if the initial link-level entanglement is very noisy, 1G repeaters using purification are likely to achieve the best possible end-to-end fidelity.
Thus, the framework is a guide for optimizing network performance with a given technology, not a strict hierarchy of quality.

In recent years, the boundaries between these generations have begun to blur.
Advancements in quantum error correction demonstration~\cite{google-surface-code-one-qubit, quantinuum-microsoft-logical-below-threshold, bluvstein-quera-qec} and the realization that real-world systems must handle multiple error types simultaneously suggest that future networks will likely employ hybrid strategies, combining techniques from all three generations.

\nomenclature{Quantum network interface card (QNIC)}{A quantum analogue to a classical Network Interface Card (NIC), the QNIC manages a node's quantum communication links and the communication qubits involved. Its fundamental distinction is its capacity for local quantum processing on the states it handles, effectively transforming into a specialized, limited-capability quantum computer optimized for network protocols. The QNIC is also typically equipped to interact and coordinate operations with additional computation or algorithmic qubits and processing units housed within the larger network node}

\subsection{All-photonic quantum repeaters}

Distinct from memory-based repeater paradigms, all-photonic quantum repeaters represent another major approach to extending quantum communication distances.
These schemes, most notably the repeater graph state (RGS) based approach~\cite{azuma-rgs}, aim to circumvent the need for long-lived quantum memories at intermediate nodes altogether.
Instead, they rely on the creation and manipulation of large, highly entangled multipartite photonic states that span multiple repeater segments.
Entanglement is then established via a sequence of adaptive measurements performed on these photonic states.
While offering potential advantages such as high speed and avoidance of memory decoherence, all-photonic repeaters present their own unique set of challenges, particularly concerning the generation of large photonic resource states and the management of photon loss.
A detailed exploration of all-photonic repeaters and how our architectural framework can be extended to integrate them is the subject of \cref{chapter:architecture-all-photonic}.

Having categorized the different generations of repeaters based on their approach to errors, we now delve deeper into the specific mechanisms they use for this critical function.

\section{Error management}
\label{sec:error-management}

In any real-world quantum system, imperfections give rise to errors that degrade the fidelity of quantum resources.
In the context of quantum networking, we face two primary sources of errors.
The first is the transmission channel error, dominated by photon loss, which limits the distance of communication as discussed in \cref{sec:quantum-networks:repeaters}.
The second, and the focus of this section, is operational error.
This includes decoherence in quantum memories, imperfect gate operations, and noise from the generation hardware itself.
The result is that the initial \emph{base fidelity} of a freshly generated Bell pair may be insufficient for a given application, and as these states are connected via entanglement swapping, errors from each link accumulate, further degrading the fidelity of the final end-to-end state.

To mitigate operational errors, quantum networks employ two complementary techniques: quantum error correction (QEC) and entanglement purification.
QEC proactively protects quantum states from future degradation, while purification reactively distills high-fidelity pairs from ensembles of noisy ones.
This section will review both of these principal strategies, beginning with the foundational concepts of QEC.

\subsection{Quantum error correction}
\label{subsec:managing-operational-error-with-qec}

Quantum error correction (QEC) is a technique that enables a system to resist a bounded number of errors by encoding the information of a single logical qubit into a redundant state of multiple physical qubits.
Errors can then be identified and corrected non-destructively via syndrome measurements (recall \cref{sec:stabilizer-formalism}).

However, QEC is not a panacea; its effectiveness hinges on the physical error rate being below a certain \emph{error threshold}.
If this condition is met, logical errors can be suppressed.
Above the threshold, the correction procedure introduces more noise than it removes~\cite{gottesmanSurvivingQuantumComputer}.
This is why QEC is typically associated with higher-generation repeaters, which assume higher-fidelity components.

In the context of quantum repeaters, QEC serves two distinct roles that define the different repeater generations.
Second-generation (2G) repeaters use QEC primarily to protect against local operational errors while the link-level entanglement generation is still heralded.
The process involves performing remote operations to entangle two logical qubits that reside in separate repeater nodes, assuming the initial link-level Bell pairs used are already of high quality~\cite{jiang-2g-repeater}.
Third-generation (3G) repeaters, in contrast, extend the use of QEC to combat channel loss directly.
The first 3G strategy performs remote operations using codes that tolerate some photon loss, effectively treating a failed link-generation attempt as a correctable erasure error while maintaining a high threshold for operational errors~\cite{fowlerSurfaceCodeQuantum2010,rametteFaulttolerantConnectionErrorcorrected2024,sinclairFaulttolerantOpticalInterconnects2025}.
The second 3G strategy uses loss-tolerant codes to directly encode and transmit flying logical qubits, making the initial link generation itself robust to photon loss~\cite{azuma-rgs,borregaard-3g-repeater,hilaire-logical-bsm}.

\subsection{Probabilistic entanglement purification via error detection}
\label{subsec:two-way-entanglement-purification}

When the base link-level fidelity does not satisfy the QEC threshold, a different strategy is required.
In the context of first-generation repeaters, where loss is high and base fidelity is low, networks can resort to probabilistic entanglement purification via error detection~\cite{dur-briegel-purification-review,bennett-purification,deutsch-purification}.
This process boosts the fidelity of noisy Bell pairs by consuming multiple copies to produce a smaller number of higher-fidelity ones.
This ``detect and discard'' approach serves as an error detection mechanism, where some pairs are sacrificed to certify the quality of others.
Because it is often more efficient to generate a new purified state than to repair a faulty one, this method is well-suited for preparing known target states like Bell pairs.

This method of purification is probabilistic and often referred to as \emph{two-way heralded entanglement purification}.
The ``two-way'' nature refers to the need for classical messages to be exchanged between parties to herald the success or failure of a purification attempt.
The main strategy is based on joint stabilizer measurements across multiple copies of Bell pairs~\cite{bennett-purification, deutsch-purification}.
This approach is powerful because its theoretical fidelity threshold is 50\%, assuming noiseless quantum operations.
Other two-way protocols have also been studied, such as entanglement breeding, which requires ideal Bell pairs as a catalyst~\cite{bennett-purification}, or error filtration, which relies on non-destructive measurements~\cite{gisinErrorFiltrationEntanglement2005}.

Recent work has extended this stabilizer-based approach by more explicitly using the structure of quantum error-correcting codes.
In these protocols, syndrome extraction is used to detect errors, but instead of applying a correction, a detected error still triggers the probabilistic ``discard'' step~\cite{pattison-constant-rate-purification,shi-stabilizer-purification}.
This approach is notable because it can achieve a non-vanishing distillation rate in the asymptotic limit---a property known as a constant-rate yield---offering potential performance advantages over traditional methods.

\subsection{Deterministic entanglement purification with QEC codes}

In contrast to the probabilistic, two-way protocols, another class of purification is deterministic and requires only one-way classical communication.
In these schemes, often called \emph{one-way purification}, there is no probabilistic discarding of resources.
The key distinction lies in the classical communication flow: while both parties perform local measurements on their share of the entangled states, only one party sends their measurement results to the other.
The receiving party then uses this one-way information to determine and apply the final corrective operations to distill the high-fidelity states.

The foundational protocol for this approach is \emph{entanglement hashing}, which was first proposed for the asymptotic limit of an infinite ensemble of Bell pairs~\cite{bennettMixedstateEntanglementQuantum1996}.
Hashing can, in principle, distill perfect entanglement with a high rate.
However, its reliance on random quantum codes makes the classical decoding process complex, and its intolerance to any operational noise makes the original protocol impractical for real-world hardware with a finite number of resources~\cite{zwergerRobustnessHashingProtocols2014}.

Recent approaches improve on this by employing structured quantum error-correcting codes to deterministically map many noisy Bell pairs into fewer, higher-fidelity logical ones.
In particular, quantum low-density parity-check (qLDPC) codes~\cite{breuckmannQuantumLowDensityParityCheck2021} have emerged as a promising foundation.
These codes support efficient decoding and maintain a constant distillation rate, even with finite resources~\cite{ataides-constant-rate-purification-high-rate-codes}.

Despite these advantages, one-way purification has a key limitation: it demands a higher initial fidelity of the Bell pairs.
Unlike two-way protocols, which tolerate base fidelities down to 50\%, qLDPC-based one-way purification schemes require infidelities below 5\% to yield a net fidelity gain, assuming gate and measurement errors around 0.1\%.
This positions one-way purification closer to the operating regime of full QEC, making it powerful but currently more demanding in practice.

\subsection{System-level error management strategies}
\label{subsec:system-level-strategy}

In practice, while a single technique like two-way purification is sufficient to achieve high fidelity from low-fidelity base resources, this approach is often not optimal across the entire fidelity range.
Limited by communication latency, two-way purification protocols tend to be less efficient in terms of throughput at higher fidelities compared to other approaches.
In these regimes, deterministic one-way purification or direct QEC-based schemes often become more resource-efficient and yield a higher rate.
This reality motivates a hybrid, multi-stage strategy, creating what can be termed a purification pipeline.
In this model, two-way protocols would first be used to elevate low-fidelity pairs to an intermediate fidelity, after which more efficient one-way or QEC protocols could take over to achieve the final target fidelity.
The specific architecture of such a pipeline, including the crossover points between different protocols, is highly dependent on the system's noise model and underlying hardware technology.

A more difficult challenge lies in coordinating these techniques across the entire connection path to produce an end-to-end Bell pair.
This requires solving a joint optimization problem involving both purification and entanglement swapping schedules, while managing trade-offs among fidelity, rate, and memory decoherence~\cite{shchukin-optimal-entanglement-swapping,krastanov-purification,zangNoGoTheoremsUniversal2025,zangEntanglementPurificationQuantum2025,jansen-enumerating-bilocal-clifford-purification,rdv-quantum-networking-book,gidney-tetrationally-compact-purification,jiang-dp}.
Ultimately, system-level error management must draw on a complementary toolkit.
Two-way purification plays an essential role in 1G repeaters, where it distills usable states from highly noisy physical resources.
In 2G and 3G repeaters, QEC becomes the central tool for protecting already high-fidelity states from additional errors and enabling fault-tolerant operations.
As quantum networks become more practical, hybrid strategies will likely dominate---blurring the boundaries between repeater generations and combining multiple error management techniques to adapt dynamically to the system's needs.

\section{Challenges in realizing quantum networks}
\label{sec:quantum-networks:challenges}

While the basic principles of quantum communication are well established, translating these principles into scalable, reliable, and efficient quantum networks presents a wide array of practical challenges.
These obstacles, spanning from physical hardware limitations to complex system-level control, represent the primary bottlenecks to deploying quantum technologies in real-world environments.
This section provides an overview of the most pressing engineering and architectural challenges, which collectively motivate the architectural work presented in the subsequent chapters of this thesis.

\subsection{Physical layer and link-level challenges}

At the most fundamental level, a quantum network is built upon the ability to generate high-quality entanglement over physical links.
However, this process is probabilistic and fraught with difficulty.
For memory-based systems, generating link-level entanglement is often hampered by photon loss in optical fibers and imperfections in sources and detectors.
For all-photonic repeaters, the challenge lies in creating large, multipartite graph states with flying qubits.
While deterministic generation via quantum emitters is theoretically efficient and has been demonstrated in small scale, current hardware often relies on probabilistic photonic fusion, which can incur a large resource overhead.

Beyond generation, a network's performance depends on having a clear and continuous understanding of its physical components.
This requires protocols for \emph{link characterization}---constantly monitoring parameters like fidelity, latency, and generation rate to inform dynamic decision-making.
Furthermore, the lack of \emph{standardized network elements} is a key barrier to building modular and interoperable multi-vendor networks.
Without agreed-upon standards for physical interfaces, photon wavelengths, or detector characteristics, scaling beyond bespoke laboratory testbeds remains difficult.

The difficulty of overcoming these physical-layer hurdles in experimental hardware, where iteration is slow and costly, underscores the critical need for high-fidelity simulation.
Exploring the network-wide impact of link-level noise, loss, and component imperfections is a task perfectly suited for a robust simulation platform, directly motivating the research question on \textbf{how simulation can guide the design and validation of scalable quantum networks}.

\subsection{Network control and resource management}

Beyond the physical layer, a quantum network's efficiency is determined by how it manages its scarce and fragile resources, which requires a sophisticated classical control plane.
This introduces several interdependent optimization problems.
The core challenge is the co-optimization of \emph{scheduling, routing, and multiplexing}.
Nodes must efficiently schedule entanglement attempts, select optimal routes that balance fidelity and rate, and fairly multiplex limited resources like quantum memories among competing users.

These decisions depend on having \emph{effective cost metrics for routing} that go beyond classical measures like hop-count to include quantum parameters such as fidelity, noise channel characteristics, and coherence time.
Furthermore, because link performance can fluctuate, the network must support \emph{real-time monitoring and adaptation}, allowing protocols to adjust dynamically to changing conditions.
The interplay between these tasks is complex; for instance, routing choices for an all-photonic path might prioritize different metrics than for a memory-based path, and the ``circuit-switched'' nature of the all-photonic approach presents unique multiplexing challenges that are not yet fully understood.

Finally, with multiple users sharing the network, \emph{ensuring fair resource allocation} is a critical socio-technical problem.
Mechanisms must be developed to arbitrate access to scarce resources, enforce priorities, and provide equitable service, all while accounting for the probabilistic nature of quantum operations.
Developing and testing these sophisticated control strategies in the face of such complexity is almost intractable without powerful modeling tools, highlighting how \textbf{simulation becomes indispensable for designing and validating effective resource management protocols}.

\subsection{System-level and architectural hurdles}

At the highest level, significant architectural hurdles must be overcome to ensure the network functions as a coherent, scalable system.
A primary obstacle is \emph{managing heterogeneity and ensuring interoperability}.
Real-world quantum networks will inevitably be built from a diverse mix of hardware, repeater technologies, and protocol generations.
Creating a unifying framework that can abstract away these differences while still leveraging their unique capabilities is a central architectural problem.

This challenge is amplified by the need for \emph{tight timing and synchronization}, as many quantum protocols require sub-nanosecond precision between distant nodes.
Furthermore, designing \emph{robust end-to-end protocols} is complicated by the enormous classical communication overhead some schemes introduce.
All-photonic repeaters, for instance, can require processing a massive volume of classical data at the end nodes, creating a potential bottleneck on classical processing power.
Finally, the question of \emph{end-node participation} reveals a deep architectural disconnect, as it remains unclear how memory-equipped end nodes can best interface with a memory-less all-photonic repeater, or how two networks built from different technologies should be connected.

These disconnections at the system level between different technologies and operational paradigms crystallize the core questions that motivate this thesis.
They force us to ask: \textbf{How do we design a unified architecture to manage this heterogeneity?}
And, as a key part of that, \textbf{how can all-photonic quantum repeaters become practical and interoperable with memory-based designs?}

\subsection{The need for a holistic architecture}

These challenges are deeply interconnected and do not exist in isolation.
Addressing them requires a holistic architectural vision that bridges hardware diversity, supports dynamic protocol adaptation, and enables performance reasoning at multiple levels of abstraction.
Such an architecture must be modular, programmable, and support scalable control, without being tightly coupled to a single hardware platform or repeater generation.
In the following chapters, we propose such a solution: a unifying framework built on the Quantum Recursive Network Architecture (QRNA) and programmable RuleSet-based protocols, designed to manage heterogeneous technologies and pave the way for a flexible and scalable Quantum Internet.

\clearpage
\chapter{Architecture and protocols for memory-based quantum repeaters}
\label{chapter:architecture-memory}

Quantum communication is advancing rapidly, with experimental demonstrations ranging from laboratory-scale to metropolitan-scale testbeds~\cite{valivarthiTeleportationSystemsQuantum2020,davisEntanglementSwappingSystems2025,davisTeleportationEntanglementFermilab2024,ruskucMultiplexedEntanglementMultiemitter2025,pompiliRealizationMultinodeQuantum2021,hermansQubitTeleportationNonneighbouring2022,vanleentEntanglingSingleAtoms2022,krutyanskiy-northup-ion-trap-entanglement,senaRobustHighFidelityQuantum2025,valivarthiQuantumTeleportationMetropolitan2016a,sunQuantumTeleportationIndependent2016,chungDesignImplementationIllinois2022,alshowkanReconfigurableQuantumLocal2021,craddockAutomatedDistributionPolarizationEntangled2024,liuCreationMemoryMemory2024,knautEntanglementNanophotonicQuantum2024}.
While many efforts have explored the design of protocols at various abstraction layers~\cite{aparicioProtocolDesignQuantum2011,bacciottiniLeveragingInternetPrinciples2025a,beauchampModularQuantumNetwork2025a,dahlbergLinkLayerProtocol2019,gauthierArchitectureControlEntanglement2023,gauthierControlArchitectureEntanglement2023a,illianoQuantumInternetProtocol2022,kumarQuantumInternetTechnologies2025a,liSurveyQuantumInternet2024,oslovichCompilationStrategiesQuantum2025,ruSynchronizationControlPlaneProtocol2024,dahlberg-netqasm}, few have integrated these components into a coherent, end-to-end architectural framework~\cite{delledonneOperatingSystemExecuting2025,fangQuantumNETworkTheory2023a,pirkerQuantumNetworkStack2019,vechtQoalaApplicationExecution2025a}.
If the history of the classical Internet has taught us anything, it is that hardware and architecture must evolve hand-in-hand; otherwise, the network risks becoming a patchwork of incompatible solutions that hinder scalability.
At present, architectural development is lagging behind hardware progress.
It is therefore essential to begin establishing architectural foundations now---foundations that can both guide hardware development and adapt to future applications.
While multiple network architectures will likely coexist, a true Quantum Internet will ultimately require them to interoperate under a unified framework.
\extrafootertext{Portions of this chapter have been adapted from the following publications:
\begin{itemize}
    \item \fullcite{rdv-qi-architecture}
    \item \fullcite{cocori-quisp}
    \item \fullcite{joaquin-michal-cross-validation-quantum-network-simulators}
\end{itemize}
}

In this chapter, I will outline our attempt at developing a comprehensive architecture for networks utilizing memory-based quantum repeaters.
Our framework builds upon and concretizes two powerful, pre-existing conceptual ideas.
The first is the architectural principle of the \emph{quantum recursive network architecture (QRNA)}, an adaptation by Van Meter, Touch, and Horsman~\cite{vanmeterRecursiveQuantumRepeater2011} of the classical Recursive Network Architecture (RNA) developed by Touch~\cite{touchDynamicRecursiveUnified2011,touchRNAMetaprotocol2008,touchRecursiveNetworkArchitecture2006a}.
QRNA is motivated by the principle that managing small, independent networks is easier than managing a single monolithic one; it uses layers of abstraction to hide internal complexity and simplify protocol reasoning.
The second foundational idea is the RuleSet-based protocol, initially proposed by Matsuo et al.~\cite{kaaki-ruleset-sim,kaaki-master-thesis} to manage errors for link-level entanglement.
Our collaborative work, detailed in this chapter, was to extend and concretize these concepts into a complete, end-to-end framework.

These concepts are brought to operational life by the Quantum Routing Software Architecture (QRSA), a modular framework designed to manage connection establishment, execute RuleSet-driven operations, handle routing decisions, and control the underlying quantum hardware.
A key part of this effort was the development of the Quantum Internet Simulation Package (QuISP), which implements these architectural principles.
My specific contribution within this collaborative effort was to mature the simulator by completing its error models and improving its performance, and to concretize the architecture by formalizing the RuleSet protocol for end-to-end connections.

The chapter is structured as follows.
It begins with an overview of the fundamental design decisions inherent in building any quantum network (\cref{sec:architecture-memory:design-decisions}).
Subsequent sections (\cref{sec:architecture-memory:expressing-connection-semantics-rulesets,sec:architecture-memory:network-services,sec:architecture-memory:networking,sec:architecture-memory:internetworking-qrna}) then systematically present the concrete protocols and mechanisms our framework provides.
The specifics of the Quantum Routing Software Architecture (QRSA) are detailed in \cref{sec:architecture-memory:qrsa}.
To validate our simulation-based approach, the chapter then presents a rigorous cross-validation study of QuISP against SeQUeNCe~\cite{wuSeQUeNCeCustomizableDiscreteevent2021}, another leading simulator, in \cref{subsec:architecture-memory:cross-validation}.
This result establishes the credibility of our architectural model and the simulation tools used to evaluate it.
Finally, the architecture is considered as a whole, and we discuss how other proposals in the literature can be viewed within the same framework (\cref{sec:architecture-memory:discussion}).

\section{Design decisions}
\label{sec:architecture-memory:design-decisions}

Building a coherent and scalable architecture for the Quantum Internet requires addressing several foundational design questions, many of which share similarities to the challenges faced during the development of the classical Internet.
The choices made at this stage are critical, as they profoundly influence the network's core functionality, its capacity to support diverse applications, the efficiency of resource utilization, the complexity of network management, and ultimately, its overall robustness and scalability.
Establishing these architectural underpinnings is essential not only for guiding near-term hardware development but also for ensuring the network can adapt to future technological advancements and operational insights.

Navigating the design space involves making deliberate choices across several key areas.
These decisions, which we will explore in detail throughout this chapter within the context of our proposed architecture, are briefly introduced below.

\textbf{The fundamental network service}: The fundamental services provided by the network determine how useful the network is to higher-layer applications.
What is the primary quantum resource provided to users?
Should the network focus solely on distributing Bell pairs, delegating the task of building more complex states to applications?
Or should it offer richer services like the direct generation of multipartite states (e.g., graph states or W-states) or even qubit teleportation?
This choice impacts protocol complexity and the division of complex tasks between the network and user space.
(\Cref{sec:architecture-memory:network-services})

\textbf{The nature of connections}: How should end-to-end entanglement be established and maintained?
Should paths be predetermined and resources reserved (akin to circuit-switching), which better suits repeater schemes requiring intricate scheduling of swapping and purification (like 1G/2G repeaters or certain all-photonic designs)?
Or should entanglement be extended hop-by-hop (like packet-switching), offering flexibility but potentially sacrificing guarantees, fidelity, and efficiency?
Furthermore, we must define the stateful nature of these connections, including the signaling mechanisms (handshakes) for setup and teardown, and whether control is centralized or distributed.
(\Cref{sec:architecture-memory:expressing-connection-semantics-rulesets})

\textbf{Application interface}: How do end-user applications interact with the network?
Defining a clear Application Programming Interface (API)---a quantum equivalent of the classical socket---is crucial for abstracting the underlying network complexity and enabling application development by separating appliation logic from the complexity of network resource management and the distributed quantum operations required to generate link-level Bell pairs, perform purifications, or entanglement swapping.
(\Cref{sec:architecture-memory:network-services})

\textbf{Conveying user requests}: How do applications communicate their needs (e.g., target fidelity, rate, state type, number of pairs) to the network?
This involves designing request protocols and establishing clear naming and tracking conventions for quantum resources, especially when multiplexing resources among multiple concurrent applications.
(\Cref{sec:architecture-memory:expressing-connection-semantics-rulesets})

\textbf{Network node roles}: What distinct functional roles should different network components play?
Defining clear roles for end nodes, various generations and classes of quantum repeaters and quantum routers (memory-based and all-photonic), and other devices is essential for modularity and interoperability, separating technological capabilities from architectural function.
(\Cref{sec:architecture-memory:networking})

\textbf{Routing strategy}: How are paths selected through the network to connect end nodes?
While seemingly analogous to classical routing, quantum routing algorithms must also account for unique challenges like qubit decoherence, probabilistic operations, biased quantum error channels, and the tight coordination between quantum operations and classical signaling.
(\Cref{subsec:architecture-memory:routing})

\textbf{Resource multiplexing}: How are shared quantum resources---memories, channels, entangled states---allocated among competing user requests?
Various strategies exist, from reserving resources (circuit-switching) to multiplexing with dynamic allocation based on time, buffer space, or statistical priority, each with implications for efficiency, fairness, and Quality of Service (QoS).
Mechanisms for authentication, authorization, and accounting (AAA) are also intrinsically linked to resource allocation.
(\Cref{subsec:architecture-memory:multiplexing})

\textbf{Security considerations}: Beyond the inherent security features like QKD, quantum networks introduce novel vulnerabilities~\cite{satohAttackingQuantumInternet2021,blakelyQuantumInformationSystem2024,suzukiClassificationQuantumRepeater2015}.
We must consider threats ranging from physical side-channel attacks on devices to architectural vulnerabilities like denial-of-service on quantum resources or impersonation via classical control channels, particularly in the context of distributed quantum protocols.
(\Cref{sec:architecture-memory:networking})

\textbf{Internetworking}: How can independently operated quantum networks, potentially based on different technologies or protocols, connect and interoperate to form a larger, global Quantum Internet?
This requires standardized protocols, cross-domain addressing schemes, and mechanisms for establishing trust and translating services between networks.
(\Cref{sec:architecture-memory:internetworking-qrna})

While this overview of design decisions is not exhaustive, addressing these fundamental questions provides a framework for constructing a robust and adaptable quantum network architecture.
We will delve into these areas in the following sections, presenting our approach based on the RuleSet-based protocols and the Quantum Recursive Network Architecture (QRNA) to provide concrete solutions and mechanisms addressing these design decisions.


\section{Defining quantum network services}
\label{sec:architecture-memory:network-services}

Entanglement is the resource that fuels quantum applications such as QKD, teleportation, quantum sensing, delegated quantum computation, or distributed quantum computation.
Continuous, reliable and efficient replenishment of this resource is one of the primary tasks of a quantum network.
However, entangled states come in many shapes and sizes~\cite{horodeckiQuantumEntanglement2009,maMultipartiteEntanglementMeasures2024,amicoEntanglementManybodySystems2008,plenioIntroductionEntanglementMeasures2006,hajdusekDirectEvaluationPure2013a,hajdusekEntanglementPureThermal2010}.

This section delves into the first fundamental design decision: defining the services a quantum network offers to its users.
We first explore the challenges and trade-offs inherent in this choice and then describe the approach taken within our proposed architectural framework.

\subsection{Fundamental network service and semantics}
\label{subsec:architecture-memory:fundamental-service}

The design of a quantum network must begin with a clear definition of its fundamental services---what quantum states or capabilities the network is expected to provide to end users.
These decisions determine the complexity of the protocols at the network layer and the applications that run above it.
A minimalist design may treat \emph{Bell pairs} as the primary network-level service.
Bell pairs serve as the smallest unit of entanglement and the foundation for nearly all quantum communication protocols.
Restricting the service to Bell pair distribution simplifies the network's responsibilities.
However, this approach shifts complexity to the applications, which must construct multipartite or fault-tolerant states themselves and manage the coordination overhead that entails.

At the other end of the spectrum, networks may offer richer services such as multipartite entangled states (e.g., GHZ, W-states~\cite{durThreeQubitsCan2000}, graph states~\cite{hein-graph-state-pra,hein-graph-state-arxiv}), or fault-tolerant state teleportation.
While applications can, in theory, synthesize these states from Bell pairs, direct network-level support may offer efficiency gains and reduce sensitivity to noise by internalizing complex procedures like direct graph state generations or supporting the delivery of error-correcting code encoded logical qubits.

In our architecture, we adopt Bell pair distribution as the core network service, as it allows for a well-scoped, foundational architectural framework.
As we will see later in \cref{sec:architecture-memory:discussion}, our architecture is also extensible to multipartite entanglement services.
These extensions build naturally on the base protocols discussed here and suggest a path for future network capabilities.

Importantly, the semantics of Bell pair distribution are not merely those of passive delivery.
Even in this basic model, distributed quantum computation occurs along the path via entanglement swapping, possibly combined with purification at intermediate repeater nodes.
A proper service definition must account for this processing, as it directly affects fidelity, latency, and trust assumptions in the network.

In addition to quantum state delivery, timing information is often a critical part of the service.
Applications in distributed quantum sensing~\cite{degen-sensing,proctor-quantum-sensing,proctorMultiparameterEstimationNetworked2018,giovannetti-advances-in-metrology,ge-distributed-metrology,gottesman-longer-baseline-telescopes} and clock synchronization~\cite{Ilo-okeke-quantum-clock} require precise knowledge of when entanglement was established or when measurement events occurred.
Hence, high-precision timestamps may need to be bundled into the service interface offered by the network.

In summary, while this work focuses on Bell pair distribution as the baseline network service, it should be emphasized that its semantics involve nontrivial quantum processing along the route and may include auxiliary data like timestamps.
This careful service definition sets the stage for scalable and extensible quantum network architectures.

\subsection{Quantum sockets}
\label{subsec:architecture-memory:sockets}

With Bell pair distribution defined as the fundamental service of the quantum network, the next question is how applications at the end nodes interact with this service.
This is the role of the Quantum Socket API~\cite{satohQuantumSocketsNISQ2018}---an abstraction layer that mediates between quantum applications and the underlying network infrastructure, much like how classical socket APIs decouple application logic from network protocols in the Internet.

The Quantum Socket provides a uniform interface for creating entanglement-based connections between nodes.
Applications can use it to request Bell pairs without needing to manage the probabilistic behavior of physical-layer protocols or the complexity of entanglement swapping and purification.
If the network supports richer services beyond Bell pair delivery, such as teleporting quantum states or accessing shared randomness, these could also be provided via the socket API.

To accommodate a range of use cases, the API is designed to be node-type agnostic, supporting a variety of application endpoints.
We can categorize applications into three types~\cite{vanmeterOptimizingTimingHighSuccessProbability2017}: (i) extracting correlated classical data (for QKD and sensor applications), (ii) distributed Clifford operations where correction operations can be performed posthoc, and (iii) distributed non-Clifford operation where states need full confirmation of success and correction before being consumed.
For extracting correlated classical results, interaction with the network is largely synchronous from the application's point of view, as classical results are obtained directly after measurement.
These applications are tolerant of the stochastic timing of Bell pair delivery and can selectively discard measurements performed on unentangled states without affecting their correctness.
In contrast, computational applications (the second and third types) typically require asynchronous interfaces to coordinate with their local quantum processes.
This coordination involves complex decision-making, such as whether to wait for a resource or proceed, and how to handle errors---either by propagating corrections forward in time or by reversing the computation to a state prior to the consumption of the faulty resource.
Developing robust programming models that handle this variability, including callbacks or deferred execution, remains an open research problem and is beyond the scope of this thesis.

By exposing a consistent API to applications regardless of application and node types, the Quantum Socket interface promotes modularity and simplifies development.
It defines the boundary between application logic and the complex distributed operations that the network must perform behind the scenes.
This thesis focuses on the design and operation of the quantum network \emph{up to this boundary}---the delivery of Bell pairs or provision of core entanglement-based services---with the understanding that applications themselves are responsible for the subsequent, higher-level utilization of these resources.

\section{Expressing connection semantics via RuleSets}
\label{sec:architecture-memory:expressing-connection-semantics-rulesets}

With the fundamental network service defined, we now turn to the question of how end-to-end entangled connections are established and managed.
This process involves a series of coordinated quantum operations across multiple, physically separated nodes---a form of \emph{distributed computation}.
Each node has access only to its local quantum resources and classical control, yet must participate in a global task to generate and maintain shared entangled states.

To address the complex problems of this distributed computation, this section details
our adoption of a connection-oriented model and introduces the RuleSet protocol, built from the work in~\cite{kaaki-master-thesis,kaaki-ruleset-sim}, as the core mechanism for orchestrating the distributed operations required.

\subsection{A RuleSet-driven connection-oriented model}
\label{subsec:architecture-memory:ruleset-driven-model}

To support flexible scheduling of complex operations and enable non-local coordination, such as entanglement purification extending beyond directly connected neighbors, our architecture employs a \emph{connection-oriented} model.
In this paradigm, an end-to-end path is established first, with all intermediate nodes along this path explicitly agreeing to participate in the protocols for distributing Bell pairs for that specific connection.
This model is chosen primarily for the enhanced reliability and more predictable state fidelity it offers for the resulting end-to-end Bell pairs, particularly when contrasted with connectionless approaches which typically restrict purification to adjacent nodes and may lack end-to-end fidelity guarantees.
(A more detailed comparison between connection-oriented and connectionless models is provided in~\cref{sec:architecture-memory:discussion}.)

To realize this connection-oriented model while maximizing node autonomy and minimizing the latency associated with synchronous control or excessive classical communication round-trips, we adopt an event-driven operational scheme orchestrated by \emph{RuleSets}~\cite{kaaki-master-thesis,kaaki-ruleset-sim}.
Upon the establishment of a connection, each participating node is provisioned with a specific RuleSet.
This RuleSet dictates the node's behavior in response to local events---such as the successful generation of link-level entanglement, the arrival of a classical message, or a timeout.
This RuleSet-driven approach is central to our architecture as it enables largely decentralized control and asynchronous operation; nodes react dynamically based on their local context and pre-defined rules rather than waiting for explicit step-by-step instructions or relying on global lockstep coordination.

\subsection{RuleSet protocol}
\label{subsec:architecture-memory:ruleset-definition}

A RuleSet consists of a collection of Rules where each Rule contains a \emph{condition clause} and an \emph{action clause}.
When a Rule's condition is satisfied by the local state or event context, its corresponding action is executed.
A collection of such RuleSets, one at each participating node, collectively constitutes the full distributed program.
A key distinction from classical analogies like Software Defined Networking (SDN) OpenFlow~\cite{mckeownOpenFlowEnablingInnovation2008} or P4~\cite{bosshartP4ProgrammingProtocolindependent2014} is that our RuleSet design explicitly incorporates the management of complex, probabilistic local state, including quantum states, which is indispensable for quantum repeater functionality.
This management of local state can be cast as the protocol's ``side effects''.
A conceptual depiction of a RuleSet is shown in~\cref{fig:ruleset-example}.
\begin{figure}[hbt!]
    \centering
    \includegraphics[width=\textwidth]{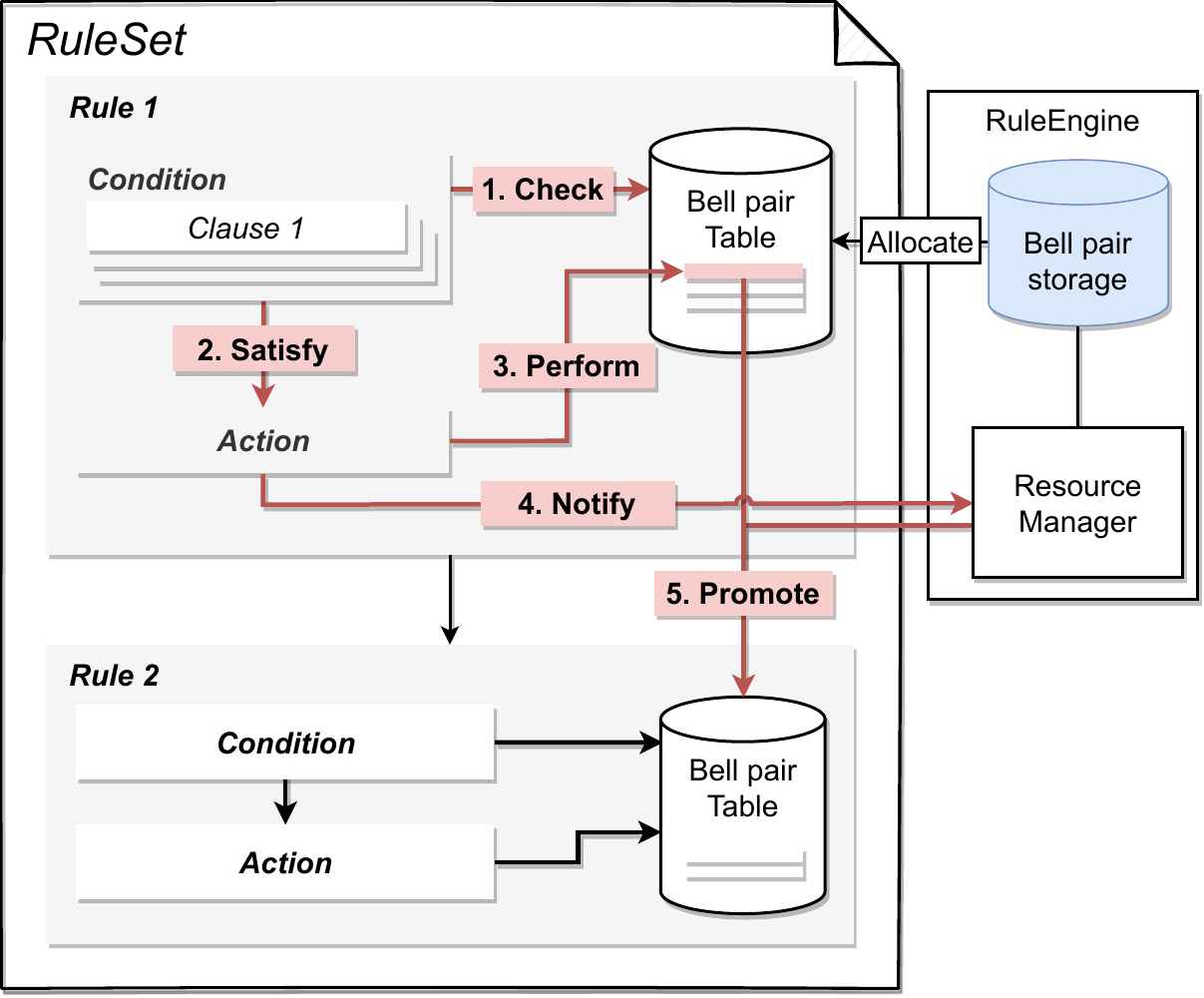}
    \caption[Workflow of the RuleSet execution cycle]{
        Workflow of the RuleSet execution cycle.
        The process unfolds in a deterministic sequence:
        (1) The Rule Engine begins an evaluation phase, checking the Condition clauses of its Rules sequentially.
        (2) If a Rule's conditions are fully satisfied, the process advances to the next step; otherwise, the engine waits for a new event (e.g., resource allocation) to trigger a new evaluation cycle.
        (3) The selected Rule executes its Action clause.
        (4) Upon completion, the Action notifies the Rule Engine of the outcome.
        (5) Finally, the Rule Engine performs a state update, promoting the resource to its next stage in the protocol.
        (From~\cite{cocori-quisp}.)}
    \label{fig:ruleset-example}
\end{figure}

Once a quantum resource, such as a link-level Bell pair, is assigned to a specific RuleSet for a connection, that assignment is typically exclusive and persists until the resource is consumed or released by that RuleSet.
The initial assignment of resources to RuleSets falls under the purview of a separate multiplexing scheme (will be discussed in~\cref{subsec:architecture-memory:multiplexing}).
Consequently, RuleSets, along with any qubits at a node that are currently assigned to a particular connection, constitute vital \emph{connection state} that must be maintained at each repeater or router along the path.
While the scalability implications of such statefulness require careful assessment (see \cref{subsec:architecture-memory:aaa} for security considerations), maintaining some form of state at intermediate nodes appears unavoidable for most approaches to robust quantum networking.

A fundamental principle of this model is that there is one RuleSet dedicated to each connection traversing a node.
Once a quantum resource, such as a link-level Bell pair, is assigned to a RuleSet, it is also tagged with a label (called a \emph{tag}) and given a \emph{sequence number} to track its role in the protocol.
Tags are used to define resource pools, and the movement of resources from one tag to another encodes the logical progression of the protocol.
The tag ordering is unidirectional and intentional: it serves as a scaffold for reasoning about RuleSet behavior and ensures termination of the protocol.
The sequence number, meanwhile, helps maintain resource ordering within a tag and supports coordination across nodes without additional messaging.
This prevents mismatches, such as two nodes selecting different Bell pairs for the same operation, a problem known as leapfrogging~\cite{rdv-quantum-networking-book}.

Thus, a RuleSet acts as a self-contained program governing the lifecycle of entangled states for a connection.
Resources flow from one tag (pool) to another as operations proceed.
When a Rule consumes or updates a resource, it either moves it to a new tag (signaling advancement in the protocol), or marks it as defunct (releasing physical qubits for link-level generation reuse).
This unidirectional tag flow ensures that every quantum state associated with a RuleSet either progresses toward successful end-to-end delivery or is discarded without ambiguity.
Note that this model does not eliminate the possibility of livelocking, where Bell pairs are repeatedly allocated and appear to progress through tag transitions, but are eventually all discarded without reaching the final tag or being delivered to the application.
An example of resource flow within a RuleSet is depicted in \cref{fig:architecture-memory:resource-flow-in-ruleset-livelock}.
\nomenclature{Livelock}{A type of resource starvation where changes in local state indicate that progress is being made with respect to local events, while the overall global task is not progressing}
\begin{figure}
    \centering
    \includegraphics[width=\textwidth]{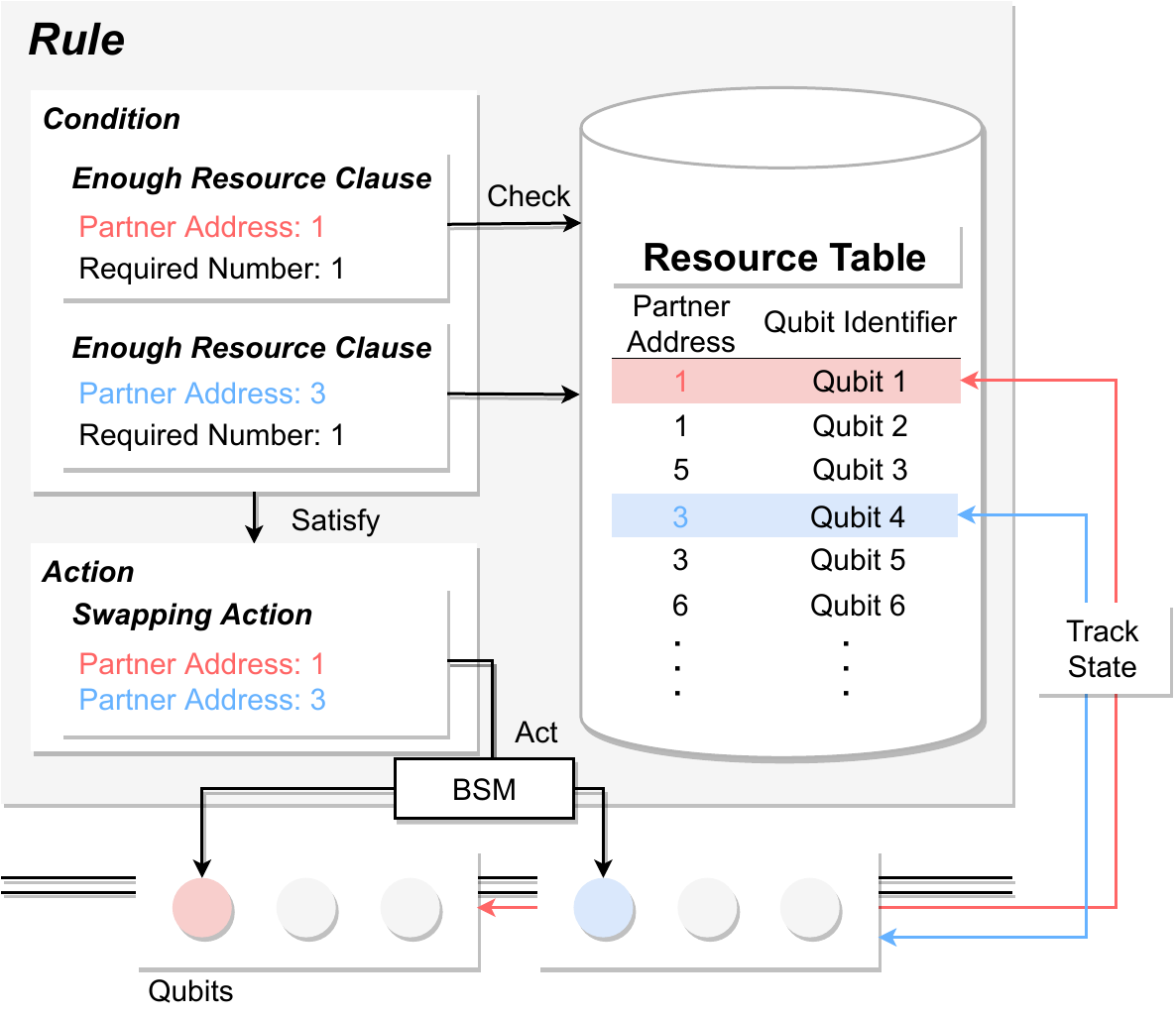}
    \caption[Resource flow as an effect of a Rule]{
        An illustration of resource flow as a result of a Rule's execution.
        The diagram shows how a resource is first evaluated against a Rule's Condition clause.
        Upon a successful match, the corresponding Action clause is executed, causing the resource to transition to a new state or tag.}
    \label{fig:architecture-memory:resource-flow-in-ruleset-livelock}
\end{figure}
Because each connection is managed by its own RuleSet and associated quantum resources, nodes must maintain a small but critical amount of connection state, including active RuleSets and their associated qubits.
While this statefulness has implications for scalability (discussed in \cref{sec:architecture-memory:internetworking-qrna}), it appears unavoidable for robust and flexible quantum networking, particularly when enabling asynchronous and autonomous repeater operations.

\subsection{Two-pass connection setup}
\label{subsec:architecture-memory:two-pass-connection-setup}

To establish a connection in our architecture, we employ a two-pass protocol for establishing connections~\cite{vanmeterConnectionSetupQuantum2019}.
This process separates resource discovery and path selection from the distribution and activation of the operational rules.
The first pass, or outbound leg, is initiated by the node requesting the connection (the \emph{Initiator}).
As the connection request message traverses the network towards the intended \emph{Responder}, guided by underlying routing algorithms, it gathers information about the links along paths, such as available quantum resources, node capabilities, and estimated link fidelities.

Upon reaching the Responder, all the collected path and resource information is processed.
The Responder then takes on the crucial role of devising the specific RuleSets required for every intermediate node (and the end nodes) along the chosen path to collectively establish (distributed computation along the path) the end-to-end entangled state.
This centralized generation of RuleSets by the Responder allows for sophisticated optimization and coherent planning of the entire distributed protocol.
Following this, the second pass, or return leg, involves the distribution of these tailored RuleSets from the Responder back to each node along the selected path.
Once all nodes have received and installed their respective RuleSets, the distributed operations for the connection can commence, driven by local events and the logic outlined in the rules.
\begin{figure}
    \centering
    \includegraphics[width=\textwidth]{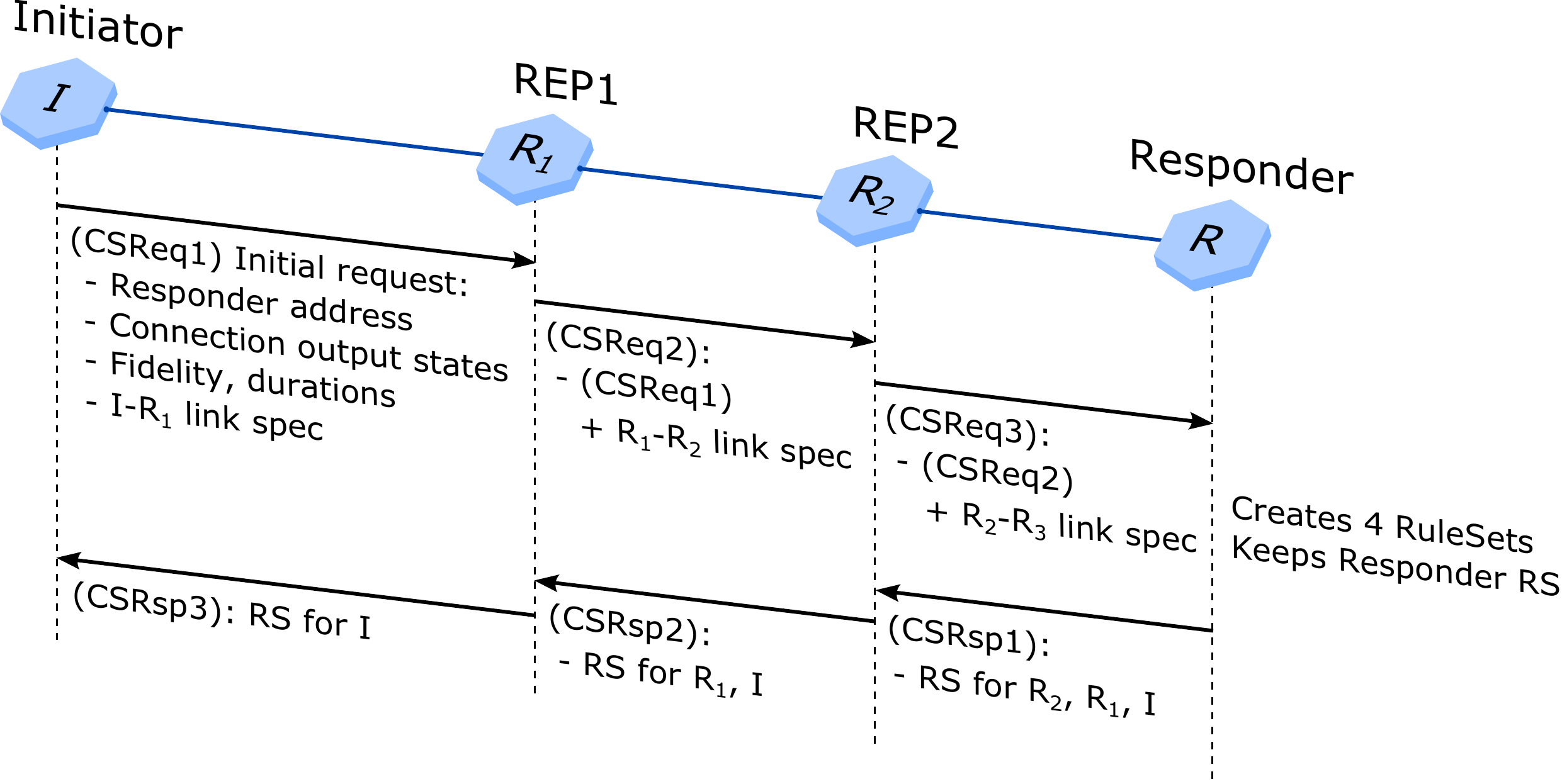}
    \caption[Two-pass connection setup within a network]{Two-pass connection setup within a network along a path.
    (From~\cite{rdv-qi-architecture}.)}
    \label{fig:architecture-memory:two-pass-connection-setup}
\end{figure}

This two-pass connection setup is illustrated for a single network in~\cref{fig:architecture-memory:two-pass-connection-setup}.
As in the classical Internet, we expect that most connections will be initiated by a client node reaching out to a server.
Placing the responsibility for RuleSet creation with the Responder (often the server) provides a natural point for service providers to innovate.
By developing more efficient or robust RuleSets, providers can offer superior connection performance (e.g., faster setup, higher fidelity, better resource utilization), creating a basis for competitive differentiation.

\section{Defining quantum network protocols with RuleSets}
\label{sec:architecture-memory:ruleset-definition-examples}

Now that we have a conceptual understanding of what a RuleSet is, this section delves into a detailed specification of the RuleSet protocol itself.
We will dissect the core components that constitute individual Rules, examine the mechanisms for quantum state representation and classical inter-node communication that RuleSets rely upon, and illustrate their operational semantics through practical examples.
This will provide a comprehensive understanding of how RuleSets function as a programmable framework for distributed quantum operations.

The focus of this specification is on the \emph{logical structure and operational semantics} of the RuleSet protocol.
This includes detailing the essential information content and parameters that messages must convey between nodes for the protocol to function correctly.
However, we will not prescribe specific low-level implementation choices such as particular message serialization formats (e.g., JSON, protocol buffers, or custom binary representations), precise bit-level encodings, or the underlying classical network transport mechanisms (e.g., TCP/IP or other datagram services) used to carry these messages.
For the purposes of defining the RuleSet protocol here, we assume the existence of reliable classical communication channels capable of delivering these logically defined messages between participating quantum network nodes.

\subsection{Fundamental constructs: Rules and RuleSets}
\label{subsec:architecture-memory:rules-and-ruleset-contruct}

At the core of the RuleSet framework are individual Rules, which define quantum operations, specify how resources are used, and state the conditions under which those operations are triggered.
A RuleSet is simply an ordered list of such Rules, constructed programmatically and associated with a unique RuleSet ID (or a connection ID).
Each Rule within the RuleSet also receives a Rule ID, allowing it to be referenced for inter-node communication.

The RuleSet life cycle starts and ends with the connection, thus a single connection will have a single RuleSet.

\subsection{State naming and management in rulesets}
\label{subsec:architecture-memory:state-naming-in-ruleset}

Managing and agreeing upon the specific qubits for joint operations like entanglement purification or swapping among collaborating nodes are critical responsibilities handled by RuleSets.
While we could draw inspiration from the IP architecture by employing network-wide unique naming for qubits, RuleSets have local execution contexts, making global naming for individual qubits unnecessary.

Instead, RuleSets use local logical names for qubits.
At the lowest level, physical qubits within a QNIC can be addressed by the local node using a tuple like \verb|<QNICAddress,QubitIndex>|.
However, this level of detail is not, and should not be, visible to the RuleSet logic.

Within a RuleSet instance on a given node, the RuleSet mechanism identifies a quantum resource primarily by its \emph{tag}---a label that categorizes the resource into a specific pool corresponding to its current role or stage in the protocol (as discussed in \cref{subsec:architecture-memory:ruleset-definition}).
To ensure an unambiguous and absolute ordering of multiple resources that may share the same tag (e.g., several Bell pairs awaiting purification), each resource, upon being allocated to a tag, is automatically assigned a \emph{sequence number} by the local RuleSet engine.
The combination \verb|<tag,seq no.>| thus serves as a distinct and locally unique key within that RuleSet instance for referencing a resource and associating it with relevant metadata, such as its tracked fidelity.
Actions specified within Rules typically reference these resources based on their tag and their relative order within that tag, for instance, by operating on the resource with the smallest sequence number (i.e., the oldest available in that pool).
This structured internal naming is vital for deterministic local operations and for preventing local operational mismatches that could lead to problems like the leapfrogging issue~\cite{rdv-quantum-networking-book} if not handled carefully at a higher protocol design level.

The RuleSet execution mechanism itself, therefore, only requires and manages this local \verb|<tag,seq no.>| system for resource identification within a single node's RuleSet instance; it does not impose any inherent restrictions or requirements on how shared entangled states are identified between nodes.
However, we \emph{recommend} a strategy wherein the RuleSet creator (e.g., the Responder during connection setup) devises RuleSets such that the \verb|<tag,seq no.>| is kept consistent across the RuleSets of all nodes involved in a particular joint operation on that resource.
The responsibility thus lies with the RuleSet author to ensure that RuleSets are designed---potentially using this consistent naming strategy---to correctly manage shared resources and achieve the desired end-to-end protocol behavior.
We believe this minimal local naming scheme provided by the RuleSet engine, combined with judicious design of the RuleSets themselves, is sufficient for implementing a wide range of complex quantum network protocols while allowing creativity with minimal restrictions on RuleSet creator parts, as will be illustrated with examples in \cref{sec:architecture-memory:ruleset-definition-examples}.

\subsection{Inter-node communication}
\label{subsec:architecture-memory:ruleset-internode-communication-messaging}

Distributed protocols, including RuleSets, rely on classical messages for coordination.
In the RuleSet framework, messages are sent as part of a Rule's Action clause, and messages are received through Condition clauses that can trigger subsequent rules.
These messages are crucial for coordinating their executions and making progress toward the successful distribution of end-to-end shared Bell pairs.

Each RuleSet coordination message includes the following fields.
\begin{itemize}
\item A \textbf{recipient address} indicating the destination node.
\item A \textbf{message context ID}  to help the receiving Rule Engine match the message to the correct logical context (this is distinct from the resource tag used to identify qubits).
\item A \textbf{structured content} (key-value pairs) payload, which may contain measurement results, Pauli frame updates, or other operational data.
\end{itemize}
While the underlying message framework may support a generic dictionary-like structure for its content, for clarity in protocol design and ease of RuleSet construction, we define a set of conventional message types or commands (as listed in \cref{table:architecture-memory:ruleset-internode-messages}).
These common patterns provide convenience for frequently occurring inter-node communications, making the RuleSets more human-readable and their intended logic easier to distinguish and manage.
The utility of these messages is evident in standard repeater operations such as entanglement purification and swapping, for heralding success and information about Pauli frame.
\begin{table*}[!t]
\centering
\caption[Inter-node messages in the RuleSet-based protocol framework]{A summary of the primary message types used for inter-node communication within the RuleSet protocol framework. These messages are sent as part of an Action clause and are used to trigger Condition clauses on remote nodes.}
\label{table:architecture-memory:ruleset-internode-messages}
\begin{tabularx}{\textwidth}{l p{3cm} p{3cm} X}
\toprule
\textbf{Name} & \textbf{Descriptive Name} & \textbf{Arguments} & \textbf{Comments} \\
\midrule
\multicolumn{4}{l}{\textbf{Remote Events (Message Transmission)}} \\
\midrule
FREE & Release Remote Qubit & Partner addr., resource IDs & Instructs a partner to release their half of a shared state. \\
\addlinespace
UPDATE & Pauli Frame Update & Partner addr., resource IDs, Pauli frame correction & Notifies a partner about a Pauli frame correction. \\
\addlinespace
MEAS & Measurement Outcome & Partner addr., resource IDs, result & Exchanges classical measurement outcomes between partners, primarily for purification. \\
\addlinespace
TRANSFER & Entanglement Transfer & Partner addr., resource IDs & Notifies a partner about the result of an entanglement swap. \\
\bottomrule
\end{tabularx}
\end{table*}

\begin{table}[!t]
\centering
\caption[Action clauses in the RuleSet-based protocol framework]{A summary of the Action clauses available within the RuleSet framework. Actions are grouped into local classical software operations and local quantum hardware operations.}
\label{tab:actions}
\begin{tabularx}{\textwidth}{l p{2.5cm} p{2.25cm} X}
\toprule
\textbf{Name} & \textbf{Descriptive Name} & \textbf{Arguments} & \textbf{Comments} \\
\midrule
\multicolumn{4}{l}{\textbf{Local Software Actions (Classical)}} \\
\midrule
SETTIMER & Set timer & Timer ID & Use with caution, as distributed race conditions can occur. \\
PROMOTE  & Promote qubits & Qubit IDs, Rule ID & Transfers control/ownership of qubits to a different rule or stage. \\
FREE     & Free qubits & Qubit IDs & Releases qubits back to the pool of unallocated resources. \\
SET      & Set variable & Variable ID, Value & Changes a RuleSet variable; useful for tracking statistics like measurement counts for tomography. \\
\midrule
\multicolumn{4}{l}{\textbf{Local Hardware Actions (Quantum)}} \\
\midrule
MEAS     & Measure qubits & Qubit IDs, Basis & Measures one or more qubits in a specified or randomly chosen basis. \\
QCIRC    & Apply quantum circuit & Qubit IDs, Circuit & Applies a general unitary operation. Used for BSM, purification, swapping, and encoding logical qubits. \\
\bottomrule
\end{tabularx}
\end{table}

\begin{table*}[!t]
\centering
\caption[Event triggers in the RuleSet-based protocol framework]{A summary of the primary Condition clauses that act as triggers within the RuleSet framework. These conditions are evaluated by the Rule Engine to initiate actions.}
\label{tab:conditions}
\begin{tabularx}{\textwidth}{l p{2.5cm} p{2.25cm} X}
\toprule
\textbf{Name} & \textbf{Descriptive Name} & \textbf{Arguments} & \textbf{Comments} \\
\midrule
\multicolumn{4}{l}{\textbf{Internal State Conditions}} \\
\midrule
CMP & Compare RuleSet Variable & Variable ID, operator (e.g., ==, >), value & Evaluated when the specified internal variable is updated. The rule triggers if the new value satisfies the comparison. Used to manage protocol logic, such as counting purification rounds or messages received. \\
\addlinespace
TIMER & Timer Expiration & Timer ID & Triggers a rule after a set duration has passed. Must be used with caution to handle timeouts and avoid race conditions in distributed protocols. \\
\midrule
\multicolumn{4}{l}{\textbf{Quantum Resource Conditions}} \\
\midrule
RES & Resource Availability & Number of pairs, Partner addr. (or wildcard), min. Fidelity & The primary trigger for quantum operations. Matches when a specified number of Bell pairs meeting a minimum fidelity constraint are available from a given partner. Used for purification, swapping, and final delivery to an application. \\
\bottomrule
\end{tabularx}
\end{table*}

\subsection{Condition clauses as triggers}
\label{subsec:architecture-memory:condition-clauses}

Condition clauses within Rules define the specific circumstances or triggers that must be met for a Rule's associated Action clause to be executed.
They can be conceptualized as predicates that enable transitions or operations within the distributed protocol governed by the RuleSet.
\Cref{tab:conditions} enumerates the types of Condition clauses supported.
A Condition clause might depend on various factors, including the arrival of specific messages, the successful generation or availability of one or more local entangled states meeting certain criteria (e.g., having a particular number of resource with specific tag label), or internal events like timeouts.
For instance, a Rule designed to deliver an entangled pair to an application might have a Condition clause that matches a single Bell pair labeled with a terminal tag indicating readiness.
More commonly, Condition clauses need to match two entangled states simultaneously: for example, two Bell pairs shared with the same neighbor node for purification or two Bell pairs connecting to different remote nodes for entanglement swapping, and these can be realized by looking at the tags.
The precise logic of the condition ensures that actions are only taken when the necessary quantum and classical state prerequisites are fulfilled.

\subsection{Action clauses for operations}
\label{subsec:architecture-memory:action-clauses}

Action clauses define the operations a node performs when the conditions of a Rule are satisfied.
At this point, the Rule is guaranteed exclusive access to any quantum resources it is meant to use, as all resources must be reserved through the prerequisite Condition clause.
Actions may include local quantum operations (e.g., applying quantum circuits), the generation of classical messages (to other nodes or to self), and updates to local classical state.

To ensure correctness and enable automated protocol validation, Action clauses must manage all resources they consume.
This means each allocated quantum resource must, by the end of the Action, be either: (1) returned to the free resource pool, (2) assigned a new tag to indicate its updated role in the protocol, or (3) retained under its existing tag with updated metadata.
Likewise, messages that triggered the Rule are considered consumed and must be either discarded or incorporated into outgoing messages or state changes.

Although the logic inside an Action clause can be expressive, its operations are constrained to a well-defined set of primitives supported by the Rule Engine (see \cref{tab:actions}) and are deliberately not Turing-complete.
These constraints, which include operations like applying a local quantum circuit (QCIRC) or measuring qubits (MEAS), are a key architectural choice.
By restricting the complexity of Actions, it becomes easier to analyze the protocol's properties, such as correctness, termination, and deadlock, and to prevent common bugs.
Because many of these primitive quantum operations yield classical information, generating classical messages is an essential and frequent outcome of Action clauses.

\subsection{Ruleset execution logic and resource management}
\label{subsec:architecture-memory:ruleset-execution-resources}

The operational correctness of RuleSet-based protocols depends on two core properties: deterministic execution and precise resource accounting.
To ensure predictable behavior, our model enforces that only one Rule fires at a time, even if multiple Rule Condition clauses are simultaneously satisfied.
Rule selection proceeds by evaluating Condition clauses in a strict, predefined order---typically determined by a unique Rule ID.
Once a Rule is selected, its Action clause executes atomically and to completion.

Although we assume that executing an Action clause incurs negligible latency compared to the delays from quantum operations or message transmission, this sequential execution model introduces potential performance limitations.
In high-throughput regimes---where quantum resources are produced or messages arrive faster than actions can be executed---this could become a bottleneck.
Addressing such scaling challenges while preserving deterministic behavior remains an open avenue for future work.

Each RuleSet instance maintains an internal resource table (or equivalent state tracker) to monitor the current state of all relevant local and remote resources (messages).
This enables Rule Conditions to express complex dependencies---for example, waiting for a classical message from another node while simultaneously checking the availability of a specific class of Bell pairs.
The table is updated dynamically as resources are generated, consumed, tagged, or modified by executed actions, and as messages are received.

To avoid conflicts and ensure independent progress of concurrent connections, we enforce resource encapsulation; once a quantum resource is assigned to a RuleSet instance (i.e., a connection context), it is not shared with other RuleSets.
This strict partitioning prevents interference between protocols running in parallel on the same node.
As will be discussed in the context of multiplexing strategies (\cref{subsec:architecture-memory:multiplexing}), such isolation allows each RuleSet to operate autonomously on its subset of resources and advance according to its local protocol logic.

\subsection{Illustrative examples of RuleSets}
\label{subsec:architecture-memory:rulesets-concrete-examples}

To make the operational principles of RuleSets more concrete, I provide simplified examples of how Rules can be defined for common quantum network tasks below.
These examples illustrate the event-driven nature of RuleSet execution and show how quantum resources transition through the stages of a protocol.
A conceptual illustration of a simple, sequential protocol for a four-hop connection is shown in \cref{fig:example-ruleset-4-hop}.
In this example, purifications are performed on the link-level Bell pairs, again after the first swap, and a final time on the end-to-end pair before it is consumed by a tomography application.
Note that the final tomography block is merely illustrative and could be replaced by a rule that delivers the resource to any quantum application.

\begin{figure}
    \centering
    \includegraphics[width=\textwidth]{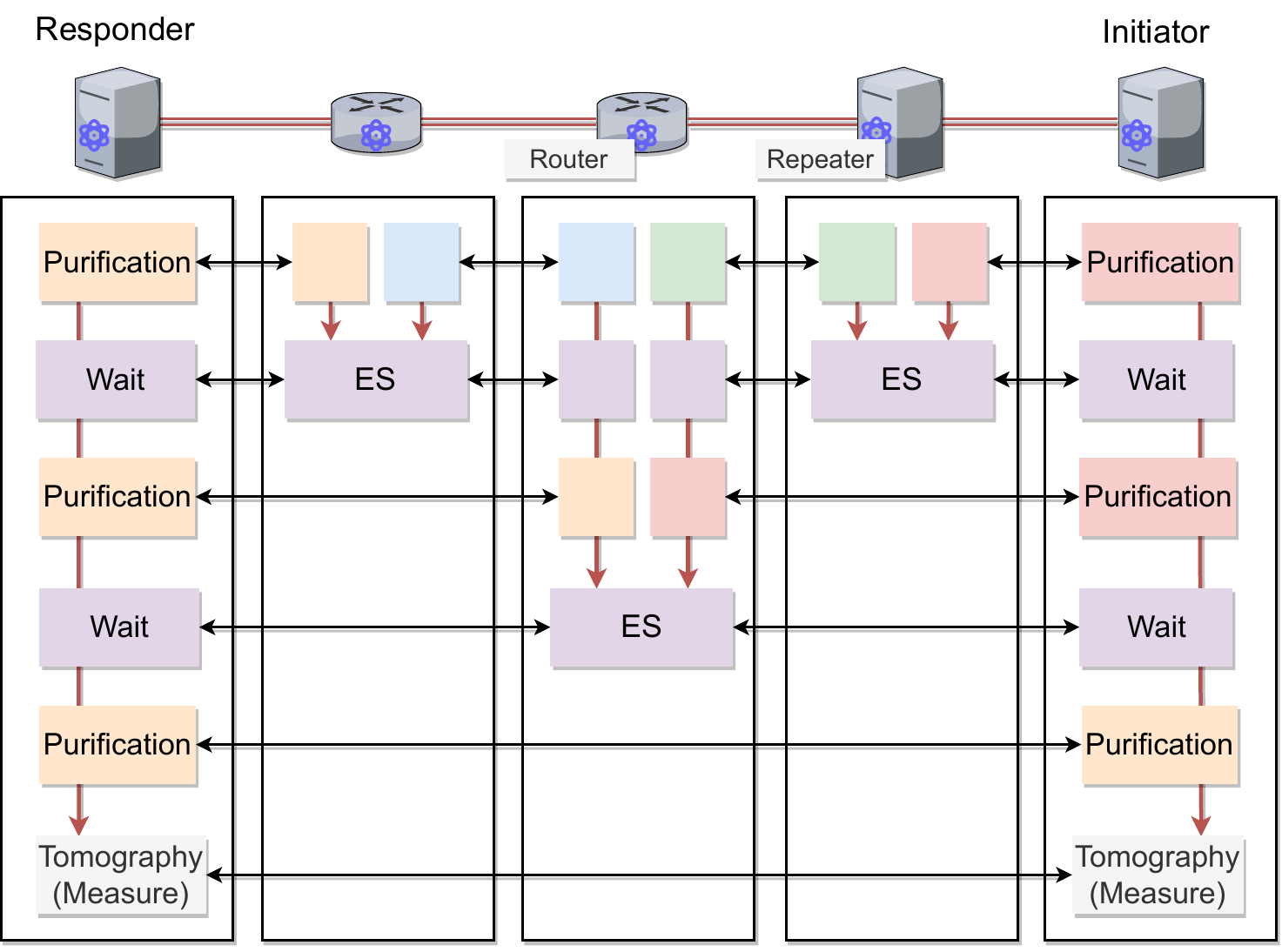}
    \caption[A conceptual depiction of RuleSets installed in nodes for a four-hop connection]{A conceptual illustration of how RuleSets could be installed for four-hop connection to create end-to-end Bell pairs. The RuleSets show flexibility by indicating purifications to be performed at three different stages of the protocol: at the link-level; before the second-swap; and end-to-end. (From~\cite{cocori-quisp}.)}
    \label{fig:example-ruleset-4-hop}
\end{figure}
While this example is simplistic, it demonstrates a general principle: once the desired flow of a quantum resource is defined---from its link-level creation, through intermediate processing steps like purification and swapping, to its final consumption---a corresponding sequence of RuleSet blocks can be constructed to implement it.
More complex protocols, such as those that transform Bell pairs into graph states~\cite{clement-graph-state}, could be realized with a similar modular approach.
For instance, to accommodate multiple rounds of purification at the link level, several purification blocks could be chained sequentially in the RuleSet.

However, more advanced purification strategies like pumping or banding require more careful construction of Rules.
An illustration of the conditions and actions for entanglement purification pumping is shown in \cref{fig:example-ruleset-pumping}.
In this model, the resource flow is not strictly sequential, as rules for different purification rounds must compete for the same pool of lower-level resources.
To ensure correctness, the ordering of Rule evaluation becomes critical; for example, a Rule for a round-$k$ purification must be evaluated before a Rule for a round-$(k-1)$ purification to ensure it gets priority access to the necessary resources.

\begin{figure}
    \centering
    \includegraphics[width=\textwidth]{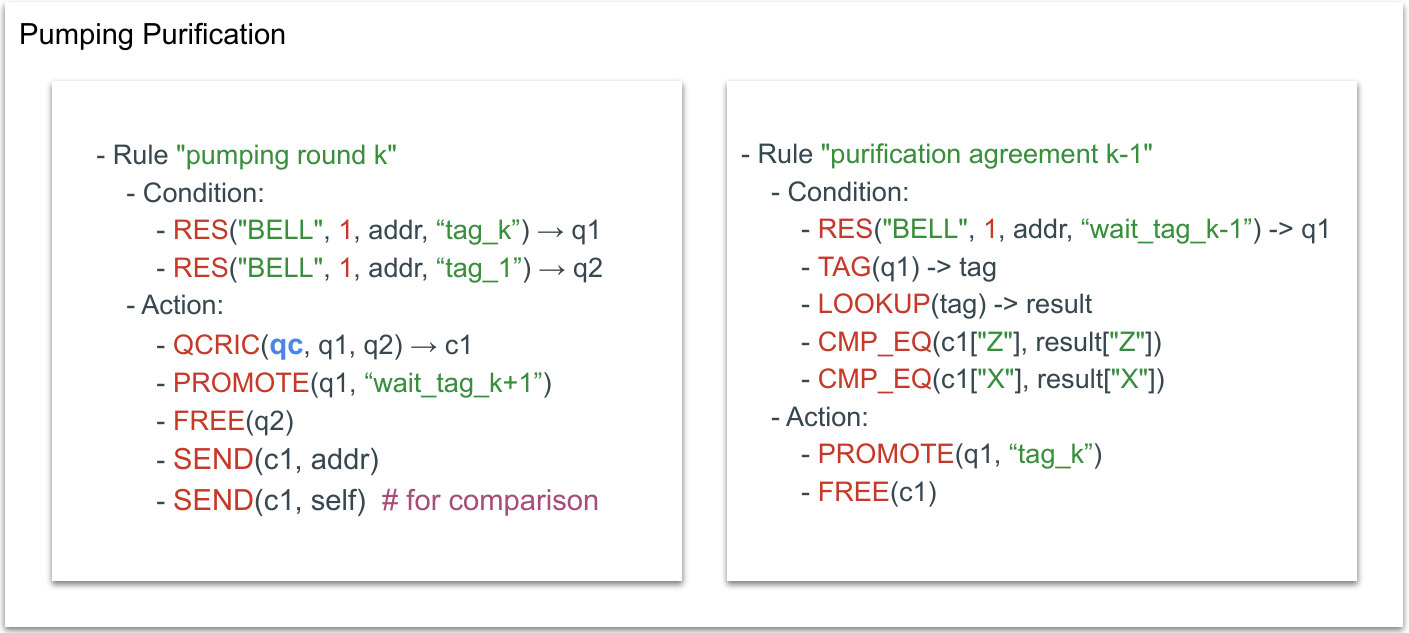}
    \caption[A conceptual depiction of purification pumping Rules]{
        A conceptual depiction of a RuleSet for entanglement purification pumping.
        The diagram illustrates two key features of the protocol: the competition for base-level resources between different purification rounds (left), and the two-step process of a local action followed by a rule that waits for a heralded message to confirm the outcome (right).}
    \label{fig:example-ruleset-pumping}
\end{figure}
A key implementation detail highlighted in this example is that any purification attempt is a two-step distributed process.
First, a Rule performs the local quantum operation and moves the resource to a waiting state (e.g., ``wait\_tag\_k'' in \cref{fig:example-ruleset-pumping}).
A second Rule is then triggered upon receiving the partner's heralded message.
If the message content agrees with the local measurement results, this second rule confirms a successful outcome and promotes the resource to the next stage of purification.

\subsection{Summary and implications of RuleSet-based semantics}
\label{subsec:architecture-memory:rulesets-summary-implication}

In summary, the RuleSet-based mechanism provides a structured and expressive framework for defining and managing the complex connection semantics required in quantum repeater networks. By decomposing distributed protocols into sets of local, event-driven condition-action rules, this approach facilitates significant node autonomy and asynchronous operation, which are crucial for practical network efficiency and responsiveness. The explicit inclusion of local state management, including quantum states and Stage-specific variables, allows RuleSets to capture the nuanced requirements of quantum protocols like entanglement purification and swapping.

The clear separation of concerns---with RuleSets defining the "program" for a connection, distinct elements for state naming, inter-node messaging, condition evaluation, and prescribed actions, all orchestrated by a well-defined two-pass setup protocol---promotes modularity and aids in reasoning about protocol correctness and performance. Furthermore, the ability to centrally devise RuleSets (e.g., at the Responder) offers opportunities for optimization and adaptation of protocols based on network conditions or application requirements.

While this framework offers considerable flexibility, it also highlights the inherent statefulness of quantum repeater nodes, a factor with significant implications for resource management and scalability. Assessing the resource footprint (both quantum memory and classical processing/storage for RuleSets) under various load conditions remains an important area for ongoing research. Nevertheless, the RuleSet paradigm provides a powerful abstraction for taming the complexity of distributed quantum computation inherent in establishing end-to-end entanglement, paving the way for more robust and capable quantum networks.

\section{Networking}
\label{sec:architecture-memory:networking}

Building upon the RuleSet mechanisms for establishing and managing individual connections, this section broadens the scope to address network-level operational considerations.
We will explore how to construct and operate robust quantum networks in complex topologies, considering diverse traffic patterns and the interplay of various network components.
This involves defining the roles of different node types, strategies for routing and resource allocation across the network, and fundamental mechanisms for security and resource accounting.

\subsection{Node types}
\label{subsec:architecture-memory:networking}

The architecture of a quantum network is defined by its constituent nodes and their specialized functions.
For clarity in describing our system, we group these nodes according to their primary contributions to network operation.
In our architecture, we classify nodes into three main types: end nodes, for application interaction; repeater and router nodes, for extending entanglement and path management; and support nodes, for auxiliary operational tasks.

End nodes represent hosts that wish to execute a quantum application such as quantum key distribution, secret sharing and blind quantum computation~\cite{broadbent-bfk-protocol,fitzsimonsPrivateQuantumComputation2017}.
The technological maturity required of an end node heavily depends on the desired application.
There are four major kinds of end nodes:

\vfill{}

\begin{wrapfigure}{o}{0.15\textwidth}
\includegraphics[width=0.9\linewidth]{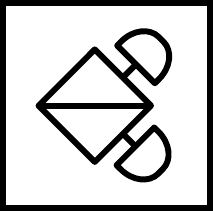}
\label{fig:icon-meas-node}
\end{wrapfigure}
\textbf{A measurement (MEAS) node} is the most basic type of quantum end node, designed primarily for protocols that do not require quantum state storage.
Its core capability is to receive individual photons and perform measurements on them in at least two different bases.
Lacking quantum memory, MEAS nodes are well-suited for applications like quantum key distribution (QKD) or as simple terminals in certain forms of secure delegated computation protocols~\cite{morimaeBlindQuantumComputation2013,fitzsimonsPrivateQuantumComputation2017}, typically interacting with the network in a synchronous manner where measurement results directly yield classical data.

\vfill{}

\begin{wrapfigure}{o}{0.15\textwidth}
\includegraphics[width=0.9\linewidth]{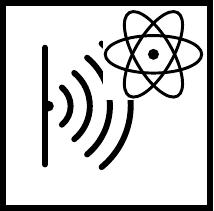}
\label{fig:icon-snsr-node}
\end{wrapfigure}
\textbf{A sensor (SNSR) node} is a specialized end node designed to utilize entangled states, often shared with distant parties, for high-precision measurements of physical quantities, for clock synchronization tasks, or for distributed sensing tasks~\cite{proctorMultiparameterEstimationNetworked2018,degen-sensing,proctor-quantum-sensing,giovannetti-advances-in-metrology,ge-distributed-metrology}.
These nodes typically feature limited quantum memory to hold working qubits (e.g., one half of an entangled pair) and possess specific quantum processing capabilities tailored for sensing protocols.
Such capabilities include performing joint measurements, like Bell State Measurements (BSMs), between their stored qubits and photons that have interacted with the environment~\cite{gottesman-longer-baseline-telescopes,huang-imaging-stars-with-qec}.
While an SNSR node's internal processing is specialized, certain sensing applications may also necessitate high-rate entanglement generation from the network to achieve desired performance.
For SNSR nodes, precise timing information is almost invariably a critical component of the service they provide or require, and their operation typically culminates in outputting classical data that corresponds to the sensed phenomenon.

\vfill{}

\begin{wrapfigure}{o}{0.15\textwidth}
\includegraphics[width=0.9\linewidth]{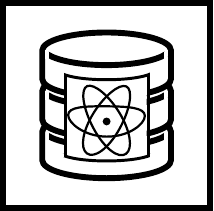}
\label{fig:icon-store-node}
\end{wrapfigure}
\textbf{A store (STOR) node} is a specialized end node whose primary function is to serve as a high-fidelity quantum data repository.
Its core capabilities are the long-term storage of quantum states---often prepared and teleported from other locations---and the ability to teleport these states out on demand.
While a STOR node does not require a universal gate set, it must support certain gates (e.g., Clifford gates) for active quantum error correction to preserve the stored quantum data.
This includes using quantum error-correcting codes to protect against decoherence, along with mitigation strategies for correlated errors from events such as cosmic ray strikes~\cite{martinisSavingSuperconductingQuantum2021,saneFightFlightCosmic2023,xuDistributedQuantumError2022a,vepsalainenImpactIonizingRadiation2020a,wuMitigatingCosmicRaylike2025,liCosmicrayinducedCorrelatedErrors2025a}.
In a network context, STOR nodes may function as data servers, enabling asynchronous applications where valuable states are prepared and stored for later retrieval.

\clearpage
\begin{wrapfigure}{o}{0.15\textwidth}
\includegraphics[width=0.9\linewidth]{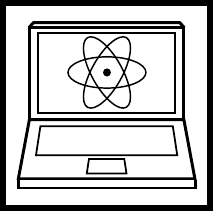}
\label{fig:icon-comp-node}
\end{wrapfigure}
\textbf{A computational (COMP) node} represents a full-fledged quantum processing endpoint within the network.
Equipped with quantum memory and additional algorithmic qubits, it can store, manipulate, and perform complex computations on quantum states received from the network or generated locally.
COMP nodes support a wide range of advanced quantum network applications~\cite{ambainisMultipartyQuantumCoin2004,taherkhaniResourceawareSystemArchitecture2017,mayersUnconditionalSecurityQuantum2004,christandlQuantumAnonymousTransmissions2005}, including distributed quantum algorithms, more general forms of blind quantum computation~\cite{broadbent-bfk-protocol,fitzsimonsPrivateQuantumComputation2017,mahadevClassicalHomomorphicEncryption2023,dulekQuantumHomomorphicEncryption2018}, and potentially fault-tolerant quantum computing~\cite{shapourianQuantumDataCenter2025,sutcliffeDistributedQuantumError2025,yoderTourGrossModular2025b,kimFaultTolerantMillionQubitScale2024}, often requiring asynchronous interfaces to coordinate their local quantum workloads with network operations.

Next, we consider quantum repeater and router nodes, which are essential for overcoming distance limitations and managing network paths.
We classify these into three primary types, briefly summarized here (for a more detailed foundational discussion, please refer to~\cref{sec:quantum-networks:repeaters}).

\begin{wrapfigure}{o}{0.15\textwidth}
\includegraphics[width=0.9\linewidth]{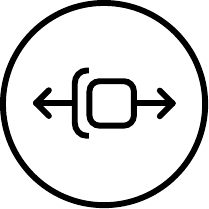}
\label{fig:icon-rep1}
\end{wrapfigure}
\textbf{A first-generation repeater (REP1)} is a network node with two quantum interfaces, whose main task is to extend entanglement.
Its primary operations include generating link-level Bell pairs with its neighbors, performing entanglement swapping to connect these segments, and managing errors through heralded entanglement purification on physical qubits along the connection path.

\begin{wrapfigure}{o}{0.15\textwidth}
\includegraphics[width=0.9\linewidth]{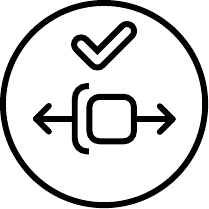}
\label{fig:icon-rep2}
\end{wrapfigure}
\textbf{A second-generation repeater (REP2)} also focuses on entanglement swapping to bridge distances but generally requires a larger quantum memory capacity and higher fidelity local quantum operations than a REP1.
It utilizes quantum error correction (QEC) to manage operational errors in conjunction with heralded link-level entanglement generation.
Instead of, or in addition to, purifying link-level physical Bell pairs, a REP2 node operates on logical qubits, where quantum information is encoded across multiple physical qubits to protect it from errors.
This approach inherently demands more sophisticated hardware and advanced computational capabilities for encoding, decoding, and error correction routines during swapping.

\begin{wrapfigure}{o}{0.15\textwidth}
\includegraphics[width=0.9\linewidth]{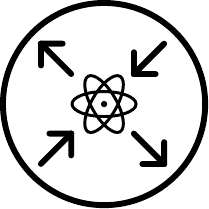}
\label{fig:icon-rtr-node}
\end{wrapfigure}
\textbf{A quantum router (RTR)} is a more complex and versatile node, possessing all capabilities of a quantum repeater and typically featuring three or more quantum network interfaces, enabling it to make sophisticated path selection decisions in complex topologies.
Architecturally, an RTR may consist of multiple line cards and a quantum backplane, allowing it to run a full suite of network operation protocols.
Beyond basic repeating functions like entanglement swapping, an RTR can govern network borders potentially interfacing between different repeater generations or technologies, participate in generating multipartite entangled states if the network provides such a service~\cite{clement-graph-state,bugalhoDistributingMultipartiteEntanglement2023,fischerDistributingGraphStates2021,fanOptimizedDistributionEntanglement2025}, and may act as a Responder in connection setups by rewriting or generating new RuleSets for different network domains (as will be discussed further in~\cref{sec:architecture-memory:internetworking-qrna}).

The last type, supporting nodes, provide essential functions that enable or enhance the operations of end nodes and repeater nodes in the entanglement distribution process.
Key types of support nodes we list are as follows.

\begin{wrapfigure}{o}{0.15\textwidth}
\includegraphics[width=0.9\linewidth]{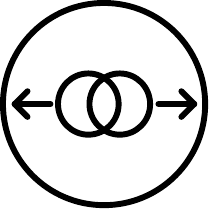}
\label{fig:icon-epps-node}
\end{wrapfigure}
\textbf{An entangled photon pair source (EPPS)} is a device dedicated to generating pairs of entangled photons, commonly through processes like Spontaneous Parametric Down-Conversion (SPDC).
These entangled photons are then typically distributed over quantum channels to be captured or measured at link endpoints, forming the initial resource for entanglement-based protocols.
EPPS nodes can be deployed in various scenarios, including terrestrial fiber links or free-space satellite-to-ground communication~\cite{khatriSpookyActionGlobal2021,fittipaldi-sattelite-quisp,haldarGlobalTimeDistribution2023,yinSatellitebasedEntanglementDistribution2017}.

\begin{wrapfigure}{o}{0.15\textwidth}
\includegraphics[width=0.9\linewidth]{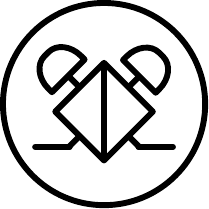}
\label{fig:icon-bsa-node}
\end{wrapfigure}
\textbf{A Bell state analyzer (BSA)} is a crucial component for performing projective measurements on two incoming photons, ideally projecting their combined state into one of the four Bell states.
BSAs are fundamental for realizing photonic entanglement swapping~\cite{zukowski-entanglement-swap}, a primary process that creates link-level entanglement, used particularly to convert memory-photon entanglement into memory-memory entanglement between distant quantum memories.
The efficiency and complexity of a BSA depend on the optical implementation and the specific photonic qubit encoding used.

\begin{wrapfigure}{o}{0.15\textwidth}
\includegraphics[width=0.9\linewidth]{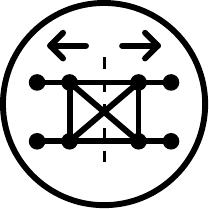}
\label{fig:icon-rgss-node}
\end{wrapfigure}
\textbf{A Repeater Graph State Source (RGSS)} is a specialized source that generates multipartite entangled photonic states, specifically tailored for all-photonic (memory-less) quantum repeater architectures~\cite{azuma-rgs,hilaire-rgs-optimizing-gen-time,buterakos-graph-generation,hilaire-logical-bsm}.
An RGSS typically distributes segments of the generated repeater graph state to adjacent network nodes, where subsequent measurements on these photonic qubits are performed to establish long-distance entanglement without relying on quantum memories.

\begin{wrapfigure}{o}{0.15\textwidth}
\includegraphics[width=0.9\linewidth]{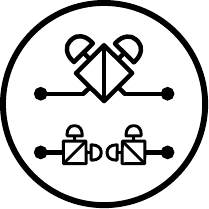}
\label{fig:icon-absa-node}
\end{wrapfigure}
\textbf{An advanced Bell state analyzer (ABSA)} represents a more sophisticated version of a BSA, particularly required in advanced all-photonic repeater protocols based on repeater graph states~\cite{azuma-rgs,hilaire-rgs-optimizing-gen-time,buterakos-graph-generation,hilaire-logical-bsm}.
Unlike basic BSAs, an ABSA must be capable of performing measurements on single or multiple photons in dynamically selectable bases.
The choice of measurement basis often depends on the outcomes of prior measurements within the network and the specific structure of the repeater graph state being utilized, implying more complex hardware and real-time classical control logic.

\begin{wrapfigure}{o}{0.15\textwidth}
\includegraphics[width=0.9\linewidth]{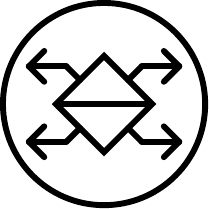}
\label{fig:icon-osw-node}
\end{wrapfigure}
\textbf{An optical switch (OSW)} is a device that can passively route photons from input optical fibers or paths to different output paths without performing measurements on them~\cite{mia-switch-design-paired-egress-bsa-pools}.
OSWs, which can be based on technologies like nanomechanical systems or nanophotonic circuits~\cite{mia-switch-design-paired-egress-bsa-pools}, can be integrated as components within other node types (e.g., routers or complex end nodes) or can function as standalone elements in the network to dynamically reconfigure optical pathways.

This list is by no means exhaustive but covers the main components of a quantum network.
The division into end, repeater and support nodes is not mutually exclusive, as there may be some overlap in functionality.
For example, the RGSS may be viewed as a type of repeater node as well, as it realizes the task of entanglement swapping and link-level generation for all-photonic quantum repeater (discuss in~\cref{sec:quantum-networks:repeaters,chapter:architecture-all-photonic}).
The RGSS and ABSA may participate in the connection planning process; the simpler BSA, on the other hand, is tasked only with notifying two nodes about the success of entanglement creation, and need not be visible to nodes farther away in the path.

\subsection{Routing}
\label{subsec:architecture-memory:routing}

Routing in quantum networks is the process of determining the communication path between a given set of end nodes.
This typically involves two distinct routes: a quantum path consisting of quantum nodes for entanglement distribution, and a separate classical path connecting the control mechanisms of these quantum devices~\cite{munroDesigningTomorrowsQuantum2022}.
It is generally assumed that nodes with quantum channel connections also possess a corresponding classical channel.
As discussed previously, a quantum router's primary task includes participating in finding these routes.
\Cref{fig:full-router} depicts such a router, capable of interfacing with multiple link types or different repeater generations and potentially possessing workspace memory for advanced processing like multipartite entanglement generation or complex error management.
Selecting a route, or connection path, from one end node to another is often termed the entanglement routing problem~\cite{abaneEntanglementRoutingQuantum2025,karRoutingProtocolsQuantum2023}.
\begin{figure}[t]
    \centering
    \includegraphics[width=\textwidth]{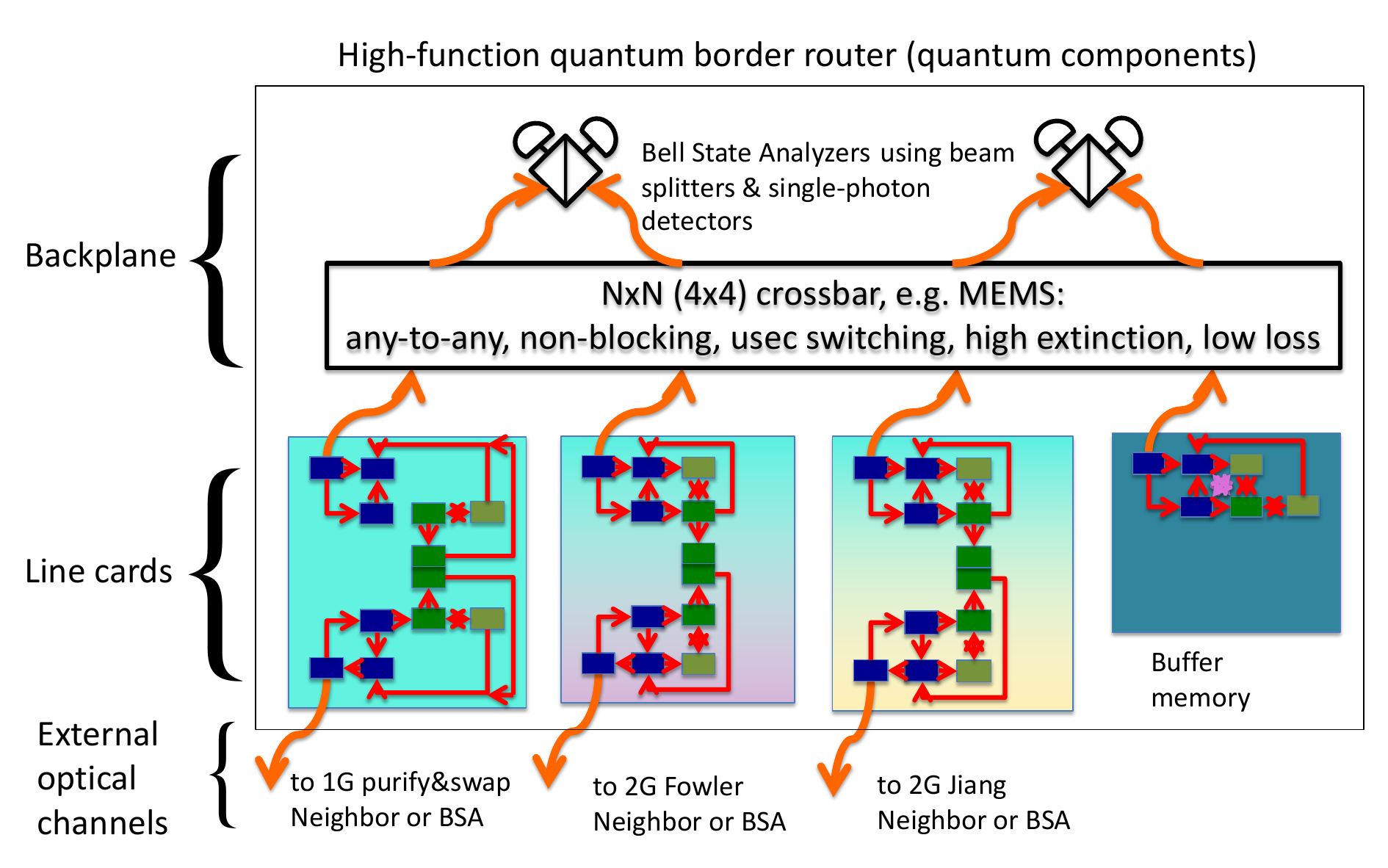}
    \caption[A schematic of a full quantum router]{A full quantum router with hardware architecture similar to today's commercial Internet routers will have QNICs (line cards) coupled via a backplane consisting of optical ports, an optical switch, and Bell State Analyzer measurement devices.
     Using the BSAs, the qubits in the backplane buffers at the top of the line cards are entangled while the transceiver qubits in the lower portion attempt to create entanglement with neighboring nodes.
     Once both backplane and neighbor entangled states are made, entanglement swapping is used within each line card to splice the long-distance entanglement.
    A number of steps in hardware complexity (and cost) will exist between the minimal configuration of and this one. (From~\cite{rdv-qi-architecture}.)}
    \label{fig:full-router}
\end{figure}

A common approach to quantum network routing is to adapt classical shortest-path finding algorithms, which shifts the primary challenge to defining an appropriate cost metric for use with algorithms like Dijkstra's.
This method, sometimes referred to as qDijkstra~\cite{vanmeterPathSelectionQuantum2013,difrancoOptimalPathQuantum2012a}, often employs a cost metric combining the fidelity and generation rate of link-level Bell pairs.
While the rate is relatively easy to characterize, fidelity is a more difficult metric to obtain accurately in practice.
It requires constant link monitoring, as noise sources in quantum links can drift and change over time.
Quantum state tomography~\cite{darianoQuantumTomography2003,jamesMeasurementQubits2001,altepeterQubitQuantumState2004} provides an accurate but resource-intensive way to measure the fidelity of generated Bell pairs at the link level.
Consequently, developing lower-cost means of characterizing quantum states is an active research area~\cite{eisertQuantumCertificationBenchmarking2020a,joshua-distimator,ananda-distimator-sigmetrics,casapao-double-selection-distimator}, and the broader field of quantum state and process learning remains dynamic with many open questions~\cite{huangLearningQuantumUniverse2023}.
Despite these measurement challenges, including fidelity in link metrics allows route calculations to automatically account for the trade-off between high-rate, low-fidelity links and low-rate, high-fidelity ones.
The qDijkstra approach, using such metrics, has demonstrated good agreement between calculated path costs and simulated throughput across various paths with heterogeneous links~\cite{vanmeterPathSelectionQuantum2013}.

Although qDijkstra shows good results for certain link characteristics, it is not always optimal, and defining a universally effective and easily determinable cost metric remains an open problem.
Using only link fidelity and rate as costs is often insufficient, as the specific shape and characteristics of noise also play a vital role in the final end-to-end entanglement fidelity~\cite{difrancoOptimalPathQuantum2012a}.
Recognizing these limitations, recent research has begun to move beyond simple qDijkstra-like approaches, attempting to integrate multiple factors---such as the probabilistic nature of entanglement swapping, memory decoherence during waiting periods, and purification scheduling---into more comprehensive global cost metrics~\cite{caleffiOptimalRoutingQuantum2017}.
These newer strategies draw from a diverse toolkit of optimization and algorithmic techniques. For example, some utilize linear programming and mixed-integer constrained linear programming~\cite{zhaoE2EFidelityAware2022,chakrabortyEntanglementDistributionQuantum2020a,zengMultiEntanglementRoutingDesign2022}, while others employ greedy algorithms~\cite{pant-routing-entanglement,chakrabortyDistributedRoutingQuantum2019a,gyongyosiDecentralizedBasegraphRouting2018}.
Furthermore, AI-based approaches~\cite{gyongyosiEntanglementGradientRoutingQuantum2017,leDQRADeepQuantum2022} and various multi-path routing methodologies~\cite{pant-routing-entanglement,pirandola-e2e-rate-limit,acin-percolation,patil-distance-ind-rate,sutcliffeMultiuserEntanglementDistribution2023,halderConcurrentMultipathEntanglement2024,lopiparoQuantumAggregationTemporal2024,leone-multi-path-routing-qunet} are also being actively investigated.
Thus, defining appropriate and readily determinable cost metrics continues to be an important area of exploration for efficient quantum network routing~\cite{azumaNetworkingQuantumNetworks2025a}.

\subsection{Multiplexing and resource allocation}
\label{subsec:architecture-memory:multiplexing}

The term multiplexing can carry different meanings depending on context.
At the physical layer of a quantum network, it might refer to techniques such as wavelength-division multiplexing (WDM) or frequency-division multiplexing (FDM) in optical fibers or spatial multiplexing via multiple cores or modes~\cite{dharaSubexponentialRateDistance2021,chakrabortySpectrallyMultiplexedQuantum2025,fanQuantumEntanglementNetwork2025,khodadadkashiFrequencybinencodedEntanglementbasedQuantum2025,lago-riveraTelecomheraldedEntanglementMultimode2021,ruskucMultiplexedEntanglementMultiemitter2025,clarkEntanglementDistributionQuantum2023}.
However, in this thesis, multiplexing refers specifically to \emph{network-level resource multiplexing}: the dynamic allocation of shared quantum resources---such as quantum memories at repeater nodes, link-level Bell pairs, or intermediate entangled states---to multiple, concurrent end-to-end user connections.

Several multiplexing strategies have been proposed to manage these quantum network resources effectively~\cite{aparicio-master-thesis,aparicioMultiplexingSchemesQuantum2011,collins-multiplexing-of-memories}.
In a circuit-switching approach, resources along the entire path are reserved exclusively for a given request for the full duration of its lifetime.
Time-division multiplexing (TDM) alternates resource allocation across requests in a round-robin per Bell pair or time-sliced manner.
Buffer-space multiplexing takes a more dynamic approach, assigning available quantum memories to requests as they arrive, with per-connection buffers resembling the notion of network slicing~\cite{barakabitze5GNetworkSlicing2020}.
Finally, statistical multiplexing---inspired by best-effort forwarding in classical networks---allocates resources probabilistically, using heuristics based on request type, urgency, or priority.
This can lead to higher overall throughput under load by letting separate parts of the network operate independently while sharing bottlenecks~\cite{pathumsootModelingMeasurementbasedQuantum2020}.

Each of these schemes comes with trade-offs.
Circuit switching, while simple and reliable, underutilizes resources in high-load regimes.
Buffer-space and statistical multiplexing have shown higher aggregate throughput in simulation studies, especially in scenarios with heterogeneous traffic and competing demands~\cite{aparicioMultiplexingSchemesQuantum2011, aparicio-master-thesis}.
Statistical multiplexing in particular has been shown to outperform circuit switching by exploiting temporal variations in demand.
Xiao et al.~\cite{xiaoConnectionlessEntanglementDistribution2024} compared buffer-space and TDM approaches, although their study focused on throughput with sequential entanglement swapping and did not account for fidelity.
These investigations remain mostly limited to small-scale networks.
More recently, mixed-integer linear programming (MILP) techniques have been applied to jointly optimize routing and multiplexing~\cite{huangNearOptimalSwappingPurifying2024,liNarrowGapReducingBottlenecks2025}, suggesting a path toward more integrated control.

However, many open questions remain.
Current models often neglect the complexities introduced by traffic variability, fidelity constraints, or the interaction between routing and resource allocation.
There is concern that uncoordinated multiplexing could lead to phenomena similar to congestion collapse, where short-range connections monopolize limited entanglement resources, starving long-range paths.
Multiplexing decisions must therefore be made in conjunction with routing logic and security mechanisms---including authentication, authorization, and accounting (the AAA model, discussed next).
Non-blocking multiplexing protocols are preferable where possible, but any protocol that allows for extended resource occupancy must be governed by policies that consider user identity, priority, and possibly even payments or credit-based quotas.

\subsection{Authentication, authorization, and accounting}
\label{subsec:architecture-memory:aaa}

As previously discussed, the gap between available quantum resources and anticipated demand is likely to necessitate fixed or semi-fixed resource allocation in early-stage quantum networks.
Under such constraints, mechanisms for authentication, authorization, and accounting (AAA) will be essential components of the network architecture~\cite{fajardoDiameterBaseProtocol2012}.
These systems will not only manage access control and usage tracking but also help enforce policies that determine who can reserve scarce resources and under what conditions.

In the absence of a carefully designed AAA framework, economic factors alone may dictate access, potentially privileging users or institutions with the ability to pay the most.
To ensure equitable participation and avoid early monopolization, it will be important to explore AAA architectures that incorporate fairness criteria, priority based on mission-critical applications, or other policy-driven metrics beyond purely financial bids.

\subsection{Security}
\label{sucsec:architecture-memory:security}

Quantum mechanics offers the promise of unprecedented confidentiality between communicating parties---an advantage that has made quantum key distribution (QKD) a major focus in both theoretical physics and computer science.
However, the intense focus on QKD has shaped a somewhat narrow view of quantum network security, one that overlooks broader architectural and protocol-level vulnerabilities.
While it is true that quantum mechanics enables new capabilities---such as the detection of eavesdropping and certain forms of tampering---it also introduces entirely new attack surfaces~\cite{satohAttackingQuantumInternet2021,blakelyQuantumInformationSystem2024,jiangSidechannelSecurityPractical2024,zhouQuantumNetworkSecurity2022}.

Some of these threats are rooted in the physical nature of quantum systems.
For example, although QKD is in principle information-theoretically secure, real-world implementations are susceptible to side-channel attacks~\cite{fungPhaseremappingAttackPractical2007,jiangSidechannelSecurityPractical2024}, Trojan-horse attacks~\cite{sushchevTrojanhorseAttackRealworld2024,gisinTrojanHorseAttacks2006}, and vulnerabilities such as photon number splitting~\cite{lutkenhausQuantumKeyDistribution2002}, which arise due to imperfections in sources and detectors~\cite{lutkenhausFocusQuantumCryptography2009}.
These are not weaknesses of the protocol itself but of the hardware, software, and deployment environment---a fact that holds true for classical cryptographic systems as well.
Other classes of vulnerabilities stem from network architecture and protocol design choices, such as denial-of-service (DoS) attacks on key quantum channels or repeaters, and impersonation of nodes via classical control messages.
These threats highlight the need to consider both quantum-specific and classical attack vectors in any serious security model.

Moreover, distributed quantum applications---such as secure multi-party computation, quantum consensus protocols~\cite{taherkhaniResourceawareSystemArchitecture2017}, and leader election~\cite{taniExactQuantumAlgorithms2005,taniExactQuantumAlgorithms2012}---rely heavily on assumptions about trust, timing, and node availability.
Any deviation from these assumptions can be exploited by adversaries, potentially undermining the correctness of the protocol.
Thus, even in quantum networks, security is not just about encryption; it also involves resilience against disruption, impersonation, and strategic manipulation of the protocol environment.

All of this reinforces the critical importance of robust authentication and tamper-resistance mechanisms in early quantum network deployments.
Whether classical privacy is necessary in all contexts is still an open question, but given the history of security oversights in classical networks, it is clear that a comprehensive security framework should be in place early---ideally well before operational quantum networks scale.
These security concerns are closely intertwined with multiplexing strategies and AAA (authentication, authorization, and accounting) systems.
In particular, as discussed previously, any form of persistent or exclusive resource allocation (such as circuit switching) will inevitably raise issues of identity, access rights, and possibly even economic fairness---each of which must be secured against misuse and abuse.

\section{Internetworking and recursive architecture}
\label{sec:architecture-memory:internetworking-qrna}

\begin{figure}[!t]
    \centering
    \includegraphics[width=\columnwidth]{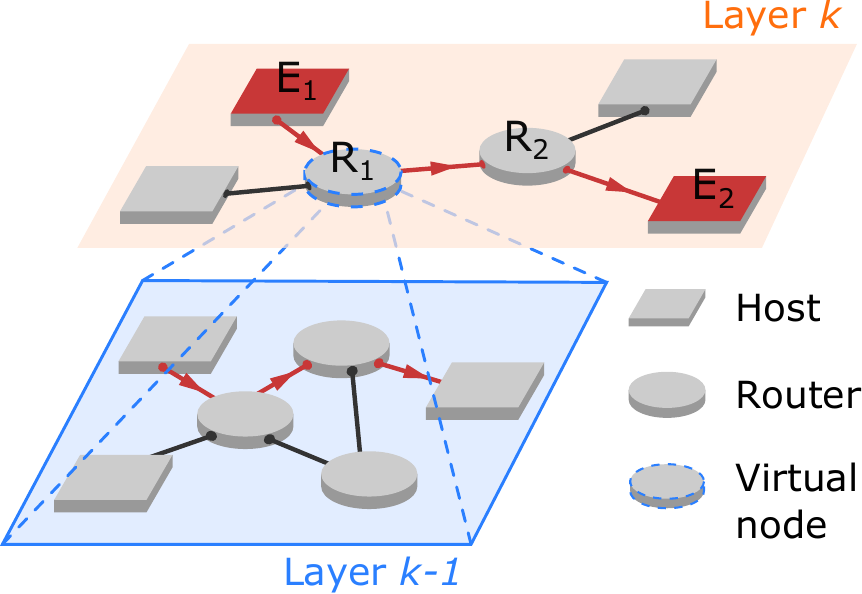}
    \caption[Recursive architecture illustration of QRNA]{\label{fig:qrna} QRNA adopts a fully recursive architecture that enables an entire network to be virtualized as a node.
    QRNA can extend down to the link layer, and interoperates across internally heterogeneous networks as long as they conform to the protocol at their borders~\cite{kozlowski-designing-network-protocol}. (From~\cite{rdv-qi-architecture}.)}
\end{figure}

The Internet is often idealized as a two-level architecture: external gateway protocols like BGP manage routing between networks, while internal gateway protocols such as OSPF and IS-IS handle routing within them.
In reality, the structure is far more intricate.
Tunneling protocols, virtual networks, switched Ethernet segments (which rely on Spanning Tree Protocol despite being ``link-layer''), and modern virtualization approaches like 5G network slicing~\cite{barakabitze5GNetworkSlicing2020} have shaped the Internet into a complex, multi-tiered system with ad hoc inter-layer interactions.

Given the opportunity to redesign the system from first principles~\cite{bacciottiniLeveragingInternetPrinciples2025a}---and with the benefit of hindsight---we would likely prefer a unified framework that explicitly supports clean abstraction boundaries and systematic composition across layers.

One such framework is the Recursive Network Architecture (RNA), proposed by Touch \textit{et al.}~\cite{touchDynamicRecursiveUnified2011,touchRNAMetaprotocol2008,touchRecursiveNetworkArchitecture2006a}.
Inspired by RNA, Van Meter, Touch, and Horsman proposed the Quantum Recursive Network Architecture (QRNA)~\cite{vanmeterRecursiveQuantumRepeater2011} as a foundational model for the Quantum Internet which we are building upon.
QRNA offers a scalable blueprint for connecting logically and physically heterogeneous networks while preserving autonomy, privacy, and security through recursion.

Recursion in QRNA permeates both naming (\cref{subsec:architecture-memory:state-naming-in-ruleset}) and routing (\cref{subsec:architecture-memory:routing}), producing a hierarchy of address spaces modeled as a directed acyclic graph.
This hierarchical structure supports scalable routing, localized policy enforcement, and layered isolation---key to privacy and trust.

In traditional internetworking, connections are typically either boundary-to-boundary (for transit) or boundary-to-end-node (for termination).
QRNA unifies these under a recursive abstraction; both are viewed as end-to-end connections at different levels of the stack.
As shown in \cref{fig:qrna}, a host node $E_1$ initiating a connection to $E_2$ at layer $k$ perceives a direct path via routers $R_1$ and $R_2$.
When the request reaches the border router, it is translated and embedded into layer $k-1$, where the path is resolved internally before being passed back up to layer $k$.
This nesting can continue recursively across multiple layers, allowing composition of arbitrarily complex internetworks.

\subsection{Connection setup with rewrite}
\label{subsec:architecture-memory:connection-setup-with-rewrite}

The recursive model must integrate with the two-pass connection setup described in~\cref{subsec:architecture-memory:two-pass-connection-setup}.
We achieve this via \textit{RuleSet rewriting} at recursion boundaries---specifically, border routers translate detailed link information from internal layers into abstracted virtual links for higher layers.

\begin{figure}
    \begin{center}
    \includegraphics[width=\textwidth]{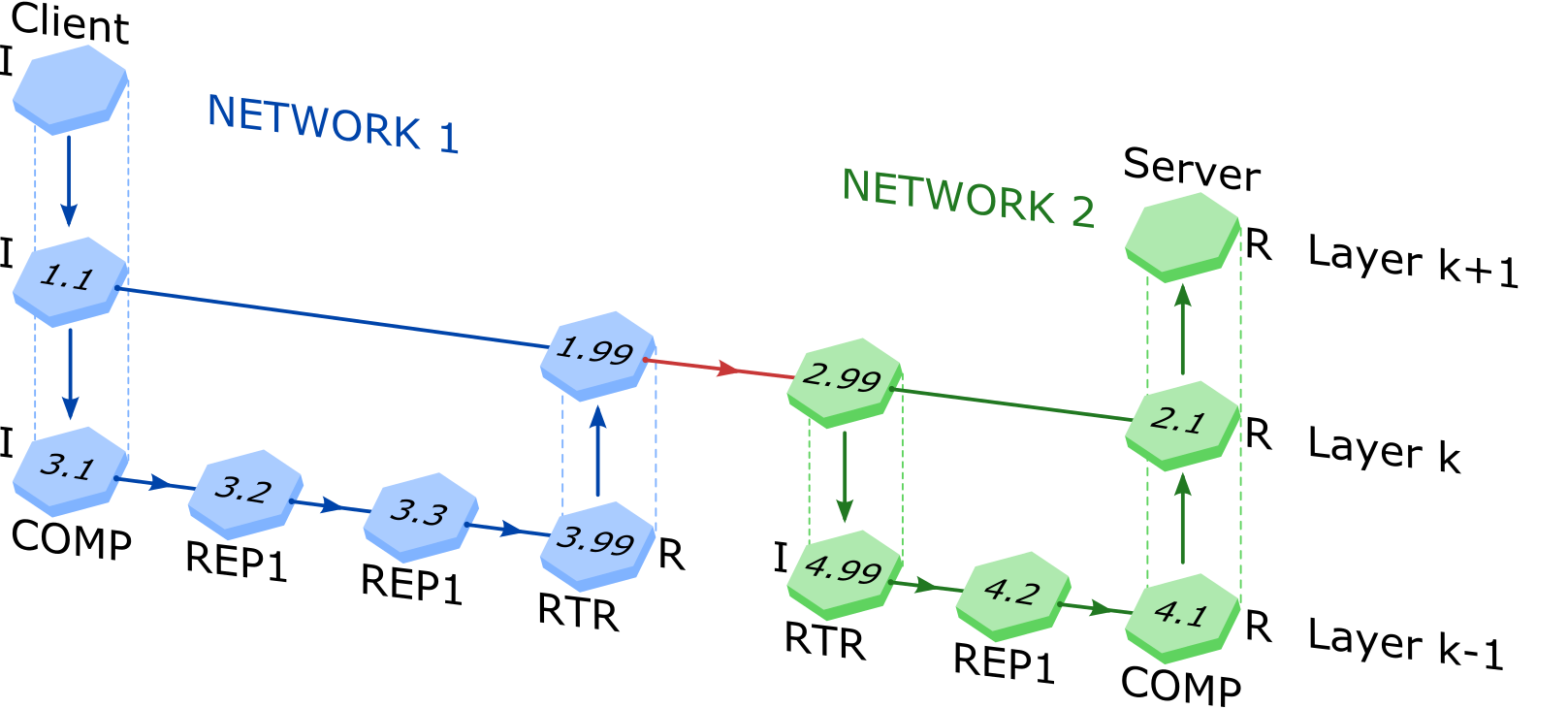}
    \end{center}
    \caption[Two-pass connection setup with rewrite in a recursive internetwork]{\label{fig:recursive-setup} Two-pass connection setup in a recursive internetwork.
    Arrows indicate how the connection request moves across and between layers.
    Each recursion layer encapsulates its own Initiator (I) and Responder (R), resulting in multiple I-R pairs at layer $k-1$. (From~\cite{rdv-qi-architecture}.)}
\end{figure}

This process is illustrated in \cref{fig:recursive-setup}.
To preserve network autonomy and scalability, each border router rewrites its segment of the path into a single logical hop.
This abstraction mimics how BGP hides internal topology while advertising route metrics.
The border router acts as a Responder to the upstream Initiator, reporting a performance metric (e.g., Bell pair generation rate at a given fidelity) that summarizes its internal network's behavior.

In the Quantum Internet, especially in first- and second-generation repeaters, connections are stateful distributed computations.
All intermediate nodes participate in operations such as purification and entanglement swapping, requiring connections to be pre-established and maintained via state along the entire path.

During setup, each node receives a Responder (destination) address and selects the next hop based on local policy.
If the received destination address is a virtual one (i.e., not a direct link-level neighbor), the node translates the connection into a request for the layer below. 
At each recursion level, new Initiator-Responder pairs are instantiated.
When a request reaches a Layer $k-1$ Responder (e.g., nodes 3.99 or 4.1 in \cref{fig:recursive-setup}), a provisional RuleSet is constructed.
This RuleSet governs local behavior and is distributed backward along the setup path.
Upon approval by the destination, the final RuleSet is confirmed and installed.

This recursive encapsulation provides several benefits:
\textbf{(i)} It bounds the scope of routing and state information, enhancing scalability and autonomy.
\textbf{(ii)} It allows each administrative domain to innovate locally within the RuleSet framework.
\textbf{(iii)} It enables systematic reasoning about complex paths via layered abstraction.
\textbf{(iv)} It supports interoperability between different repeater generations, enabling seamless deployment of new technologies~\cite{shota-interopability}.
\textbf{(v)} It facilitates interworking between fundamentally different network architectures~\cite{kozlowski-designing-network-protocol}.

\subsection{Routing, multiplexing, and AAA}
\label{subsec:architecture-memory:routing-multiplexing-aaa}

Routing in QRNA generalizes classical inter-domain models by requiring a local routing protocol at each recursion level.
While classical Internet routing relies on BGP and internal protocols like OSPF, QRNA uses full recursion, with each layer managing its own topology, policy, and performance metrics.

At higher recursion levels, path selection is typically guided by virtual link costs, often using the qDijkstra metric.
This abstraction enables scalable inter-network routing by treating entire networks as single links at higher layers.
When a network is virtualized in this way, it exposes performance metrics---such as entanglement generation rate or fidelity---that reflect internal policies.
For instance, a network may choose to perform additional purification internally to present a higher-fidelity virtual link to upper layers.
Prior work on purification scheduling and swap optimization~\cite{vanmeterSystemDesignLongLine2009,poramet-quanta-qcnc,poramet-quanta-extended,zangAnalyticalPerformanceEstimations2024,shchukin-optimal-entanglement-swapping,haldar-purify-swap-schedule,chenOptimumEntanglementPurification2024} has shown that performing purification closer to the physical link layer can significantly improve end-to-end throughput and fidelity.

At the lowest levels, within individual networks, routing strategies may vary and are not limited to distributed models like qDijkstra.
Centralized or policy-driven approaches can often yield better performance and may be preferable depending on local infrastructure and goals (see \cref{subsec:architecture-memory:routing}).

An intriguing direction for future research at higher recursion levels is the adaptation of segment routing techniques to quantum networks.
In classical networks, segment routing~\cite{filsfilsSegmentRoutingArchitecture2018,caiEvolveCarrierEthernet2014,carpaSegmentRoutingBased2014} enables a source to encode a sequence of forwarding instructions---called segments---directly into the packet, allowing for flexible, programmable routing without requiring per-flow state at intermediate nodes. Applied to QRNA, this could involve embedding segment identifiers corresponding to recursive virtual links, streamlining route computation and enabling richer policy expression, such as fidelity constraints or preferred purification strategies.

Multiplexing is more challenging.
As with classical networks, many concurrent connections must share infrastructure, and mechanisms for fairness, scheduling, and resource allocation are essential.
Unfortunately, the classical Internet provides limited guidance: while inter-domain QoS mechanisms have been discussed for decades, they remain sparsely deployed~\cite{clarkTussleCyberspaceDefining2002,xiaoInternetQoSBig1999,}.
As such, scalable, secure multiplexing remains a critical open research problem in quantum networking.

\section{Quantum routing software architecture (QRSA)}
\label{sec:architecture-memory:qrsa}

The operational capabilities of quantum repeater nodes---particularly the execution of complex, stateful protocols like RuleSets within the broader principles of the Quantum Recursive Network Architecture (QRNA)---are realized through a well-defined software framework. Classical control software plays a critical role in designing and deploying efficient repeater-based quantum networks.

To this end, we propose the Quantum Routing Software Architecture (QRSA), a modular and extensible framework that provides the essential functionality for managing quantum resources, executing distributed protocols, and interfacing with physical hardware. QRSA serves as the node-level realization of the QRNA principles discussed in earlier sections, with software components designed to reflect the abstractions introduced in RuleSets and the distributed coordination requirements of quantum networking.

As illustrated in~\cref{fig:architecture-memory-qrsa-diagram}, QRSA consists of five primary components: the Connection Manager (CM), Routing Daemon (RD), Hardware Monitor (HM), Rule Engine (RE), and Realtime Controller (RC). Together, these modules form the minimal functional configuration needed for a quantum repeater or end node to participate in a scalable quantum network. The following subsections describe each component in detail.

\begin{figure}[htbp]
  \centering
  \includegraphics[width=0.8\textwidth]{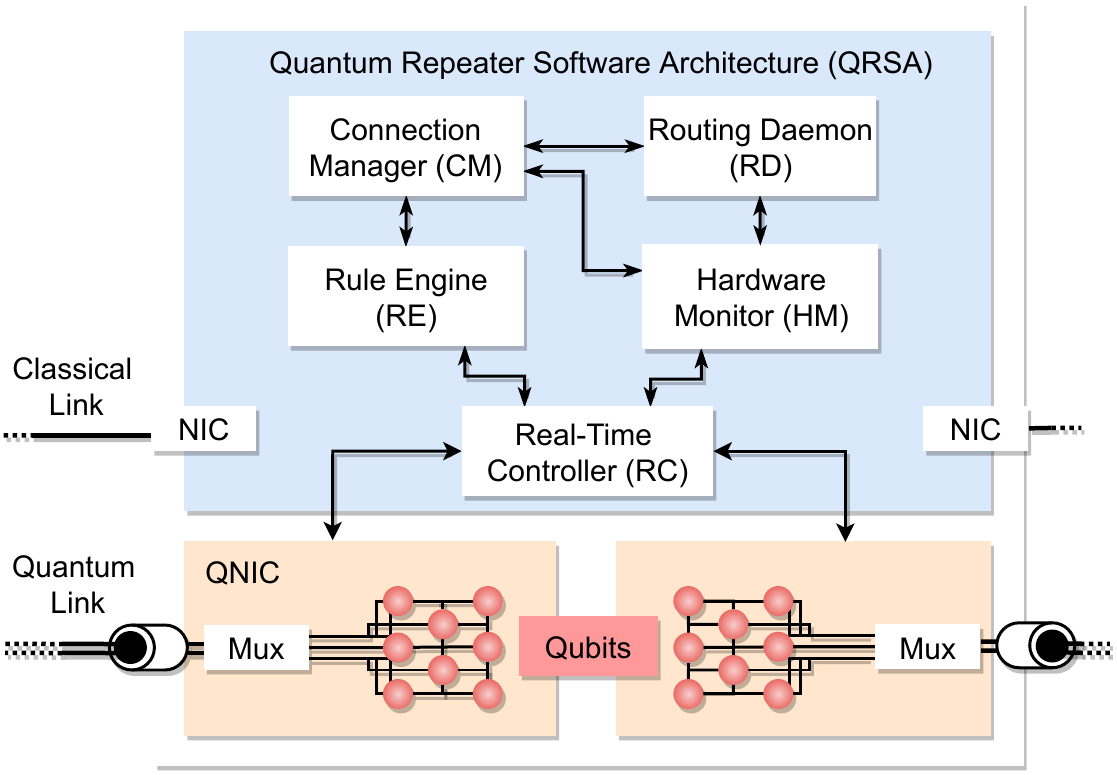}
  \caption[A depiction of quantum routing software architecture]{The Quantum Routing Software Architecture (QRSA) comprises five interacting software components: Connection Manager (CM), Routing Daemon (RD), Hardware Monitor (HM), Rule Engine (RE), and Realtime Controller (RC), managing both classical and quantum communication interfaces (QNICs). (From~\cite{cocori-quisp}.)}
  \label{fig:architecture-memory-qrsa-diagram}
\end{figure}

\subsection{Connection Manager}
\label{subsec:architecture-memory:connection-manager}

The Connection Manager (CM) is responsible for orchestrating connection establishment between an Initiator and Responder, spanning all intermediate nodes along the selected path. Its duties include handling incoming connection requests, assessing feasibility based on current resource availability and node policies, and coordinating the two-pass connection setup protocol detailed in~\cref{subsec:architecture-memory:two-pass-connection-setup,subsec:architecture-memory:connection-setup-with-rewrite}.

During the outbound pass of the connection setup, the CM at each intermediate node appends local link information---such as available QNICs and link characteristics---to the traversing request. Once the annotated request reaches the Responder node and is accepted, the CM's most critical role begins; generating the RuleSets for each node along the path~\cite{SfcaquaRuleSetSpec2024}.

The CM at the Responder uses the aggregated link information to determine whether the network can meet the requested connection parameters (e.g., target fidelity, number of Bell pairs). If feasible, it constructs RuleSets that encode the protocol logic and resource coordination for each participating node and distributes them during the return pass. If the connection cannot be supported, the CM returns a rejection to the Initiator.

\subsection{Routing Daemon}
\label{subsec:architecture-memory:routing-daemon}

The Routing Daemon (RD) maintains path information to other nodes in the network, focusing on quantum-capable links. It builds and updates routing tables by exchanging reachability or link-state information with peer daemons at neighboring nodes via a suitable routing protocol (see~\cref{subsec:architecture-memory:routing}). The RD integrates local link metrics provided by the Hardware Monitor, such as entanglement rate or fidelity, and reflects network topology constraints or policies.

The RD's information is vital during connection setup. When the CM receives a request to establish a connection, it interfaces with the RD to determine the next hop toward the Responder. This interaction can be a direct query, or it can be achieved by the CM reading from a shared data structure, such as a routing table, that the RD maintains and continuously updates. The level of detail in this data can vary depending on the routing algorithm, ranging from minimal reachability information to full path-level metrics.

\subsection{Hardware Monitor}
\label{subsec:architecture-memory:hardware-monitor}

The Hardware Monitor (HM) is responsible for tracking the performance and operational characteristics of local quantum hardware and link interfaces. Since link quality directly impacts the feasibility and efficiency of end-to-end quantum protocols, the HM provides timely updates on the status of each QNIC, including entanglement generation success probabilities, link-level fidelities, noise models, and memory decoherence parameters.

This information is consumed by both the Routing Daemon (RD) and the Connection Manager (CM). The RD uses it to update link cost metrics and inform routing decisions, while the CM relies on it to tailor RuleSets to prevailing link conditions. For example, if a link exhibits high loss or reduced fidelity, the CM may increase the number of purification rounds or opt for alternative purification circuits. In the case of second- and third-generation repeaters, the HM may also provide logical qubit fidelity data under different error-correcting code distances.

Continuously characterizing links, however, can be resource intensive---especially in dynamic environments where link characteristics fluctuate due to external factors. Moreover, some diagnostic procedures (such as quantum state tomography) may temporarily render a link unavailable during the measurement process, reducing overall network throughput.

Recent work suggests alternatives that can reduce this overhead. For instance, there are methods proposed for inferring link properties through passive observation of entanglement purification outcomes~\cite{ananda-distimator-sigmetrics,joshua-distimator,casapao-double-selection-distimator}. This avoids the need to suspend normal operation while still gaining useful information. Similar ideas have been applied to quantum error correction, where techniques have been developed to estimate noise properties from syndrome data~\cite{wagnerOptimalNoiseEstimation2021,wagnerPauliChannelsCan2022a}, offering relevance for second- and third-generation repeaters.

Developing efficient, non-intrusive methods for detecting and adapting to link-level changes remains an important open problem. Advances in this area would improve real-time responsiveness to hardware fluctuations, ultimately enhancing the robustness and performance of quantum networks.

\subsection{Rule Engine}
\label{subsec:architecture-memory:rule-engine}

The Rule Engine (RE) is the component responsible for executing RuleSets at each node. Once a RuleSet is delivered by the CM, the RE instantiates and manages its lifecycle, monitoring available resources and triggering actions when their associated conditions are met. The RE tracks the state of local qubits---via tags and sequence numbers as described in~\cref{subsec:architecture-memory:state-naming-in-ruleset}---as well as classical messages, which may originate from the network or other software components.

When a Rule's condition clause is satisfied, the RE executes the associated action clause. This may involve invoking quantum operations via the Realtime Controller, updating resource metadata (e.g., re-tagging or releasing qubits), or generating new classical messages. In this way, the RE drives the autonomous, event-driven logic of the repeater node in accordance with the RuleSet's semantic specification.

A key function of the RE is managing resource multiplexing, ensuring that qubits are allocated and released according to active RuleSets. It also participates in connection teardown---when an end node determines that a RuleSet has completed, a teardown message is propagated through the path to instruct intermediate nodes to clean up associated state. Designing efficient teardown and resource agreement protocols is non-trivial but beyond the scope of this thesis, we refer curious readers to check some preliminary work in~\cite{SfcaquaRuleSetSpec2024}.

\subsection{Realtime Controller}
\label{subsec:architecture-memory:realtime-controller}

The Realtime Controller (RC) is the lowest layer in QRSA, responsible for translating logical commands from the RE into precise control signals for the quantum hardware. It executes quantum operations such as gate applications, Bell measurements, qubit initialization, and photon emissions. Because many quantum operations are time-sensitive, the RC must meet strict real-time performance constraints.

The RC generates the appropriate physical signals (e.g., laser pulses, microwave pulses, or hardware triggers) and synchronizes timing with neighboring nodes during link-level entanglement generation. It also handles qubit re-initialization, measurement readout, and low-level coordination of hardware components. In this way, it serves as the operational interface between abstract protocol logic and the concrete capabilities of the quantum device.

\subsection{Component interactions and overall operation}
\label{subsec:architecture-memory:qrsa_discussion}

The five components of QRSA work in concert to support distributed quantum networking. The process typically begins with the Connection Manager receiving a connection request. It queries the Routing Daemon for the next hop, consults the Hardware Monitor for link quality, and determines if the path can support the desired protocol. If the Responder approves the connection, the CM generates and distributes RuleSets to all participating nodes.

At each node, the Rule Engine loads and executes the RuleSet, monitoring resource availability and triggering quantum operations through the Realtime Controller. The RC, in turn, interfaces directly with the hardware to perform gate operations, entanglement generation, measurements, and qubit resets. Classical messages generated during protocol execution are routed via the system's classical NICs and interpreted by the REs at destination nodes.

This modular software architecture is designed to be extensible, interoperable, and suitable for various generations of quantum repeater protocols. Each component has clearly defined responsibilities and interfaces, enabling clean abstraction layers and enabling research into protocol-level behavior, resource allocation, and network scheduling. QRSA also serves as the implementation target for the QuISP simulator (\cref{sec:architecture-memory:evidence}), which models these components in a realistic time-driven simulation environment. We believe QRSA represents a minimal, yet powerful, foundation for the implementation of scalable quantum repeater nodes.

\section{Simulation as evidence for architectural design}
\label{sec:architecture-memory:evidence}

The architectural framework presented in this thesis, comprising the Quantum Recursive Network Architecture (QRNA), RuleSets, and the Quantum Routing Software Architecture (QRSA), aims to provide a scalable and modular approach for designing with the ability to analyze quantum networks.
Throughout this chapter, we have described the architectural principles and protocol semantics in detail.
In this section, we turn to simulation as a means of evaluating the behavioral fidelity and practicality of these ideas.

Simulation plays two complementary roles: it allows us to reason about the correctness of protocol behavior under realistic conditions, and it enables empirical validation of our models through comparison with analytical predictions and alternative implementations.
In what follows, we first describe how RuleSets enable local reasoning about distributed behavior, then introduce the QuISP simulator~\cite{cocori-quisp,kaaki-ruleset-sim}, explain its timing and error models, and present the results of a comparative validation study~\cite{joaquin-michal-cross-validation-quantum-network-simulators} with another quantum network simulator, the SeQUeNCe simulator~\cite{wuSimulationsPhotonicQuantum2019,wuSeQUeNCeCustomizableDiscreteevent2021,wuParallelSimulationQuantum2024}.

\subsection{Reasoning with RuleSets and the role of simulation}
\label{subsec:architecture-memory:reasoning-and-simulation}

In an ideal world, the long-term validation of a network architecture would come through widespread real-world adoption.
However, many well-designed systems---whether processors, operating systems, or communication protocols---have failed to gain traction due to external factors unrelated to technical merit.
More importantly, retrospective success offers little help when evaluating architectures prospectively.
For quantum networks, where deployments are fledging and experimental feedback is currently limited, it is crucial to develop abstractions and tools that support structured reasoning and prospective validation.

A core motivation for RuleSets, introduced in \cref{sec:architecture-memory:ruleset-definition-examples}, is to provide a programmable and analyzable abstraction for specifying node behavior in a distributed quantum network.
RuleSets define node logic through event-driven condition-action rules, enabling modular and locally inspectable implementations of complex protocols such as entanglement swapping and purification.
Their structure supports reasoning about protocol semantics, resource state transitions, timing constraints, and inter-node coordination.

Beyond supporting the implementation of correct behavior, RuleSets also help identify failure scenarios.
Timing-sensitive issues such as leapfrogging or premature resource discard (see \cref{subsec:architecture-memory:rulesets-summary-implication}) can be addressed by carefully designing condition clauses and discard timers relative to classical messaging delays.
This makes RuleSets a practical framework for protocol design that balances abstraction with operational precision.

Nevertheless, this kind of design-time reasoning does not guarantee correctness in deployed systems.
The real behavior of a network depends on the execution semantics provided by the underlying QRSA software stack, including the Rule Engine, classical messaging, and physical resource control.
Bugs or mismatches may arise from several sources:
\begin{itemize}
    \item Misinterpretation of Rule Engine behavior by the RuleSet author;
    \item Implementation flaws in QRSA components (e.g., Connection Manager or Rule Engine);
    \item Incorrect assumptions about timing, failure handling, or hardware error rates;
    \item Unanticipated interactions or edge-case behaviors in concurrent execution.
\end{itemize}

To mitigate these risks, simulation plays a critical role.
It allows us to evaluate protocol behavior under noisy, asynchronous, and resource-constrained conditions---revealing issues that are difficult to anticipate analytically.
In this way, simulation complements human reasoning and provides an essential bridge between protocol design and practical implementation, supporting the iterative refinement of both logic and architecture.

\subsection{QuISP and the QRSA execution model}
\label{subsec:architecture-memory:quisp}

To support simulation-based study of RuleSets and QRNA, we developed \emph{Quantum Internet Simulation Package (QuISP)}, an open-source quantum network simulator implemented on top of the OMNeT++ discrete event simulation engine~\cite{vargaOVERVIEWOMNeTSIMULATION2008,vargaOMNeT2010a}.
QuISP is explicitly designed to reflect the QRSA model described in \cref{sec:architecture-memory:qrsa}, with a modular structure comprising the five core software components: the Connection Manager (CM), Routing Daemon (RD), Hardware Monitor (HM), Rule Engine (RE), and Realtime Controller (RC).

Each quantum node inside QuISP operates autonomously using RuleSets to guide its logic.
When a condition is satisfied (e.g., a message arrives, a resource becomes available), an associated action is triggered; this may involve performing a local quantum operation, sending a message, or tagging or discarding a qubit.
QuISP tracks both quantum and classical state over time, simulating the behavior of nodes in full accordance with QRSA semantics.

Importantly, while QuISP supports the execution of full RuleSet-based logic, our current use of the simulator has focused on small-scale networks, primarily linear repeater chains.
The primary goal has been to confirm that the RuleSet abstraction and QRSA model result in the expected low-level behavior in the presence of timing delays, resource constraints, and noise.
The simulator provides a platform for testing both protocol logic and architectural behavior under these controlled settings.
It is of interest to extend the use cases of QuISP to the study of emergent behaviors of large-scale quantum network in the future.

\subsection{SeQUeNCe and its execution model}
\label{subsec:architecture-memory:sequence-overview}

To assess the validity of QuISP's physical and timing models, we conducted a comparative study with \emph{SeQUeNCe: Simulator of QUantum Network Communication}, an independently developed open-source quantum network simulator by Argonne National Laboratory and the University of Chicago~\cite{wuParallelSimulationQuantum2024,wuSeQUeNCeCustomizableDiscreteevent2021,wuSimulationsPhotonicQuantum2019}.
Although SeQUeNCe does not implement the QRSA model or use RuleSet-based protocols, it operates with a similar rule-based resource manager that governs quantum memory allocation and coordination between its internal modules (see \cref{fig:architecture-memory:sequence-software-architecture}).

\begin{figure}[!t]
    \centering
    \includegraphics[width=0.9\linewidth]{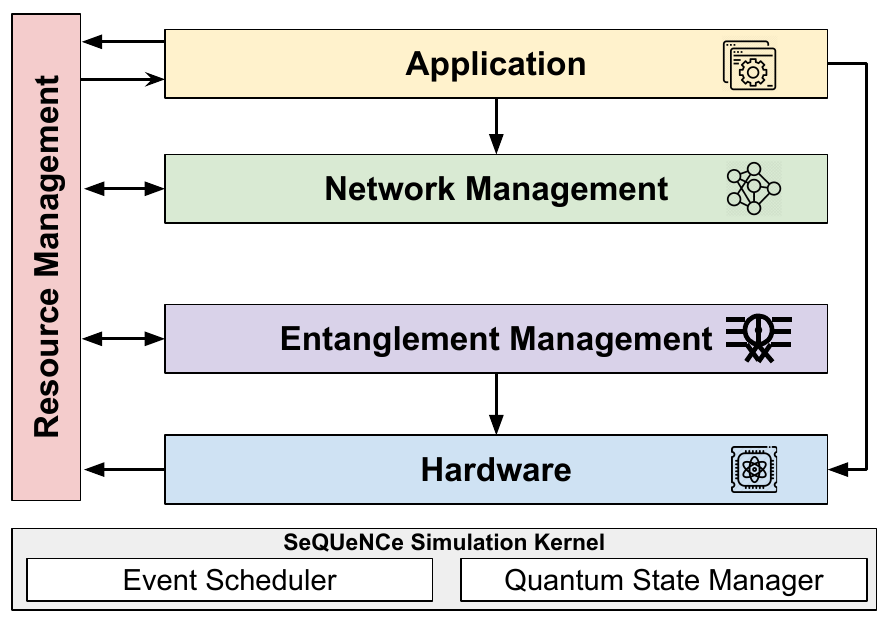}
    \caption[SeQUeNCe software framework]{SeQUeNCe software framework comprises a simulation kernel and five functional modules. (From~\cite{joaquin-michal-cross-validation-quantum-network-simulators}.)}
    \label{fig:architecture-memory:sequence-software-architecture}
\end{figure}

SeQUeNCe is structured around a modular architecture that enables flexible customization and reuse of simulation components. It consists of five primary modules---hardware, entanglement management, resource management, network management, and application---built atop a discrete event simulation kernel. This kernel handles event scheduling and includes a quantum state manager that maintains the evolving quantum state of all components. Quantum states in SeQUeNCe can be represented in either bra-ket or density matrix formalism, supporting detailed modeling of quantum operations, noise, and decoherence.

The hardware module models the physical characteristics of quantum components such as QNICs and quantum memories, including their errors and limitations. The entanglement management module implements protocols for generating and managing entangled pairs. Resource and network management modules coordinate local and inter-node resource usage using classical control messages. SeQUeNCe's resource manager follows a rule-based logic to update memory state, functionally similar to the memory handling mechanism of QuISP's Rule Engine.

The application module models end-user quantum programs and their service demands. Additionally, SeQUeNCe integrates analytical noise models for efficient simulation of noisy bipartite entanglement, enabling simulations to maintain both physical realism and scalability. Together, these components allow SeQUeNCe to simulate a wide range of quantum network scenarios with high configurability and physical fidelity.

\subsection{Timing models of QuISP and SeQUeNCe}
\label{subsec:architecture-memory:timing-model}

Both QuISP and SeQUeNCe employ discrete-event simulation frameworks to capture the timing behavior of quantum and classical network operations, but they differ in execution strategy and modeling granularity.

\textbf{QuISP} models quantum operations, photon emission, and classical message propagation as discrete events scheduled with configurable system parameters. Link-level delays are determined by physical distance and the speed of light in fiber, while additional classical processing latencies can be specified for components such as the Rule Engine. Qubits experience time-dependent decoherence using a sliced exponential noise model: the qubit's lifetime is divided into short slices $\delta t$, and probabilistic error channels are applied slice by slice, closely approximating continuous error accumulation.

Link-level entanglement generation in QuISP begins with a connection setup, during which the desired number of Bell pairs $N_{\text{Bell}}$ is specified. After synchronization, nodes emit photon trains toward a Bell-state analyzer (BSA), with photons separated by a time interval $t_{\text{sep}}$. Each round consists of photon emission, BSM processing, and classical feedback. The duration of one round is
\begin{equation}
    T_{\text{round}} = 2T_0 + 10t_{\text{sep}} + (N_{\text{mem}} - 1)t_{\text{sep}},
\end{equation}
where $T_0 = \max(L_A, L_B)/c$ is the one-way signal delay from BSA to the farthest node and $N_{\text{mem}}$ denoting the number of memories the node holds.
The expected number of rounds needed to generate $N_{\text{Bell}}$ pairs is
\begin{equation}
    k = \left\lceil \frac{N_{\text{Bell}}}{N_{\text{mem}} p_{\text{succ}}} \right\rceil,
\end{equation}
where $p_{\text{succ}}$ is the success probability of a BSM.
Thus, the total expected time is
\begin{equation}
    T_{\text{exp}}^{\text{QuISP}}(N_{\text{Bell}}) = T_{\text{setup}} + k \cdot T_{\text{round}},
\end{equation}
with $T_{\text{setup}} = 2L/c + T_0$ representing the initial setup delay.

\textbf{SeQUeNCe}, by contrast, initiates a fresh entanglement generation protocol for each Bell pair, using a three-way handshake between nodes and the BSA. Timing events such as photon emission and classical messaging are scheduled based on physical path lengths and configured device latencies. The handshake introduces setup latency for every round, making the per-attempt model more sensitive to communication overhead.

Each entanglement generation round in SeQUeNCe includes handshake and feedback delays totaling approximately as
\begin{equation}
    T_{\text{round}}^{\text{SeQUeNCe}} = 4L / c.
\end{equation}
The total expected time to generate $N_{\text{Bell}}$ Bell pairs is then
\begin{equation}
    T_{\text{exp}}^{\text{SeQUeNCe}}(N_{\text{Bell}}) = k \cdot T_{\text{round}}^{\text{SeQUeNCe}},
\end{equation}
where $k$ is again the expected number of successful attempts.
This round-by-round negotiation model simplifies protocol logic but adds latency, especially for long-distance links.

Message diagrams illustrating both timing models are shown in \cref{fig:architecture-memory:quisp-sequence-timing-model}, helping to visualize the key differences in protocol structure and scheduling assumptions.

\begin{figure}[!t]
    \centering
    \includegraphics[width=\textwidth]{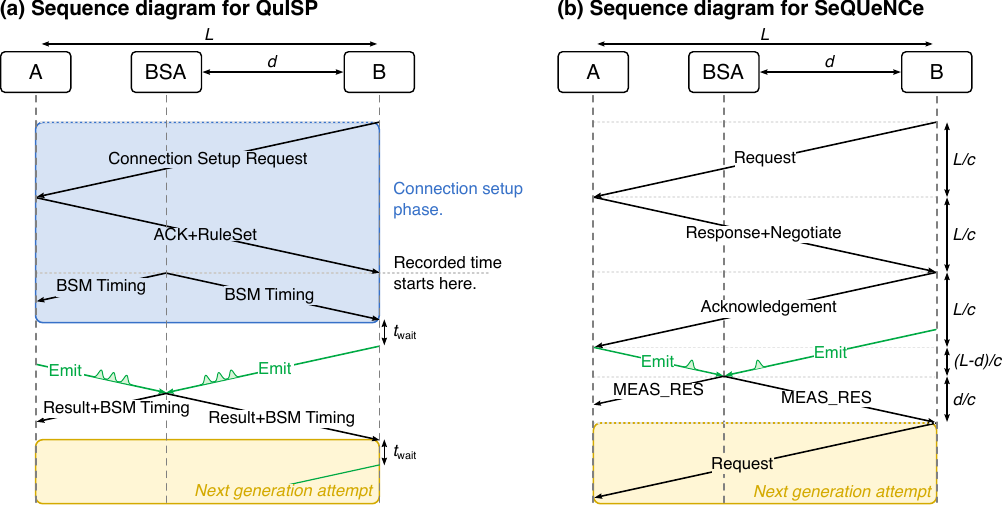}
    \caption[Message diagram for link-level entanglement generation in QuISP and SeQUeNCe]{
        Message diagrams for link-level entanglement generation in QuISP and SeQUeNCe.
        (a) In QuISP's timing model, a single connection setup phase (blue) is performed once to request a batch of Bell pairs.
        This initiates a continuous stream of photon emission and BSA attempts (white), a process that repeats (yellow) until the requested number of Bell pairs are successfully generated.
        (b) In SeQUeNCe's timing model, a three-way handshake protocol is used to generate each Bell pair individually.
        The entire handshake process (yellow) is repeated for every subsequent pair required by the connection.
        (From~\cite{joaquin-michal-cross-validation-quantum-network-simulators}.)}
    \label{fig:architecture-memory:quisp-sequence-timing-model}
\end{figure}

These differing models foreshadow the timing behaviors observed in simulation results discussed in later subsection.

\subsection{Error models of QuISP and SeQUeNCe}
\label{subsec:architecture-memory:error-models}

While both simulators implement quantum operational errors, they differ significantly in modeling philosophy and granularity.

\textbf{QuISP} adopts a probabilistic, per-qubit error model on top of graph state representation.
Each qubit is associated with an error probability vector $\vec{\pi}(t)$ that represents the likelihood of being in different error states, including Pauli errors $(X, Y, Z)$, relaxation and excitation $(R, E)$, and photon loss $(L)$.
The evolution of $\vec{\pi}(t)$ over time follows a time-sliced Markov process as
\begin{equation}
\vec{\pi}(t) = \vec{\pi}(0) Q^n, \quad \text{for } n = \left\lceil \frac{t}{\delta t} \right\rceil,
\end{equation}
where $Q$ is the single-qubit error transition matrix and $\delta t$ is the time slice interval.

Memory decoherence, single-qubit gates, and photon loss are modeled using this per-qubit transition matrix $Q$, with the evolving error vector $\vec{\pi}(t)$ tracking accumulated noise.
Two-qubit gate errors, which introduce correlated noise and thus violate the independence assumption of single-qubit channels, are treated separately; with probability $p_{\text{g}}$, a two-qubit Pauli error is sampled and applied directly to the affected qubits.
Measurement errors are implemented by flipping the reported classical outcome with probability $p_{\text{m}}$.

The fidelity of a two-hop Bell pair with noisy gates and measurements (ignoring decoherence) is
\begin{equation}
    F_\mathrm{swap} = (1 - p_{\text{g}})(1 - p_{\text{m}})^2 + \frac{3}{15} p_{\text{g}}(1 - p_{\text{m}})^2 + \frac{8}{15} p_{\text{g}}(1 - p_{\text{m}}) p_{\text{m}} + \frac{4}{15} p_{\text{g}} p_{\text{m}}^2.
\label{eq:quisp-swap-fidelity}
\end{equation}
With memory decoherence over times $t_1$, $t_2$, and total classical communication delay $T$, the final fidelity becomes
\begin{equation}
    F_{\text{swap, decoherence}} = F_{\text{swap}} \cdot \left(Q^{2t_1 + 2t_2 + 2T}\right)_{00}.
\label{eq:quisp-swap-with-decoherence-fidelity}
\end{equation}
For depolarizing errors, we have a closed form expression for any integer power of the transition matrix given by
\begin{equation}
    \left(Q^t_\mathrm{depol}\right)_{00} = \left[1+3^{1-t}(3-4 p)^t\right]/4,
\end{equation}
where $p$ is chosen such that $\left(Q^{t_\mathrm{coherence}}_\mathrm{depol}\right)_{00} = 1/e$.

\textbf{SeQUeNCe} uses analytical approximations to track fidelity.
Memory decoherence is modeled using a continuous-time depolarizing channel applied to Werner states as
\begin{equation}
    F(t) = F_{\text{in}} e^{-2t/\tau} + \frac{1 - e^{-2t/\tau}}{4},
\end{equation}
where $\tau$ is the memory coherence time, with the initial fidelity $F_{\text{in}}$.
Gate noise is applied with probability $p_{\text{g}}$ as a complete depolarization, while measurement errors flip the output with probability $p_{\text{m}}$.

For swapping two Werner states of fidelity $F_1$ and $F_2$, the post-swap fidelity is
\begin{align}
    F_\mathrm{swap,W} & = \frac{p_{\text{g}}}{4} + (1 - p_{\text{g}})\left[(1 - p_{\text{m}})^2(F_1F_2 + 3e_1e_2)\right. \nonumber \\
    & \left. + p_{\text{m}}(p_{\text{m}} - 2)(F_1e_2 + e_1F_2 + 2e_1e_2)\right],
\label{eq:sequence-swap-fidelity}
\end{align}
with $e_i = (1 - F_i)/3$.

These error models reflect differing design goals. QuISP prioritizes scalability by using per-qubit stochastic noise models that approximate quantum state evolution without tracking full quantum states.
In contrast, SeQUeNCe uses more detailed noise modeling based on Kraus operators and density matrices, enabling higher precision and allowing for analytical tractability in some small-scale cases.
Despite these differences, both simulators have been shown to produce consistent fidelity results under comparable conditions, as discussed in the next subsection.

\begin{figure}[tb]
    \centering
    \includegraphics[width=\textwidth]{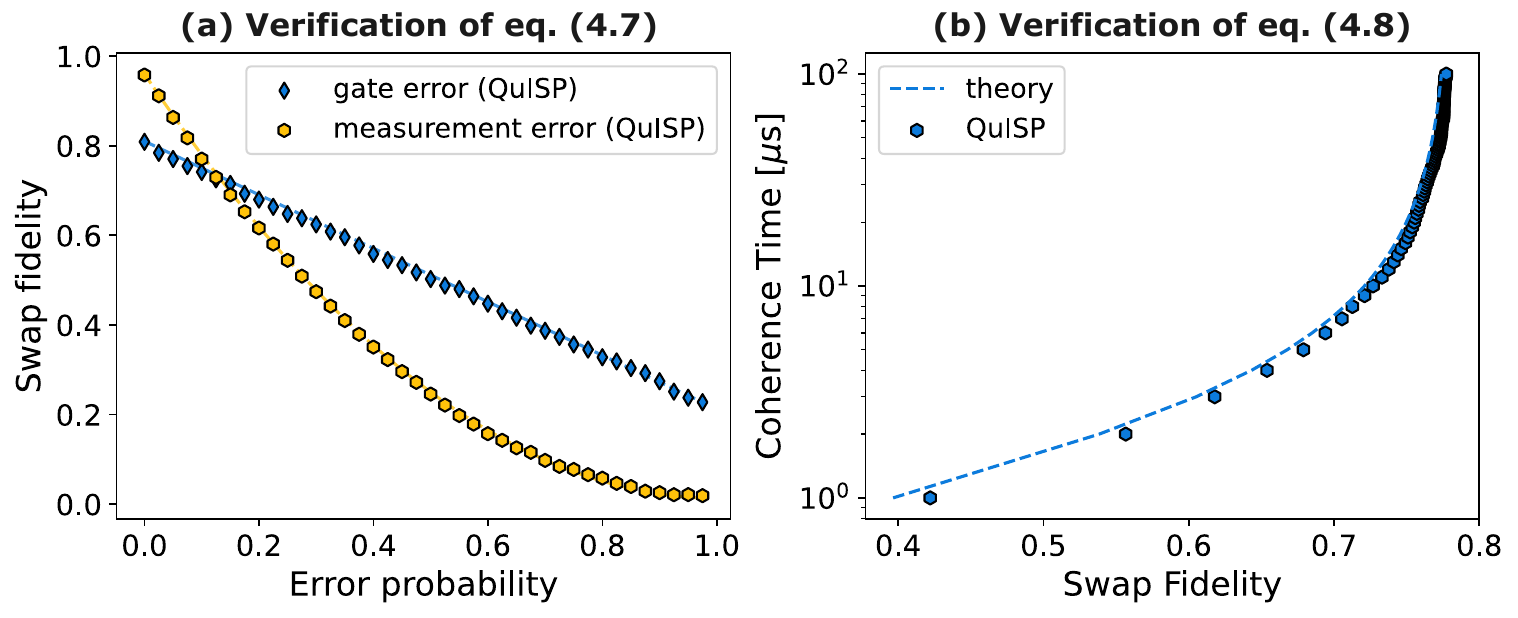}
    \caption[Plots showing the fidelities of swapped Bell pairs between output of QuISP and its analytical model varying noise types and strength]{(a) Verification of end-to-end fidelity \cref{eq:quisp-swap-fidelity} for QuISP. Blue diamonds for varying $p_{\text{g}}$ and $p_{\text{m}}=0.1$ and yellow hexagons for varying $p_{\text{m}}$ with $p_{\text{g}}=0.05$. Dashed lines represent our theoretical predictions. (b) Verification of \cref{eq:quisp-swap-with-decoherence-fidelity} for $p_{\text{g}}=0.05$ and $p_{\text{m}}=0.1$, varying degree of depolarizing noise. (From~\cite{joaquin-michal-cross-validation-quantum-network-simulators}.)}
    \label{fig:architecture-memory:quisp-model-verification}
    \vspace{\floatsep}
    \centering
    \includegraphics[width=\textwidth]{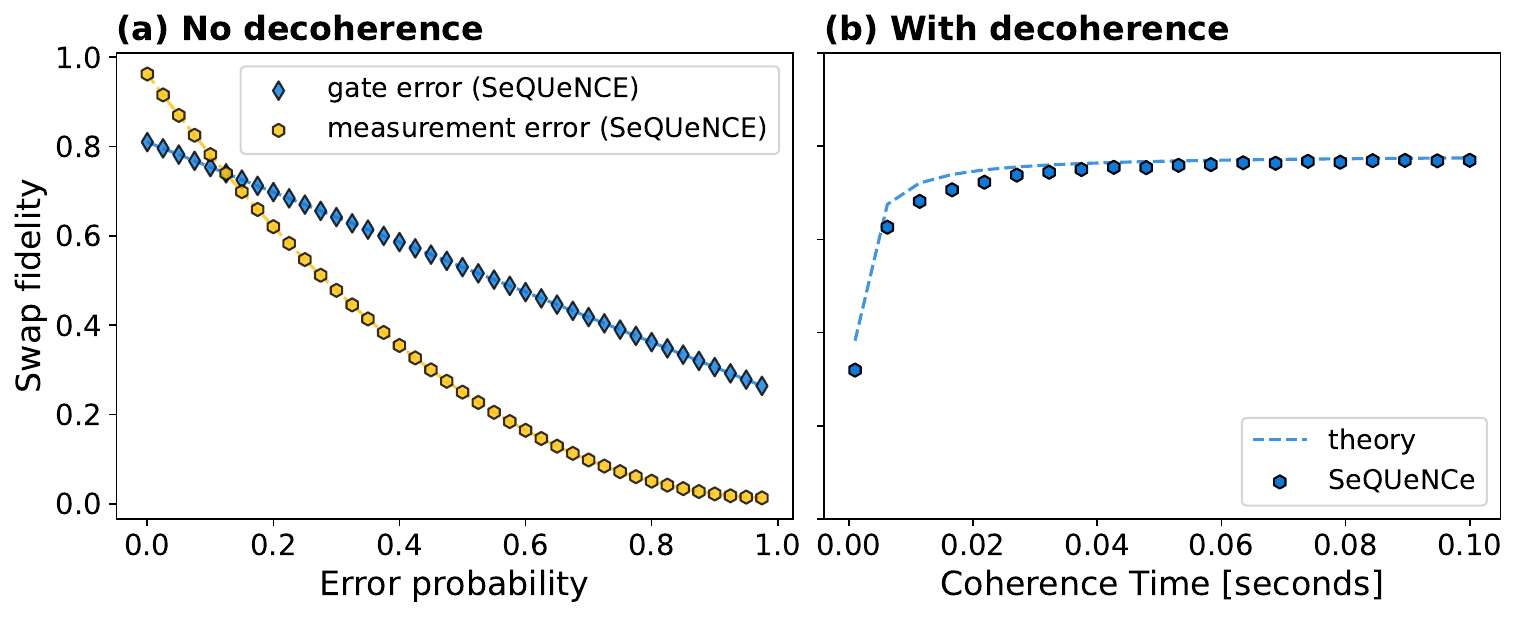}
    \caption[Plots showing the fidelities of swapped Bell pairs between output of SeQUeNCe and its analytical model varying noise types and strength]{Verification of SeQUeNCe error models. (a) Without decoherence. Blue diamonds for $p_{\text{m}}=0.1$ and yellow hexagons for $p_{\text{g}}=0.05$. (b) Including decoherence for $p_{\text{g}}=0.05$ and $p_{\text{m}}=0.1$.
    (From~\cite{joaquin-michal-cross-validation-quantum-network-simulators}.)}
    \label{fig:architecture-memory:sequence-model-verification}
\end{figure}
\Cref{fig:architecture-memory:quisp-model-verification,fig:architecture-memory:sequence-model-verification} show the end-to-end fidelity results for both QuISP and SeQUeNCe under varying gate and measurement error rates, with the other error type held fixed.
In both simulators, we observe excellent agreement between the simulated outputs and the respective theoretical models; \cref{eq:quisp-swap-fidelity,eq:quisp-swap-with-decoherence-fidelity} for QuISP, and \cref{eq:sequence-swap-fidelity} for SeQUeNCe.
Additionally, when memory decoherence is introduced via depolarizing noise, the simulated fidelity continues to closely track analytical predictions.
These results confirm the effective consistency of both simulators' error models, despite their differing approaches to model errors.

\subsection{Cross-validation of QuISP and SeQUeNCe: fidelity, timing, and error models}
\label{subsec:architecture-memory:cross-validation}

Cross-validating independent simulation platforms is essential for establishing trust in their results.
Agreement between QuISP and SeQUeNCe, which were developed separately and based on different modeling philosophies, strengthens confidence in both simulators' ability to accurately capture quantum network behavior.
This joint effort, presented in~\cite{joaquin-michal-cross-validation-quantum-network-simulators}, aimed to reproduce three canonical quantum networking experiments on both platforms and compare their outputs to analytical predictions.

The three experiment scenarios studied are as follows.
\begin{enumerate}
    \item \textbf{Symmetric entanglement generation:} Entanglement generation over a memory--interference--memory (MIM) link, where the number of quantum memories at both end nodes is varied to evaluate scaling performance under symmetric conditions.
    \item \textbf{Asymmetric entanglement generation:} End nodes each have a single memory, while the position of the Bell-state analyzer (BSA) is varied along the link to introduce asymmetry. This experiment examines how BSA placement affects timing and throughput under different classical control models.
    \item \textbf{Two-hop entanglement swapping:} Establishing an end-to-end Bell pair across two hops, with each end node having a single memory and the repeater node having two (one per link). The experiment focuses on end-to-end fidelity under varied operational errors, including gate infidelity, memory decoherence, and measurement noise.
\end{enumerate}

\begin{figure}[!t]
    \centering
    \includegraphics[width=0.8\linewidth]{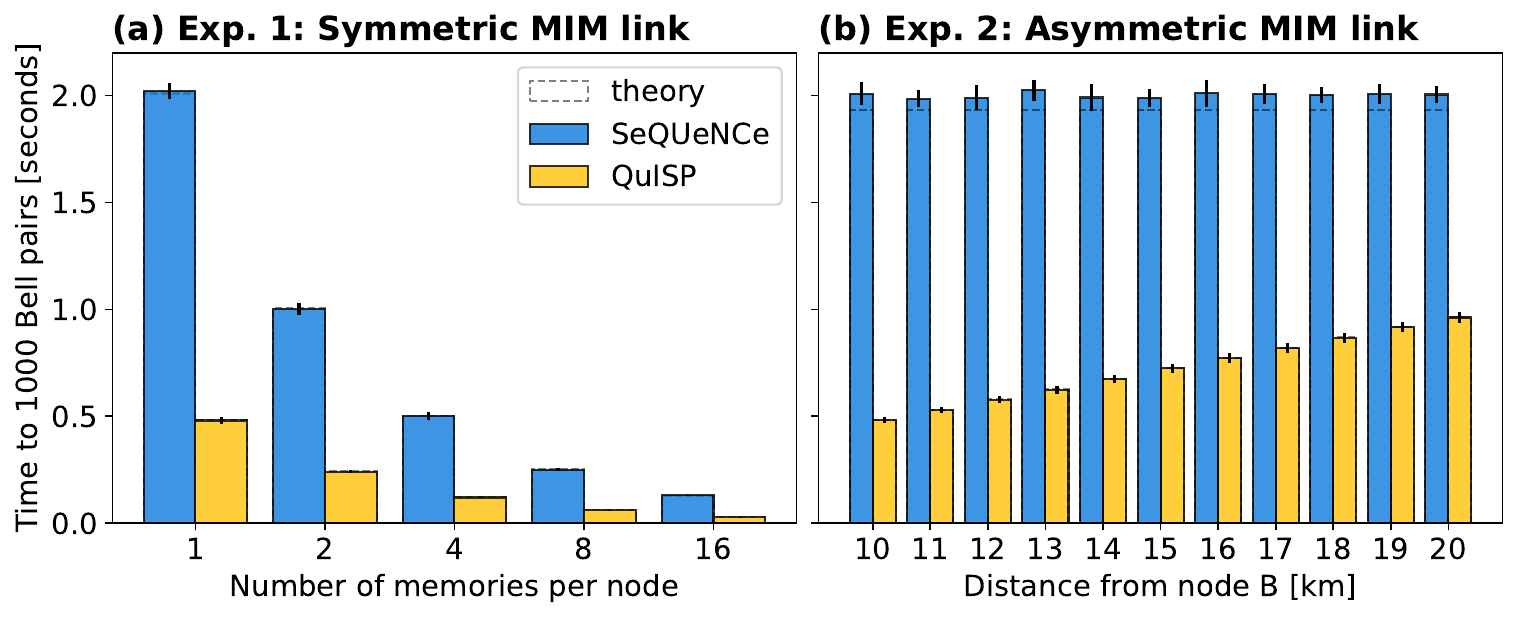}
    \caption[Plots showing time to establish 1000 Bell pairs between QuISP and SeQUeNCe]{Time to generate 1000 Bell pairs for symmetric MIM link in (a), and asymmetric MIM link in (b). QuISP simulation used $t_\text{sep} = 1 \text{ns}$.
    (From~\cite{joaquin-michal-cross-validation-quantum-network-simulators}.)}
    \label{fig:quisp_sequence_exp1}
\end{figure}
In the first two experiments, which involved only photon loss and no operational noise, both simulators produced generation rates and completion times (for requests of 1000 Bell pairs) that matched analytical expectations derived from their respective timing models.
As shown in \cref{fig:quisp_sequence_exp1}, both simulators demonstrate the same qualitative scaling behavior in the first experiment, where increasing the number of available memories leads to improved performance.

In the second experiment, which introduces asymmetry by shifting the position of the Bell-state analyzer (BSA), the two simulators diverge in their trend lines due to differences in classical control assumptions.
In SeQUeNCe, each Bell pair generation attempt is treated as an independent connection, requiring a new three-way handshake as shown in the timing diagram.
This repeated setup dominates the overall timing, making the position of the BSA effectively irrelevant.
In contrast, QuISP performs connection setup once at the beginning for the entire request of 1000 Bell pairs, and subsequent entanglement attempts proceed continuously.
As a result, the effect of BSA placement becomes visible in QuISP: a centrally placed BSA minimizes total communication delay, while asymmetrical placement leads to longer delays per attempt and a more noticeable impact from photon loss on the longer fiber path.
Despite these differences, both simulators match their own analytical models, confirming the internal consistency of their respective timing assumptions.

\begin{figure}[!t]
    \centering
    \includegraphics[width=\linewidth]{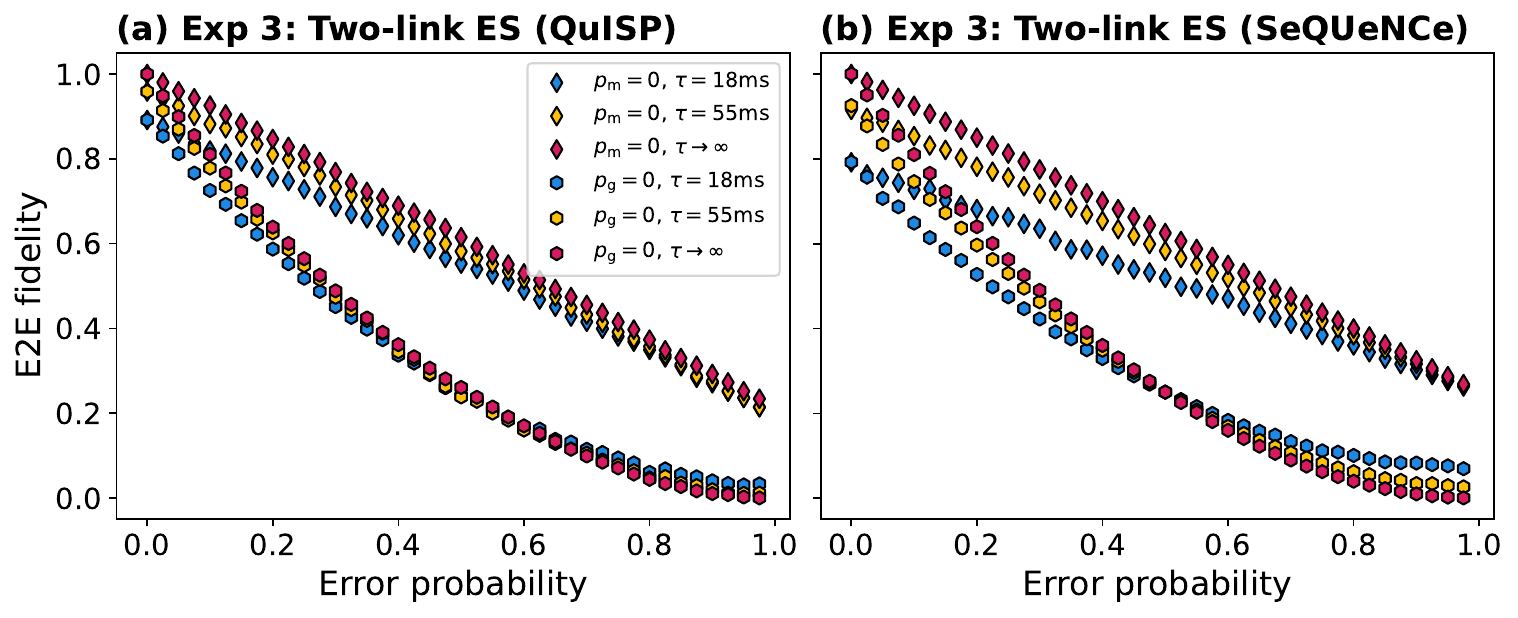}
    \caption[Plots showing fidelity of two-hop Bell pairs with varying error types and rate between QuISP and SeQUeNCe]{Experiment 3 (Two-hop swapping): End-to-end fidelity under varying gate and memory noise. Agreement between simulators supports the validity of each platform's error model and protocol behavior.
    (From~\cite{joaquin-michal-cross-validation-quantum-network-simulators}.)
    }
    \label{fig:quisp_sequence_fidelity}
\end{figure}
The third experiment, two-hop entanglement swapping, offered a more comprehensive test of distributed behavior and error modeling.
This scenario introduces memory decoherence (finite $T_1$), gate infidelity, and measurement error.
As shown in \cref{fig:quisp_sequence_fidelity}, QuISP and SeQUeNCe produced nearly identical end-to-end fidelity curves across a wide range of error parameters, once the parameters were translated to account for differences in how each simulator models noise.

This agreement is particularly noteworthy given the distinct modeling philosophies of the two simulators.
SeQUeNCe employs the density matrix formalism and analytical error models, while QuISP adopts a discretized and computationally lighter approach, using single-qubit noise channels and stochastic sampling to model correlated errors such as those arising from two-qubit gates.
In the third experiment, the primary source of discrepancy lies in how gate error probabilities are defined.
SeQUeNCe models a gate error of probability $p$ as a complete depolarization of the two-qubit state with probability $p$, whereas QuISP applies one of 15 non-identity Pauli error terms with total probability $p$, excluding the identity operator.
When this difference of definition is reconciled analytically, the two simulators produce nearly identical fidelity results across a broad parameter range.

Taken together, these findings validate the consistency and reliability of both timing and error models in QuISP and SeQUeNCe.
Despite differences in modeling granularity and simulation strategies, both platforms faithfully reproduce expected quantum behaviors.
For the purposes of this thesis, these results provide strong support for using QuISP as a platform for prototyping, simulating, and validating distributed quantum protocols within the QRSA software architecture and RuleSet-based control framework.

\subsection{Limitation and future work}
\label{subsec:architecture-memory:future}

While the experiments discussed in this section demonstrate promising agreement between QuISP, SeQUeNCe, and analytical models, they represent only an initial validation step.
Our current investigations are limited to small-scale networks, primarily linear chains and single-hop topologies, and focus on verifying timing models and fidelity predictions under controlled noise parameters.
A natural next step is to apply QuISP to simulate larger-scale networks and study emergent behaviors, resource contention across multiple connections, and the performance of routing and multiplexing protocols in more realistic settings.

It is important to emphasize that QuISP offers significantly more capabilities than demonstrated in this thesis.
Built on the QRSA architecture and RuleSet-based control framework, QuISP provides modular tools for simulating resource allocation, routing, multiplexing, and AAA (authentication, authorization, and accounting) under heterogeneous network conditions.
It supports a wide range of link types, including satellite-ground channels~\cite{fittipaldi-sattelite-quisp}, memory--source--memory (MSM) and memory--interference--memory (MIM) links~\cite{soon-heterogeneous-msm-mim,soon-optimizing-msm-protocol-qcnc-2024}, and is capable of modeling internetworking scenarios that span administrative domains~\cite{teramotoRuleSetbasedRecursiveQuantum2023}.
However, these advanced features have not yet been applied to protocol validation or network-wide policy evaluation, which presents an important direction for future research.

Several other quantum network simulators have been developed, though many focus on different goals.
Some emphasize hardware-level modeling~\cite{bartlett-squanch,coopman-netsquid}, others aim to integrate with specific applications or testbeds~\cite{dahlberg-netqasm,oslovichCompilationStrategiesQuantum2025,QuTechDelftSquidasm2025,QuTechDelftQneadk2025}, and some focus on entanglement distribution in linear repeater chains~\cite{wallnofer-requsim-faithfully-simulating,fangQuantumNETworkTheory2023a} or specific network-layer issues such as routing and multi-path scheduling~\cite{diadamo-qunetsim,leone-multi-path-routing-qunet,dahlberg-simulaqron}.
In contrast, QuISP is designed to support end-to-end architectural studies that encompass routing, multiplexing, AAA, and abstraction across link and network boundaries.

Unlike quantum error correction (QEC), where simplified noise models are supported by strong theoretical justification for fault-tolerant properties, quantum networking still lacks foundational results that explain when and why such models are appropriate~\cite{gottesman-phd-thesis,gottesmanSurvivingQuantumComputer,gottesman-ftqc-review,simon-qec-beginners,roffe-qec-intro}.
As a result, the community tends to rely on more rigorous modeling approaches such as full density matrix simulations.
We argue that further validation of simulators like QuISP, particularly through experimental testbeds and real-world deployments, is essential for increasing confidence in their predictive value.
Building this empirical and theoretical foundation remains a key challenge for advancing simulation-based architecture design in quantum networking.

\section{Discussion}
\label{sec:architecture-memory:discussion}

Having presented the details of our architecture throughout the chapter, in this last section, we now reiterate its core principles and revisit the key design decisions introduced in \cref{sec:architecture-memory:design-decisions}.
This reflection highlights the rationale behind our architectural choices, the internal coherence of the overall design, and how the key components---RuleSets, QRNA, and QRSA---together provide a robust and extensible foundation for memory-based quantum networks.
While this chapter has focused on one specific architecture grounded in RuleSets and recursive abstraction, it is not the only viable approach.
Later, we broaden the discussion to consider alternative architectural paradigms and other operational models of quantum networks that offer distinct trade-offs in generality, performance, and implementation complexity.
This comparative discussion places our design choices within the broader landscape of quantum network architectures.

\subsection{Revisiting the design decisions}
\label{subsec:architecture-memory-revisting-design-decisions}

Our architecture adopts a connection-oriented model in which connections are established through a two-pass setup protocol.
These connections are managed by RuleSets distributed along the path between the Initiator and Responder nodes.
This design treats Bell pair delivery as the fundamental network service, but it also leaves room for generalizing to the generation and distribution of multipartite entangled states, should a network operator choose to support such use cases (see \cref{sec:architecture-memory:ruleset-definition-examples}).
This addresses the fundamental nature of the network's services and how connections are conceptualized and controlled.

RuleSets lie at the heart of our programming model, serving as distributed, event-driven control logic that defines local behavior and governs resource usage.
They are generated based on user requests at the Responder and installed during connection setup, enabling localized yet coordinated behavior across the network.
This allows connections to be customized and analyzable.
At the application level, users specify high-level constraints, such as target fidelity or the number of Bell pairs, without dealing with low-level protocol details.
These intents are translated by the Connection Manager into executable RuleSets.
Furthermore, the Rule Engine, which allocates local resources to specific RuleSet instances (each corresponding to a single connection), enables multiplexing by tagging and isolating resources.
This mechanism also lays the foundation for integrating fairness, access control, and scheduling policies across shared infrastructure---addressing the questions of how application requests are conveyed, how resources are multiplexed, and how security and AAA considerations can be incorporated.

In addition to defining connection semantics and control logic, our architecture clarifies node roles and facilitates scalable network composition through recursion.
QRNA defines a recursive abstraction boundary that allows networks to be grouped and treated as virtual links at higher levels of hierarchy.
This enables independent domains, each managed by different operators, to interconnect while preserving control over local behavior.
Within each recursion level, the QRSA stack (especially the Connection Manager) rewrites RuleSets and coordinates connection setup through wrapper functions, supporting seamless internetworking across domains.
Each recursion layer makes routing decisions based on its own topology and policy, enabling flexible and scalable routing strategies.
At domain boundaries, RuleSets are translated to preserve end-to-end connection semantics, allowing interoperability across heterogeneous hardware and administrative zones.
This modular and layered approach supports both architectural consistency and extensibility---addressing the questions of defining network roles, supporting flexible routing strategies, enabling internetworking, and managing security and privacy at domain boundaries.

Taken together, these architectural mechanisms form a cohesive foundation for building programmable, distributed quantum networks that are scalable, analyzable, and technology-agnostic.
The deliberate use of modularity, abstraction, and cross-domain compatibility ensures that the architecture can support a broad range of quantum networking applications and deployment scenarios.

\subsection{Operation scheduling paradigms}
\label{subsec:other-operational-paradigms}

The RuleSet-based architecture central to this thesis, as detailed in this chapter, operates on what can be best described as a dynamic scheduling paradigm. In this model, network activities are intrinsically event-driven, reacting to local information and resource availability, which fosters significant node autonomy in decision-making (albeit guided by the centrally-authored RuleSets). This approach offers considerable flexibility and responsiveness, mirroring the adaptive nature often sought in complex distributed systems.

To better contextualize this design choice and understand its advantages and trade-offs, it is instructive to examine it alongside other prominent paradigms for scheduling operations within a quantum network. These alternative approaches often differ fundamentally in how they manage the timing and coordination of key tasks---such as link-level entanglement generation, purification, and entanglement swapping---ranging from globally synchronized, discrete-time operations to more decentralized, continuous-time mechanisms. In the following, we will discuss three principal scheduling paradigms: the one-shot model, the time-slot based model, and will further elaborate on the characteristics of the dynamic model as exemplified by our work.

The \emph{one-shot based paradigm}, explored in the work of Pant et al.~\cite{pant-routing-entanglement} and Patil et al.~\cite{patil-distance-ind-rate}, operates with operations occurring in predefined stages within synchronized slots.
Its defining characteristic is that all unused entangled pairs (i.e., those not part of a successfully completed end-to-end entanglement by the slot's conclusion) are discarded at the end of each slot.
This approach is particularly relevant for near-term quantum network implementations, especially where quantum memory lifetimes are short and cannot reliably hold qubits across multiple time slots.
By ensuring each new operational round effectively starts with fresh quantum resources, this method minimizes the impact of memory decoherence on the final entanglement.
Consequently, it tends to offer higher fidelity for successfully established end-to-end pairs, though often at the cost of slower overall throughput compared to approaches that preserve resources across slots.
Furthermore, this operational model aligns well with assumptions for distributed quantum computation at the datacenter scale, where protocols may rely on Bell pairs being generated and made available within specific time windows to perform remote operations between distinct quantum processing units or computers~\cite{jnane-multicore-qc,leoneUpperBoundsClock2023,leoneResourceOverheadsAttainable2025,shapourianQuantumDataCenter2025,pouryousefNetworkAwareSchedulingRemote2025a}.

In contrast, \emph{the time-slot based paradigm}, as analyzed in the work of Cicconetti et al.~\cite{cicconetti-request-scheduling} and Dhara et al.~\cite{dharaSubexponentialRateDistance2021}, also dictates that the network operates in globally coordinated time slots, with each slot typically divided into several predefined stages for sequential operations (e.g., link-level generation, purification, swapping, classical communication).
However, unlike the one-shot model, entangled Bell pairs that are generated but not immediately consumed are retained in quantum memories and carried over to subsequent time slots.
This approach necessitates quantum memories with longer coherence times capable of preserving qubit states reliably across these extended durations.
The global coordination inherent in this paradigm often implies the need for a central decision-making unit to schedule and synchronize operations across the network.
While such centralization can simplify the implementation of complex, multi-stage coordination strategies, the classical messaging overhead involved in communicating instructions and feedback between the central unit and distributed network nodes can become a significant bottleneck, potentially limiting performance and scalability, especially for wide-area networks.
Nevertheless, for smaller or more controlled network segments, this paradigm can offer high pipeline efficiency.

Finally, the \emph{dynamic paradigm}, which aligns with our RuleSet-driven architecture, schedules operations asynchronously.
This approach is explored in the work of Zhan et al.~\cite{zhanDesignSimulationAdaptive2025}, J. Li et al.~\cite{li-connection-oriented-quantum-network}, Z. Li et al.~\cite{liConnectionOrientedConnectionlessRemote2022}, and Andrade et al.~\cite{andradeAnalysisQuantumRepeater2024}.
In this model, actions are primarily event-driven and triggered by local conditions, such as resource availability or the arrival of messages.
Link-level Bell pairs are generated whenever quantum memories become free and opportunities arise, and subsequent operations like purification or swapping occur opportunistically.
This approach offers high flexibility and efficient memory utilization, embodying the adaptive and decentralized ``spirit of the Internet.''
However, its asynchronous and highly state-dependent nature poses significant challenges for analytical performance modeling and prediction, especially at the multi-hop or network-wide scale.

\subsection{Connection-oriented versus connectionless models}
\label{subsec:architecture-memory:connection_vs_connectionless}

Two primary paradigms have emerged for managing entanglement distribution across quantum networks: the \emph{connection-oriented model} and the \emph{connectionless model}.
Our proposed architecture adopts the connection-oriented model, as it offers strong guarantees in terms of coordination, fidelity, and scheduling flexibility.
Understanding the fundamental distinctions between these two paradigms is crucial for appreciating the design rationale behind our framework.

The connection-oriented model~\cite{rdv-qi-architecture,cicconetti-request-scheduling,liConnectionOrientedConnectionlessRemote2022,li-connection-oriented-quantum-network,zhangConnectionorientedConnectionlessQuantum2022}, as detailed for our architecture in \cref{sec:architecture-memory:expressing-connection-semantics-rulesets}, emphasizes establishing a confirmed path and securing resource commitments from all intermediate nodes before the attempt to distribute end-to-end Bell pairs begins.
This pre-established agreement, orchestrated, for instance, by a Responder generating and distributing RuleSets, aims to enhance the reliability of entanglement distribution and provide more predictable end-to-end rate and state fidelity.
This model is particularly well-suited for scenarios requiring non-local coordination, such as multi-hop entanglement purification and more complex scheduling.

In contrast, the connectionless model~\cite{liConnectionOrientedConnectionlessRemote2022,zhangConnectionorientedConnectionlessQuantum2022,bacciottiniLeveragingInternetPrinciples2025a,diadamoPacketSwitchingQuantum2022} operates without prior resource reservation or explicit commitment from all nodes along an entire path.
Instead, operations like Bell pair distribution typically proceed on a hop-by-hop basis, analogous to packet switching in classical networks.
While this approach is often associated with certain third-generation repeater schemes designed for high-speed attempts, it could, in principle, be adapted for first or second-generation repeaters.
If utilized with a RuleSet-like protocol, for example, it would likely require RuleSets to be more dynamically generated or updated, involving ongoing coordination between the current end-node of an extending Bell pair and the next hop node it is attempting to reach, rather than being entirely predetermined by a central Responder.

The hop-by-hop nature of the connectionless model generally restricts entanglement purification to occur only between adjacent nodes.
This can make it challenging to achieve high end-to-end fidelity guarantees, as errors might accumulate along the path.
If an operation within a segment fails, or if the fidelity of an intermediate entangled pair drops below a usable threshold and cannot be salvaged by link-level purification alone, the attempt for that path segment is aborted and may be retried.
Consequently, connectionless approaches often impose more stringent requirements on the baseline error thresholds of individual quantum operations, the initial fidelity of elementary links, and the coherence times of any quantum memories involved.
While many studies on connectionless models~\cite{liConnectionOrientedConnectionlessRemote2022,zhangConnectionorientedConnectionlessQuantum2022} have focused on their potential advantages in terms of entanglement generation rates, comprehensive analyses of their end-to-end fidelity characteristics, especially across multi-hop paths with realistic noise models, are an ongoing area of investigation.

\clearpage
\chapter{Integrating all-photonic repeaters into the architecture}
\label{chapter:architecture-all-photonic}

A key element that enables quantum repeaters to circumvent the fundamental repeaterless bound is their ability to store short-distance Bell pairs between adjacent nodes inside quantum memories, and later splice them together via entanglement swapping to extend entanglement across longer distances.
This principle underlies both first- and second-generation repeaters.
Other approaches---such as those proposed in third-generation repeaters---aim to significantly relax memory lifetime requirements, reducing them from the duration of an end-to-end round-trip to merely the time needed to receive and forward quantum states at each node, typically on the scale of link-level communication delays.
These methods achieve this by leveraging fault-tolerant teleportation~\cite{muralidharanUltrafastFaultTolerantQuantum2014,muralidharanOvercomingErasureErrors2017}, performing logical operations through surface code-style measurement-based computation~\cite{fowlerSurfaceCodeQuantum2010}, or transmitting photonic qubits encoded in fault-tolerant codes directly through the network in hop-by-hop manner decoding and re-encoding at every repeater node~\cite{munroQuantumCommunicationNecessity2012a,borregaard-3g-repeater}.
\extrafootertext{Portions of this chapter have been adapted from the following publications:
\begin{itemize}
    \item \fullcite{naphan-engineering-challenges-in-rgs}
    \item \fullcite{naphan-half-rgs}
    \item \fullcite{naphan-integrating-purification-rgs}
\end{itemize}
}

Challenging the assumption that quantum memories are indispensable, Azuma, Tamaki, and Lo proposed a new type for quantum repeaters, the \emph{all-photonic quantum repeaters}, built entirely from flying photons~\cite{azuma-rgs}.
Their insight was that the use of memories arises primarily to deal with the probabilistic nature of entanglement generation and to facilitate the delayed operations required for high-fidelity entanglement swapping.
By embedding these capabilities directly---near-deterministic link-level generation---into large photonic graph states, called the \emph{repeater graph states (RGS)}, they outlined how repeaters could be realized without intermediate storage, instead relying on photonic loss-tolerant graph states and adaptive measurements.
\nomenclature{Repeater graph state (RGS)}{A specific multipartite photonic graph state that serves as the central resource in a time-reversed memoryless quantum repeater architecture, characterized by a set of inner qubits connected as a complete graph, with each inner qubit also entangled with a corresponding outer qubit for near-deterministic link-level entanglement creation}

\nomenclature{Quantum emitter}{A stationary qubit that can be controlled to sequentially emit multiple photons, with each emission coherently extending the entanglement to form a single multipartite state between the emitter and the entire sequence of photons}
In this chapter, we will revisit this all-photonic paradigm and focus on addressing the gap between conceptual designs and practical implementations~\cite{naphan-engineering-challenges-in-rgs}.
We first review the original proposal by Azuma, Tamaki, and Lo, then present our contributions: a more modular approach to repeater graph state generation based on the introduction of half repeater graph states (half-RGS)~\cite{naphan-half-rgs,naphan-rgs-tqe}, which can be deterministically generated using quantum emitters.
This design allows not only greater timing flexibility and modularity but also provides a natural way to hand off photonic entanglement to memory-equipped end nodes (addressing challenges raised in \cref{sec:quantum-networks:challenges}).
Leveraging this abstraction, we demonstrate how paths composed entirely of all-photonic repeaters can be treated as virtual link-level connections from the perspective of higher-level protocols.
This integration point allows such paths to be seamlessly embedded into the recursive architecture introduced in \cref{chapter:architecture-memory}, despite their very different underlying hardware assumptions and operations.
Finally, we address a long-standing question in the literature---how to support non-neighbor purification in an all-photonic segment---by showing that half-RGS offers the operational flexibility needed to realize purification and swapping even when limited to photonic-only infrastructure~\cite{naphan-integrating-purification-rgs} in \cref{sec:rgs-with-purification}.

\section{Preliminaries to understanding all-photonic quantum repeaters}
\label{sec:preliminaries-for-all-photonic-quantum-repeaters}

To understand the design and operation of all-photonic quantum repeaters, we begin by revisiting a key concept---the graph state formalism.
While the basics of graph states were introduced in \cref{sec:preliminaries:graph-states}, we will extend this discussion by introducing an additional measurement pattern, the XX measurement on adjacent qubits, and its corresponding graph transformation rule, both of which are central to the all-photonic quantum repeater protocol in this chapter.

\subsection{Revisiting graph state transformation rules}
\label{subsec:xx-measurement-rules}

Let us recall a basic transformation rule in the graph state formalism, the effect of a Z basis measurement on a qubit.
In a graph state, measuring a qubit in the Z basis removes the corresponding vertex from the graph, and induces a Pauli $Z^r$ operator as a side effect on each of the measured qubit's neighbors, where $r \in \{0,1\}$ is the classical measurement outcome.
An example of this Z measurement effect is depicted in \cref{fig:graph-state-example-in-rgs} (bottom left).

\begin{figure}[htb]
    \centering
    \includegraphics[width=\textwidth]{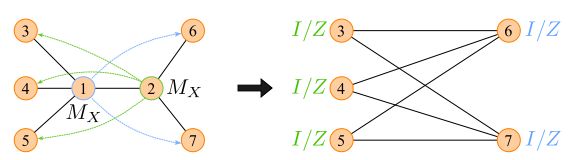}
    \caption[Examples of graph states]{Top: Two equivalent graph state representations of the same quantum state. Vertices represent qubits, and edges correspond to controlled-phase gates. A green edge in the top left highlights one such gate. The top right graph is obtained by applying local complementation on vertex 3, replacing edges between its neighbors (e.g., deleting $(1,4)$) with new ones (e.g., adding $(1,5)$ and $(4,5)$, shown in red). The states remain equivalent under local Clifford operations. Bottom: A $Z$ measurement and its effect (left), and an $XX$ measurement on two qubits in a different graph (right). Side effects on qubits 3, 4, and 5 depend on the measurement of qubit 2; those on qubits 6 and 7 depend on qubit 1. Blue and green arrows indicate these dependencies.
    (From~\cite{naphan-half-rgs}.)}
    \label{fig:graph-state-example-in-rgs}
\end{figure}

To fully describe the operational mechanics of the all-photonic repeater scheme, however, we require an additional graph transformation rule---the \emph{XX measurement rule}.
This rule applies when two adjacent qubits in a graph state, say $a$ and $b$, are each measured in the X basis, under the condition that these two qubits do not share any common neighbors.
It is important to distinguish this sequence of two independent single-qubit X basis measurements from a joint measurement of the $XX$ Pauli operator (a stabilizer measurement) often encountered in quantum error correction.

As illustrated in \cref{fig:graph-state-example-in-rgs} (bottom right), performing these X basis measurements on adjacent qubits $a$ and $b$ transforms the graph such that the distinct neighborhoods $N_a \setminus \{b\}$ and $N_b \setminus \{a\}$ become fully interconnected, forming a complete bipartite graph between them, while $a$ and $b$ are removed.
While such graphical transformation rules are powerful tools for bypassing the often cumbersome direct manipulation of stabilizer tableaus, we will derive this XX measurement rule explicitly using the stabilizer formalism.
This derivation is particularly more convenient as X basis measurements involve multiple non-trivial steps to represent purely through local complementations or simple vertex and edge removal rules.

The derivation proceeds by analyzing the effect of the X basis measurements on $a$ and $b$ on the stabilizer generators of the initial graph state.
We can categorize these generators as follows:
\begin{enumerate}
    \item The stabilizer generator $g_a$ associated with qubit $a$.
    \item The stabilizer generator $g_b$ associated with qubit $b$.
    \item Generators $g_j$ associated with qubits $j \in N_a \setminus \{b\}$ (neighbors of $a$, excluding $b$).
    \item Generators $g_k$ associated with qubits $k \in N_b \setminus \{a\}$ (neighbors of $b$, excluding $a$).
    \item Generators associated with all other qubits outside of $N_a \cup N_b \cup \{a,b\}$.
\end{enumerate}
Generators in the last category commute with the X basis measurements on $a$ and $b$, thus remaining unaffected, and can be excluded from this specific derivation.
For clarity, and with slight abuse of notation, we denote the Pauli string for stabilizer generators of each group (omitting the last group) in \cref{table:xx-measurement-stabilizer-table}
\begin{table}[htb]
    \caption[The stabilizer of graph states before XX measurement]{The stabilizer of graph states before XX measurement on adjacent qubits $a$ and $b$. Operators acting on qubit groups: $M_j$, $N_a \setminus \{b\}$, $a$, $b$, $N_b \setminus \{a\}$, and $M_k$ from left to right, with $M_j$ ($M_k$) denoting qubits in the neighborhood of qubts in $N_a$ ($N_b$) not including $a$ ($b$).}
    \label{tab:ch5_stabilizers_before_xx}
    \centering
    \begin{tabular}{l >{$}c<{$} >{$}c<{$} >{$}c<{$} >{$}c<{$} >{$}c<{$} >{$}c<{$}}
        \toprule
        \textbf{Generator Type} & \multicolumn{6}{c}{\textbf{Pauli Operators on Qubit Groups}} \\
        \cmidrule(lr){2-7} 
        & M_j & N_a \setminus \{b\} & a & b & N_b \setminus \{a\} & M_k \\
        \midrule
        $g_a$                   & I & Z^a_j & X_a & Z_b & I & I \\
        $g_b$                   & I & I & Z_a & X_b & Z^b_k & I \\
        \addlinespace 
        $g_j$ (for $j \in N_a \setminus \{b\}$) & Z_{M_j} & X^a_j & Z_a & I & I & I \\
        \addlinespace
        $g_k$ (for $k \in N_b \setminus \{a\}$) & I & I & I & Z_b & X^b_k & Z_{M_k} \\
        \bottomrule
    \end{tabular}
\label{table:xx-measurement-stabilizer-table}
\end{table}

To derive the effect of the XX measurement, we apply the stabilizer update procedure from \cref{subsec:preliminaries:stabilizer-measurement} sequentially: first for the measurement of \( P_a = X_a \), then for \( P_b = X_b \).  
Let the measurement outcomes be \( s_a, s_b \in \{+1, -1\} \).  
After the measurements, the stabilizer group of the entire system will include \( s_a X_a \) and \( s_b X_b \).  
Our goal is to determine the updated stabilizer generators that describe the system post-measurement.

Starting with the measurement on qubit \( a \), the system is projected into an eigenstate of \( X_a \).  
We observe that \( g_b \) and all \( g_j \) anticommute with \( X_a \).  
Following the rules from \cref{sec:stabilizer-formalism}, we update the stabilizer by multiplying each \( g_j \) with \( g_b \), giving updated generators \( g_j' = g_b \cdot g_j \).  
We then remove \( g_b \) from the stabilizer set.  
After this first measurement, the stabilizer becomes
\begin{equation}
    \left\langle\begin{array}{c}
                Z^a_j X_a Z_b  \\
                s_a X_a \\
                Z_{M_j}  X^a_j  I_a X_b Z^b_k  \\
                Z_b X_k^b Z_{M_k}
    \end{array}\right\rangle.
\end{equation}
Here, each generator \( g_j \) associated with a neighbor \( j \in N_a \setminus \{b\} \) has been updated, effectively transferring its entanglement to qubit \( b \).

Next, we measure qubit \( b \), obtaining outcome \( s_b \) for \( X_b \).  
We choose \( g_a \), which anticommutes with \( X_b \), and multiply it into all \( g_k \), yielding the updated stabilizer:
\begin{equation}
    \left\langle\begin{array}{c}
                s_b X_b  \\
                s_a X_a \\
                Z_{M_j}  X^a_j  I_a X_b Z^b_k  \\
                Z^a_j X_a X_k^b Z_{M_k}
    \end{array}\right\rangle,
\end{equation}
To simplify, we now eliminate \( X_a \) and \( X_b \) from the remaining generators by multiplying the transformed \( g'_j \) and \( g'_k \) with \( s_a X_a \) and \( s_b X_b \), respectively. This results in:
\begin{equation}
    \left\langle\begin{array}{c}
                s_b X_b  \\
                s_a X_a \\
                s_b Z_{M_j}  X^a_j Z^b_k  \\
                s_a Z^a_j X_k^b Z_{M_k}
    \end{array}\right\rangle,
\end{equation}

The net effect of these two measurements is a new stabilizer set for the remaining unmeasured qubits.  
For each original neighbor \( j \in N_a \setminus \{b\} \), its updated generator becomes:
\begin{equation}
g'_j = s_b X_j \left( \prod_{v \in N_j \setminus \{a\}} Z_v \right) \left( \prod_{k \in N_b \setminus \{a\}} Z_k \right).
\end{equation}
Similarly, for each original neighbor \( k \in N_b \setminus \{a\} \), its updated generator is:
\begin{equation}
g'_k = s_a X_k \left( \prod_{v \in N_k \setminus \{b\}} Z_v \right) \left( \prod_{j \in N_a \setminus \{b\}} Z_j \right).
\end{equation}
Interpreting these in the graph state formalism (see \cref{eq:graph-state-stabilizer-generators}), we see that each qubit \( j \) retains its original neighbors (excluding \( a \)) and gains new edges to all qubits in \( N_b \setminus \{a\} \), as indicated by the \( \prod_{k \in N_b \setminus \{a\}} Z_k \) term.  
Likewise, each \( k \) in \( N_b \setminus \{a\} \) becomes connected to all \( j \in N_a \setminus \{b\} \), as shown by the \( \prod_{j \in N_a \setminus \{b\}} Z_j \) term.  
The resulting graph corresponds precisely to a complete bipartite connection between the two original neighborhoods (as shown in the bottom right of \cref{fig:graph-state-example-in-rgs} as claimed), confirming the graphical transformation rule.

Having now extended the ruleset for graph state transformations and their associated effects under measurement, we focus on two key measurement sequences---single-qubit Z measurements and two-qubit XX measurements, to analyze how Bell pairs are distributed in the all-photonic quantum repeater scheme.

\section{Original all-photonic quantum repeaters}
\label{sec:original-rgs-recap}

We begin by reviewing the original proposal for all-photonic quantum repeaters by Azuma, Tamaki, and Lo~\cite{azuma-rgs}.
The core component of this scheme is a highly entangled photonic graph state known as the \emph{repeater graph state} (RGS).
For brevity, we will refer to this all-photonic quantum repeater scheme as the \emph{RGS scheme} throughout the rest of this thesis.

\begin{figure}[htb]
\centering
\includegraphics[width=\textwidth,keepaspectratio]{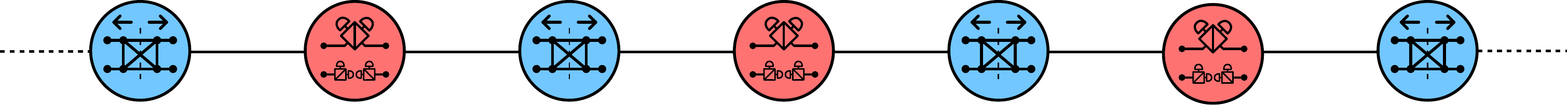}
\caption[An illustration of a connection path segment in the RGS scheme network]{An illustration of a segment of a connection path in a network based on the RGS scheme. Blue nodes are RGS source nodes (RGSS), and red nodes are Advanced Bell State Analyzers (ABSAs). Note that end nodes are often omitted in diagrams of the original scheme, as it was primarily designed for applications like QKD where the end nodes can be simple measurement devices.}
\label{fig:rgs-scheme-path}
\end{figure}

Using the terminology and node types defined in \cref{chapter:architecture-memory}, we can describe the RGS scheme as follows.
The nodes responsible for generating repeater graph states are referred to as repeater graph state sources (RGSS).
The intermediate measurement nodes, which must be capable of performing both Bell state measurements and adaptive single-qubit measurements, are known as advanced Bell state analyzers (ABSAs).
An example of a connection path segment in the RGS scheme is shown in \cref{fig:rgs-scheme-path}.
It is important to note that the original scheme was designed with applications like QKD in mind, and thus often implicitly assumes the end nodes (e.g., Alice and Bob) are measurement-based terminals, which are typically not drawn in the diagrams.

We will now delve into the original RGS scheme in more detail, highlighting its key innovations, particularly its resilience to photon loss and operational errors.
We will also describe the quantum error-correcting code it employs and briefly outline the conceptual and technological assumptions underlying the scheme.

\subsection{Overview of the RGS scheme}
\label{subsec:overview-of-rgs-scheme}

A repeater graph state (RGS) consists of two sets of qubits, \emph{inner} and \emph{outer} qubits, each containing $2m$ qubits (as shown in Step 1 of \cref{fig:rgs-scheme-overview}).
The inner qubits form a complete graph, and each inner qubit is connected to an outer qubit.
We refer to $m$ as the \emph{number of arms} of the RGS.
This complete graph among the inner qubits encodes the entanglement swapping, before photons are transmitted to establish link-level entanglement between adjacent neighbor nodes.
This leads to the interpretation of the RGS scheme as a ``time-reversed'' version of conventional quantum repeaters, where entanglement swapping occurs \emph{before} the link-level Bell pairs are generated.

To distribute a Bell pair between two end nodes (e.g., Alice and Bob), the RGS scheme employs a linear path made up of repeater graph state source (RGSS) and advanced Bell state analyzer (ABSA) nodes, analogous to a path consisting of memory--interference--memory (MIM) links in memory-based networks.

\begin{figure}[htb]
    \centering
    \includegraphics[width=\textwidth,keepaspectratio]{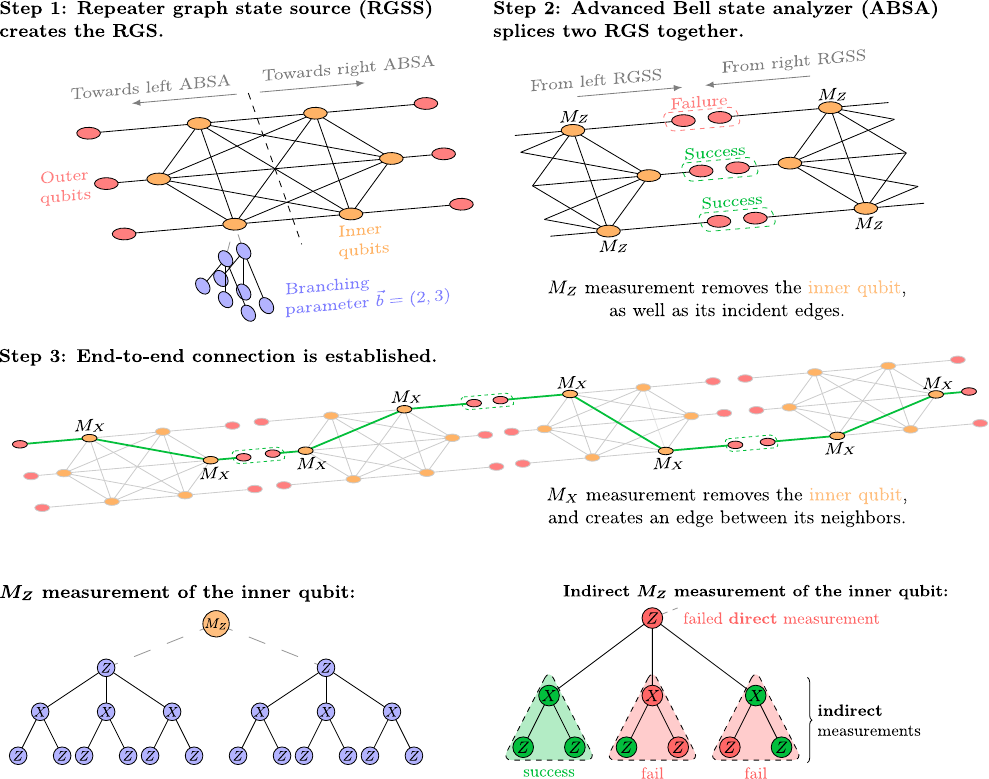}
    \caption[Overview of the RGS scheme]{An overview of the RGS scheme. The three steps shown here have corresponding analogs in memory-based repeaters. Inner qubits play the role of quantum memories, while outer qubits act as flying qubits. Step 1 (at RGSS) mirrors entanglement swapping, though no specific inner qubit pairs are selected. Step 2 (at ABSA) generates link-level entanglement via Bell state measurements (BSMs) on outer qubit pairs. The Z and X basis measurements in Steps 2 and 3 effectively determine which inner qubit pairs are swapped. Logical Z measurements are indicated at the bottom. Z/X labels on blue qubits refer to physical measurement bases; logical X measurements swap these roles. The triangle structures in the bottom right show how indirect Z basis measurement is possible if photon loss occurs.
    (From~\cite{naphan-half-rgs}.)}
    \label{fig:rgs-scheme-overview}
\end{figure}

As illustrated in \cref{fig:rgs-scheme-path}, each RGSS along the path generates an RGS.
Each RGS is split in two: one half is sent to the left ABSA and the other to the right ABSA.
Each ABSA receives the halves of \emph{two} RGSs, one from each neighboring RGSS, and performs \emph{three stages of adaptive measurement}.

In the first stage, the ABSA performs \emph{rotated Bell state measurements (BSMs)} between each pair of outer qubits from the two RGS halves (Step 2 of \cref{fig:rgs-scheme-overview}).
These rotated BSMs are equivalent to adding an edge between the two outer qubits and then performing an XX measurement.
This is repeated for all $m$ pairs.
Since BSMs implemented by linear optics are inherently probabilistic and photons may be lost in transit from RGSS to ABSA, the success of each measurement is uncertain.

If \emph{at least one} BSM succeeds at an ABSA, the protocol proceeds to the second stage.
Z basis measurements are performed on the inner qubits connected to outer qubits where the BSM failed.
This is because a Z measurement effectively removes a vertex (and its edges) from the graph state.
If \emph{no} BSMs succeed at an ABSA node, the connection attempt fails and must be restarted.
Inner qubits connected to successfully measured outer qubits are also Z measured, except for one surviving pair per ABSA.
This procedure prunes the graph, yielding a \emph{linear chain graph state} (as shown with green path in Step 3 of \cref{fig:rgs-scheme-overview}).

In the third and final stage (Step 3), the remaining inner qubits, now arranged in a linear chain, are measured in the X basis.
Since only two inner qubits remain per ABSA, these measurements effectively implement XX measurements.
The final result is a \emph{Bell pair} (i.e., a two-qubit graph state) shared between the two end nodes, up to local Pauli corrections.

Having discussed the conceptual foundations of the RGS scheme, we now turn to how it mitigates photon loss and operational errors---features that qualify it as a third-generation quantum repeater, albeit without support for hop-by-hop, packet-switching-style connections.

\subsection{Robustness of the RGS scheme: encoding and measurement of inner qubits}
\label{subsec:rgs-scheme-loss-and-error-mitigation}

The RGS scheme incorporates several mechanisms to address photon loss during link-level entanglement generation and to mitigate operational errors.
A primary strategy against photon loss in the Bell state measurements (BSMs) on the RGS's outer qubits is the use of multiple parallel ``arms,'' a parameter denoted by $m$.
Increasing $m$ enhances the probability that at least one BSM succeeds at each advanced Bell state analyzer (ABSA), a necessary condition for establishing an end-to-end entangled connection.

Successfully measuring the outer qubits, however, is only one aspect of the challenge.
The inner qubits of the RGS, which are also photonic and therefore susceptible to loss, must undergo successful measurement to complete the entanglement swapping protocol and prune the graph state into a final end-to-end Bell pair.
To guarantee arrivals of these vital inner qubits, the RGS scheme typically employs the \emph{loss-tolerant tree code encoding}~\cite{varnava-counter-factual-measurement}.
In this encoding, each logical inner qubit is represented by a tree-like graph state composed of multiple physical photonic qubits.
The specific structure of this tree, defined by its branching parameters $\vec{b} = (b_0, b_1, \ldots, b_{n-1})$, provides redundancy.
This redundancy, in conjunction with the number of arms $m$, determines the overall loss and error tolerance of the RGS scheme.
An illustrative example of such tree-encoded inner qubits can be seen in \cref{fig:rgs-scheme-overview}.

The logical state of these encoded inner qubits is determined by performing a pattern of single-qubit measurements on their constituent physical qubits.
All physical qubits at the same level within a tree are measured in the same basis, with the measurement basis (e.g., X or Z) alternating between successive levels, as depicted in \cref{fig:rgs-scheme-overview}.
For instance, to perform a logical measurement in the X basis (Z basis), we measure the first-level physical qubits in the X basis (Z basis), the second-level physical qubits in the Z basis (X basis), and so on, alternating the basis for each subsequent level.
The logical outcome of the measurement is then inferred from the parity of a specific subset of these physical measurement outcomes.
A logical Z basis measurement outcome is derived from the parity of eigenvalues from all physical qubits in the first level of the tree, whereas a logical X basis outcome is determined by the parity of an X basis measurement on one first-level qubit combined with Z basis measurements of its neighbors at the second level.

Even if a physical qubit is lost, its measurement result can sometimes still be deduced through a technique known as \emph{counterfactual or indirect measurement}~\cite{varnava-counter-factual-measurement}.
This principle, rooted in the stabilizer formalism of graph states, allows the intended measurement outcome of a physical qubit $a$ (which is part of a stabilizer generator $g_i$) to be inferred if all other physical qubits in the support of $g_i$ have been successfully measured, either directly or indirectly.
The inherent redundancy of the tree code thereby ensures that the logical outcome can often be consistently determined despite some physical qubit losses.
For details of how logical operators are defined, readers are referred to the supplementary material of~\cite{azuma-rgs}.

Beyond mitigating photon loss, the tree encoding also offers a degree of resilience against quantum operational errors that affect physical measurements.
For instance, errors in physical measurement outcomes in the Z basis on a qubit at level $k$ within the tree can often be corrected by performing a majority vote.
This vote utilizes the indirect measurement results for that qubit, which are inferred from its neighboring qubits at level $k+1$ and, in turn, their neighbors at level $k+2$.

Considering the \emph{logical} measurement outcomes, results from logical measurements in the X basis are amenable to correction through majority voting.
Since the logical X outcome is determined by the parity of a measurement in the X basis on a single first-level physical qubit and the measurements in the Z basis of its neighbors in the second level; if no errors occur, any such valid grouping for parity calculation will yield the same logical result, allowing for error detection and correction via a majority vote across these possible derivations.
In contrast, logical measurements in the Z basis, which depends on the combined parity of \emph{all} physical measurements in the Z basis on the first-level qubits, are not directly correctable by a similar majority vote mechanism at the logical level.
However, the individual physical measurements in the Z basis contributing to this logical outcome can still benefit from the aforementioned indirect measurement and majority voting techniques at the physical qubit level.
Consequently, the overall error tolerance of the RGS scheme is a function of both the specific tree code structure and the prevailing photon loss probability.

\subsection{Transmission order of photons in the RGS}

The measurement basis selection of the inner qubits is dependent on the BSM outcome of the outer qubits.
As long as the outer qubits arrive at the ABSA before their connected inner qubit, the measurement basis can be determined locally by the ABSA.
The physical qubits composing the inner qubits can be sent in any order as long as it is known to the ABSA beforehand, this is fixed and does not require any delay line to correct the ordering in our work as we will show the generation sequence in \cref{sec:generation-of-half-rgs}.
We note that even if the photons are lost, assuming that the photons are well separated temporally and time tagged, the ABSA can deterministically flag the loss event.
This well-separated assumption is also commonly adopted in memory-based repeater schemes for multiplexing photons from multiple memories into a single fiber~\cite{rdv-qi-architecture}.

\section{Half-RGS building block}
\label{sec:half-rgs}

While the original repeater graph state (RGS) connects its inner qubits in a complete graph, such full connectivity is not strictly necessary for entanglement swapping.
Prior work has demonstrated that a complete bipartite structure (biclique), formed by linking two halves of the inner qubits, is sufficient to achieve the same functionality~\cite{russo-rgs-biclique-generation, tzitrin-rgs-biclique-equivalent}, as illustrated in \cref{fig:half-rgs-biclique-rgs}.
This relaxation preserves the core capabilities of the RGS while potentially reducing the complexity of its generation~\cite{ghanbari-hoi-kwong-rgs-optimization, li-entangled-photon-factory}, and may also lower error rates given the same baseline fidelity of photon sources and detectors.

This insight motivates our first and most central contribution: the \emph{half-Repeater Graph State (half-RGS)}---a modular, emitter-friendly building block for constructing biclique RGSs.
The half-RGS is more than just a simplification, it is the foundation upon which nearly all contributions in this chapter are built.
It directly addresses key engineering and architectural challenges in the RGS framework: it simplifies photonic generation procedures, enables clean timing synchronization, aligns naturally with end-node constraints, and allows direct interfacing with memory-based repeaters.
Moreover, it introduces new possibilities for entanglement purification and modular scalability (\cref{sec:rgs-with-purification}).

In this section, we describe the structure and capabilities of the half-RGS, and demonstrate how this modular abstraction leads to practical, scalable designs in both pure all-photonic and heterogeneous quantum networks---between all-photonic and memory-based repeaters.
The modularity of the half-RGS building block offers significant advantages in addressing practical engineering challenges for RGS schemes, particularly in optimizing end-node resource requirements, simplifying timing synchronization, and facilitating the integration of all-photonic segments with memory-based repeater networks.

\begin{figure}[htb]
    \centering
    \includegraphics[width=\textwidth]{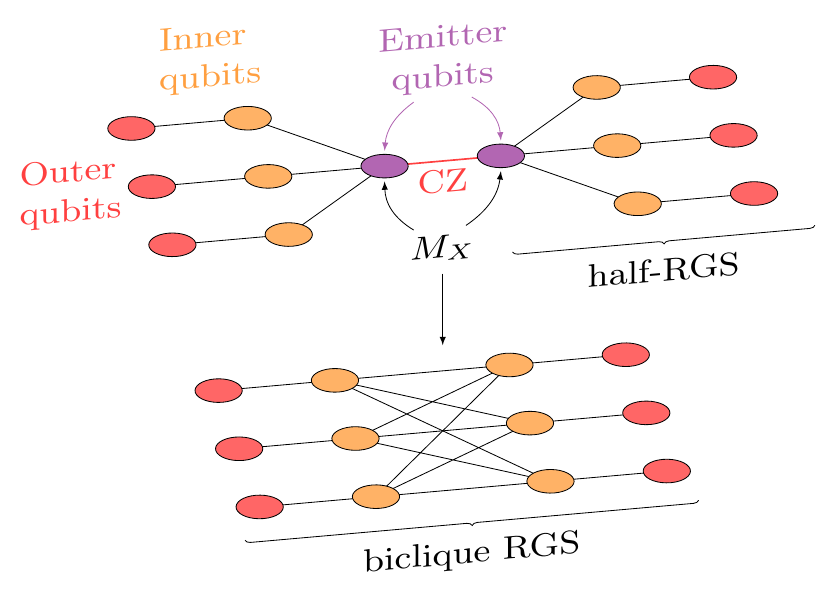}
    \caption[An illustration of a half-RGS and its transformation to a biclique RGS]{An example of a half-RGS (left and right structures) and the transformation of two such half-RGSs into a biclique RGS (center). The anchor emitter qubits (blue) of the two half-RGSs are joined via an application of a CZ gate followed by an XX measurement, resulting in the biclique RGS suitable for entanglement swapping. (From~\cite{naphan-half-rgs}.)}
    \label{fig:half-rgs-biclique-rgs}
\end{figure}

\subsection{The half-RGS construct}
\label{subsec:half-rgs-definition}

The \emph{half-RGS}, as illustrated in \cref{fig:half-rgs-biclique-rgs}, is a multipartite entangled state derived from one half of a standard Repeater Graph State (RGS) with a key structural modification.
A stationary quantum emitter or memory qubit, referred to as the \emph{anchor qubit} (shown in purple), is introduced and becomes entangled with the photonic portion of that half, which consists of a set of inner qubits and their corresponding outer qubits.
At a given repeater graph state source (RGSS) node, two such half-RGSs can be generated independently and then deterministically fused by applying a CZ gate between their anchor qubits followed by an XX measurement (i.e., a rotated Bell measurement).
This operation combines the two halves into a full biclique RGS, hence the name ``half-RGS.''
As we will show next, this modular structure addresses several longstanding challenges associated with the broader RGS framework~\cite{naphan-engineering-challenges-in-rgs}.

\subsection{Improving resource requirements at end nodes}
\label{subsec:half-rgs-end-nodes}

The integration of compute-capable end nodes into the RGS scheme, extending beyond simpler measurement-only nodes for QKD, has not been fully explored.
These compute nodes require quantum memories to store entangled states while awaiting classical messages from distant ABSAs and RGSSs before Bell pairs are ready for use.
One approach, explored in~\cite{zhan-graph-based-repeater-analysis} (and depicted at the top of \cref{fig:proposed-architecture}), necessitates that end nodes be equipped with at least $m$ (the number of RGS arms) emissive quantum memories per trial.
To sustain high trial rates, comparable to the RGS scheme's full potential, an additional $rm$ memories must be held in reserve.
This is because photons from previous trials would have already been emitted while the system awaits correction messages, where $r$ denotes the ratio of classical message travel time from the farthest ABSA to the RGS generation time.
For instance, based on the analysis in~\cite{hilaire-rgs-optimizing-gen-time} for a 1000~km separation between end nodes, each end node would require approximately 150 quantum memories to sustain high-rate operation, assuming $m = 15$ arms and a message delay ratio of $r = 10$.

\begin{figure}[tbp]
    \centering
    \includegraphics[width=\textwidth]{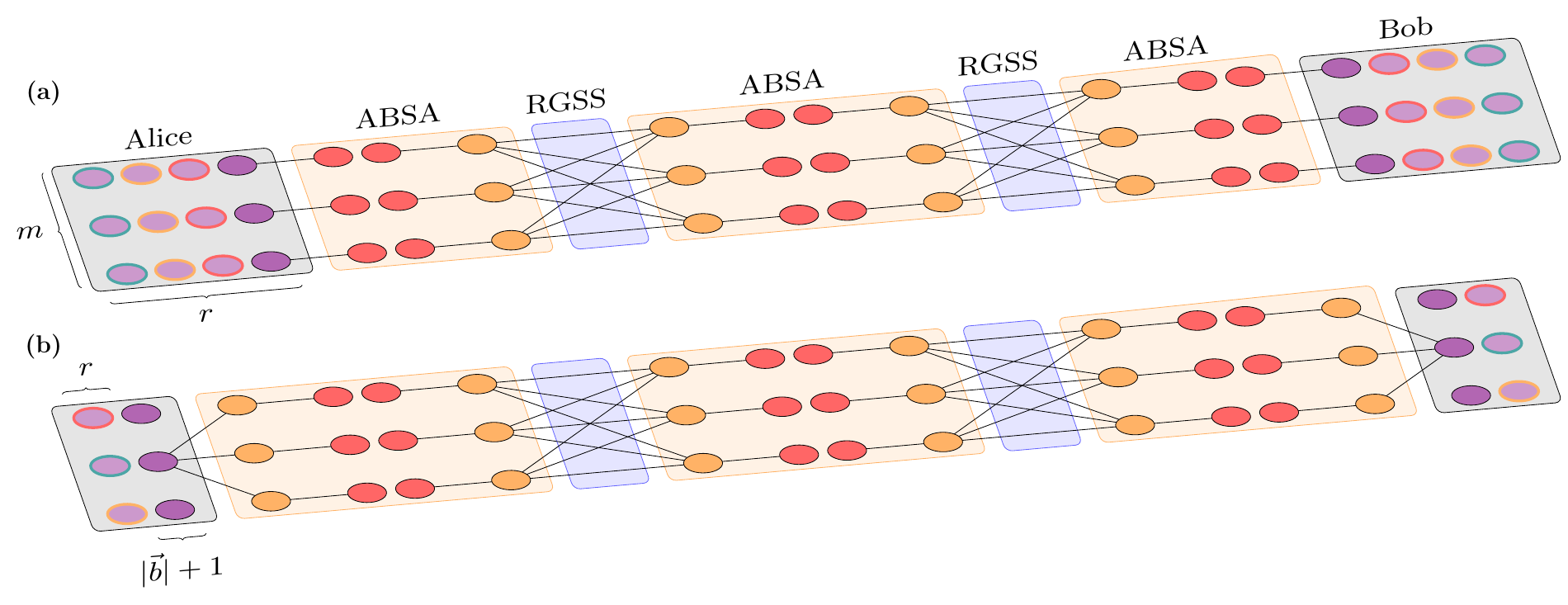}
    \caption[A comparison between two end node architectures for RGS schemes]{Architectures supporting the RGS scheme where the photonic states generated at end nodes are different. (a) The architecture proposed in~\cite{zhan-graph-based-repeater-analysis}, where end nodes require $m$ emissive memories for each trial, along with $rm$ idle memories that have already been utilized in prior trials and are awaiting messages from all ABSAs. (b) Proposed architecture, featuring end nodes equipped with half-RGS building blocks. In this setup, end nodes require $|\vec{b}|+1$ quantum emitters (where $|\vec{b}|$ represents the depth of the tree encoding plus one), while reserving $r$ memories awaiting notification messages from prior trials. Purple circles with black borders represent the quantum emitters utilized in the current trial. Purple circles with colored borders in both (a) and (b) indicate idle memories awaiting messages from ABSAs, with the same colors denoting memories participating in the same trials. (From~\cite{naphan-half-rgs}.)}
    \label{fig:proposed-architecture}
\end{figure}
In contrast, our proposed architecture leveraging the half-RGS building block (shown at the bottom of \cref{fig:proposed-architecture}) streamlines this requirement.
An end node would need only $r$ quantum memories (for storing states from previous trials awaiting classical corrections) and $|\vec{b}| + 1$ emitters (where $|\vec{b}|$ is the depth of the tree encoding) for generating its local half-RGS, resulting in a total of $r+|\vec{b}|+1$ combined memory and emitter units.
Crucially, the $r$ quantum memories do not need to be emissive if the state of the half-RGS anchor emitter can be swapped into them.
Furthermore, the emitters themselves do not require coherence times as long as those of quantum memories, given their role is primarily the generation of photonic qubits for the half-RGS.
This configuration can effectively double the entanglement trial rate compared to the $rm$ emissive memories approach (or $2rm$ if multiplexed fiber is used to match our rate) discussed in~\cite{zhan-graph-based-repeater-analysis}, as illustrated in \cref{fig:proposed-architecture}.

While this approach does imply that end nodes must be equipped with half-RGS generation capability, the hardware is identical to what is already required at repeater nodes, making the assumption reasonable.
A more end-node friendly approach would be that end nodes connected to a repeater bridging the RGS scheme and the port connecting to end nodes can be memory-based.

While this strategy implies that end nodes themselves must be equipped with Half-RGS generation capabilities, this requirement can be considered reasonable given that such hardware is functionally identical to that already necessary for the RGSS nodes within the all-photonic segments. Alternatively, for end nodes lacking this specific generation hardware, a more accommodating approach involves utilizing a memory-equipped repeater node to act as an interface. Such a repeater would bridge the all-photonic RGS segment to a conventional memory-based quantum link terminating at the end node, thereby allowing simpler, memory-based end nodes to connect to networks that leverage advanced RGS technology (see also \cref{subsec:half-rgs-path-virtualization}).

\subsection{Simplifying timing synchronization.}
\label{subsec:half-rgs-simplifying-timing-synchronization}

For entanglement swapping via Bell state measurements (BSMs) to succeed, the involved photons must arrive at the measurement station within an indistinguishable time window.
This requirement is universal across all quantum repeater schemes and stems from the physics of linear-optical two-photon interference, not from the design of any particular architecture.

What makes the RGS scheme uniquely challenging, however, is that the full RGS is typically generated by a single set of emitters~\cite{buterakos-graph-generation,hilaire-rgs-optimizing-gen-time}.
In this configuration, both halves of the photonic graph---destined for neighboring nodes on either side---are entangled with the same stationary qubits and must be emitted in a tightly coordinated sequence.
For the outer photons to arrive simultaneously at their respective Bell state measurement stations (ABSAs), the emitter must coordinate timing with both neighboring nodes.
This means that a single RGSS node must synchronize with both its left and right ABSAs, implicitly creating a dependency chain across the full connection path.
Timing messages must propagate along the entire route, turning synchronization into a global coordination problem.
This requirement imposes significant engineering complexity and hinders scalability~\cite{moriScalableTimingCoordination2024a}.

Our half-RGS approach resolves this issue by introducing modularity at the level of RGS generation.
Instead of creating both sides of the RGS from a single set of emitters, the biclique RGS is formed from two independently generated half-RGSs---each constructed by its own set of emitters (different QNICs) within the same RGSS node.
This design localizes timing control as synchronization is now reduced to standard neighbor-to-neighbor coordination, just as in conventional memory-based repeater schemes.
Timing messages are only exchanged between adjacent nodes and the local controller managing emitter operations, eliminating the need for global synchronization across the entire connection path.

\subsection{Integration with memory-based repeaters through virtualization of links}
\label{subsec:half-rgs-path-virtualization}

The integration of all-photonic repeaters---particularly the RGS scheme---into memory-based quantum repeater networks remains an underexplored but essential step toward scalable and heterogeneous quantum networks.
By reframing an entire RGS segment as a single virtual link-level connection between two memory-equipped devices---rather than as a chain of network-level repeaters---we enable its seamless use within existing memory-based protocols, as illustrated in \cref{fig:rgs-virtual-link}.
This abstraction allows the Bell pairs generated over the RGS chain to serve as fundamental link-level resources, similar to those produced by memory-based links like MIM, MSM~\cite{cody-msm}, or even Sneakernet~\cite{simon-sneakernet}.

Once generated, these Bell pairs can be directly handed off to the control stack (e.g., QRSA and RuleSets) at the terminating memory-equipped nodes, allowing seamless integration with existing entanglement management protocols.
This abstraction aligns with the design principles of the future Quantum Internet---modularity, heterogeneity, and interoperability---by allowing RGS links to be treated identically to other link technologies within connection setup and resource scheduling processes, as discussed in \cref{chapter:architecture-memory}.

\begin{figure}[htb]
\centering
\includegraphics[width=\textwidth,keepaspectratio]{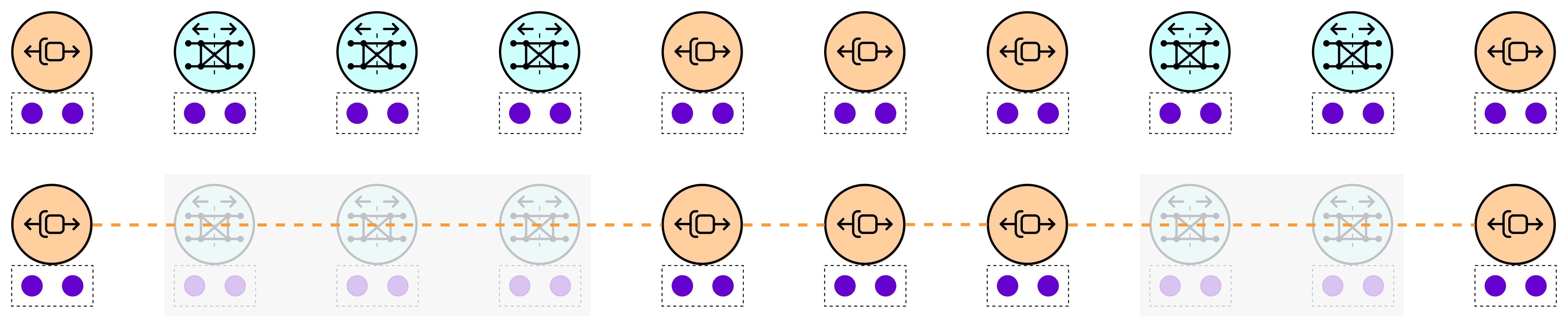}
\caption[Virtual link abstraction for RGS scheme segments]{Segments constituting RGSS (repeater graph state source, in light blue) and ABSA (advanced Bell state analyzer, omitted) nodes are treated as virtual links to memory-equipped repeaters (orange). The RGS scheme link is invisible at the connection-level protocol, reducing the cost of connection setup and management. (From~\cite{naphan-half-rgs}.)}
\label{fig:rgs-virtual-link}
\end{figure}

This virtual link abstraction fits naturally within the QRNA framework described in \cref{chapter:architecture-memory}, where RGS-based subpaths can be treated as single hop-level links between memory-equipped nodes.
At the boundaries of these subpaths, RuleSets and connection messages are translated into the lower-level RGS protocol---mirroring the internetworking process outlined in \cref{sec:architecture-memory:internetworking-qrna}.
Conceptually, an RGS subnetwork may function like a router, due to its internal coordination, or like an optical switch (OSW), owing to its abstraction of link-level and tight real-time signaling requirement.
While implementation specifics are left for future work, the key insight here is that encapsulating all-photonic repeater chains as virtual links is both technically viable and architecturally beneficial, paving the way for unified protocol designs across heterogeneous quantum technologies.

\section{Generation of half-RGS}
\label{sec:generation-of-half-rgs}

Generating photonic RGS remains one of the most challenging aspects of implementing the RGS scheme. Numerous generation methods have been proposed and refined over time~\cite{azuma-rgs,pant-rate-dist-tradeoff,buterakos-graph-generation,zhan-graph-gen-delay-line,shapourianModularArchitecturesDeterministically2023}, alongside optimization techniques that reduce resource consumption and generation time~\cite{li-entangled-photon-factory,kaur-patil-guha-rgs-generation,ghanbari-hoi-kwong-rgs-optimization}.

In this work, we adopt the approach of Buterakos et al.~\cite{buterakos-graph-generation}, which enables the deterministic generation of photonic graph states via quantum emitters.

To support deterministic generation of half-RGS using emitter-based schemes, we begin by clarifying the operational assumptions and primitive transformations used in our construction.

\subsection{Assumptions on the operations}
\label{subsec:assumptions-of-operations-in-half-rgs-generation}

We assume the following capabilities for the quantum emitters, summarized in \cref{fig:rgs-gen-sequence}.
Emitters are arranged in a linear topology $(Q_a, Q_e, Q_0, Q_1, \ldots, Q_{|\vec{b}|-1})$, where $Q_a$ is the anchor qubit and $Q_e$ is the emitter for outer qubits.
As a reminder, anchor qubits refer to emitters that are part of the half-RGS and serve to anchor photonic qubits (\cref{fig:half-rgs-biclique-rgs}).
Emitters support single-qubit Hadamard gates and controlled-phase gates between adjacent emitters.
Photon emission is modeled as a controlled-Not operation to a newly initialized qubit in the $\ket{0}$ state (\cref{fig:rgs-gen-sequence}(a)).
While we model this emission as deterministic for the purposes of the protocol logic, we acknowledge that no physical emitter can achieve perfect quantum efficiency.
Any inefficiency in single-photon emission is treated as a form of initial photon loss, a type of error that the RGS scheme is inherently designed to tolerate.

In addition to the standard graph state manipulation rules, we use \emph{push-out operations}, where an emitter is effectively replaced by its emitted photon and then disconnected from the rest of the graph~\cite{enocomou-2d-cluster-generation,russo-arbitrary-graph-state}.
This is implemented by emitting a photon followed by a Hadamard gate, as shown in \cref{fig:rgs-gen-sequence}(b).

\begin{figure}[htb]
\centering
\includegraphics[width=\textwidth,keepaspectratio]{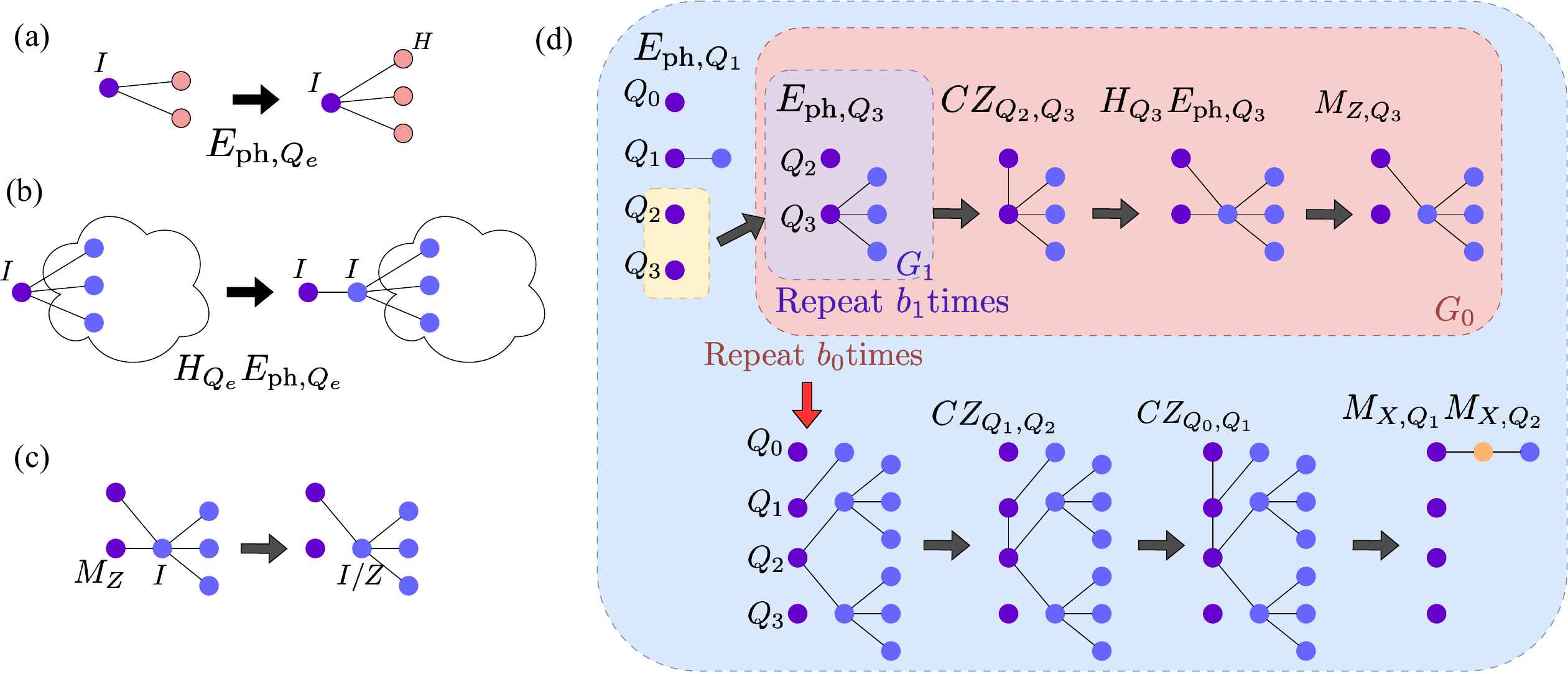}
\caption[An illustration of half-RGS generation via deterministic quantum emitters and side effects]{(a) Photon emission with side effects labeled. (b) Push-out operation: photon emission followed by a Hadamard gate, introducing no side effect. (c) Push-out followed by Z measurement, the main sequence used in half-RGS generation. (d) Generation sequence (adapted from~\cite{hilaire-rgs-optimizing-gen-time}) for one arm of half-RGS with $\vec{b} = (2, 3)$. Repeat $m$ times to construct $m$ arms. (From~\cite{naphan-half-rgs}.)}
\label{fig:rgs-gen-sequence}
\end{figure}

The Bell-state measurement (BSM) employed in our protocol is equivalent to applying a controlled-phase gate between the two qubits followed by an XX measurement.
Due to limitations of linear optics, the BSM succeeds with probability 50\%, which we model by treating odd-parity outcomes (i.e., 01 or 10) as success.
This choice is arbitrary; any two of the four Bell basis outcomes would suffice.

With these assumptions in place, we now describe the full generation sequence for a half-RGS state with arbitrary branching parameters and multiple arms.

\subsection{Sequences for generating the half-RGS}
\label{subsec:half-rgs-generating-sequences}

To generate a half-RGS with $m$ arms and branching vector $\vec{b} = (b_0, b_1, \ldots, b_{n-1})$, anchored at $Q_a$, we use a sequence adapted from~\cite{buterakos-graph-generation, hilaire-rgs-optimizing-gen-time}, as also illustrated in \cref{fig:rgs-gen-sequence}(d), which is given by
\begin{equation}
  \begin{aligned}
    & \left( M_{X,Q_e} M_{X,Q_0} CZ_{Q_a, Q_e} CZ_{Q_e, Q_0} G^{b_0}_{0} E_{\mathrm{ph},Q_e} \right)^m, \\
    & \text { with } G_k = M_{Z, Q_{k+1}} H_{Q_{k+1}} E_{\mathrm{ph}, Q_{k+1}} CZ_{Q_{k}, Q_{k+1}} G_{k+1}^{b_{k + 1}} \\
    & \text { and } G_{n-1}=E_{\mathrm{ph}, Q_{n-1}},
  \end{aligned}
\end{equation}
where $E_{\mathrm{ph}, Q_i}$ denotes photon emission from emitter $Q_i$, and $M_{P, Q_i}$ denotes measurement in basis $P$ on emitter $Q_i$.

Joining two half-RGS states into a full RGS is achieved by applying a CZ gate between the two anchor qubits ($Q_a$) and performing an XX measurement.
Although our generation sequence requires more emitters than some prior approaches, it simplifies both timing and resource management at ABSA nodes.
In particular, photon emission proceeds in the same order as measurements, which reduces the need for photon storage and mitigates loss in optical fibers.

Compared to previous methods~\cite{hilaire-rgs-optimizing-gen-time}, our design uses twice as many emitters per RGSS and two additional emitters for generating the outer qubits.
However, this cost is offset by significantly improved performance.
Our scheme cuts the generation time in half and eliminates the need for optical delay lines or efficient photon storage at ABSA nodes, as required in~\cite{buterakos-graph-generation, hilaire-rgs-optimizing-gen-time, zhan-graph-based-repeater-analysis, li-entangled-photon-factory}.
All photons are emitted in the same order as their corresponding measurements, ensuring minimal travel distance through fiber and reducing susceptibility to loss.

While this construction ensures deterministic generation and efficient timing, it introduces certain Clifford side effect operators during entangling gates and measurements.
These side effects must be tracked and corrected to ensure the proper formation of Bell pairs, as discussed next.

\subsection{Side effects created during the generation}
\label{sec:generation-of-half-rgs:side-effect}

The generation of half-RGS states and the subsequent measurements at ABSA nodes introduce Clifford side effects that must be carefully tracked and corrected to ensure proper state preparation.
During generation, leaf nodes corresponding to the inner and outer qubits inherit an $H$ side effect (\cref{fig:rgs-gen-sequence}(a)).
These can be corrected either optically, by applying Hadamard gates to the emitted photons, or logically, by swapping the X and Z measurement bases at ABSA.

Beyond the $H$ side effects, additional $Z$ side effects may arise.
In particular, qubits produced via a push-out operation followed by Z basis measurement of the emitter can acquire a Z side effect with 50\% probability (\cref{fig:rgs-gen-sequence}(c)).
Outer qubits may also pick up Z side effects from the XX measurement used to attach them to the inner qubits.

Further side effects are introduced when joining two half-RGS states to form a biclique RGS.
This step involves a CZ gate and an XX measurement between the anchor qubits (\cref{fig:half-rgs-biclique-rgs}), which can collectively toggle the Z side effect on all physical qubits in the first tree level (\cref{fig:graph-state-example-in-rgs}, bottom right).
There are two ways to handle these effects: either (1) track a logical Z side effect on the inner qubit or (2) track the side effect status of each individual physical qubit.
We adopt the second strategy---toggling the physical qubits---when describing the general protocol in \cref{subsec:agnostic-rgs-pauli-frame-correction}, and the first---tracking logical Z---when analyzing the protocol tailored to our half-RGS design in \cref{subsec:half-rgs-pauli-frame-correction}.

\section{The full RGS protocol}
\label{sec:full-rgs-protocol}

We now describe the communication protocol to realize the RGS scheme given our assumptions discussed above.
First, we will describe a biclique RGS generation agnostic protocol outlining where information about measurement and side effects originate and how they should be passed to intermediate and end nodes.
Later, in the following subsection, we will go into the generation specific with our proposed half-RGS interface leveraging graph state manipulation rules to argue about the calculations of Pauli frame corrections in arm-by-arm manner and hop-by-hop manners.

\subsection{Protocol overview}
\label{subsec:side-effect-correction}

A communication protocol for the RGS scheme must support fast generation and soft real-time processing.
In particular, ABSAs must be able to determine measurement bases based solely on the arrival order of photons, without waiting for classical instructions from the RGSS.
At the same time, the protocol must ensure that Pauli frame corrections are correctly tracked---even as the number of measurement results scales with the size of the RGS---so that end nodes can deterministically correct the final entangled state.

In each trial, the RGSS sends an RGS to its adjacent ABSA, followed by a classical message detailing side effects associated with each physical qubit.
These side effects---such as byproduct Pauli $Z$ operators introduced during graph state generation---are determined at the RGSS regardless of the specific generation method used~\cite{azuma-rgs, pant-rate-dist-tradeoff,buterakos-graph-generation,zhan-graph-gen-delay-line,li-entangled-photon-factory}.
Importantly, in some architectures (e.g., full RGSs with complete graphs), certain generation methods may impose additional side effects on emitted photons.
For example, if the emitter qubit is measured in the Y basis as in~\cite{hilaire-rgs-optimizing-gen-time}, then the first photon in a tree-encoded arm must be measured in the Y basis rather than X.
This means ABSAs must be aware in advance of the exact form of side effects that may appear on each photon, in order to choose the correct measurement basis.
Hence, side-effect information must be communicated ahead of time or synchronized on a per-trial basis between adjacent RGSS--ABSA pairs.

Once the RGS reaches the ABSA, Bell-state measurements (BSMs) are performed on the outer qubits.
The outcomes of these BSMs determine the measurement bases of inner qubits, and the ABSA proceeds to measure them accordingly.
To process these outcomes, each ABSA constructs a measurement tree per RGS arm, recording measurement outcomes and bases (or marking qubits as lost).
This tree mirrors the logical tree encoding and enables the ABSA to determine the logical measurement result locally.
The $Z$ side effects from generation, along with the measurement tree data, can then be compressed into just two classical bits per ABSA per trial: one indicating whether the logical $X$ result should be flipped, and another indicating whether decoding was successful.
These bits are sent to the end nodes, which then apply the appropriate Pauli corrections to recover a deterministic Bell pair.

\subsection{Generation-agnostic Pauli frame calculations}
\label{subsec:agnostic-rgs-pauli-frame-correction}

We now explain how Pauli frame corrections can be computed in a generation-agnostic manner---without relying on the internal steps of RGS creation.
This tracking is a crucial but often overlooked part of RGS-based protocol design, as many works focus solely on the final entangled state up to local Clifford equivalence, sidestepping the operational details of how Pauli frame information is extracted and propagated during execution.

While photons are physically measured in the order they arrive at the ABSA, we can conceptually treat them as being measured in any order, provided each outer qubit is measured before its corresponding inner qubits.
We now describe how to compute the resulting Pauli frame for the memory qubits at the end nodes.
\cref{fig:side-effect-propagations-and-corrections} outlines the procedure below.

\begin{figure}[htb]
    \centering
    \includegraphics[width=\textwidth]{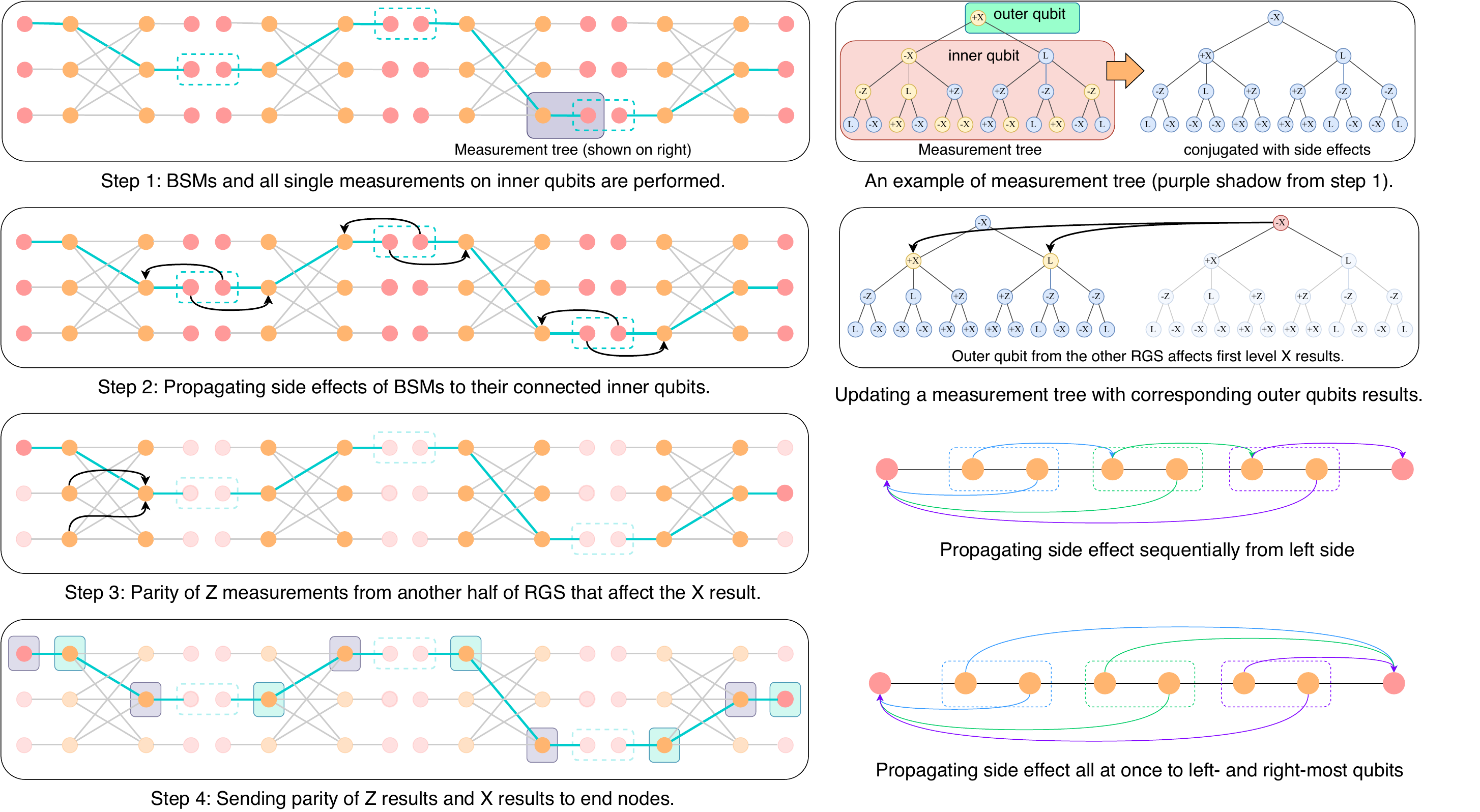}
    \caption[A depiction of Pauli frame correction procedure for generation agnostic RGS protocol]{The tracking of side effect propagation at each step, the resolution of physical and logical measurement results, and the rationale behind the Pauli frame corrections applied at end nodes, as illustrated in Steps 1--4 on the left side of the figure. Examples of measurement trees and their updating during the Pauli frame calculation process are presented in the top two panels on the right side. Yellow-filled vertices represent qubits with a $Z$ side effect. The sequential propagation of side effects, viewed hop-by-hop, is contrasted with an equivalent view of simultaneous propagation to the outermost qubits, depicted in the bottom two panels on the right side. (From~\cite{naphan-half-rgs}.)}
    \label{fig:side-effect-propagations-and-corrections}
\end{figure}

First, we assume that the success or failure of all BSMs on the outer qubits is known, so the measurement bases for all inner qubits are determined.
Before we begin tracking state changes due to measurements, each ABSA constructs a measurement tree for each inner--outer qubit pair in every RGS arm. This tree includes the side effect information about all physical qubits, as provided by the RGSS and end nodes.
The tree stores the measurement basis and result, or marks the qubit as lost. Measurement outcomes are denoted by $+$ (for 0) or $-$ (for 1), along with the basis (X or Z); loss events are marked by $L$.
If a qubit has a $Z$ side effect from generation (shown as a yellow-filled vertex in~\cref{fig:side-effect-propagations-and-corrections}), then its X basis measurement result is flipped, while Z basis results and loss entries remain unchanged.

Next, we propagate the BSM results obtained from the outer qubit pair to their connected inner neighbors, as shown in Step 2 of \cref{fig:side-effect-propagations-and-corrections}.
A BSM has the same effect as an XX measurement in \cref{fig:graph-state-example-in-rgs} and may introduce a $Z$ side effect.
Successful BSM outcomes propagate across the RGSs to the inner qubits, indicated by the black arrows in Step 2.
If the measurement outcome on the outer qubit from the left (right) is 1 (i.e., $-$X in the figure), then the measurement results of the first-level physical qubits in the right (left) tree are flipped.
We do not propagate $Z$ side effects to inner qubits whose outer neighbors were part of a failed or discarded BSM, since these inner qubits are measured in the Z basis and Z side effects do not affect Z measurement results.

The results of Z measurements on the inner qubits of one half of the RGS affect the X outcomes of the other half; i.e., the two halves sent to different ABSAs.
However, these measurements are performed at separate ABSAs, making their outcomes locally unavailable.
This cross-arm dependency is illustrated in Step 3 of \cref{fig:side-effect-propagations-and-corrections}.
Fortunately, as shown in the bottom two panels on the right side of the figure, the $Z$ propagation from XX measurements can be handled at the end nodes via the parity product of measurement outcomes.
This same logic applies to any $Z$ side effect introduced by inner qubits measured in the Z basis.
Each ABSA computes the parity of its local $Z$ side effects (i.e., the product of the individual Pauli $Z$) needing corrections, and sends this partial parity to the end node, which combines them to determine the total parity.

Therefore, each ABSA needs to send only two bits of information to the end nodes per trial: one bit indicating whether a $Z$ side effect propagates to that side (based on the parity), and one bit indicating whether the procedure was successful.
The number of classical bits that each ABSA receives is equal to the number of photons it measured.
We will discuss this total amount of classical bits to be sent in \cref{subsec:one-vs-two-stage-correction}

\subsection{Half-RGS specific Pauli frame calculation}
\label{subsec:half-rgs-pauli-frame-correction}

We now show how our proposed half-RGS building block, generated using quantum emitters, enables straightforward calculation of Pauli frame corrections and naturally leads to a simple graphical method for determining the end-to-end fidelity of the protocol (discuss in \cref{sec:half-rgs-fidelity-derivation}).
As in the generation-agnostic approach previously described, we begin with Step 1 of \cref{fig:side-effect-propagations-and-corrections}, but the subsequent steps differ.

\begin{figure}
    \centering
    \includegraphics[width=\linewidth]{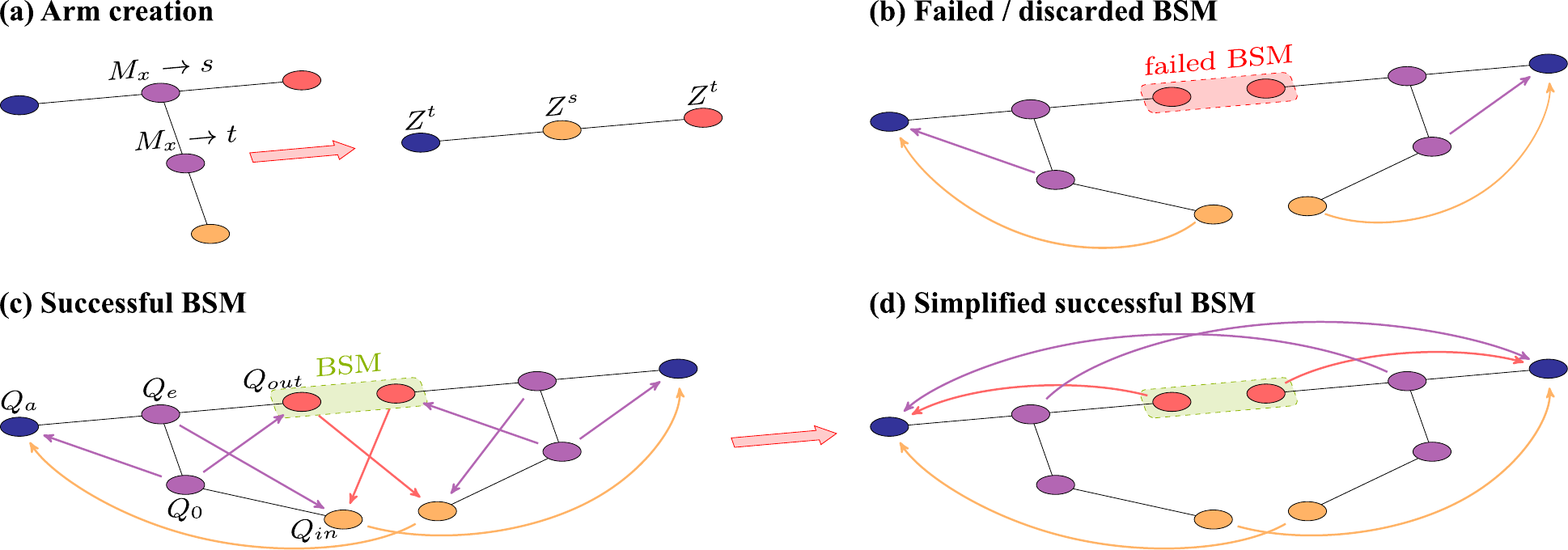}
    \caption[An illustration of side effect propagation for a single arm of half-RGS generation]{Side effect propagation for an arm of a one-hop RGS link from half-RGS. An arrow from $a$ to $b$ indicates that if the measurement result of $a$ is 1, the result of $b$ should be flipped. (a) Side effect propagation when creating an arm attached to anchor qubits via XX measurements on emitters joining inner and outer qubits. (b) Propagation for failed or discarded BSMs, where inner qubits are measured in the Z basis. (c) Side effect propagation for a successful BSM arm, with inner qubits measured in the X basis. (d) A simplified view of (c), where the tips of arrows are moved to anchor qubits by following the arrow chain.}
\label{fig:side-effect-half-rgs}
\end{figure}
With our construction of half-RGS blocks, all photons are generated in the same order they are sent and measured at an ABSA, eliminating the need for optical delay lines.
Each arm of the half-RGS is created sequentially, which allows us to track and propagate $Z$ side effects of each arm to the anchor qubit during the generation process.
Importantly, the half-RGS anchor qubit never emits photons, and the only operation performed on it---namely, the CZ (used to join and form the biclique RGS at the final step)---commutes with $Z$ side effects.
Thus, this tracking approach is sound.

First, let us define the qubits involved in this one-hop picture, as shown in \cref{fig:side-effect-half-rgs} (corresponding to the second-to-last step of \cref{fig:rgs-gen-sequence}, where tree-encoded physical qubits are represented as a single logical orange qubit).
Let $Q^L_a$, $Q^L_e$, $Q^L_0$, $Q^L_{out}$, and $Q^L_{in}$ denote the anchor qubit of the half-RGS, emitter for outer qubits, anchor of the inner logical qubits, outer qubit, and inner logical qubit, respectively, for the left side incoming to the ABSA.
Similarly, $Q^R_j$ denotes the corresponding qubits for the right incoming side.

Recall the side effects introduced by the generation sequence of the half-RGS as described in \cref{sec:generation-of-half-rgs:side-effect}.
Graphically (see \cref{fig:side-effect-half-rgs}), if a measurement in the RGS protocol introduces a side effect (flipping the result of another qubit's measurement) when the outcome is ``1,'' we draw an arrow from the measured qubit to the affected one.
The arrows from inner and outer qubits correspond to the view in the last step of \cref{fig:rgs-gen-sequence}, after $Q_1$ and $Q_2$ have been measured out.
The arrows from $Q_0$ and $Q_1$ follow from the generation sequence in the same figure, where inner and outer qubits become part of the RGS arm attached to the anchor.

The figure depicts two measurement patterns: one where inner qubits are measured in the X basis and another where they are measured in the Z basis (left and right halves, respectively).
From this graphical picture, we observe that a single flipped measurement result propagates through the chain of arrows and always affects an anchor qubit.
This lets us simplify the entire effect of each measurement to a possible flip of the $Z$ side effect on only the anchor qubits (lower half of \cref{fig:side-effect-half-rgs}).

By repeating this process for each arm of a one-hop RGS link between two half-RGSs, we can compute the full Pauli frame corrections.
To generalize this to a multi-hop linear chain of RGS links, recall that at each RGSS node, the two halves of the half-RGS are joined to form a biclique RGS using a CZ gate followed by XX measurements.
Because the tracked $Z$ side effect in the one-hop picture commutes with CZ and flips X basis measurement outcomes, we simply flip the reported measurement result if the anchor has a $Z$ side effect.
These measurement results, which are then sent to the end nodes for post-processing, indicate the Pauli frame corrections---similar to Step 4 in \cref{fig:side-effect-propagations-and-corrections}.

\subsection{One-Stage vs Two-Stage Pauli frame correction}
\label{subsec:one-vs-two-stage-correction}

An important aspect of RGS protocols is how classical information from measurements is processed to determine the Pauli frame corrections at the end nodes.
Each measurement generates a classical bit and may introduce side effects that must be tracked.
We now compare two strategies for processing and communicating this classical data: the conventional \emph{One-Stage Correction Method}, and the \emph{Two-Stage Correction Method} we propose (\cref{subsec:agnostic-rgs-pauli-frame-correction,subsec:half-rgs-pauli-frame-correction}).

In the conventional One-Stage Method, implied in most literature~\cite{azuma-rgs,hilaire-rgs-optimizing-gen-time}, every intermediate node---including all Advanced Bell State Analyzers (ABSAs) and Repeater Graph State Sources (RGSSs)---forwards all measurement outcomes to the end nodes without performing any local post-processing.
This includes results from every physical photon measured along the entire RGS chain and their side effects during the generation, potentially numbering in the hundreds or thousands depending on the encoding parameters (e.g., $m=14$, $\vec{b}=(10,5)$) per RGS.
The end nodes then process this full dataset to reconstruct the logical qubit outcomes and calculate the required Pauli frame corrections.
While conceptually straightforward and easier to implement, this approach incurs significant classical communication overhead and places a heavy processing burden on the end nodes, which must interpret the entire history of side effect propagation.

In contrast, our Two-Stage Method introduces local pre-processing at each ABSA.
Rather than transmitting all raw photon measurement results, the ABSA collects and reduces this data---along with any necessary side-effect tracking from neighboring RGSSs---down to a small number of summary bits.
Specifically, each ABSA computes and forwards only two bits: the effective parity of $Z$ side effect for end nodes and a bit inidiciating the success of the procedure.
These summary bits are then passed along the chain, and the end nodes use this compressed information to determine the final Pauli corrections.

\begin{figure}[htbp]
    \centering
    \includegraphics[width=\textwidth, keepaspectratio]{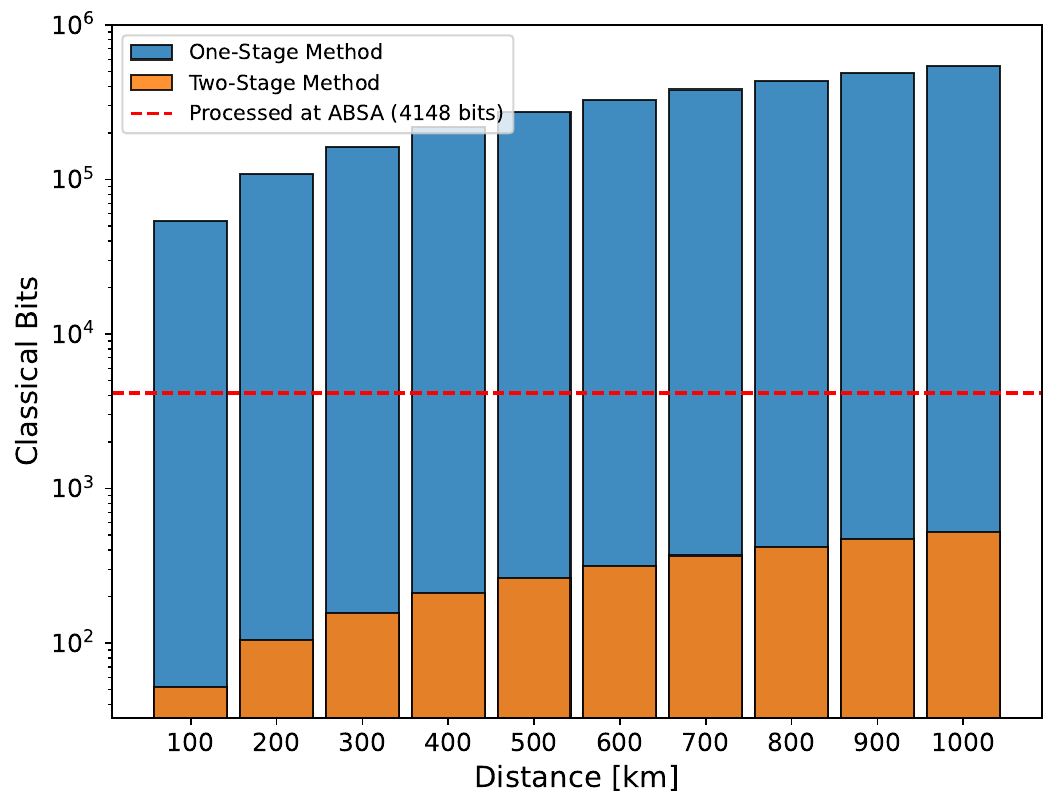}
    \caption[Classical communication bits comparison One-Stage and Two-Stage Pauli frame correction of RGS protocol]{The number of classical bits to be processed (equivalent to the total number of photonic qubits involved in generating the RGS arms for one Bell pair attempt) versus the separation distance between two end nodes. This assumes a near-optimal RGS structure with $m = 14$ arms and inner qubit tree-encoding parameters $\vec{b} = (10, 5)$. Blue bars represent the classical bits processed by end nodes in the One-Stage Correction Method. Orange bars show the significantly reduced bits processed by end nodes in the Two-Stage Correction Method. The red dashed line indicates the constant, relatively small number of bits processed by each intermediate ABSA in the Two-Stage Method. (From~\cite{naphan-engineering-challenges-in-rgs}.)}
    \label{fig:communication-cost}
\end{figure}
This distributed correction strategy has two major advantages.
First, the end nodes receive and process far fewer classical bits---only two bits per ABSA instead of all individual measurement results---resulting in a much lighter computational load.
Second, the total classical communication volume across the network is significantly reduced.
As illustrated in \cref{fig:communication-cost}, the Two-Stage Method cuts down the number of bits the end nodes must process by up to three orders of magnitude compared to the One-Stage Method.
The constant and relatively small processing load per ABSA, shown as the red dashed line, makes this a scalable and efficient approach.

By optimizing classical side-effect management in this way, our Two-Stage Method improves the practicality of RGS-based quantum repeater protocols, particularly for large-scale quantum networks where classical bandwidth and processing resources are constrained.

\subsection{Correctness of the protocol and the propagation rules}
\label{subsec:half-rgs-simulation-validation}

We validated the correctness of the two protocols, generation-agnostic and half-RGS building block, using the stabilizer tableau simulator Stim~\cite{gidney-stim}.
The code for this validation can be found at \cite{rgs-github-repo}.
In the generation-agnostic simulation, each physical qubit forming an inner qubit is directly created from quantum circuits using Hadamard and CZ gates without measurements.
To simulate the stochastic $Z$ side effect on each physical qubit within an inner logical qubit, a $Z$ gate is applied with 50\% probability after state generation.
Subsequently, inner--outer qubit arms are generated following the procedure in \cref{sec:generation-of-half-rgs} until half-RGSs are obtained and transformed into biclique RGSs.
Once all RGSs are generated at each RGSS, along with the two half-RGSs from the end nodes, BSMs are performed on all outer qubit pairs, and measurement bases are selected for all inner qubits.
Finally, all physical qubits in the system are measured simultaneously in a single circuit layer.

Instead of directly simulating message passing, the correction procedure is executed step by step, as outlined in \cref{subsec:agnostic-rgs-pauli-frame-correction}.
For each RGS arm, a measurement tree is created, followed by updating the measurement tree with side effect information and $Z$ propagation from the BSM, and the decoding of the logical result.
Finally, the 4 bits of information at each ABSA are processed, and correction operations are applied to the end nodes' memories accordingly.

For the half-RGS specific protocol---where biclique RGS is created via our proposed half-RGS building blocks---we simulate it hop-by-hop similarly to how we described it.
This helps speed up the simulation as the total number of qubits during the simulation is bounded by only a single hop independent of how many hops we want to simulate.
All the photons are generated exactly as described in \cref{sec:generation-of-half-rgs}, photon-by-photon, confirming that our Pauli frame calculation and also the generation sequence create the right state.

To confirm that the final state between two end nodes is a two-vertex graph state without any side effect, we used Stim to examine the stabilizer generators of the final two qubits.
We verified that the stabilizer generators are $X_a Z_b$ and $Z_a X_b$ as expected, where the subscripts $a$ and $b$ denote the qubit held by Alice and Bob, respectively.
Simulating the photon loss event is achieved by probabilistically marking qubits as lost with some probability.
If a qubit is marked as lost, we randomly apply Pauli $X$, $Y$, or $Z$ before the measurement, and the results are then removed from the measurement tree, excluding their participation from the logical qubit decoding process.

\section{Fidelity derivation of the end-to-end Bell pairs}
\label{sec:half-rgs-fidelity-derivation}

The fidelity of end-to-end Bell pairs is a key performance metric in quantum networks, quantifying the closeness of the final shared quantum state to an ideal Bell state.
In this section, we will derive this fidelity using the single-hop half-RGS architecture, guided by the Pauli frame propagation behavior illustrated in \cref{fig:side-effect-half-rgs}.

\subsection{Error modeling assumptions}
\label{subsec:rgs-error-modeling-assumptions}

\begin{figure}
    \centering
    \includegraphics[width=0.6\columnwidth]{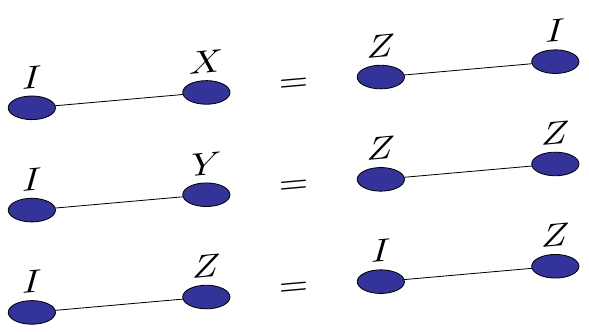}
    \caption[An illustration of equivalent side effects on two-qubit graph states]{The equivalence descriptions of graph states with side effects. The graph states on the left side represent Bell pairs with $X$, $Y$, and $Z$ Pauli errors, while the right side shows how we model errors in this work using graph state representations.}
    \label{fig:two-qubit-graph-state-error-model}
\end{figure}
Our error model assumes that errors manifest as independent single-qubit Pauli error channels acting on each photonic qubit involved in the RGS.
For simplicity, and in line with standard approaches~\cite{azuma-rgs,zhan-graph-based-repeater-analysis,hilaire-rgs-optimizing-gen-time}, we consider quantum emitters, their associated gate operations, and measurement devices to be ideal and noiseless.
This simplification is justified because noise from these components can be effectively approximated as an increased effective noise channel on the emitted photons themselves~\cite{lindner-rudolf-machine-gun-graph-state}.

While standard approaches for modeling Bell pair errors typically apply Pauli errors ($X$, $Y$, or $Z$) to just one qubit---exploiting the symmetry of Bell states---we adopt a different viewpoint tailored to our protocol.
Since our target state is a two-vertex graph state (equivalent to a Bell state), we can equivalently track only $Z$ type side effects accumulating on each of the two end-node qubits.
This modeling choice simplifies error propagation across hops without loss of generality.

Indeed, for a two-qubit graph state, these approaches are equivalent up to stabilizer transformations.
As shown in \cref{fig:two-qubit-graph-state-error-model}, a Pauli $X$, $Y$, or $Z$ error on Alice's qubit corresponds respectively to a $Z$ error on Bob's, $Z$ errors on both, or a $Z$ error on Alice's qubit---each equivalent under multiplication by stabilizer generators $Z_a X_b$ or $X_a Z_b$.
A key insight in our fidelity analysis is that the error probability for a single-hop half-RGS link can be decomposed into two contributions:
(i) inner qubit contributions, which are local and uncorrelated between the two end nodes; and (ii) outer qubit contributions, which induce correlated errors across the link.
The error on a memory qubit at an end node depends on whether the Pauli frame at its corresponding anchor is flipped.
This occurs when an \emph{odd} number of side-effect-inducing errors affect the qubits pointing to the anchor.
An even number of such errors cancels out.
Thus, the error probability can be computed by evaluating the parity of errors influencing each memory qubit.

\subsection{Error propagation from outer qubit operations}
\label{subsec:rgs-error-prob-outer-qubits}

We first analyze how errors arising from operations on the \emph{outer qubits} contribute to the effective error on the Bell pair after link-level entanglement swapping.
These outer qubits participate in Bell state measurements (BSMs) that connect adjacent half-RGS segments, analogous to elementary links in traditional repeater chains.
For this analysis, we can conceptually consider all inner qubits to have been measured out first (as measurement order can be rearranged without affecting the final state apart from classical processing of outcomes), allowing us to isolate the BSMs on outer qubits.

Each outer qubit is subjected to an independent single-qubit Pauli error channel.
Focusing on successful transmission events (i.e., no photon loss for these qubits), each resulting Bell pair formed by these outer qubits before a BSM can be described by a probabilistic mixture of Pauli errors.
In our graph state framework, this is represented by a four-element probability vector $\mathbf{e}$ given by
\begin{align}
    \mathbf{e} &= [w, x, y, z]^T \\
               &\coloneq w\ket{\psi_{II}}\bra{\psi_{II}} + x\ket{\psi_{ZI}}\bra{\psi_{ZI}} + y\ket{\psi_{ZZ}}\bra{\psi_{ZZ}} + z\ket{\psi_{IZ}}\bra{\psi_{IZ}},
\label{eq:bell-pair-error-vector}
\end{align}
where each $\ket{\psi_{P}}$ is a two-qubit graph state (Bell pair) affected by an effective Pauli error $P \in \{II, ZI, ZZ, IZ\}$, with $w, x, y, z$ being their respective probabilities.
When a BSM (specifically, the XX measurement following a CZ gate, as common in RGS protocols) is performed between two such Bell pairs, the output error vector is computed via the transformation given by
\begin{align}
    \text{BSM}_{XX} \left(
        \begin{bmatrix}
            w_1 \\
            x_1 \\
            y_1 \\
            z_1
        \end{bmatrix},
        \begin{bmatrix}
            w_2 \\
            x_2 \\
            y_2 \\
            z_2
        \end{bmatrix}
    \right)
    =
    \begin{bmatrix}
        w_1 w_2 + x_1 x_2 + y_1 y_2 + z_1 z_2 \\
        w_1 x_2 + x_1 w_2 + z_1 y_2 + y_1 z_2 \\
        w_1 y_2 + y_1 w_2 + x_1 z_2 + z_1 x_2 \\
        w_1 z_2 + z_1 w_2 + x_1 y_2 + y_1 x_2
    \end{bmatrix}.
\label{eq:bsm-graph-state}
\end{align}
To compute the error across the entire path, we can apply \cref{eq:bsm-graph-state} iteratively at each hop between the end nodes, ultimately yielding the final error probability vector.

For a special case of $N-1$ hops with the same depolarizing errors for every photons.
We can derive the final error probability vector for an end-to-end Bell pair established across a chain of $N$ initial elementary links.
Each elementary link is assumed to start with an identical error vector $\mathbf{e}^{(0)} = [w_0, x_0, y_0, z_0]^T$.
The end-to-end entanglement is created by performing $N-1$ successive BSMs, each extending the entanglement by one link.

Let $\mathbf{e}^{(k)} = [w^{(k)}, x^{(k)}, y^{(k)}, z^{(k)}]^T$ denote the error vector of an entangled pair that spans $k+1$ elementary links; i.e., the result of $k$ BSM operations.
The process is recursive; the $k$-th BSM operation combines the entangled pair from the previous $k-1$ swaps (with error vector $\mathbf{e}^{(k-1)}$) with a fresh elementary link (error vector $\mathbf{e}^{(0)}$).
Using the BSM transformation rule defined previously, this can be expressed as a linear recurrence relation
\begin{equation}
    \mathbf{e}^{(k)} = M \cdot \mathbf{e}^{(k-1)},
\end{equation}
where the transformation matrix $M$ is determined by the components of the initial error vector $\mathbf{e}^{(0)}$ given by
\begin{equation}
    M = \begin{pmatrix}
    w_0 & x_0 & y_0 & z_0 \\
    x_0 & w_0 & z_0 & y_0 \\
    y_0 & z_0 & w_0 & x_0 \\
    z_0 & y_0 & x_0 & w_0
    \end{pmatrix}.
\end{equation}
The final state after $N-1$ swaps, which spans all $N$ initial links, is therefore $\mathbf{e}^{(N-1)} = M^{N-1} \mathbf{e}^{(0)}$.

To compute $M^{k_s}$ where $k_s = N-1$ is the number of swaps, we can diagonalize $M$ as matrix $M$ is real and symmetric, and its structure allows for straightforward determination of its eigenvalues.
These eigenvalues are
\begin{align}
    \lambda_0 &= w_0 + x_0 + y_0 + z_0 = 1 \quad \left(\text{as } \sum p_i = 1\right), \label{eq:lambda0} \\
    \lambda_1 &= w_0 + x_0 - y_0 - z_0, \label{eq:lambda1} \\
    \lambda_2 &= w_0 - x_0 + y_0 - z_0, \label{eq:lambda2} \\
    \lambda_3 &= w_0 - x_0 - y_0 + z_0. \label{eq:lambda3}
\end{align}
The initial state $\mathbf{e}^{(0)}$ can be decomposed in the eigenbasis of $M$. After $k_s = N-1$ applications of the transformation $M$, the components of the final error vector $\mathbf{e}^{(N-1)} = [w^{(N-1)}, x^{(N-1)}, y^{(N-1)}, z^{(N-1)}]^T$ are given by the standard solution for such iterated Pauli channels given by
\begin{align}
    w^{(N-1)} &= \frac{1}{4} \left[ \lambda_0^{N} + \lambda_1^{N} + \lambda_2^{N} + \lambda_3^{N} \right] \nonumber \\
              &= \frac{1}{4} \left[ 1 + \lambda_1^N + \lambda_2^N + \lambda_3^N \right], \label{eq:final_w} \\
    x^{(N-1)} &= \frac{1}{4} \left[ 1 + \lambda_1^N - \lambda_2^N - \lambda_3^N \right], \label{eq:final_x} \\
    y^{(N-1)} &= \frac{1}{4} \left[ 1 - \lambda_1^N + \lambda_2^N - \lambda_3^N \right], \label{eq:final_y} \\
    z^{(N-1)} &= \frac{1}{4} \left[ 1 - \lambda_1^N - \lambda_2^N + \lambda_3^N \right]. \label{eq:final_z}
\end{align}
In these equations, $\lambda_1, \lambda_2, \lambda_3$ are defined by \cref{eq:lambda1,eq:lambda2,eq:lambda3} using the components $w_0, x_0, y_0, z_0$ of the initial error vector $\mathbf{e}^{(0)}$ of each elementary link, and $N$ is the total number of elementary links in the chain. The fidelity of the final distributed Bell pair with respect to the ideal state $\ket{\psi_{II}}$ is given by $w^{(N-1)}$. This derivation quantifies how the initial Pauli error probabilities on elementary links propagate through a chain of $N-1$ entanglement swaps.

\subsection{Error probability from inner qubits}
\label{subsec:rgs-error-prob-inner-qubits}

The error contributions from inner qubits arise from measurements on their constituent physical photons and the subsequent impact on Pauli frame tracking.
These contributions are effectively independent for Alice's and Bob's sides of the final entangled pair, allowing their respective error probabilities to be defined and calculated separately.
To facilitate a compact derivation, we will assume an identical depolarizing channel affects all physical photonic qubits within the inner structures.
Under this symmetric model, if $p_\text{dep}$ is the single-photon depolarization probability, the effective single-photon measurement error probability relevant for our analysis $p_{\text{ph}}$ can be considered as $\frac{2}{3} p_\text{dep}$ since $X$ and $Y$ error would flip the results when measuring in the Z basis, and similarly for X and Y measurements.
While this specific derivation assumes uniformity, the general approach can be extended to cases with non-identical or asymmetric Pauli error channels across different hops or within distinct tree code structures.

Let $p_{ZI}$ and $p_{IZ}$ denote the probability that a single-hop link's inner qubit measurements induce an effective $Z$ error on the left (Alice's side) and right (Bob's side) memory (anchor qubit), respectively.
Assuming there are $N$ hops between the two end nodes and independent errors across these hops, the total end-to-end error probabilities for Alice's $\left(p^{\text{e2e}}_{ZI}\right)$ and Bob's $\left(p^{\text{e2e}}_{IZ}\right)$ memory qubits are given by summing over an odd number of such errors as
\begin{align}
    p^{\text{e2e}}_{ZI}
        &= \sum_{k \text{ odd}}^N \binom{N}{k} p_{ZI}^k (1 - p_{ZI})^{N-k} \nonumber \\
        &= \frac{1}{2} \left( 1 - (1 - 2 p_{ZI})^N \right), \label{eq:e2e-error-prob-zi} \\ 
    p^{\text{e2e}}_{IZ}
        &= \sum_{k \text{ odd}}^N \binom{N}{k} p_{IZ}^k (1 - p_{IZ})^{N-k} \nonumber \\
        &= \frac{1}{2} \left( 1 - (1 - 2 p_{IZ})^N \right), \label{eq:e2e-error-prob-iz} 
\end{align}
assuming $p_{ZI}, p_{IZ} < 0.5$, as to avoid a completely mixed state at the end.

The single-hop error probabilities, $p_{ZI}$ and $p_{IZ}$ (which are equal to a common $p_{\text{in-hop}}$ in the symmetric case), capture the likelihood that an odd number of physical errors within the $m$ arms of an RGS segment's inner qubit measurements results in a net flip of the Pauli frame for that hop.
This probability is given by
\begin{align}
    p_{\text{in-hop}} = p_{ZI} = p_{IZ} &= \frac{1}{2} \left(1 - \left(1 - 2p_{\text{in}}^{X}\right) \left(1 - 2p_{\text{in}}^Z\right)^{m - 1} \right),
\label{eq:half-rgs-inner-error-single-hop} 
\end{align}
where $p_{\text{in}}^{P}$ denotes the probability that a logical measurement in basis $P \in \{\text{X}, \text{Z}\}$ on a tree encoded inner qubit is decoded incorrectly.
The average error probability of logical measurements $p_{\text{in}}^{P}$ can be derived from underlying physical error parameters (such as $p_{\text{ph}}$ or $p_\text{dep}$), taking into account both photon loss probabilities and physical measurement error probabilities within the tree code structure, as detailed in analyses such as~\cite{hilaire-rgs-optimizing-gen-time}.
While \cref{eq:half-rgs-inner-error-single-hop} provides the per-hop probability for identical links, in a general network where these probabilities might vary from hop to hop, the end-to-end calculation would involve a more generalized combination of these per-hop error events, though the fundamental principle of summing over configurations leading to an odd number of total Pauli frame flips remains.

\subsection{Combined end-to-end fidelity derivation}
\label{subsec:rgs-combined-fidelity-derivation}

Having separately characterized the error propagation from outer qubits (resulting in the error vector $\mathbf{e}_{\text{out}}^{(N-1)}$) and the effective $Z$ type errors from inner qubit measurements (leading to independent $Z$ error probabilities $p_A \equiv p^{\text{e2e}}_{ZI}$ and $p_B \equiv p^{\text{e2e}}_{IZ}$ on Alice's and Bob's sides, respectively, as given by \cref{eq:e2e-error-prob-zi,eq:e2e-error-prob-iz}), we now combine these contributions to determine the overall error probability vector of the final Bell pair.

The effect of the independent $Z$ errors from inner qubits can be described as a transition matrix (linear transformation) applied to the error vector $\mathbf{e}_{\text{out}}^{(N-1)}$ from the outer qubit chain.
Let these inner qubit error probabilities be
\begin{align}
    w'_{\text{in}} &= (1-p_A)(1-p_B) \quad (\text{Prob. of } II \text{ from inner channel}), \label{eq:ch5_w_prime_in} \\
    x'_{\text{in}} &= p_A(1-p_B) \quad (\text{Prob. of } ZI \text{ from inner channel}), \label{eq:ch5_x_prime_in} \\
    y'_{\text{in}} &= p_A p_B \quad (\text{Prob. of } ZZ \text{ from inner channel}), \label{eq:ch5_y_prime_in} \\
    z'_{\text{in}} &= (1-p_A)p_B \quad (\text{Prob. of } IZ \text{ from inner channel}). \label{eq:ch5_z_prime_in}
\end{align}
The final end-to-end error vector $\mathbf{e}_{\text{e2e}} = [w_{\text{e2e}}, x_{\text{e2e}}, y_{\text{e2e}}, z_{\text{e2e}}]^T$ is then obtained by applying the transformation $M_{\text{in}}$ to $\mathbf{e}_{\text{out}}^{(N-1)}$ as
$\mathbf{e}_{\text{e2e}} = M_{\text{in}} \cdot \mathbf{e}_{\text{out}}^{(N-1)}$, where the transition matrix $M_{\text{in}}$ due to inner qubit errors is given by
\begin{equation}
    M_{\text{in}} = \begin{pmatrix}
    w'_{\text{in}} & x'_{\text{in}} & y'_{\text{in}} & z'_{\text{in}} \\
    x'_{\text{in}} & w'_{\text{in}} & z'_{\text{in}} & y'_{\text{in}} \\
    y'_{\text{in}} & z'_{\text{in}} & w'_{\text{in}} & x'_{\text{in}} \\
    z'_{\text{in}} & y'_{\text{in}} & x'_{\text{in}} & w'_{\text{in}}
    \end{pmatrix}.
    \label{eq:ch5_M_in_transition_simplified}
\end{equation}
This matrix $M_{\text{in}}$ has the same structural form as the BSM transformation matrix $M$, with its components now being the probabilities derived from the inner qubit errors ($w'_{\text{in}}, x'_{\text{in}}, y'_{\text{in}}, z'_{\text{in}}$). Applying this transformation explicitly yields the components of the final error vector as
\begin{align}
    w_{\text{e2e}} &= w'_{\text{in}} w_{\text{out}}^{(N-1)} + x'_{\text{in}} x_{\text{out}}^{(N-1)} + y'_{\text{in}} y_{\text{out}}^{(N-1)} + z'_{\text{in}} z_{\text{out}}^{(N-1)}, \label{eq:ch5_w_e2e_final_s} \\
    x_{\text{e2e}} &= x'_{\text{in}} w_{\text{out}}^{(N-1)} + w'_{\text{in}} x_{\text{out}}^{(N-1)} + z'_{\text{in}} y_{\text{out}}^{(N-1)} + y'_{\text{in}} z_{\text{out}}^{(N-1)}, \label{eq:ch5_x_e2e_final_s} \\
    y_{\text{e2e}} &= y'_{\text{in}} w_{\text{out}}^{(N-1)} + z'_{\text{in}} x_{\text{out}}^{(N-1)} + w'_{\text{in}} y_{\text{out}}^{(N-1)} + x'_{\text{in}} z_{\text{out}}^{(N-1)}, \label{eq:ch5_y_e2e_final_s} \\
    z_{\text{e2e}} &= z'_{\text{in}} w_{\text{out}}^{(N-1)} + y'_{\text{in}} x_{\text{out}}^{(N-1)} + x'_{\text{in}} y_{\text{out}}^{(N-1)} + w'_{\text{in}} z_{\text{out}}^{(N-1)}, \label{eq:ch5_z_e2e_final_s}
\end{align}
with the end-to-end fidelity given by $F_{\text{e2e}} = w_{\text{e2e}}$.

The graphical representation in \cref{fig:side-effect-half-rgs} and these error expressions also inform efficient Monte Carlo simulation methods by tracking $Z$ flips hop-by-hop.
This aligns with techniques in QuISP, prior work~\cite{kaaki-ruleset-sim, cocori-quisp, addala-krastanov-faster-than-clifford}, and Stim~\cite{gidney-stim}.
Our simulations using this approach show consistency with average error models like those in~\cite{hilaire-rgs-optimizing-gen-time}.
The code for this simulation can be found at~\cite{rgs-github-repo}.

\section{Purification-enhanced RGS scheme}
\label{sec:rgs-with-purification}

Entanglement purification is a powerful tool for improving the fidelity of long-distance quantum communication.
However, integrating it beyond the link level into all-photonic architectures remains a significant challenge~\cite{azuma-rgs,hilaire-rgs-optimizing-gen-time}.
In this section, we introduce a purification-enhanced framework for the RGS scheme that leverages our half-RGS building block to address this problem.
Our approach uses optimistic purification protocols to achieve flexible purification between distant nodes without introducing memories or significantly increasing latency, opening the door to a new class of high-performance all-photonic repeaters.

\subsection{Motivation for integrating purification into the RGS scheme}
\label{subsec:half-rgs-purification-motivation}

To understand the need for a new approach, we first consider conventional methods for purification in repeater networks.
A natural strategy in RGS-based networks is to perform two-way purification at the end nodes once a long-distance entangled state has been established.
While conceptually simple, this method suffers from increased latency due to one-way end-to-end classical communication and degradation from memory decoherence during the waiting period.

Alternatively, a measurement-based approach using purification-embedded graph states can potentially avoid quantum memories but is limited to purifying entanglement between adjacent nodes only~\cite{zwerger-measurement-based-repeater}.
This method still requires additional optical Bell-state measurements, which reduces success rates and lowers the overall generation rate.
Such limitations highlight the fundamental difficulty of incorporating purification into RGS protocols, which are otherwise valued for their high loss tolerance and memory-free design.

We address this challenge by leveraging the half-RGS building block to enable purification directly on the generated half-RGS primitives.
This is achieved using optimistic purification, which is executed immediately on the anchor qubits of multiple half-RGSs without awaiting heralding signals~\cite{mobayenjarihani-optimistic-purification-qce,hartmann-optimistic-purification,razavi-optimistic-purification-pumping}.
Success signals and Pauli frame corrections propagate only at the end of the process, preserving the fast rate and memoryless structure of the RGS scheme.
This design allows for purification between non-neighboring nodes and supports flexible scheduling strategies---a capability previously limited to memory-based repeaters.

Crucially, the added overhead is modest, with generation time scaling proportionally with purification rounds while total latency remains effectively unchanged.
These features enable RGS-based networks to benefit from a large body of prior work on purification scheduling and optimization, bringing established memory-based techniques into the all-photonic domain~\cite{poramet-quanta-qcnc,poramet-quanta-extended,haldar-purify-swap-schedule,wangEfficientSchedulingScheme2024}.
The key enabler for this framework is the optimistic purification protocol, which we now review in more detail.

\subsection{Optimistic purification}
\label{subsec:optimistic-purification-review}

Optimistic entanglement purification, also known as blind purification, is a variant of the heralded approach where purification proceeds without waiting for confirmation that the initial Bell pairs were successfully established~\cite{hartmann-optimistic-purification, razavi-optimistic-purification-pumping, mobayenjarihani-optimistic-purification-qce}.
As illustrated in \cref{fig:purification-schematic}(b), multiple purification rounds can be executed sequentially, with a single round of classical messaging exchanged only at the very end to report the combined outcome.

\begin{figure}[htp]
    \centering
    \includegraphics[width=\textwidth]{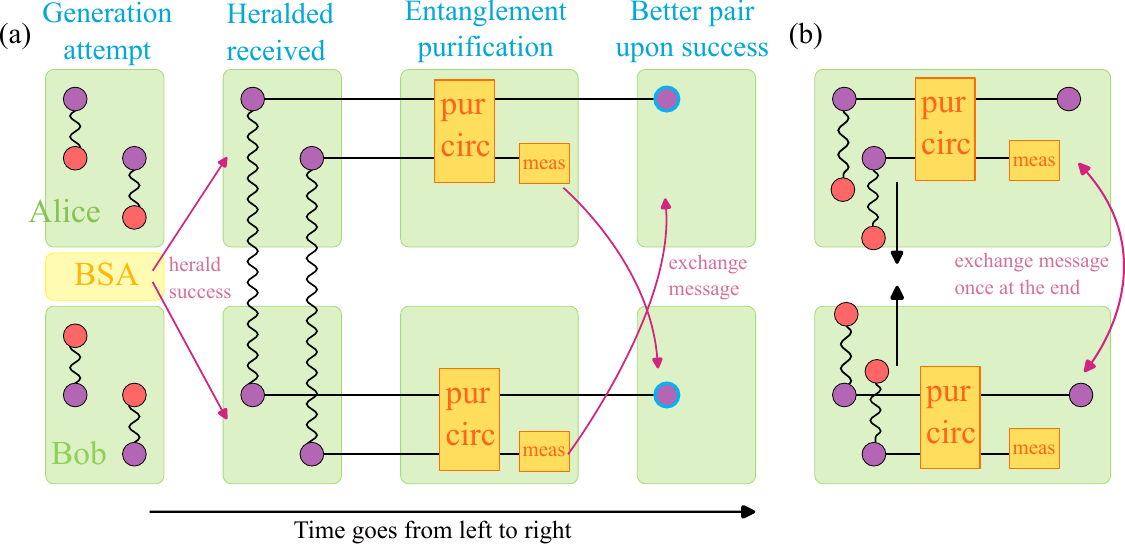}
    \caption[Schematic of two-way entanglement purification and its optimistic variant]{(a) Schematic of two-way heralded entanglement purification. The process begins with the probabilistic distribution of Bell pairs, followed by heralding signals from intermediate nodes confirming successful pair creation. A purification circuit is then applied, during which some Bell pairs are measured and consumed. The measurement results are exchanged between the two parties; if the outcomes agree (i.e., correlate or anti-correlate as expected), the remaining Bell pair is kept, resulting in improved fidelity. (b) Schematic of optimistic purification. In this approach, distribution and purification operations are performed immediately in sequence, without waiting for heralding confirmation. Classical outcomes are compared only once at the end, combining both the heralding of Bell pair generation and the results of purification in a single round of communication. (From~\cite{naphan-integrating-purification-rgs}.)}
    \label{fig:purification-schematic}
\end{figure}

While this strategy may lower the overall success probability, it offers significant advantages in high-latency settings by minimizing memory idle time and the cumulative effects of decoherence~\cite{razavi-optimistic-purification-pumping, mobayenjarihani-optimistic-purification-qce}.
This trade-off can lead to higher output fidelity compared to traditional purification, even with fewer successful attempts.

The optimistic approach is particularly well-suited to scenarios with high or near-deterministic Bell pair generation rates, a key feature of the RGS framework.
When pairs are created reliably, waiting for heralding signals offers little benefit and only increases latency.
Having established the advantages of this protocol, we now detail how our half-RGS architecture is modified to incorporate it.

\subsection{Modification to the RGS generation process}
\label{sec:rgs-with-purification:generation}

The standard half-RGS generation process, detailed in \cref{sec:generation-of-half-rgs}, concludes by immediately joining two half-RGSs.
To integrate optimistic purification, we modify this procedure by delaying the joining step to insert a purification stage.
Instead of creating a single half-RGS per side, we generate multiple half-RGSs, requiring additional anchor emitters.
Once all half-RGSs are created, their anchor qubits are fed as inputs to a purification circuit.
After purification, the retained high-fidelity half-RGSs from both sides are paired and joined.
While this modification enables purification, it introduces new operational steps whose impact on resource overheads must be analyzed.

\begin{figure}[t]
    \centering
    \includegraphics[width=\textwidth]{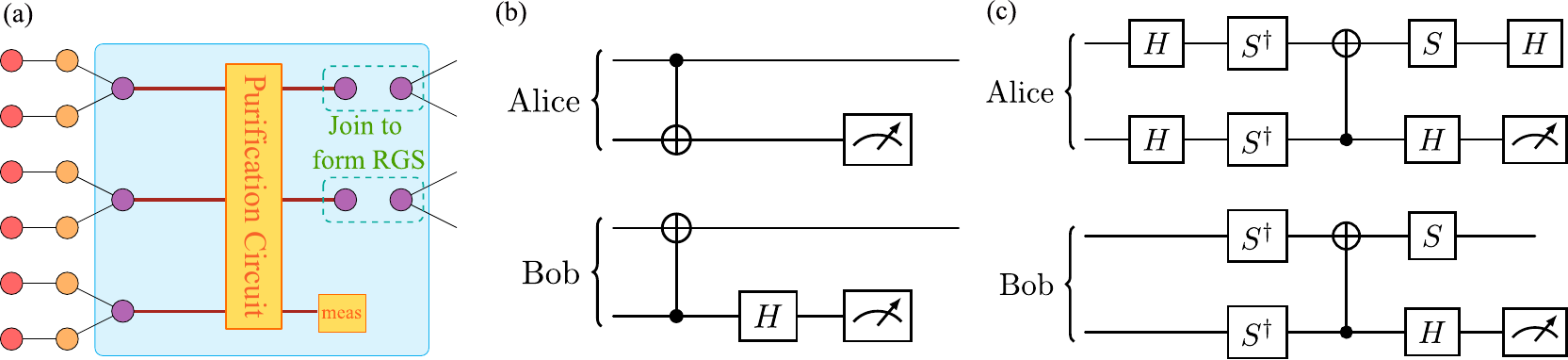}
    \caption[A depiction of entanglement purification on half-RGSs and purification circuits utilized in purification of bipartite graph states]{
        (a) Integration of entanglement purification into the RGS generation process.
        Instead of joining left and right half-RGSs immediately after their creation as in the original protocol~\cite{naphan-half-rgs}, multiple half-RGSs are generated on each side.
        Their anchor emitters (purple qubits) serve as inputs to a purification circuit.
        After purification, the retained high-fidelity half-RGSs from the left and right are paired and joined to form complete RGSs.
        (Note that dark red lines connecting purple qubits in the blue shaded box represent the circuits or their evolution in time and do not represent entangled properties of graph states.)
        (b) Graph state purification circuits that measures the parity of stabilizer $Z_a X_b$ where $a$ represents the Alice's side and $b$ for Bob~\cite{dur-graph-state-purification, kruszynska-graph-state-purification} Top (bottom) qubits of Alice and Bob are Bell pairs in the two-qubit graph state.
        The top qubits for both side are retained after the purification circuit.
        The purification circuit for $X_a Z_b$ are similar, by exchanging the circuits that Alice and Bob locally performed.
        (c) Graph state purification circuit that measures the parity of $Y_a Y_b$ stabilizers.
        (From~\cite{naphan-integrating-purification-rgs}.)
    }
    \label{fig:purification-enhanced-rgs}
\end{figure}

\subsection{Overhead analysis of the purification-enhanced RGS scheme}
\label{sec:rgs-with-purification:overhead-analysis}

We now analyze the time overhead of our purification-enhanced framework, focusing on the qubit coherence time required at the end nodes and the RGSS nodes.
We evaluate how long these qubits must remain coherent to support applications requiring high-fidelity entanglement.
For a direct comparison, we consider three scenarios: (1) the \textbf{raw RGS scheme} without purification; (2) a \textbf{baseline purification} approach where purification is applied only at the end nodes; and (3) our \textbf{proposed optimistic purification} framework.

First, in the \emph{raw RGS scheme}, the memory coherence time $t^{\text{(raw)}}_{\text{mem}}$ must span the half-RGS generation and the arrival of remote Pauli frame correction messages.
Second, the \emph{baseline purification} approach requires a significantly longer coherence time $t^{\text{(p-base)}}_{\text{mem}}$ as it is dominated by the round-trip communication needed for heralding.
Finally, our \emph{proposed optimistic purification scheme} has a memory coherence time of $t^{\text{(p-opt)}}_{\text{mem}}$ that scales much more favorably than the baseline, as it avoids the heralding delay.
This analysis reveals that our protocol significantly reduces the demands on quantum memory, a crucial advantage that sets the stage for evaluating its overall performance in terms of fidelity and rate.

\section{Performance evaluation}
\label{sec:error-modeling}

Having established the architectural modifications and favorable overheads of our protocol, we now evaluate its performance and demonstrate its flexibility through numerical simulations.
We showcase this flexibility by applying purification both at the link-level and between non-neighboring nodes.
While a full optimization of purification scheduling is beyond the scope of this thesis, our analysis focuses on the fidelity of the distributed Bell pairs and the rate at which a minimum fidelity threshold is met, providing a clear benchmark of the scheme's capabilities.

\subsection{Scope and assumptions}

To ground our evaluation, we first define the scope of our simulation and the key assumptions made.
We consider a repeater chain of 10 hops between two memory-equipped end nodes, Alice and Bob, connected by nine intermediate RGSS nodes.
We model the outcome of half-RGS generation by applying effective errors directly to the anchor qubits, following the derivation in \cref{sec:half-rgs-fidelity-derivation}.
The half-RGS parameters are tuned to achieve a near-deterministic generation success rate ($\geq 0.999$).
Our key assumptions are as follows:
\begin{itemize}
    \item Nodes are spaced 2 km apart, with a channel loss of 0.2 dB/km.
    \item All other sources of photon loss, such as from coupling or emission failure, are considered negligible.
    \item Each half-RGS uses 18 arms and a tree encoding with branching parameters $(16, 14, 1)$ to ensure high success probability.
    \item All photons are subject to a uniform depolarizing noise channel.
    \item Quantum gates and measurements are assumed to be noiseless.
    \item End-node memories are assumed to be ideal, with no decoherence.
\end{itemize}
The justification for these ideal assumptions is twofold.
First, operational noise can be effectively modeled as increased depolarizing error on the emitted photons.
Second, recent advances in quantum error correction~\cite{google-surface-code-one-qubit, quantinuum-microsoft-logical-below-threshold, bluvstein-quera-qec} make the availability of highly coherent logical memories a reasonable projection for future systems.
With these assumptions defined, we can now specify the purification strategies used in our model.

\subsection{Purification circuits and scheduling}

We restrict our noise model to Pauli channels, a choice justified because most error-correcting codes convert general noise into an effective Pauli channel during syndrome measurement~\cite{gottesmanSurvivingQuantumComputer}.
Both the baseline and our enhanced schemes use the same three 2-to-1 stabilizer-based purification circuits, which operate on two-qubit graph states (Bell pairs).
These circuits, shown on the right in \cref{fig:purification-enhanced-rgs}, measure the joint stabilizers $Z_a X_b$, $X_a Z_b$, and $Y_a Y_b$~\cite{dur-graph-state-purification, kruszynska-graph-state-purification}.

For the analysis, we adopt an error model where all errors are represented as $Z$-type Pauli operators that can affect either or both qubits of the Bell pair.
This model simplifies the analysis of error propagation and aligns with the Pauli frame tracking used in our protocol.
The purification circuits are sensitive to different error signatures; for instance, the ZX circuit detects $Z$ errors on qubit $b$, the XZ circuit detects them on qubit $a$, and the YY circuit detects an odd number of total $Z$ errors.

Our purification schedule for the baseline setting is inspired by the entanglement pumping protocol from~\cite{gidney-tetrationally-compact-purification}.
The schedule applies these circuits iteratively in the sequence YY, ZX, YY, XZ.
The rationale for starting with the YY circuit is that while the initial depolarizing noise is symmetric, the effective errors contributed by the inner qubits are biased towards independent $Z$ errors on each end node ($IZ$ and $ZI$), as shown in \cref{sec:half-rgs-fidelity-derivation}.
The YY circuit is therefore most effective at detecting these dominant error types in the early rounds.

\subsection{Error vector formalism for purification}
\label{subsec:error-modeling-for-half-rgs-purification}

To evaluate the different purification schedules, we now define the mathematical rules for purification within the error vector formalism established in \cref{sec:half-rgs-fidelity-derivation}.
Each 2-to-1 purification circuit is characterized by its success probability and the corresponding transformation it applies to the error vectors of the two input Bell pairs, $\mathbf{e_1}$ and $\mathbf{e_2}$.
The success probabilities for the three stabilizer measurement circuits ($ZX$, $XZ$, and $YY$) are given by
\begin{align}
    P_{ZX}(\mathbf{e_1}, \mathbf{e_2}) &= (w_1 + x_1) (w_2 + x_2) + (z_1 + y_1) (z_2 + y_2), \\
    P_{XZ}(\mathbf{e_1}, \mathbf{e_2}) &= (w_1 + z_1) (w_2 + z_2) + (x_1 + y_1) (x_2 + y_2), \\
    P_{YY}(\mathbf{e_1}, \mathbf{e_2}) &= (w_1 + y_1) (w_2 + y_2) + (x_1 + z_1) (x_2 + z_2).
\end{align}
Conditioned on a successful outcome, the error vector of the retained Bell pair is given by one of the following transformations
\begin{align}
    \text{Pur}_{ZX}(\mathbf{e_1}, \mathbf{e_2}) &= \frac{1}{P_{ZX}(\mathbf{e_1}, \mathbf{e_2})}\begin{bmatrix}
           w_1  w_2 + z_1  z_2 \\ z_1  w_2 + w_1  z_2 \\ x_1  y_2 + y_1  x_2 \\ x_1  x_2 + y_1  y_2
        \end{bmatrix}, \\
    \text{Pur}_{XZ}(\mathbf{e_1}, \mathbf{e_2}) &= \frac{1}{P_{XZ}(\mathbf{e_1}, \mathbf{e_2})}\begin{bmatrix}
           w_1  w_2 + x_1  x_2 \\ z_1  z_2 + y_1  y_2 \\ z_1  y_2 + y_1  z_2 \\ x_1  w_2 + w_1  x_2
        \end{bmatrix}, \\
    \text{Pur}_{YY}(\mathbf{e_1}, \mathbf{e_2}) &= \frac{1}{P_{YY}(\mathbf{e_1}, \mathbf{e_2})}\begin{bmatrix}
           w_1  w_2 + y_1  y_2 \\ x_1  z_2 + z_1  x_2 \\ y_1  w_2 + w_1  y_2 \\ x_1  x_2 + z_1  z_2
        \end{bmatrix}.
\end{align}
These new rules for purification, combined with the BSM transformation from \cref{eq:bsm-graph-state}, form the complete analytical toolkit needed to model the resulting fidelity between baseline and our purification-enhanced scheme.

\subsection{Evaluation methodology}

With the analytical tools established, we now outline the methodology for our numerical simulation.
Our simulation tracks the effective error vectors, conditioned on the successful generation of one-hop Bell pairs.
The freedom to reorder measurements allows us to model each half-RGS as an effective two-qubit graph state between an anchor and an outer qubit, with noise from inner qubits applied directly to the anchors using the average logical error probability~\cite{azuma-rgs, hilaire-rgs-optimizing-gen-time}.
The process to establish an end-to-end link involves first applying depolarizing noise to the outer qubits, then performing BSMs to connect the segments, and finally updating the anchor error vectors with the noise contributions from the inner qubits, per the rules in \cref{sec:half-rgs-fidelity-derivation}.

After establishing the initial noisy end-to-end Bell pairs, we simulate the scheduled purification protocols for each of the three scenarios.
For the baseline case, purification is applied to the final Bell pairs.
For our enhanced scheme, purification is performed at intermediate stages according to the multi-stage schedule detailed in the next section.
This consistent methodology ensures a direct and fair comparison of the end-to-end fidelities achieved by each framework.

\section{Numerical results}
\label{sec:numerical-results}

\begin{figure}[htp]
    \centering
    \includegraphics[width=\textwidth]{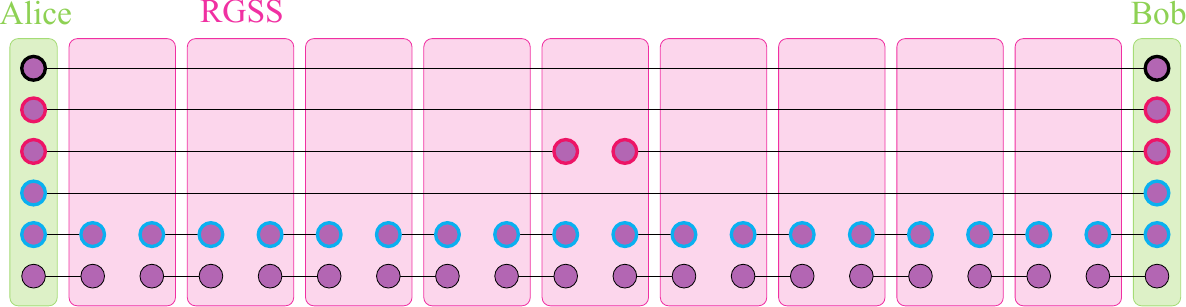}
    \caption[An example of flexible purification scheduling enabled by purification-enhanced RGS]{
        Illustration of flexible purification scheduling enabled by our purification-enhanced RGS framework.
        Using five half-RGSs per side at each hop (bottom line represents a copy), three end-to-end Bell pairs are created via different purification schedules, resulting in varying fidelities.
        These pairs then undergo further purification at the end nodes to yield a single high-fidelity Bell pair.
        The first is constructed by purifying ($YY$) each link-level connection and stitching the purified links across the full path (blue border).
        The second is formed by stitching link-level pairs into two copies of two 5-hop segments, purifying ($YY$) within each segment, and then connecting them (pink border).
        The third is created by directly stitching raw link-level pairs without intermediate purification (black thick border).
        These three pairs are then used in a pumping-like purification sequence: the first and second are purified (with the second as the sacrificial pair with $ZX$), followed by purification between the result and the third (with the third as the sacrificial pair with $YY$).
        The final output is a high-fidelity end-to-end Bell pair, significantly improved over the unpurified case upon success. (From~\cite{naphan-integrating-purification-rgs}.)
        }
    \label{fig:flexible-scheduling}
\end{figure}

\begin{figure}[!htp]
    \centering
    \includegraphics[width=0.8\columnwidth]{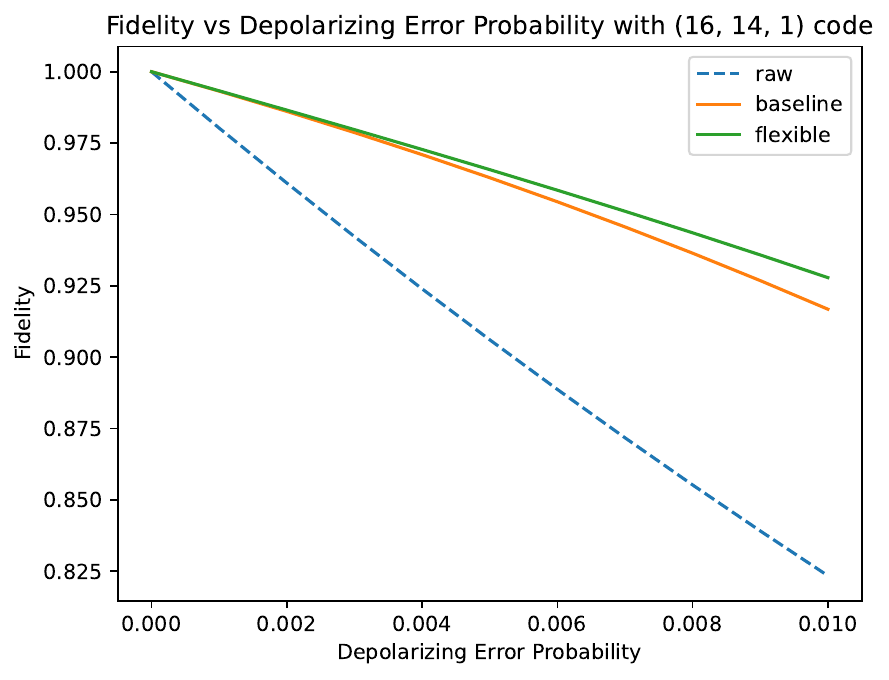}
    \caption[Fidelity comparison between different RGS purification strategies given photon depolarizing error probability]{
        End-to-end fidelity comparison of repeater chain with ten hops between unpurified Bell pairs (dashed blue lines), baseline purification using direct end-to-end Bell pairs (orange), and our purification-enhanced scheme (green).
        Both of the purified schemes consume five half-RGSs per end-to-end Bell pair instead of one in the unpurified setting.
        (From~\cite{naphan-integrating-purification-rgs}.)
    }
    \label{fig:fidelity-plot}
\end{figure}

\begin{figure}[!htp]
    \centering
    \includegraphics[width=0.8\columnwidth]{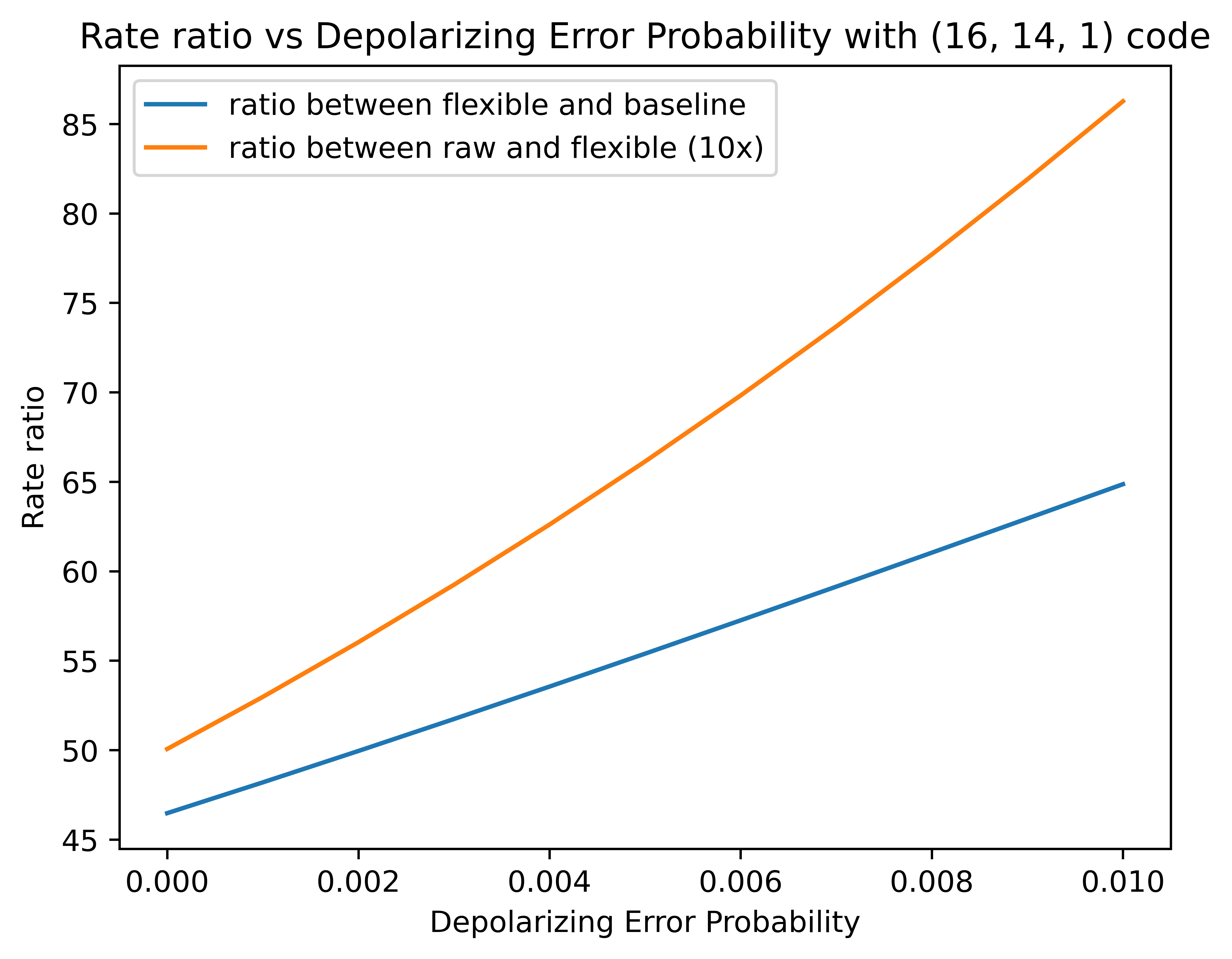}
    \caption[End-to-end Bell pair generation rate ratio between RGS purification strategies]{
        Comparison of Bell pair distribution rates for different purification strategies over a 40 km (10-hop) repeater chain.
        The blue curve shows the rate enhancement of the purification-enhanced strategy over a baseline pumping approach, achieving an improvement of 45 to 65 times, depending on the noise level.
        The orange curve, scaled by 10x for visibility, illustrates the overhead of purification by comparing its rate to that of a raw, unpurified link.
        This overhead corresponds to a rate reduction factor of 5 to 8.5.
        (From~\cite{naphan-integrating-purification-rgs}.)
    }
    \label{fig:rate-ratio-plot}
\end{figure}

To quantify the advantages of our architecture, we present a detailed performance comparison between our purification-enhanced scheme, the baseline approach, and the unpurified RGS protocol.
For the purification-enhanced scheme, we implement a sophisticated multi-stage schedule, illustrated in \cref{fig:flexible-scheduling}.
This strategy involves creating three distinct end-to-end Bell pairs using different intermediate purification schedules before combining them in a final purification sequence at the end nodes.
This flexible, multi-layered approach exemplifies our framework's ability to strategically apply purification at the most effective points in the network.

The simulation results, presented in \cref{fig:fidelity-plot,fig:rate-ratio-plot}, clearly illustrate the critical trade-offs between fidelity and rate.
As seen in \cref{fig:fidelity-plot}, both purification schemes substantially improve fidelity over the unpurified case, consistently achieving values above 0.9 for the noise parameters considered ($0 \leq p_{dep} \leq 0.01$).
Notably, our enhanced scheme maintains a consistent fidelity advantage over the baseline approach.
This improvement is a direct consequence of its ability to suppress errors early in the repeater chain, thereby preventing noise from accumulating and propagating across the entire connection.

This performance advantage becomes even more pronounced when considering the distribution rate.
As shown in \cref{fig:rate-ratio-plot}, the purification-enhanced scheme outperforms the baseline strategy by a factor of 45 to 65 across the tested noise levels.
This significant gain arises because our optimistic protocol eliminates the need for multiple, long-distance classical communication rounds required by the baseline's schedule, pushing the exchange of all purification results to the very end.

Crucially, this comparison is based on equal resource consumption, as both purification schemes use five half-RGSs per side at each RGSS and end node.
The performance gap is therefore attributable to the architectural design of our protocol.
By exploiting the structure of repeater graph states, our framework embeds purification directly within the entanglement generation process.
This approach enables multiple stages of error suppression that enhance fidelity while simultaneously reducing the time overhead from classical communication.

These results, obtained without a globally optimized schedule, suggest that further improvements are readily achievable through tailored scheduling strategies~\cite{krastanov-purification}.
Nonetheless, they provide compelling evidence that the purification-enhanced RGS framework can simultaneously meet the dual goals of high fidelity and a high entanglement distribution rate.
This result successfully overcomes a fundamental challenge in all-photonic quantum repeater design, marking a significant step toward their practical implementation.

In summary, by integrating optimistic purification directly into the generation process, our framework transforms the all-photonic repeater chain from a conceptual protocol into a robust, high-performance building block for larger quantum networks.
The demonstrated ability to generate high-fidelity entanglement at high rates means that an entire RGS-based path can now be reliably treated as a single, abstracted virtual link.
This capability is essential for creating scalable, heterogeneous networks, as it allows all-photonic segments to be seamlessly composed with memory-based systems under a unified control architecture, bridging a critical gap between hardware-level photonic protocols and the requirements of a future Quantum Internet.

\section{Discussion and future outlook}
\label{sec:discussion-and-variants}

In \cref{chapter:architecture-memory}, we established a comprehensive architecture for memory-based quantum networks built on the principles of the quantum recursive network architecture (QRNA) and managed by RuleSet-based protocols.
A central tenet of that architecture is the ability to abstract entire network segments as virtual links, enabling scalable and heterogeneous internetworking.
However, integrating memoryless, all-photonic repeaters into such a stateful, connection-oriented framework presents a significant challenge.
In this chapter, we have addressed this challenge directly by developing a robust, high-performance RGS scheme.
Our key contributions---the modular half-RGS building block and the integration of optimistic purification---provide the necessary tools to transform a volatile all-photonic path into a reliable architectural component.

\subsection{Hardware progress and the experimental landscape}

Recent experimental progress has begun to bridge the gap between theoretical all-photonic repeater protocols and their physical realization.
Early proof-of-concept demonstrations of the RGS scheme have been achieved, for instance, over a single hop using a passive, specialized measurement apparatus~\cite{hasegawa-passive-absa-experiment} and in a two-arm configuration~\cite{jian-wei-pan-rgs-experiment}.
A critical challenge remains the deterministic generation of the large-scale photonic graph states required.
Several hardware platforms are being pursued, with proposals for using arrays of quantum emitters~\cite{tiurevHighfidelityMultiphotonentangledCluster2022} and promising demonstrations using semiconductor quantum dots as sources~\cite{istratiSequentialGenerationLinear2020,chenHeraldedThreePhotonEntanglement2024}.

The most resource-efficient theoretical proposals for deterministic graph state generation often rely on a single, highly-capable quantum emitter that sequentially emits entangled photons, a concept sometimes called a ``photonic machine gun''~\cite{lindner-rudolf-machine-gun-graph-state,reisererCavitybasedQuantumNetworks2015}.
This approach has seen remarkable experimental success recently, with demonstrations in the optical domain using trapped atoms and quantum dots~\cite{yangSequentialGenerationMultiphoton2022,thomasEfficientGenerationEntangled2022,thomasFusionDeterministicallyGenerated2024,coganDeterministicGenerationIndistinguishable2023,costeHighrateEntanglementSemiconductor2023,suContinuousDeterministicAllphotonic2024a,mengDeterministicPhotonSource2024} and in the microwave domain with superconducting qubits~\cite{ferreiraDeterministicGenerationMultidimensional2024,osullivanDeterministicGeneration20qubit2024}.
Among the most advanced experimental feats to date are the generation of reconfigurable 2D cluster states of up to 20 photons~\cite{huetDeterministicReconfigurableGraph2025} and deterministically generated, redundantly encoded linear graph states that enable near-deterministic fusion operations with built-in error correction~\cite{petterssonDeterministicGenerationConcatenated2025, hilaire-near-deterministic-fusion}.

\subsection{Alternative all-photonic architectures}

Beyond the RGS framework, a rich landscape of alternative all-photonic repeater architectures exists.
These approaches can be broadly distinguished by their core operational strategy~\cite{munro-inside-quantum-repeaters,zwergerMeasurementbasedQuantumCommunication2016}, which has significant implications for their integration into a QRNA-style architecture.

The first category, like the RGS scheme, is based on pre-distributing entangled Bell pairs, a model that aligns naturally with the connection-oriented paradigm of \cref{chapter:architecture-memory}.
This family includes protocols that refine the crucial BSM step with adaptive measurement strategies on logically encoded states~\cite{hilaire-logical-bsm}.
Other variants in this family introduce short-term photonic memory by having inner qubits travel through spools of optical fiber, allowing for more complex operations like boosted BSMs or parallel transmission of RGS arms through multi-channel fibers~\cite{pant-rate-dist-tradeoff,patilImprovedDesignAllPhotonic2024}.

The second category~\cite{niu-all-photonic-one-way-repeater,ewertUltrafastFaulttolerantLongdistance2017,borregaard-3g-repeater,leeFundamentalBuildingBlock2019,ewertUltrafastLongDistanceQuantum2016} operates via the direct, hop-by-hop transmission of encoded logical states, an approach akin to the packet-switching model of third-generation repeaters~\cite{munro-inside-quantum-repeaters}.
The core idea is to utilize a ``decode-and-re-encode'' strategy at each repeater station to correct for errors, a strategy that can be implemented with various optimized loss-tolerant codes~\cite{ralphLossTolerantOpticalQubits2005,bellOptimizingGraphCodes2023}.
While this model potentially offers high speed, integrating such connectionless protocols into the stateful, connection-oriented architecture presented in \cref{chapter:architecture-memory} is non-trivial.
Our framework relies on end-to-end RuleSets established via a two-pass setup, but a hop-by-hop protocol could be emulated by sequentially creating and tearing down a series of short-lived, single-hop connections.
This would, however, introduce significant protocol overhead, especially for managing the rapid instantiation of RuleSets and coordinating handoffs at network boundaries.

Finally, other distinct approaches exist that rely on different physical encodings entirely.
These include protocols based on bosonic GKP states~\cite{fukuiAllopticalLongdistanceQuantum2021,albertPerformanceStructureSinglemode2018}, coherent states which can naturally provide deterministic BSMs~\cite{rozpedek-all-photonic-bosonic-gkp-code}, or multipartite entangled states such as W-states~\cite{miguel-ramiroQuantumRepeater-for-W-states}, each representing a unique avenue of research.

\subsection{Strategies for protocol enhancement}

Further improvements to all-photonic repeaters can be pursued along several research directions, primarily by addressing key operational bottlenecks.
A primary challenge in many schemes, including RGS, is the probabilistic nature of linear-optical Bell-state measurements (BSMs).
Boosting the BSM success probability beyond the standard 50\% limit is a major avenue for enhancement, as it would directly reduce the resource overhead of the virtual links presented to the higher-level network architecture.
Several proposals exist to achieve this, using techniques like ancillary photons~\cite{ewert34Efficient2014,griceArbitrarilyCompleteBellstate2011}, weak non-linearities~\cite{barrettSymmetryAnalyzerNondestructive2005}, or squeezed states~\cite{zaidiNeardeterministicCreationUniversal2015, kilmer-guha-boosting-linear-optimal-bsm-squeezing}.
However, these advanced techniques involve significant experimental complexity and are constrained by fundamental no-go theorems~\cite{lutkenhausBellMeasurementsTeleportation1999} and theoretical limits on loss tolerance~\cite{hilaireLinearOpticalLogical2023}.

Another strategy for enhancement is to improve the resource efficiency of the RGS protocol itself.
Currently, only one Bell pair can be extracted from a successful RGS trial, even if multiple BSMs succeed at every hop.
This limitation has motivated the proposal of a generalized RGS framework where the tree code for inner qubits is replaced with a classical low-density parity-check (LDPC) code~\cite{liGeneralizedQuantumRepeater2025}.
For certain families of LDPC codes, this would enable the creation of multiple, locally-CNOT equivalent Bell pairs, thus allowing for the extraction of more than one end-to-end pair from a single trial.
Ultimately, weighing the architectural trade-offs between these two enhancement strategies---one improving a fundamental operation and the other the protocol's core structure---will be crucial for developing the most practical and resource-efficient quantum repeaters.

\subsection{Conclusion}

In this chapter, we addressed several key architectural and operational challenges of the all-photonic repeater graph state (RGS) scheme.
Our primary contribution is the introduction of the \textbf{half-RGS}, a modular construct that resolves fundamental issues related to monolithic graph state generation and global timing synchronization by decomposing the problem into independent, localized operations.
This building block provides a more practical and hardware-friendly foundation for constructing all-photonic repeater links.

Building upon this modular foundation, we then tackled the long-standing problem of incorporating entanglement purification into all-photonic repeater paths.
The architectural flexibility afforded by the half-RGS enabled the seamless integration of optimistic purification directly into the generation process.
As our numerical results demonstrated, this purification-enhanced framework achieves both high fidelity and significantly higher distribution rates---outperforming baseline approaches by a factor of 45 to 65---without incurring additional latency or memory requirements.
This performance gain is a direct result of the protocol's efficient design, which eliminates the need for multiple long-distance classical communication rounds.

Ultimately, these contributions advance the all-photonic repeater chain from a theoretical protocol towards its realization as a robust, high-performance component for large-scale networks.
The demonstrated ability to generate high-fidelity entanglement at high rates validates the abstraction of an entire RGS subpath as a well-behaved \textbf{virtual link}.
This provides the crucial missing piece for integrating this entire class of technology into the scalable, recursive architecture laid out in \cref{chapter:architecture-memory}.
The half-RGS, guided by optimistic purification, becomes the concrete physical-layer and link-layer implementations that can be controlled and composed by the higher-level Connection Management and RuleSet protocols, truly bridging the gap between hardware-specific paradigms and the architectural requirements of a universal Quantum Internet.

\clearpage
\chapter{An application of distributed quantum computation: an optimization algorithm perspective} 
\label{chapter:qaoa-and-network-resources} 

Having established an architectural framework for heterogeneous quantum networks in the preceding chapters, we now take a step back and turn to the applications that such a network is designed to support.
It would be incomplete to discuss the architecture without considering the computations it must enable.
A primary motivation for building a quantum network is to connect individual quantum processors into distributed quantum computer clusters capable of performing large-scale computations that are infeasible on a single machine.
While truly large-scale algorithms like Shor's factoring algorithm~\cite{shor-factoring} or complex chemical simulations may represent a distant goal, \emph{combinatorial optimization} problems offer a compelling and industrially relevant use case for near- to medium-term distributed quantum computing.
\extrafootertext{Portions of this chapter have been adapted from the following publications:
\begin{itemize}
    \item \fullcite{naphan-spo}
    \item \fullcite{naphan-lower-bound-qaoa}
\end{itemize}
}

This chapter bridges the analysis of quantum optimization algorithms with the practical requirements of the network architectures designed to run them.
To build a network capable of executing these algorithms, it is essential to understand their runtimes and resource consumption.
This knowledge is critical for creating realistic traffic models for network simulators like QuISP, informing resource allocation strategies, and ultimately understanding the demands that distributed computation will place on a future Quantum Internet.

To this end, this chapter presents two primary contributions.
First, we explore potential advancements in optimization by developing an algorithm based on a generalization of amplitude amplification.
This method uses a non-trivial phase shift oracle, which we term the \emph{subdivided phase oracle} (SPO), to directly amplify states corresponding to high-quality solutions.
Interestingly, we find that the SPO algorithm, with its repeated application of a problem-specific phase oracle and a general mixing operator, can be viewed as a specific instance of the QAOA framework~\cite{andreas-grover-mixer-qaoa}.

Recognizing that our SPO algorithm is a member of this broader class, the second key contribution of this chapter is a detailed analysis of the general QAOA framework's performance, which has direct implications for network design.
Specifically, we derive some of the first analytical lower bounds on the number of QAOA rounds required to achieve a guaranteed approximation ratio.
This provides crucial, concrete data on algorithm runtimes needed for network planning.

This chapter is structured to unfold this narrative.
We will first provide a detailed background on Grover's algorithm, QAOA, and the SPO.
We will then present our main results on the lower bounds for QAOA rounds for different classes of problems and mixers.
Finally, we will discuss the broader implications of these algorithmic runtimes for designing and managing resources in the quantum networks of the future.

\section{Introduction: why study a heuristic optimization algorithm?}
\label{sec:introduction}

Solving optimization problems is a cornerstone computational task, with a wide range of applications across industry, science, and technology.
A primary application for quantum computers is to tackle these problems, and the quantum approximate optimization algorithm~\cite{farhi-qaoa} and its generalization to constrained optimization problems, the quantum alternating operator ansatz (QAOA)~\cite{hadfield-qaoa}, have emerged as leading candidates.
QAOA is a hybrid quantum-classical approach designed to find high-quality approximate solutions to combinatorial optimization problems.
Its structure involves alternating between a phase separator, which encodes the problem's objective function into a Hamiltonian $H_1$, and a mixing operator, $H_0$.
A classical optimizer then tunes the parameters for a fixed number of rounds, $p$, to guide the quantum state toward a low-energy solution that represents a high-quality answer.

While algorithms based on Grover's search can find exact solutions with a provable quadratic speedup, this advantage must be considered in context.
It has been argued that a quadratic speedup alone may be insufficient to provide a practical advantage, as the immense overheads associated with fault-tolerant quantum error correction could nullify the gains over classical algorithms~\cite{babbushFocusQuadraticSpeedups2021}.
Furthermore, Grover-style algorithms solve optimization problems heuristically by treating the problem as a black box and forgoing any exploitable structure; they are the quantum equivalent of brute force and do not provide approximation guarantees.
This has motivated the development of other approaches, such as QAOA, which are designed to find high-quality \emph{approximate} solutions and are seen as a potential path to achieving practical performance benefits.

A common misconception is that variational heuristic algorithms like QAOA are merely a stop-gap for the noisy intermediate-scale quantum (NISQ) era~\cite{preskillQuantumComputingNISQ2018a}, to be supplanted by algorithms with provable guarantees once fault-tolerant quantum computers (FTQCs) are available.
This view is mistaken.
In fact, running QAOA on problem sizes of industrial relevance will likely \emph{require} fault-tolerant hardware~\cite{augustinoStrategiesRunningQAOA2024}.
Furthermore, concerns that the classical optimization of QAOA's parameters may be a bottleneck for large-scale problems are being addressed by promising research into techniques like transfer learning~\cite{ruslan-parameter-transfer-qaoa}.
Indeed, a growing body of numerical evidence supports QAOA as a leading candidate for the go-to optimization algorithm in future quantum computing scenarios, making it a key potential application for the quantum networks we envision~\cite{john-numerical-speedup-constrained,john-sat-qaoa,ruslan-labs-evidence-of-qaoa-scaling,boulebnaneSolvingBooleanSatisfiability2024}.

Given the long-term relevance of QAOA, a deeper theoretical understanding of its performance is crucial.
However, theoretical results regarding its runtime---specifically, the number of rounds $p$ required to achieve a given approximation ratio---are still relatively sparse, especially beyond the low-depth regime.
A key contribution of this thesis is to address this gap by deriving some of the first analytical lower bounds on $p$.
We achieve this by leveraging a powerful connection between the digital, round-based structure of QAOA and the continuous-time evolution of quantum annealing, building on recent results that established bounds on annealing times.
This analysis will provide crucial insights into how to plan for and manage the resources needed to run large-scale optimization tasks on the distributed quantum networks of the future.

\section{Background on quantum optimization algorithms}
\label{sec:background-optimization-algorithms}

To understand the computational workloads a quantum network must support, we first review several foundational algorithms for quantum optimization.
This review will establish the key concepts, performance metrics, and operational principles that are relevant to the analysis in the subsequent sections of this chapter.

\subsection{The combinatorial optimization problem}

The core task we consider is the combinatorial optimization problem.
We define this problem on bitstrings of length $n$, with an objective function $C(x)$ that maps a given bitstring $x$ to a non-negative integer value.
The goal is to find an input $z$ that maximizes this objective function, yielding the value $C_{max} = \max_z C(z)$.
For our analysis, it is also useful to define the average of the objective values, $C_{avg} = \frac{1}{N} \sum_z C(z)$, and their standard deviation, $\sigma_C = \sqrt{\sum_z \frac{(C(z) - C_{avg})^2}{N}}$.
We will consider unconstrained problems where the search space $F$ consists of all $N = 2^n$ possible bitstrings.

\nomenclature{Approximation ratio}{A performance measure used to evaluate an approximation algorithm's average-case behavior. For a given optimization problem, it is expressed as the multiplicative factor comparing the average output of the algorithm against the true optimal solution. A value close to one indicates that the algorithm consistently finds solutions that are, on average, very close to the optimal result, signifying a high-quality approximation}
To approach this with a quantum algorithm, it is natural to translate the objective function into a corresponding problem Hamiltonian.
We first define a diagonal cost Hamiltonian $H_C$ whose eigenvalues are the objective function values
\begin{equation}
    H_C = \sum_{z \in F} C(z) \ket{z}\bra{z}.
\end{equation}
A common way to evaluate the performance of a heuristic optimization algorithm is its \emph{approximation ratio}, $\lambda \in [0, 1]$, which we define as
\begin{equation}
    \lambda = \frac{\braket{H_C}_p}{C_{max}} = \max_{\ket{\psi_p}} \frac{\braket{\psi_{p} | H_C | \psi_p}}{C_{max}},
\label{eq:approximation-ratio}
\end{equation}
where $\ket{\psi_p}$ is the state prepared by $p$ rounds of an algorithm.
Since it is often more convenient to work with Hamiltonians whose ground state represents the solution, we define our primary problem Hamiltonian $H_1$ as
\begin{equation}
    H_1 = \mathbbm{1} C_{max} - H_C.
\label{eq:h1-definition}
\end{equation}
With this definition, the task of maximizing $C(x)$ is equivalent to finding the ground state of $H_1$.

\subsection{Grover's algorithm and adaptive search}

One quantum approach to optimization is to treat it as a search problem.
Grover's search algorithm is a celebrated quantum algorithm that provides a provable quadratic speedup for finding a marked item in an unstructured database~\cite{grover-alg}.
It can find one of $t$ target inputs among $N$ total possibilities in just $O(\sqrt{N/t})$ queries to an oracle, compared to the $O(N/t)$ required classically.
It was later understood that Grover's algorithm is an optimal special case of a more general technique known as \emph{amplitude amplification}~\cite{brassard-bhmt-amplitude-amplification-estimation}.
While powerful for finding exact solutions, Grover's algorithm in its basic form does not provide approximation guarantees.
However, it can be used as a heuristic optimizer in a strategy called \emph{Grover Adaptive Search}, where the algorithm is run repeatedly to find solutions of increasing quality~\cite{durrQuantumAlgorithmFinding1999,gilliam-grover-adaptive-search-cpbo}.

\subsection{Quantum annealing}

One leading paradigm for quantum optimization is quantum annealing (QA).
The protocol starts in an easy-to-prepare ground state of a simple ``mixer'' Hamiltonian, $H_0$.
The system is then slowly evolved according to a time-dependent Hamiltonian 
\begin{equation}
    H(t) = [1-g(t)]H_0 + g(t)H_1,
\end{equation}
where $H_1$ is the problem Hamiltonian whose ground state encodes the optimal solution.
According to the adiabatic theorem, if this evolution is sufficiently slow, the system remains in the instantaneous ground state throughout the process and ends in the desired solution state of $H_1$.
The primary limitation of this approach is that for many hard problems, the required ``annealing time'' to ensure adiabaticity can become exponentially long, negating the potential speedup~\cite{bernaschiQuantumTransitionTwodimensional2024,albashAdiabaticQuantumComputation2018a}.

\subsection{The quantum approximate optimization algorithm}

The quantum approximate optimization algorithm (QAOA) is a leading hybrid quantum-classical algorithm designed to find high-quality approximate solutions to combinatorial optimization problems~\cite{farhi-qaoa}.
It is structured around $p$ rounds of alternately applying two parameterized unitary operators: a \emph{phase separator} $e^{-i\gamma H_1}$, which encodes the problem Hamiltonian, and a \emph{mixing operator} $e^{-i\beta H_0}$, which drives transitions between different computational basis states~\cite{hadfield-qaoa}.
The initial state $\ket{\psi_0}$ is typically the ground state of the mixer $H_0$.
The state after $p$ rounds is thus given by
\begin{equation}
\label{eq:qaoa-circuit}
    \ket{\psi (\vec{\beta}, \vec{\gamma})}_p = (e^{-i \beta_p H_0} e^{-i \gamma_p H_1}) \dots (e^{-i \beta_1 H_0} e^{-i \gamma_1 H_1})  \ket{\psi_0}.
\end{equation}
The goal of the QAOA is to prepare a state that minimizes the energy expectation value $\langle H_1 \rangle_p = \braket{\psi (\vec{\beta}, \vec{\gamma})| H_1 | \psi (\vec{\beta}, \vec{\gamma})}_p$.
This is achieved through a classical optimization loop that varies the angle parameters $\vec{\gamma}$ and $\vec{\beta}$ to find the minimum.

Conceptually, QAOA is deeply connected to Quantum Annealing (QA).
In the limit of infinite rounds ($p \rightarrow \infty$), QAOA can be seen as a Trotterized discretization of a continuous-time annealing process.
More specifically, the alternating application of the two Hamiltonians in QAOA corresponds to a "bang-bang" schedule in QA, where only one driving Hamiltonian is turned on at a time~\cite{innocenti-bang-bang}.
This connection is significant, as the bang-bang schedule is conjectured to be the optimal annealing schedule for many optimization problems~\cite{yang-chamon-bang-bang-optimal-vqa}.

On the other hand, the QAOA does not need to arise from the Trotterization or discretization of any annealing schedule; instead, it can be understood purely as a variational ansatz designed for optimization~\cite{cerezo-vqa-review}.
Research in the finite-$p$ regime has shown that for certain problems, QAOA can achieve a constant approximation ratio with a constant number of rounds~\cite{farhi-qaoa, farhi-e3lin2, wurtz-love-max-cut-constant-depth-qaoa}.
Other active areas of investigation include understanding the algorithm's expressibility (i.e., the subset of Hilbert space accessible by a given ansatz)~\cite{akshay-reachability-deficit-prl, akshay-circuit-depth-scaling-sat-density, niu-chuang-qaoa-state-transfer}, the phenomenon of parameter concentration~\cite{basso-constant-depth-qaoa-analysis, anshu-metger-concentration-bound-qaoa}, and techniques for improving performance through better initialization or insights from counterdiabatic driving~\cite{cain-qaoa-get-stuck, sack-serbyn-qaoa-initialization-via-qa, wurtz-love-counterdiabaticity-qaoa}.

\section{Amplitude amplification with subdivided phase oracle for optimization problems}
\label{sec:subdivided-phase-oracle}

A promising avenue for quantum optimization lies in extending the principles of amplitude amplification beyond a simple binary distinction between ``good'' and ``bad'' states.
Recent work has investigated the use of oracles that apply non-trivial phase shifts, acting uniquely on each input based on a cost function~\cite{satoh-spo,koch-gaussian-amplitude-amplification,shyamsundarNonBooleanQuantumAmplitude2023a}.
Building on this, we introduce and analyze a specific implementation of this concept, which we term the \emph{subdivided phase oracle} (SPO)~\cite{satoh-spo}.
This approach replaces the canonical Grover oracle in the amplitude amplification process, providing a new method for solving combinatorial optimization problems.
While the word ``oracle'' typically refers to a black-box function that recognizes a solution, we use the term here to describe the phase-shifting operator to maintain consistency with the language of amplitude amplification and related literature.

\subsection{Oracle construction}

The core idea of the SPO is to embed the objective value $f(x)$ of each potential solution $\ket{x}$ directly into its quantum phase.
Instead of applying a uniform phase shift to all "good" states, the SPO applies a unique phase to every computational basis state in proportion to its quality.
The action of the phase oracle, $O_{\phi}$, on a state $\ket{x}$ is formally defined as
\begin{equation}
\label{eq:phase_oracle}
    O_\phi \ket{x} \rightarrow e^{i \phi(x)}\ket{x},
\end{equation}
where $\phi(x)$ is a function that maps the input $x$ to a phase angle.
This mapping is constructed from the objective function $f(x)$ via a simple linear relationship
\begin{equation}
    \label{eq:phi_mapping}
    \phi(x) = k f(x) + \theta,
\end{equation}
where $k$ is an adjustable real-valued parameter and $\theta$ represents an offset, typically corresponding to the minimum objective value.
Since a constant phase offset applied to all states amounts to an irrelevant global phase, the $\theta$ term can be ignored.
The action of the SPO is therefore fully characterized by the scaling parameter $k$, and we can write the operator as
\begin{equation}
\label{eq:spo_combined}
    O_\phi(k) \ket{x} \rightarrow e^{i k f(x) }\ket{x}.
\end{equation}
The oracle's action is invariant under shifts in the objective values (e.g., making them all positive) and permutations of the input labels.
For simplicity in our analysis, we can therefore assume without loss of generality that $f(x)$ is a non-decreasing function with a minimum value of $f(x_{min}) = 0$.

\subsection{The SPO optimization algorithm}

Our algorithm incorporates the SPO into the standard amplitude amplification framework, replacing the binary oracle but retaining the conventional diffusion operator.
The full process is outlined in Algorithm~\ref{alg:algorithm-1}.

\begin{algorithm}
\caption{Optimization via SPO-based Amplitude Amplification}
\label{alg:algorithm-1}
\begin{algorithmic}[1]
\State Initialize state $\ket{s} = \ket{0}^{\otimes n}$
\State Prepare uniform superposition $\ket{\Psi_0} \gets H^{\otimes n}\ket{s}$
\For{a predetermined number of iterations, $j$}
    \State Apply the Subdivided Phase Oracle: $\ket{\Psi'} \gets O_\phi (k) \ket{\Psi_{j-1}}$
    \State Apply the Diffusion Operator: $\ket{\Psi_{j}} \gets D \ket{\Psi'}$
\EndFor
\State Measure the final state $\ket{\Psi_j}$
\end{algorithmic}
\end{algorithm}

The process begins by preparing an equal superposition of all possible states, $\ket{\Psi_0}$.
Then, in each iteration, two operations are applied in sequence.
First, the SPO, $O_\phi(k)$, rotates the phase of each computational basis state as described above.
Second, the standard Grover diffusion operator, $D$, is applied, which performs an inversion about the average of the amplitudes.
The diffusion operator is given by
\begin{equation}
    D = 2\ket{\Psi_0}\bra{\Psi_0} - I,
\end{equation}
where $\ket{\Psi_0}$ is the uniform superposition state and $I$ is the identity operator.
After a set number of iterations, the final state is measured in the computational basis, with a higher probability of yielding states with high objective values.

\subsection{Amplification mechanism and visualization}

To build intuition for how the SPO leads to optimization, it is helpful to visualize its effect on the state amplitudes.
Unlike the standard Grover algorithm where all amplitudes remain real, the SPO introduces complex phases.
We can represent the amplitude of each basis state $\ket{x}$ as a vector $\vec{\alpha}_x$ in the 2D complex plane.
Initially, in the uniform superposition state $\ket{\Psi_0}$, all $N$ amplitude vectors are identical, real, and have a magnitude of $1/\sqrt{N}$, as shown in \cref{fig:spo-evolution}(a).

The application of the SPO, $O_\phi(k)$, corresponds to rotating each of these vectors around the origin by an angle $k f(x)$, as seen in \cref{fig:spo-evolution}(b).
States with higher objective values are rotated more.
This step only changes the relative phases and does not by itself alter the measurement probabilities.
The amplification occurs upon applying the diffusion operator, $D$, which can be understood as a geometric operation in the complex plane:
\begin{enumerate}
    \item First, the mean of all $N$ amplitude vectors, $\vec{\alpha}_{mean}$, is calculated.
    \item Then, each individual amplitude vector $\vec{\alpha}_x$ is reflected about this mean.
    \item The new amplitude vector, $\vec{\alpha}'_x$, is given by the transformation $\vec{\alpha}'_x = 2\vec{\alpha}_{mean} - \vec{\alpha}_x$.
\end{enumerate}
This process is depicted in \cref{fig:spo-evolution}(c-f).
The reflection preferentially increases the magnitude of vectors that are pointed in a direction roughly opposite to the mean vector.
Since states with high objective values are rotated the most by the SPO, they are most likely to end up in this favorable position for amplification.
The vector corresponding to the optimal solution (most saturated red) is significantly increased in magnitude, thus amplifying its measurement probability.

\begin{figure}[htbp]
    \centering
    \includegraphics[width=0.7\textwidth]{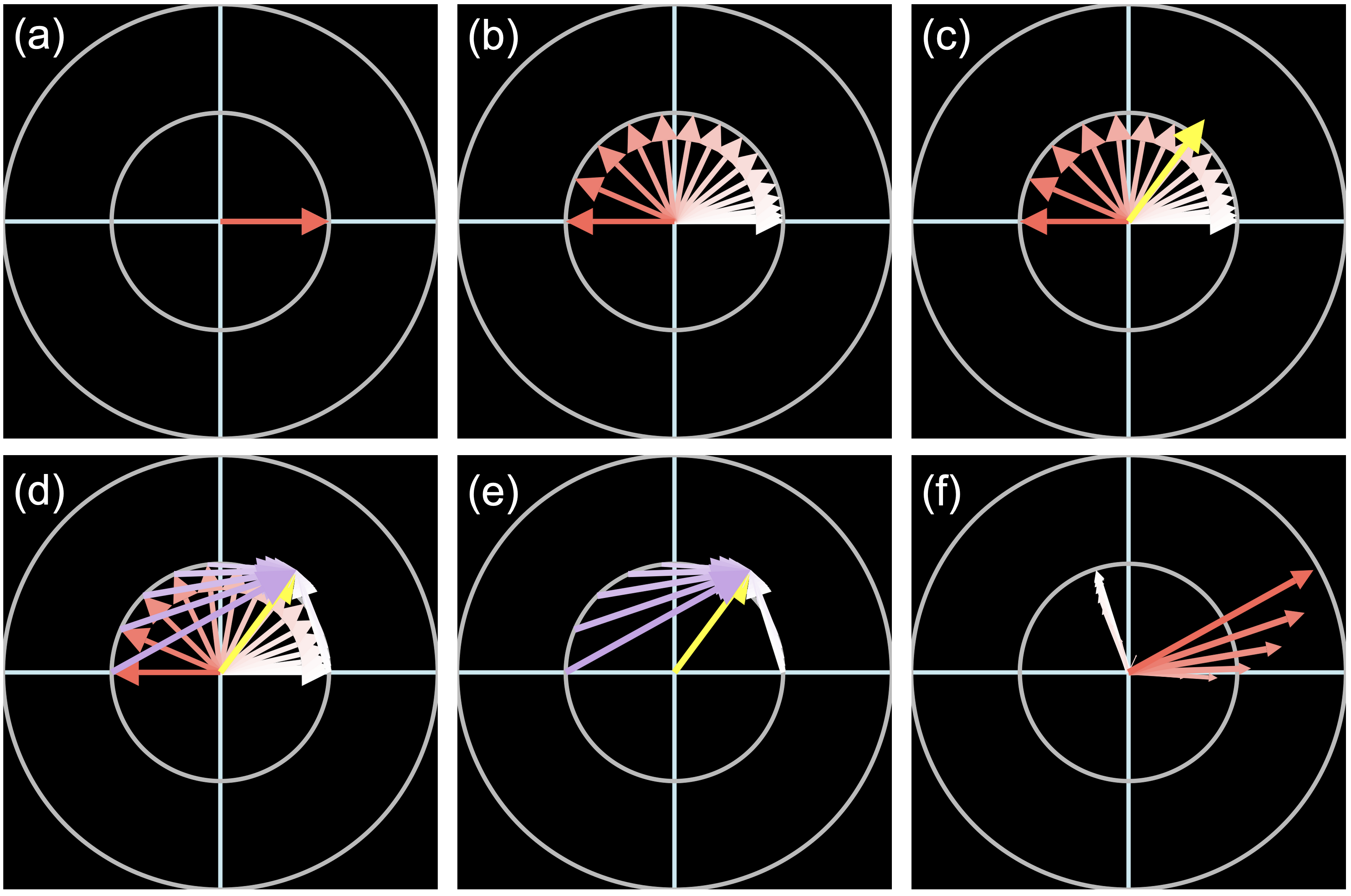}
    \caption[2D complex plot of amplitude evolution with an SPO]{
    A 2D complex plot illustrating how the amplitudes of $N=16$ states change during one iteration of amplitude amplification using an SPO with $f(x) = x^2$.
    The color of each vector, from white to red, denotes its objective value.
    (a) Initial amplitudes in a uniform superposition.
    (b) Amplitudes after the SPO applies a phase rotation proportional to each state's objective value.
    (c) The mean amplitude vector ($\vec{\alpha}_{mean}$) is calculated, shown here scaled by a factor of two (yellow arrow).
    (d, e) The diffusion operator reflects each amplitude vector about the mean.
    (f) The final amplitudes after one full iteration, showing a significant increase in the magnitude of the optimal state's amplitude.
    (From~\cite{naphan-spo}.)}
    \label{fig:spo-evolution}
\end{figure}

This visualization also makes the condition for successful amplification clear.
The amplitude of the optimal state, $\vec{\alpha}_{best}$, will increase in magnitude only if it is pointing generally away from the mean vector, $\vec{\alpha}_{mean}$.
More formally, significant amplification occurs when the angle between $\vec{\alpha}_{best}$ and $\vec{\alpha}_{mean}$ is in the range $[\pi/2, 3\pi/2]$.
The choice of the parameter $k$ in the SPO is what determines this crucial angle.
Finding an optimal $k$ is therefore essential for the algorithm's performance, as it ensures the optimal state is consistently positioned for amplification in each iteration.

\section{Performance analysis for common objective distributions}
\label{sec:results}

Before presenting the simulation results, it is important to contextualize the problems being studied.
The success of the Subdivided Phase Oracle (SPO) algorithm relies heavily on the distribution of the objective function values, as this distribution dictates the behavior of the mean amplitude vector during the amplification process.
To explore the algorithm's performance, we therefore categorize our simulations based on the statistical properties of the objective function, focusing on three representative classes: symmetric (normal), asymmetric (skew normal), and heavy-tailed (exponential).
These distributions are chosen not only for their analytical tractability but also because they serve as effective proxies for the types of objective landscapes found in real-world optimization problems.

\subsection{Symmetric distributions: the normal distribution}

We begin by considering objective values that follow a normal distribution, as described by the probability density function
\begin{equation}
    \phi_{\text{norm}} (x) = \frac{1}{\sigma \sqrt{2\pi}} e^{-\frac{(x-\mu)^2}{2\sigma^2}},
\end{equation}
where $\mu$ is the mean and $\sigma$ is the standard deviation.
This is a natural starting point, as many instances of NP-hard optimization problems, when formulated as Ising problems, exhibit normally distributed objective values for large, random instances.
This behavior arises from the central limit theorem, as the total cost is a sum of many local interaction terms sampled from a single distribution.

The key feature of the normal distribution is its symmetry.
This symmetry makes it possible to choose a parameter $k$ for the SPO such that the phase difference between the mean amplitude vector and the optimal solution's amplitude vector is exactly $\pi$, leading to highly effective amplification.
Using the visualization technique presented in \cref{fig:spo-evolution}, one can see that an appropriate choice of $k$ can ensure that both the mean vector and the optimal solution's vector remain on the real axis throughout the amplification process.
However, this symmetry has a notable side effect: because the distribution is symmetric around its mean, the algorithm not only amplifies the probability of the best state $\ket{N-1}$ but also amplifies the probability of the worst state $\ket{0}$ to an equal degree.

\begin{figure}[htbp]
    \centering
    \includegraphics[width=\columnwidth]{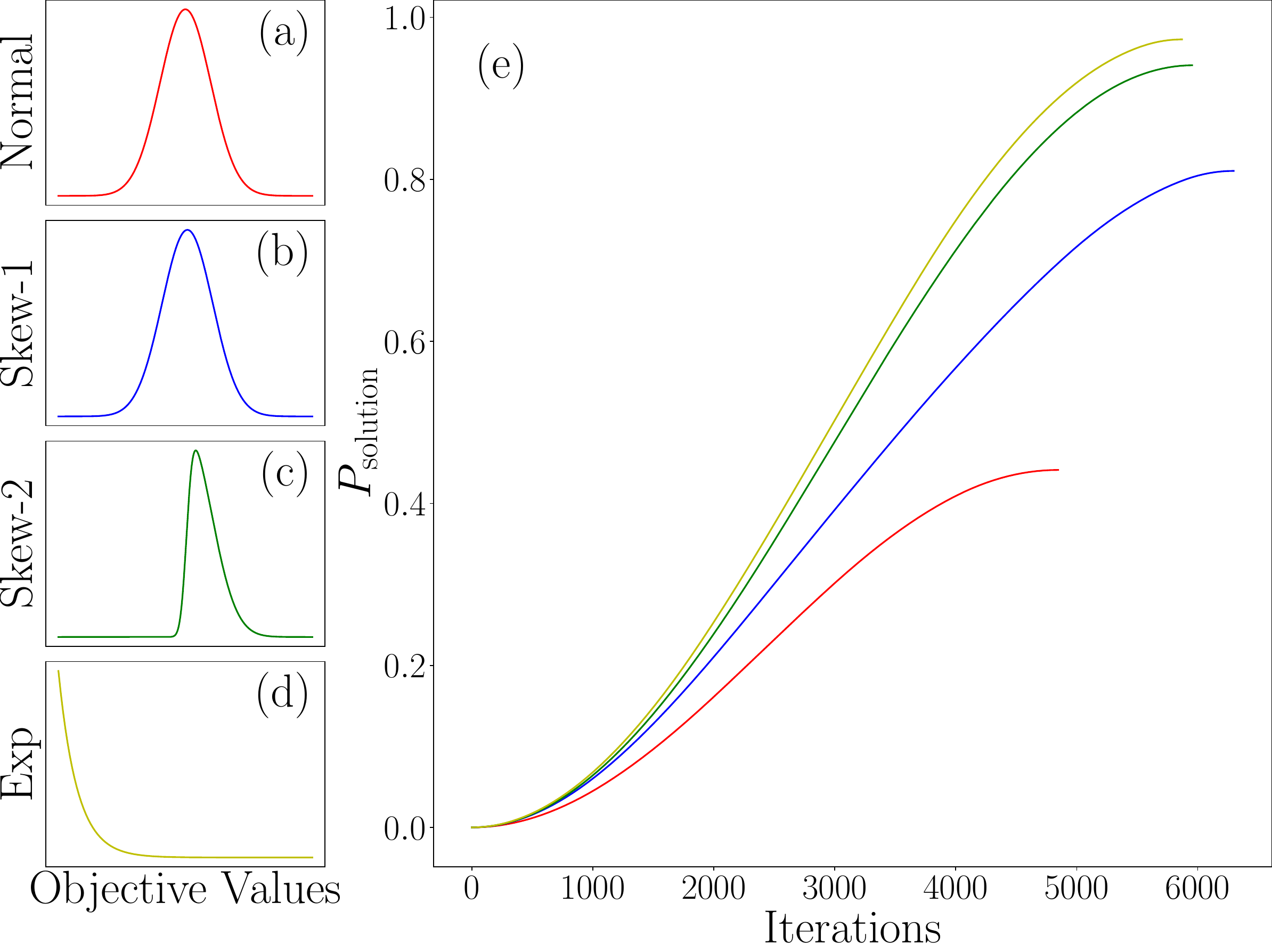}
    \caption[Amplitude amplification performance for various objective distributions]{
    The probability of measuring the optimal solution, $P_{\text{solution}}$, as a function of the number of iterations for an optimization problem with search space $N \approx 2^{25}$.
    The objective values are sampled from four different distributions shown on the left: (a) Normal ($\sigma = 10$), (b) Skew Normal ($\alpha=0.1$), (c) Skew Normal ($\alpha=5$), and (d) Exponential.
    Plot (e) shows that the SPO algorithm successfully amplifies the solution probability to a significant degree for all distributions, though the maximum probability and required iterations vary.
    (From~\cite{naphan-spo}.)
    }
    \label{fig:prob-various-dist}
\end{figure}

\subsection{Asymmetric distributions: skew normal and exponential}

While the normal distribution provides a useful baseline, real-world optimization problems are unlikely to possess perfect symmetry.
We therefore evaluate our algorithm's performance on two classes of asymmetric distributions.
We find that not only do these asymmetries not pose an obstacle, they often prove to be advantageous for the algorithm's performance.

First, to capture a slight deviation from normality, we use the \emph{skew normal distribution}, defined as
\begin{equation}
    \phi_{\text{skew}}(x;\alpha) = 2 \phi_{\text{norm}}(x)\Phi(x;\alpha),
\end{equation}
where $\Phi(x;\alpha)$ is the cumulative distribution function of the standard normal distribution and the parameter $\alpha$ controls the degree and direction of the skew.
As shown in \cref{fig:prob-various-dist}(b) and (c), a non-zero $\alpha$ breaks the symmetry of the distribution.
Second, we consider the \emph{exponential distribution},
\begin{equation}
    \phi_{\text{exp}} (x; \lambda) = \lambda e^{-\lambda x},
\end{equation}
where $\lambda$ is the rate parameter.
This distribution is characterized by a long tail, representing problems where most solutions are poor, but a few are exceptionally good.
For both of these asymmetric distributions, we observe that there exists an optimal parameter $k$ that amplifies the probability of the optimal solution to a significant degree, as depicted in \cref{fig:prob-various-dist}(e).

\subsection{Common observations and performance analysis}

By comparing the results across the three distributions, we can identify several common behaviors and key performance characteristics of the SPO algorithm.
\cref{fig:prob-various-dist}(e) shows how the probability of measuring the optimal solution, $P_{\text{solution}}$, evolves as the algorithm iterates.
Starting from an initial probability of $P_{\text{solution}} \approx 1/N$ (here, $N=2^{25}$), the algorithm successfully amplifies the solution in all cases.
The specific distribution, however, affects both the number of iterations required and the maximum achievable probability, $P_{\text{solution}}^*$.
For the symmetric normal distribution, $P_{\text{solution}}^* \approx 0.5$.
In contrast, the two asymmetric distributions yield much stronger results: $P_{\text{solution}}^* = 0.94$ for the skew normal distribution and $P_{\text{solution}}^* = 0.97$ for the exponential distribution.
This demonstrates that asymmetry in the objective landscape is beneficial, as it reduces the amplification of sub-optimal states and concentrates probability more effectively on the single best solution.

If the algorithm is run beyond the point of maximum amplification, $P_{\text{solution}}$ begins to oscillate, reminiscent of the behavior in standard Grover search---exhibiting the souffle problem behavior.
This brings us to the most critical factor for performance: the choice of the parameter $k$---obtaining the maximum possible amplification requires selecting an optimal $k$ for a given objective distribution.
The algorithm is extremely sensitive to this choice, and deviating even slightly from the optimal value can cause the achievable amplification to decrease rapidly.
\Cref{fig:angle-sensitivity} illustrates this high sensitivity, showing a sharp peak in performance at the optimal $k$ for each of the three distributions.
Finding this optimal $k$ efficiently is a non-trivial task and a key challenge for practical implementations of this algorithm.
Finally, a crucial observation is that for these broad distributions, the effectiveness of the algorithm---that is, the maximum achievable $P_{\text{solution}}^*$---is not affected by the size of the input space $N$.
This is a highly desirable feature, signifying the potential applicability of the SPO approach to optimization problems with extremely large search spaces.

\begin{figure}[htbp]
    \centering
    \includegraphics[width=\textwidth]{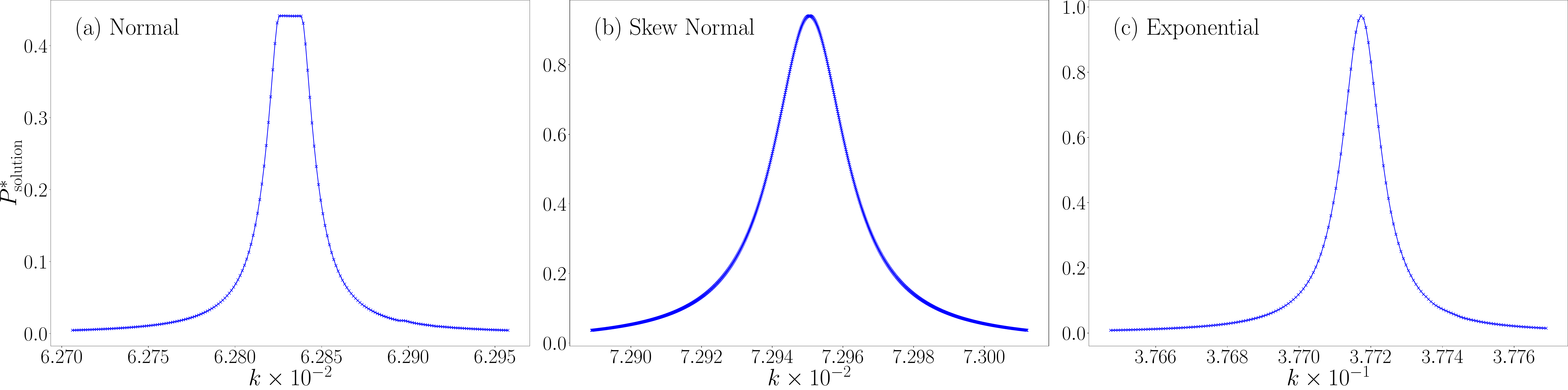}
    \caption[Performance sensitivity of the SPO algorithm to the parameter k]{
    Sensitivity of the SPO algorithm's performance, measured by the maximum achievable solution probability $P_{\text{solution}}^*$, as a function of the scaling parameter $k$.
    The plots for (a) Normal, (b) Skew Normal, and (c) Exponential distributions all display an extremely sharp peak, demonstrating that performance rapidly degrades when deviating from the optimal $k$ value for a given distribution.
    (From~\cite{naphan-spo}.)}
    \label{fig:angle-sensitivity}
\end{figure}

Now that we have introduced the optimization algorithm based on subdivided phase oracle (SPO) and discussed its behavior under fixed-angle amplitude amplification, we turn to the question of formally bounding its runtime. Specifically, we aim to understand whether the fixed-angle strategy is inherently suboptimal, and whether allowing the angles~$k$ to vary across rounds---i.e., treating them as variational parameters---could enable performance that surpasses the Grover-like speedup.

From another perspective, the entire SPO algorithm can be viewed as a specific instance of the Quantum Approximate Optimization Algorithm (QAOA).
The process of alternating a problem-dependent phase operator ($O_{\phi}$) with a problem-independent mixing operator (the Grover diffusion operator) fits directly into the QAOA ansatz structure, where the diffusion operator acts as a ``Grover-style'' mixer~\cite{andreas-grover-mixer-qaoa}.
This insight provides a powerful bridge between the two approaches and motivates a deeper analysis of the broader QAOA framework.

Recognizing that our SPO algorithm is a member of this more general class of protocols naturally leads to a crucial question: What are the performance limitations of the QAOA framework as a whole?
While our results for the SPO are promising, they are specific to one type of mixer.
To understand the full potential of QAOA for practical optimization, we must analyze its performance more broadly.
A key unresolved issue in the field is the lack of theoretical guarantees on the runtime of QAOA---specifically, the number of rounds $p$ required to achieve a given solution quality.

\section{Lower bounds on QAOA runtimes}
\label{sec:main-results}

While QAOA is a promising heuristic, a deeper theoretical understanding of its performance is crucial for evaluating its potential.
A key performance metric is the number of rounds, $p$, required to achieve a desired solution quality, as this directly relates to the algorithm's time complexity.
In this section, we present the second main contribution of this chapter: some of the first analytical lower bounds on $p$ required to achieve a guaranteed approximation ratio, $\lambda$.
We derive these bounds by establishing and exploiting a powerful connection between the digital, round-based structure of QAOA and the continuous-time evolution of quantum annealing.

\subsection{The connection to quantum annealing}
\label{subsec:qaoa-qa-connection}

The link between QAOA and Quantum Annealing (QA) provides the foundation for our analysis.
Recall that QA evolves a system under a time-dependent Hamiltonian $H(t) = (1 - g(t))H_0 + g(t)H_1$.
The QAOA protocol, with its alternating application of unitaries $e^{-i\beta_j H_0}$ and $e^{-i\gamma_j H_1}$, can be viewed as a specific implementation of QA that uses a "bang-bang" schedule.
In this type of schedule, the function $g(t)$ only takes the values 0 or 1, meaning only one Hamiltonian, $H_0$ or $H_1$, is active at any given time.

Under this interpretation, the sum of the absolute values of the QAOA angle parameters corresponds directly to the total annealing time, $t_{anneal}$, of the equivalent QA process:
\begin{equation}
    t_{anneal} = \sum_{j=1}^p \left( |\beta_j| +  |\gamma_j| \right).
    \label{eq:sum-of-angles-to-annealing-time}
\end{equation}
This equivalence allows us to apply theoretical results from the study of general quantum annealing directly to the analysis of QAOA.

\subsection{A general lower bound on QAOA rounds}
\label{subsec:rounds-lower-bound}

Our derivation begins with a known lower bound on the annealing time required to reach a ground state with high probability, which holds regardless of the state's specific trajectory during the process.

\begin{theorem}[from \cite{luis-pedro-annealing-lower-bounds}]
Given two driving Hamiltonians with zero ground state energies, $H_0$ and $H_1$, a quantum annealing protocol that starts in the ground state of $H_0$ and ends at time $t_{f}$ is bounded by
\begin{equation}
    t_f \geq \frac{\braket{H_0}_{t_f} + \braket{H_1}_0 - \braket{H_1}_{t_f}}{\|[H_1, H_0]\|},
\end{equation}
where $\braket{H}_t$ is the expectation value with respect to the Hamiltonian $H$ at time $t$ and $\| \cdot \|$ denotes the spectral norm.
\label{theorem:annealing-lower-bound}
\end{theorem}

By combining Theorem~\ref{theorem:annealing-lower-bound} with the relationship in eq.~\eqref{eq:sum-of-angles-to-annealing-time}, we can establish a lower bound on the sum of the QAOA angles.

\begin{lemma}
\label{lemma:qaoa-ground-state-bound}
For a QAOA protocol with $p$ rounds driven by Hamiltonians $H_0$ and $H_1$ with zero ground state energies, starting from the ground state of $H_0$, the angle parameters $\vec{\gamma}$ and $\vec{\beta}$ are bounded by
\begin{equation}
    \sum_{j=1}^p \left( |\beta_j| +  |\gamma_j| \right) \geq \frac{\braket{H_0}_{p} + \braket{H_1}_0 - \braket{H_1}_{p}}{\|[H_1, H_0]\|}.
\end{equation}
\end{lemma}

\begin{proof}
This is an immediate result from Theorem~\ref{theorem:annealing-lower-bound} and eq.~\eqref{eq:sum-of-angles-to-annealing-time}.
\end{proof}

To translate this bound on angles into a bound on the number of rounds $p$, we introduce the reasonable assumption that the driving Hamiltonians are periodic.
For the problem Hamiltonian $H_1$ built from an integer-valued objective function, a period of $2\pi$ arises naturally.
Most common mixers, $H_0$, are also periodic.
This allows us to bound the meaningful range of any individual angle, leading to our main result.

\begin{theorem}
\label{theorem:qaoa-objective-bound}
Given a classical objective function $C(x)$ for a maximization task, represented by the Hamiltonian $H_C$ and encoded into the phase separator $H_1 = \mathbbm{1}C_{max} - H_C$.
Let the mixer be $H_0$, where all Hamiltonians are $2\pi$ periodic with zero ground state energies.
If a QAOA protocol with $p$ rounds reaches a state with approximation ratio $\lambda$, then
\begin{equation}
p \geq \frac{\braket{H_0}_{p} + \lambda C_{max} - C_{avg}}{4\pi \|[H_C, H_0]\|}.
\label{eq:qaoa-obj-phase-lower-bound}
\end{equation}
\end{theorem}

\begin{proof}
From Lemma~\ref{lemma:qaoa-ground-state-bound}, when imposing the periodicity constraints, we get
\begin{align}
    \sum_{j=1}^p \left( |\beta_j| +  |\gamma_j| \right) 
    &\geq \frac{\braket{H_0}_{p} + \braket{H_1}_0 - \braket{H_1}_{p}}{\|[H_1, H_0]\|} \\
p(2\pi + 2\pi)     
    &\geq \frac{\braket{H_0}_{p} + \braket{H_1}_0 - \braket{H_1}_{p}}{\|[H_1, H_0]\|} \\
p
    &\geq \frac{\braket{H_0}_{p} + \braket{H_1}_0 - \braket{H_1}_{p}}{4\pi \|[H_1, H_0]\|},
\end{align}
where the second line comes from the fact that the two Hamiltonians are $2\pi$ periodic, thus, the largest meaningful $\beta_j$ and $\gamma_j$ are bounded by $2\pi$.
Since the phase separator $H_1$ of QAOA encodes the information of the objective value into the phase, we get
\begin{align}
\braket{H_1}_0 &= C_{max} - C_{avg}, \\
\braket{H_1}_p &= C_{max} - \lambda C_{max},\\
\braket{H_1}_0 - \braket{H_1}_p 
    &= \lambda C_{max} - C_{avg}.
\end{align}
This comes directly from the definitions of the approximation ratio and $H_1$ defined in \cref{eq:h1-definition,eq:approximation-ratio}.
As for the denominator, we know that 
\begin{align}
    [A + kI, B] = [A, B]
\end{align}
holds for any constant $k$ and any matrix $A, B$. 
Because $H_1 = \mathbbm{1}C_{max} - H_C$, we get $\|[H_1, H_0]\| = \|[H_C, H_0]\|$.
\end{proof}

To the best of our knowledge, this is the first result that analytically connects the number of rounds $p$ to the approximation ratio $\lambda$ for general QAOA on combinatorial optimization problems, independent of the mixer choice or problem structure.
It provides a powerful tool for reasoning about the minimum circuit depth required for QAOA to achieve a meaningful result.

\subsection{The effect of Hamiltonian rescaling}
\label{subsec:hamiltonian-rescaling}

One might ask if this bound can be circumvented by simply rescaling the Hamiltonians, a common practice in quantum annealing to adjust energy ranges.
If we rescale $H_0$ by $\alpha_0 > 0$ and $H_1$ by $\alpha_1 > 0$, the lower bound on $p$ changes.
However, for QAOA, the effect of rescaling is different than for continuous QA.
Rescaling the Hamiltonians only changes their period, which means the angles $\beta$ and $\gamma$ can be rescaled accordingly without changing the number of rounds $p$.
The number of rounds is therefore invariant under Hamiltonian rescaling.

This observation confirms that deriving the bounds from the $2\pi$ periodic case is logically sound.
It also implies that we can make the bound tighter by choosing scaling factors that maximize the expression, giving a worst-case lower bound:
\begin{equation}
    p \geq \max_{\alpha_0, \alpha_1} \frac
     {\alpha_0 \braket{H_0}_p + \alpha_1 (\braket{H_1}_0 - \braket{H_1}_p)}
    {2\pi (\alpha_0 + \alpha_1) \|[ H_1, H_0 ]\|}.
\label{eq:worst-case-constant-bound}
\end{equation}
This analysis shows that while the constant factor of the bound can be tightened, its asymptotic behavior cannot be arbitrarily changed by simple rescaling.
Having established this general bound, we will now explore its implications for specific, concrete examples.

\section{Lower bounds on QAOA with a Grover-style mixer}
\label{sec:grover-mixer}

Having established a general lower bound on QAOA performance, we now turn our attention to a specific and important case: QAOA with a Grover-style mixer.

\subsection{The Grover mixer and its general lower bound}

One of the mixers that has been proposed for constrained optimization problems is the \emph{Grover mixer}~\cite{andreas-grover-mixer-qaoa}.
This operator is the parameterized version of the usual diffusion operator in Grover's algorithm and amplitude amplification~\cite{grover-alg, brassard-bhmt-amplitude-amplification-estimation}.
It is a natural choice for many problems because it preserves the state within the subspace of feasible solutions and is invariant to permutations of the input.
This means all states with the same objective value are treated equally, ensuring fair sampling.
This style of QAOA is also studied from the perspective of amplitude amplification, in an attempt to decrease the cost of adaptive Grover searches and to generalize ground-state-finding techniques~\cite{satoh-spo,shyamsundarNonBooleanQuantumAmplitude2023a,koch-gaussian-amplitude-amplification,naphan-spo} and also our SPO approach \cref{sec:subdivided-phase-oracle}.
While the Grover mixer's ability to rapidly explore the state space provides better scaling for small problem sizes, it has been empirically shown to be potentially detrimental for large problems when using standard classical optimization loops~\cite{john-numerical-speedup-constrained, john-sat-qaoa, akshay-reachability-deficit-prl, john-threshold-qaoa, farhi-qaoa-whole-graph-typical-case, farhi-qaoa-whole-graph-worst-case, marsh-wang-optimization-with-quantum-walk-mixer, marsh-wang-quantum-walk-assisted-approximate-algorithm}.
Nonetheless, characterizing its fundamental limitations is crucial, as finding optimal QAOA schedules can be computationally hard~\cite{bittel-vqa-depth-qcma-hard}.

The Hamiltonian for the Grover mixer is defined as the projector onto the uniform superposition of all feasible states, $\ket{\psi_0}$,
\begin{equation}
    H_{Grover} = \mathbbm{1} - \ket{\psi_0}\bra{\psi_0}.
\label{eq:grover-diffusion-hamiltonian}
\end{equation}
For this specific mixer, our general lower bound from Theorem~\ref{theorem:qaoa-objective-bound} simplifies significantly.

\begin{theorem}
Given a classical objective function $C(x)$ for a maximization task, if a QAOA protocol with $p$ rounds driven by the Grover mixer $H_{Grover}$ reaches a state with approximation ratio $\lambda$, then
\begin{equation}
    p \geq \frac{1 - \left|\braket{\psi_0|\psi_{p}}\right|^2 + \lambda C_{max} - C_{avg}}{ 4 \pi \sigma_{C} },
\end{equation}
where $C_{max}$, $C_{avg}$, and $\sigma_{C}$ are the maximum, average, and standard deviation of the objective function, respectively.
\label{theorem:grover-objective-bound}
\end{theorem}

\begin{proof}
The term $\braket{H_{Grover}}_p$ can be evaluated directly from \cref{eq:grover-diffusion-hamiltonian}.
Since $H_C$ is diagonal in the computational basis and $H_{Grover}$ is a rank-1 projector, the term $[H_C, H_{Grover}]$ evaluates to a projector of rank at most 2.
There is a closed-form formula for calculating the spectral norm of a matrix of this form given in~\cite{ipek-eigenvalue-rank-1-skew-symmetric}, which gives 
\begin{align}
    \| [H_C, H_{Grover}] \| = \sigma_C.
\label{eq:spectral-norm-commutator-grover-mixer}
\end{align}
Plugging these values into \cref{theorem:qaoa-objective-bound} gives the lower bound as shown.
\end{proof}

This theorem reveals a crucial insight: the performance of QAOA with a Grover mixer depends only on the statistical distribution of the objective values, not on the combinatorial structure of the problem itself.

\subsection{Application to k-local Hamiltonians}

We can make this bound even more concrete for problems whose objective functions can be described by a $k$-local Hamiltonian.
This is a natural assumption for many combinatorial optimization problems, where the cost function is a sum of terms that each act on at most $k$ variables~\cite{taillard-design-heuristic-algorithm-optimization-problem, lucasIsingFormulationsMany2014}.
Such a Hamiltonian can be written as
\begin{equation}
    H_C = \mathbbm{1}C_{avg} - \sum_{\nu=1}^m \alpha_\nu H_\nu,
\end{equation}
where each $H_\nu$ is a product of at most $k$ Pauli $Z$ operators, and $\alpha_\nu$ are real-valued weights.
For this class of Hamiltonians, the standard deviation $\sigma_C$ can be calculated directly from the coefficients as $\sigma_C^2 = \sum_{\nu=1}^m \alpha_\nu^2$.
This leads to a more explicit lower bound.

\begin{theorem}
Given a classical objective function $C(x)$ for a maximization task, represented by a $k$-local Hamiltonian $H_C = C_{avg}\mathbbm{1} - \sum_{\nu=1}^m \alpha_\nu H_\nu$, where each $H_\nu$ is a product of Pauli Z with at most $k$ terms, each $\alpha_\nu$ is a real number denoting the weight of the term, and $C_{avg}$ and $C_{max}$ denote the average and the global maximum of $C(x)$.
If a QAOA protocol with $p$ rounds driven by the phase separator $H_1 = C_{max}\mathbbm{1} - H_C$ and the mixer $H_{Grover}$ that starts in the state $\ket{\psi_0} = \ket{+}^{\otimes n}$ reaches a state $\ket{\psi_p}$ with approximation ratio $\lambda$, then
\begin{align}
p
    &\geq \frac
    {1 - \left|\braket{\psi_0|\psi_{p}}\right|^2 + \lambda C_{max} - C_{avg}}
    {4 \pi \sqrt{\sum_{\nu = 1}^m \alpha_\nu^2}}.
\end{align}
\label{theorem:grover-local-cost}
\end{theorem}

\begin{proof}
Since we have already established that $C_{avg}$ is the coefficient of the trivial Pauli term, we now show that the standard deviation $\sigma_C$ can be easily calculated from the problem Hamiltonian representation as well.
Realizing that
\begin{align}
\sigma_C^2
    &= \braket{s|H_C^2|s} - \braket{s|H_C|s}^2\\
    &= \left(C_{avg}^2 + \sum_{\nu=1}^m \alpha_\nu^2\right) - C_{avg}^2\\
    &= \sum_{\nu=1}^m \alpha_\nu^2,
\end{align}
where the first term of the second line is due to the fact that only when all the Pauli $Z$'s cancel out into identity that we get non-zero contributions to the sum.
\end{proof}

This formulation is powerful because it allows us to estimate the lower bound on the required circuit depth directly from the coefficients of the problem Hamiltonian.

\subsection{Implication for Max-Cut: a polynomial round requirement}

The consequences of these bounds are most apparent when applied to a specific problem like Max-Cut.
By applying Theorem~\ref{theorem:grover-objective-bound} to a Max-Cut instance on a graph with $|E|$ edges, we can derive a direct corollary.
We can form a corollary for Max-Cut to formally state this insight.
\begin{corollary}
If a QAOA protocol with $p$ rounds finds an approximate solution with approximation ratio $\lambda$ to Max-Cut of a graph with $|E|$ edges driven by the objective value phase separator $H_1$ and the Grover-mixer $H_{Grover}$ that starts in the state $\ket{\psi_0} = \ket{+}^{\otimes n}$, then
\begin{align}
    p &\geq \frac{
        1 - \left|\braket{\psi_0|\psi_{p}}\right|^2 + 
      \lambda C_{max} - |E|/2}
    { 2 \pi \sqrt{|E|} }.
\end{align}
\end{corollary}
\begin{proof}
Applying the lower bound of Theorem~\ref{theorem:grover-objective-bound} to an instance of Max-Cut with the number of edges denoted by $|E|$.
Since there are 4 possible ways to assign a cut partition to an edge, it is easy to see that the expected cut value for each edge is $1/2$, and since any pair of edges are independent of each other, the expected cut values then equals $|E|/2$.
We can use a similar line of argument to calculate the standard deviation of the cuts of a graph, which evaluates to $\sqrt{|E|}/2$.
\end{proof}

We know that for any bipartite graph, the maximum number of cut edges will be equal to the number of edges in the graph.
Using the above corollary, we get
\begin{align}
p &\geq \frac{
        1 - \left|\braket{\psi_0|\psi_{p}}\right|^2 + 
      \lambda |E| - |E|/2}
    { 2 \pi \sqrt{|E|} } \\
  &\geq \frac{(2\lambda - 1)}{4\pi}\sqrt{|E|}.
\end{align}
Since $|E|$ is at least linear in the number of vertices $n$ when the graph is connected, this shows that there exist problem instances in the class of Max-Cut that QAOA with Grover-mixer cannot get any approximation ratio guarantee with constant rounds and would require the number of rounds to be at least a polynomial in $n$ to obtain a constant approximation ratio.
We note that this conclusion is different and cannot be compared directly to results in~\cite{bravyi-rqaoa}, where it is shown that constant rounds QAOA with transverse field mixer can obtain constant approximation ratio for a certain families of regular bipartite expander graphs due to different assumptions and choice of mixer.

Although we have only shown this formally with the Max-Cut problems, similar calculations can be made for most combinatorial optimization problems.

This is a significant finding, as it demonstrates that a constant-round QAOA with this mixer is not sufficient to solve many common optimization problems with a guarantee of quality.

\subsection{Discussion on the tightness of the bound}

A natural question is whether this polynomial lower bound is tight.
While our analysis indicates a polynomial requirement, some empirical studies have suggested that finding optimal angles for QAOA on hard problem instances might require an exponential number of rounds~\cite{john-numerical-speedup-constrained, john-sat-qaoa}.
This apparent discrepancy may arise because the heuristic, round-by-round methods used to find angles in numerical simulations might not achieve the true, globally optimal parameters for a given $p$.
Furthermore, intuition from the low-depth ($p=1$) regime suggests that each round of this QAOA variant should only contribute a Grover-like progress, which does not fully capture the complex dynamics at higher depths~\cite{mcclean-low-depth-quantum-optimization}.
Therefore, while our lower bound is rigorously proven, the possibility remains that this type of QAOA could offer a significant advantage over a simple Grover adaptive search, and determining its true upper-bound complexity remains an important open question.

\section{Lower bounds applied to search problems}
\label{sec:search-lower-bound}

Having analyzed the lower bounds for general optimization problems, we now turn to the special case of search.
A search problem on $n$ qubits can be described as finding a bitstring $z$ from a set $S = \{ z | C_{search}(z) = 1 \}$, where the objective function $C_{search}(x)$ maps inputs to $\{0, 1\}$.
While using the hybrid QAOA framework to solve search is not practically efficient, analyzing this case provides a crucial test for the tightness and generality of our derived bounds.
The well-established lower bound for unstructured search is $\Omega(\sqrt{N/m})$ queries to an oracle, where $m$ is the number of marked items.
Our goal is to see if our QAOA bounds can recover this fundamental limit.

This analysis is also motivated by a desire to understand the behavior of our bounds when the contribution from the mixer Hamiltonian, $\braket{H_0}_p$, dominates, a scenario not fully explored in the previous section.
We will first show that our energy-based bound (Theorem~\ref{theorem:qaoa-objective-bound}) is indeed tight for the Grover mixer.
We will then show that for the transverse-field mixer, this same bound reveals a dependency on the problem's structure but is ultimately loose.
Finally, we will derive and apply a new, more general bound based on state overlap that correctly captures the unstructured search limit for a broad class of mixers.

\subsection{Performance with a Grover mixer: a tight bound}
\label{subsec:search-lower-bound-grover-mixer}

We begin by showing that the lower bound from Theorem~\ref{theorem:grover-objective-bound} is tight when applied to search problems using the Grover mixer.
Since Grover's algorithm can be viewed as a QAOA with fixed angles, our bound is expected to recover the well-known scaling.

\begin{corollary}
If a QAOA protocol solves a search problem with $p$ rounds, finding a marked state from $m$ marked states, with success probability $\lambda$ and is driven by the objective value phase separator $H_1$ (phase oracle with adjustable phase) and the Grover-mixer $H_{Grover}$ that starts in the state $\ket{\psi_0} = \ket{+}^{\otimes n}$, then
\begin{equation}
p  \geq \frac{\lambda}{2\pi} \sqrt{\frac{N-m}{m}} - \frac{1}{2\pi} \sqrt{\lambda (1 - \lambda)}.
\label{eq:grover-search-bound}
\end{equation}
\label{corollary:grover-search-bound}
\end{corollary}

\begin{proof}
We can see that, by rewriting the states as linear combinations of solution and non-solution states, we get
\begin{gather}
    \left|\braket{\psi_0|\psi_p}\right|^2 = \frac{1}{N}\left( \sqrt{\lambda m} + \sqrt{(1 - \lambda)(N-m)} \right)^2, \label{eq:final-state-overlap-search}\\
    C_{max} = 1, \text{ and } C_{avg} = m/N.
\end{gather}
The standard deviation for the search cost function is
\begin{equation}
\sigma_{C}
    = \frac{\sqrt{m(N-m)}}{N}.
\label{eq:standard-deviation-search}
\end{equation}
Plugging these into the lower bound in Theorem~\ref{theorem:grover-objective-bound} gives this expression, thus completing the proof.
\end{proof}

As expected, the expression in eq.~\eqref{eq:grover-search-bound} recovers Grover's search scaling of order $\sqrt{N/m}$ when $m \ll N$.
This result confirms that our energy-based lower bound is tight and consistent with the optimal query complexity for unstructured search when the appropriate mixer is used.

\subsection{Performance with a transverse-field mixer: a structure-dependent bound}
\label{subsec:search-lower-bound-tf-mixer}

Next, we examine the search problem with the transverse-field ($H_{TF}$) mixer.
Unlike the Grover mixer, $H_{TF}$ is not invariant under input permutations, and its performance depends on the structure of the marked states.
To illustrate this, we analyze two different search problems with the same number of solutions but different structures.

First, we consider a set of marked states $S_{dist-3}$ where the Hamming distance between any two marked states is at least 3.
For this problem, the lower bound evaluates to
\begin{equation}
p  \geq \frac{n(1 - 2\sqrt{\lambda(1-\lambda)}) + 2\lambda - 2m/N}{4\pi \sqrt{n}},
\end{equation}
which is on the order of $\sqrt{n}$, regardless of the number of marked states $m$.

Second, we consider a set $S_{Hamming-k}$ where all states with a fixed Hamming weight $k$ are marked.
In this case, the lower bound becomes
\begin{equation}
p \geq \frac{n(1 - 2\sqrt{\lambda(1-\lambda)}) + 2\lambda - 2m/N}{4\pi \sqrt{2k(n-k) + n}}.
\end{equation}
This bound depends on $k$; for instance, when $k=n/\log^2 n$, the lower bound is on the order of $\log n$.
These two examples show that for the same number of marked states, the lower bound can yield different scaling ($\sqrt{n}$ vs. $\log n$), revealing that the energy-based bound for the transverse-field mixer is sensitive to the problem's structure.
However, since neither recovers the expected $\Omega(\sqrt{N/m})$ worst-case scaling, we conclude that this bound is loose for the general case of unstructured search with $H_{TF}$.

\subsection{A general and tighter bound via state overlap}
\label{subsec:search-overlap}

To resolve the looseness of the energy-based bound for certain mixers, we can derive a new bound based on a different principle: the rate of change of state overlap.
This leads to the following lemma.
\begin{lemma}
Given a QAOA protocol driven by Hamiltonians $H_0$ and $H_1$, where $H_1$ is $2\pi$ periodic, and a projector $P_0$ that commutes with $H_0$, the number of rounds $p$ is bounded by
\begin{equation}
p \geq \frac{ \left| \braket{P_0}_{p} - \braket{P_0}_{0} \right| }{2\pi \left\| [P_0,H_1] \right\|}.
\end{equation}
\label{lemma:probability-based-bound-annealing}
\end{lemma}
\begin{proof}
Consider dynamics with a Hamiltonian
\begin{equation}
    H(t) = (1-g(t))H_0 + g(t) H_1,
\end{equation}
where the value of $g(t) = \{0, 1\}$ in the case of QAOA.
If $P_0$ is a projector that commutes with $H_0$, we have that
\begin{align}
\left| \frac{d}{dt}\braket{P_0}_t \right|
    &= \left| -i \tr{[H(t),\rho_t]P_0} \right| \\
    &= \left| \tr{\rho_t [P_0, g(t) H_1]} \right|\\
    &\leq |g(t)| \| [P_0,H_1] \|,
\end{align}
where the last line arrived from the cyclic property of trace and Rayleigh quotient.
Integrating over the duration $t_f$ of the protocol gives
\begin{align}
\left| \braket{P_0}_{t_f} - \braket{P_0}_{0} \right|
    &=  \left| \int_{0}^{t_f} \frac{d}{dt} \braket{P_0}_{t} dt \right| \\
    &\leq \| [P_0,H_1] \| \int_{0}^{t_f} |g(t)| dt \label{eq:probability-based-bounds-integration-line}\\
    &\leq \| [P_0,H_1] \| \sum_{j=1}^{p} |\gamma_j|\\
    &\leq \| [P_0,H_1] \| 2\pi p.
\end{align}
where we assumed $|\gamma_j| < 2\pi$ due to the periodicity.
Then,
\begin{equation}
    p \geq \frac{ \left| \braket{P_0}_{p} - \braket{P_0}_{0} \right| }{2 \pi \left\| [P_0,H_1] \right\|},
\end{equation}
completing the proof.
\end{proof}

While the numerator is always between 0 and 1, this bound proves to be much stronger and more general for search problems.

\subsection{General lower bound for search with QAOA}
\label{subsec:search-qaoa}

By choosing $P_0 = \ket{\psi_0}\bra{\psi_0}$ where $\ket{\psi_0}$ is the uniform superposition state, we can apply the above lemma to form a new, more general theorem for search.

\begin{theorem}
If a QAOA protocol solves a search problem with $p$ rounds, finding one of $m$ marked states with success probability $\lambda$, and is driven by any mixer $H_0$ that has $\ket{\psi_0} = \ket{+}^{\otimes n}$ as its ground state, then
\begin{equation}
p \geq \frac{\lambda(N-2m) + m - 2\sqrt{\lambda(1-\lambda)m(N-m)}}{2 \pi \sqrt{m(N-m)}}.
\end{equation}
\label{theorem:search-bound-uniform-superposition-initial-state}
\end{theorem}
\begin{proof}
Consider the lower bound expression in Lemma~\ref{lemma:probability-based-bound-annealing}, when choosing $P_0 = \ket{\psi_0}\bra{\psi_0}$ we see that
\begin{align}
    [P_0, H_1] = -[\mathbbm{1} - P_0, H_1] = [H_1, H_{Grover}].
\end{align}
Then using Lemma~\ref{lemma:probability-based-bound-annealing}, we get
\begin{align}
p \geq \frac{ \left| \braket{P_0}_{p} - \braket{P_0}_{0} \right| }{2 \pi \left\| [P_0, H_1] \right\|} = \frac{1 - \left|\braket{\psi_0|\psi_p}\right|^2}{2\pi \sigma_C}.
\end{align}
By plugging in the values from \cref{eq:final-state-overlap-search} and \cref{eq:standard-deviation-search} into the above inequality, we get the expression as claimed, thus completing the proof.
\end{proof}

This powerful result successfully recovers the optimal Grover scaling for the unstructured search algorithm by providing a lower bound of $\Omega(\sqrt{N/m})$ that matches the known upper bound~\cite{grover-alg}.
Crucially, it holds not only for the Grover mixer but also for the transverse-field mixer and any other mixer that begins in the standard uniform superposition state.
This confirms that while the energy-based bound can reveal problem structure, this overlap-based bound correctly captures the fundamental query complexity for unstructured search within the QAOA framework.
This result can also be adapted to apply to continuous-time quantum walk searches that start in a stationary distribution.

\section{Implications for network architecture and resource management}
\label{sec:network-implications}

The analysis of QAOA runtimes is not merely an exercise in algorithmic complexity theory; it provides fundamental, actionable insights that connect directly to the architectural challenges of designing heterogeneous quantum repeater networks.
The performance characteristics of the applications that will run on a future Quantum Internet are the primary drivers of its architectural requirements.
By establishing lower bounds on the number of rounds for a leading class of optimization algorithms, we can begin to quantitatively model the demands these applications will place on network resources.

\subsection{Computation models, architecture, and traffic implications}

The number of QAOA rounds $p$ is a direct proxy for how long a distributed quantum computation will occupy network resources.
This runtime determines the coherence time required at quantum memories, the operational lifetime of end-to-end entangled states, and the entanglement generation rate needed across elementary links to sustain the computation.

For industrially relevant problems, the number of algorithmic logical qubits often exceeds the capacity of a single device, motivating the distributed architecture central to this thesis.
The resulting traffic model between nodes depends heavily on how qubits or subcircuits are partitioned across the network, which in turn is influenced by the architecture of the quantum devices themselves---including their physical qubit connectivity, the error-correcting codes employed, and the computational model used to realize inter-node operations.

These variables span different abstraction levels and can be managed individually (e.g., via hardware-specific compilers) or through a unified optimization framework that globally schedules physical qubit operations and Bell-pair usage.
The best strategy depends on system scale and use case.
Tightly coupled clusters, such as quantum datacenters, may benefit from holistic optimization.
In contrast, wide-area networks or privacy-sensitive architectures like blind quantum computation require fundamentally different computational models that cannot be globally optimized due to the obfuscation of operational steps and the need for verification.
Since nodes may not be fully trusted, protocols must enforce privacy and correctness, which introduces translation overheads and favors coarser abstractions—significantly altering resource demands.

To better understand these trade-offs, we now focus on the datacenter setting, where high trust and physical proximity simplify coordination and optimization.
The choice of fault-tolerant computational model plays a crucial role in determining the timing, fidelity, and structure of Bell pairs needed for inter-node operations.
Gate-based models typically use Bell pairs to implement nonlocal gates (via telegates) or to teleport qubits (teledata).
Efficient partitioning and scheduling can significantly reduce entanglement consumption~\cite{davisDistributedQuantumComputing2023a,vanmeterArchitectureQuantumMulticomputer2006}.
By contrast, measurement-based models such as one-way MBQC~\cite{brown-roberts-ft-mbqc}, Pauli-based computation~\cite{litinski-game-of-surface-code}, lattice surgery~\cite{fowlerLowOverheadQuantum2019,horsmanSurfaceCodeQuantum2012,litinskiLatticeSurgeryTwist2018}, or CSS surgery~\cite{poirsonEngineeringCSSSurgery2025a} often require large graph or GHZ states as primitives, especially in topologies with dense connectivity.

An alternative approach is the \emph{active volume} model~\cite{litinski-active-volume}, which combines architectural and computational perspectives.
It can reduce overheads by orders of magnitude while requiring only Bell pairs, assuming a high-connectivity interconnect such as proposed in~\cite{sakumaOpticalInterconnectModular2024}.
This paradigm fundamentally alters the traffic model, emphasizing fewer but more strategic Bell-pair uses over graph-state generation.

At a lower level of abstraction, code-specific implementations---such as blocklet-based surface code execution~\cite{litinskiBlockletConcatenationLowoverhead2025a}, qLDPC-based schemes~\cite{yoderTourGrossModular2025b}, or distributed Floquet codes~\cite{sutcliffeDistributedQuantumError2025}---enable further optimization by aligning physical layout and scheduling with the error-correcting code structure.
These models promise reduced resource consumption but demand precise control over how logical operations are mapped and executed.
However, they currently lack efficient compilers capable of translating standard gate-based algorithms without incurring significant overheads, often quadratic in size~\cite{babbushFocusQuadraticSpeedups2021}.
New compilation efforts targeting these models are ongoing~\cite{sitong-substrate-scheduler,robertsonResourceAllocatingCompiler2025,saadatmandFaulttolerantResourceEstimation2024,saadatmandSuperconductingQubitsUtility2025,bowenDesignEfficiencyGraphstate2025}.

Ultimately, accurate traffic modeling must incorporate the underlying error-correcting code (e.g., surface code vs.\ qLDPC~\cite{yoderTourGrossModular2025b}), the nature of communication primitives, and algorithm-specific patterns.
Our findings show that QAOA runtimes span a broad spectrum: from constant-time $O(1)$ algorithms for easy instances, to polynomial-time $\text{poly}(n)$ scaling for structured problems, and even $O(\sqrt{N})$ runtimes for unstructured search-like tasks.
This variability is crucial for classifying network traffic types.
Short tasks generate bursty, lightweight entanglement demands.
In contrast, long-running computations with polynomial complexity require sustained, high-rate delivery of high-fidelity Bell pairs over extended periods.
Capturing this diversity in temporal and resource requirements is essential for building realistic traffic models in network simulators like QuISP.
To accurately predict network performance, a simulator's traffic generator must be able to model this diversity of application demands, which arise directly from the algorithmic properties we have analyzed.

\subsection{Resource allocation strategies informed by algorithmic runtimes}

Understanding these algorithmic runtime profiles directly informs the architectural choices described in \cref{chapter:architecture-memory}, particularly with respect to resource allocation and connection setup strategies.
The different runtime classes of QAOA translate naturally into different network operation modes.

Long-running applications with predictable runtimes and high entanglement demands (e.g., those with polynomial or square-root scaling) are ill-suited to a best-effort, connectionless model.
These tasks require end-to-end guarantees on fidelity and availability, motivating the use of a connection-oriented approach with dedicated resource reservation---akin to circuit switching.
This is the paradigm implemented by our RuleSet protocol and two-pass connection setup mechanism, which allocate and manage network resources over the full duration of such computations.

Conversely, shorter and bursty workloads---such as constant-round algorithms---can be handled more efficiently with opportunistic resource allocation via dynamic, connectionless protocols or statistical multiplexing.
This knowledge allows a network's Connection Manager (CM) within QRSA to make more intelligent decisions, authoring RuleSets that requrest the appropriate resource reservation strategy based on the anticipated runtime of the application.

These architectural insights are especially relevant for advanced scenarios like blind quantum computation, where a cluster of distributed quantum devices must coordinate on a large-scale computation while preserving privacy and correctness.
Such applications impose stringent requirements on entanglement quality, reliability, and scheduling, all of which depend on robust error management protocols like those described in \cref{sec:error-management}.
The runtime lower bounds derived in this work offer concrete, analytical tools to reason about total resource requirements and guide Quality of Service (QoS) decisions for these demanding use cases.

\section{Conclusion}
\label{sec:chapter-conclusion}

This chapter highlights the critical synergy between quantum algorithm analysis and quantum network design.
We have demonstrated that understanding the performance and resource requirements of key applications, such as optimization algorithms like QAOA and the Subdivided Phase Oracle, is fundamental to designing practical and efficient quantum networks.
The analytical lower bounds derived here are not just abstract results in complexity theory; they are essential engineering parameters that inform everything from traffic modeling to the choice of high-level architectural paradigms.

By quantifying the runtime demands of these algorithms, we provide the necessary tools for more effective traffic modeling, resource allocation, and overall architectural planning.
This foundational understanding allows us to intelligently apply the principles of the Quantum Recursive Network Architecture (QRNA) and the RuleSet-based control plane.
It enables a network architect to anticipate the needs of different computational workloads and provision resources accordingly, a critical capability for building the heterogeneous quantum repeater networks that will form the future Quantum Internet.
\clearpage
\chapter{Discussion and conclusion}
\label{chapter:discussion-and-conclusion}

\section{Summary of contributions}

This thesis has addressed some of the most pressing engineering and architectural challenges facing the development of a scalable, heterogeneous Quantum Internet.
By drawing on principles from classical network design and focusing on the practical requirements of real-world implementation, we have developed a cohesive framework that spans from high-level architecture to specific hardware-level protocols and their analysis.
The main ideas and contributions presented can be summarized by four key pillars.

First, building upon foundational concepts from prior research, we introduced a comprehensive network architecture for memory-based quantum repeaters, grounded in the Quantum Recursive Network Architecture (QRNA) and managed by programmable, event-driven RuleSets.
This framework, realized in the Quantum Routing Software Architecture (QRSA), provides a modular and extensible foundation for managing network resources, defining connection semantics, and enabling interoperability between diverse and independently operated quantum networks.

Second, to validate this architecture and explore protocol performance, we utilized and contributed to the development of the QuISP network simulator, a project that began prior to my joining the research group.
We established its credibility through a rigorous cross-validation study against another independent simulator, demonstrating that despite different modeling philosophies, QuISP faithfully reproduces the expected physical behaviors of quantum networks.
This work affirms that well-designed simulation is an indispensable tool for guiding the design and analysis of complex quantum network protocols and architectures.

My primary and novel contributions then extend this architectural model.
Third, we extended our architectural framework to the domain of all-photonic repeaters.
We introduced the half-Repeater Graph State (half-RGS) as a modular building block that resolves key engineering challenges in monolithic RGS generation and timing synchronization.
Building on this, we developed a purification-enhanced RGS scheme, demonstrating for the first time a practical method for non-neighbor purification in a memory-less repeater chain.
This work transforms the all-photonic segment into a robust, high-performance component that can be abstracted as a "virtual link" within the broader QRNA framework.

Finally, to ground our architectural work in the needs of real applications, I led the analysis of the performance of quantum optimization algorithms.
We presented some of the first analytical lower bounds on the number of rounds required for the Quantum Approximate Optimization Algorithm (QAOA) to achieve a guaranteed approximation ratio.
This analysis provides concrete data on the runtime and resource requirements of a key application class, enabling more realistic traffic modeling and informing resource management strategies for networks designed to support distributed quantum computation.

\section{Revisiting the research questions}

Having summarized the contributions, we now evaluate how this work addresses the primary research questions posed in the introduction of this thesis.

The first question asked: \textit{How do we design a unified network architecture that supports heterogeneous quantum repeaters?}
\Cref{chapter:architecture-memory,chapter:architecture-all-photonic} of this thesis directly address this question.
\Cref{chapter:architecture-memory} laid out the QRNA and RuleSet framework as a candidate for a unified architecture, designed specifically to manage different technologies through abstraction.
\Cref{chapter:architecture-all-photonic} then provided a crucial analysis by showing how all-photonic repeaters---a fundamentally different paradigm---can be successfully integrated into this memory-centric architecture via the half-RGS and virtual link abstractions.
While full heterogeneity remains a challenge, this work provides a concrete and viable architectural blueprint for achieving it.

The second question was: \textit{How can simulation guide the design and resource management of scalable quantum networks?}
This thesis answers this question both by example and by providing tools.
The cross-validation study in \Cref{chapter:architecture-memory} established QuISP as a reliable tool for architectural exploration.
Furthermore, the algorithmic analysis in \Cref{chapter:qaoa-and-network-resources} demonstrates how simulation can be used: by taking the runtime bounds of algorithms like QAOA, we can create realistic traffic models that allow simulators like QuISP to study network-wide phenomena like resource contention and scheduling policy performance, which would be intractable to analyze purely analytically.

The third question asked: \textit{How can all-photonic quantum repeaters become practical and interoperable with memory-based designs?}
\Cref{chapter:architecture-all-photonic} provides a partial but significant answer.
By introducing the half-RGS to simplify generation and by solving the long-standing purification problem, we make the RGS scheme substantially more robust and practical for high-error environments.
Furthermore, the ``virtual link'' abstraction provides the explicit mechanism for interoperability, defining how a memory-based control plane using RuleSets can manage and utilize a memory-less all-photonic sub-network.
The full realization of this vision still depends on experimental progress, but our work provides the necessary architectural and protocol-level solutions.

\section{Future work and open questions}

The contributions of this thesis lay the groundwork for numerous avenues of future research, ranging from near-term engineering projects to broader, more fundamental open questions for the field.

\subsection{Near-term research directions}

Several concrete projects can directly extend the work presented here:
\begin{enumerate}
    \item Integrating Segment Routing into QRNA: A promising direction is to adapt segment routing concepts from classical networks to the QRNA framework. This could involve embedding segment identifiers corresponding to recursive virtual links, potentially streamlining route computation and enabling richer, programmable routing policies.
    
    \item Expanding Simulation Capabilities in QuISP: Further work is needed to mature the QuISP simulator by expanding its error modeling and protocol support. Key areas for development include: (1) implementing real-time link monitoring, a responsibility of the Hardware Monitor that could draw on techniques from recent work~\cite{casapao-double-selection-distimator,ananda-distimator-sigmetrics,joshua-distimator}; (2) integrating routing algorithms that explicitly account for link error characteristics; (3) enabling adaptive routing to handle link failures; (4) incorporating the full range of node types defined in \cref{sec:architecture-memory:networking}; (5) adding support for quantum error-correcting codes to allow for the simulation of 2G, 3G, and all-photonic RGS repeaters; (6) implementing diverse multiplexing strategies; (7) completing the resource pool implementation to fully match the RuleSet specification; and (8) developing an advanced, configurable traffic modeling engine.
While not exhaustive, this list represents a selection of the highest-priority areas for future development.
    
    \item Bridging Simulation and Experiment: A crucial task is to develop methodologies for translating real-world experimental testbed parameters into QuISP simulations. Improving this pipeline is essential for using simulation to accurately predict or replicate the behavior of physical hardware.
    
    \item Mapping Architecture to Hardware: A key engineering challenge is to map the abstract software components of the Quantum Routing Software Architecture (QRSA) onto concrete hardware platforms, defining the specific interfaces and control logic required for a practical implementation.
\end{enumerate}

\subsection{Broader open questions for the field}

This research also highlights several deep, open questions that will shape the future of quantum networking:
\begin{enumerate}
    \item In quantum error correction (QEC), a robust theoretical foundation exists for defining fault-tolerance paradigms, which establishes that analyzing Pauli channels is often sufficient for studying both phenomenological and circuit-level noise. However, the direct applicability of this foundation to quantum networking is less clear, especially since the dominant source of error is often photon loss. A key open question is therefore: Can a similarly robust theoretical framework be developed to enable efficient network simulations that accurately represent realistic noise models, including loss, without requiring full density-matrix simulation? Answering this is essential for building confidence in our near-term simulation and design tools.

    \item This thesis has presented engineering advancements for all-photonic quantum repeaters, suggesting their operational characteristics can increasingly resemble those of memory-based networks. We have also proposed methods for managing network heterogeneity using the Quantum Recursive Network Architecture (QRNA) framework. To translate these insights into practice, a concrete research question arises: When comparing 2G and RGS-based repeaters, which technology is better suited for deployment near the network edge versus the backbone? This question requires considering realistic network topologies and the known robustness and fragility trade-offs of the classical Internet to guide optimal deployment strategies.

    \item We have demonstrated that RGS-based segments can be successfully abstracted as virtual links within the QRNA framework. This raises a further architectural question: Are there other repeater models or paradigms that cannot be so readily abstracted and would thus necessitate fundamental modifications to the conceptual levels of QRNA? For instance, how might repeaters based on the Zeno effect, or those utilizing continuous-variable (CV) quantum information or other bosonic encodings, fit within or challenge this recursive architectural model? Fully understanding the integration requirements of these diverse repeater types is necessary to test the universality of the architecture.

    \item A crucial challenge lies at the boundary between the network and its users: What is the optimal interface between the network-level operations managed by QRSA and RuleSets, and the application layer that consumes the generated entanglement? Defining appropriate "quantum socket" abstractions is critical. While recent work offers one vision for this interface by coupling application logic to specific hardware capabilities , it remains an open question whether this direction is the most promising path forward, or if more hardware-agnostic interfaces are required for a truly interoperable application ecosystem.
\end{enumerate}

\section{Perspective}

Ultimately, the journey from isolated laboratory demonstrations to a globally interconnected Quantum Internet will not be defined by a single breakthrough technology.
Instead, progress will come from our ability to compose a diverse ecosystem of specialized solutions---guided by the principle of openness and interoperability, much like the inclusive spirit of the early classical Internet.
The central argument of this thesis is that designing for such heterogeneity is not a secondary concern to be addressed after scalable quantum systems arrive---it is the foundational engineering challenge of our time.
Waiting for the experimental and system to settle is a risky strategy; history shows that once technologies are widely adopted, retrofitting for compatibility becomes prohibitively difficult.
By demonstrating that memory-based and all-photonic repeater paradigms can be unified under a single, recursive architectural framework---and by grounding that framework in the concrete demands of emerging applications---this work offers a practical roadmap for an architecture that is modular, evolvable, and prepared for future innovation.

Realizing this vision will require sustained collaboration across the distinct cultures of theoretical computer science, quantum physics, and network engineering.
The challenges remain formidable, but they are no longer merely conceptual; they are now the concrete engineering problems of interface design, protocol optimization, and resource management.
While many of the core theoretical concepts are in place, progress will be driven as much by clever engineering as by new physics.
As has been said in classical computing by Leiserson et al., there is still ``plenty of room at the top''~\cite{leisersonTheresPlentyRoom2020} for the clever optimizations and smarter implementations that can improve performance by orders of magnitude.
Good engineering, done well, is what will transform today's feasibility into tomorrow's utility.
The architectural principles and systems explored in this thesis are not the final answer, but they are a step toward it.
They offer a foundation for future researchers, developers, and engineers to build on, refine, and extend---toward a future where the Quantum Internet is not just possible, but real.

\printbibliography[heading=bibintoc]

\appendix

\clearpage
\chapter{List of Papers and Presentations}

\section{Peer-Reviewed Journals}

\begin{refsection}
\begin{refcontext}[sorting=ydnt]{}
    \nocite{naphan-engineering-challenges-in-rgs,suksen-cga-quantum-assisted-feasbility-enforcement,naphan-lower-bound-qaoa,joshua-distimator}
    \printbibliography[heading=none, env=appendixpub]
    \citereset
\end{refcontext}
\end{refsection}

\section{Conference Papers (Peer-Reviewed)}
{
    \renewbibmacro*{finentry}{%
    \usebibmacro{orig:finentry}
    \usebibmacro{print:endnote}
  }

    \begin{refsection}
    \begin{refcontext}[sorting=ydnt]{}
        \nocite{cocori-quisp,fittipaldi-sattelite-quisp,joaquin-michal-cross-validation-quantum-network-simulators}
        \nocite{mia-switch-design-paired-egress-bsa-pools,naphan-half-rgs,naphan-spo,poramet-optimizing-unequal-link-icsec-2021,poramet-quanta-qcnc}
        \nocite{rdv-qi-architecture,sitong-substrate-scheduler,soon-heterogeneous-msm-mim,soon-optimizing-msm-protocol-qcnc-2024,suksen-q-cga-experiments-icsec-2021}
        \nocite{ananda-distimator-sigmetrics,naphan-integrating-purification-rgs}
        \nocite{voy-mia-distimation-ibm-q,bank-dqi-circuit,casapao-double-selection-distimator}
        \printbibliography[heading=none, env=appendixpub]
        \citereset
    \end{refcontext}
    \end{refsection}
}

\section{Posters with Proceedings (Peer-Reviewed)}
\begin{refsection}
\begin{refcontext}[sorting=ydnt]{}
    \nocite{naphan-rgs-poster-qce-2023,soon-entanglement-ejection}
    \printbibliography[heading=none, env=appendixpub]
    \citereset
\end{refcontext}
\end{refsection}

\section{Posters without Proceedings (Peer-Reviewed)}
\begin{enumerate}[parsep=0pt, partopsep=0pt]
    \item ``Lower Bounds on Number of QAOA Rounds Required for Guaranteed Approximation Ratios,'' Quantum Information Processing Conference (QIP), Raleigh, NC, 25 February 2025.
    \item ``Flexible Repeater Graph States for Non-uniform End-to-end Paths,'' Workshop for Quantum Repeaters and Networks (WQRN3), Chicago, IL, 21 August 2022.
\end{enumerate}

\section{Manuscripts under Review and In Preparation}
\begin{refsection}
\begin{refcontext}[sorting=ydnt]{}
    \nocite{naphan-rgs-tqe,shingy-dirac-simulation,poramet-quanta-extended}
    \printbibliography[heading=none, env=appendixpub]
    \citereset
\end{refcontext}
\end{refsection}

\section{Talks}

\begin{enumerate}[parsep=0pt, partopsep=0pt]
    \item ``Architecture and Protocols for All-Photonic Quantum Repeaters,'' Invited talk at Mahidol University Summer School on Quantum Computing and Quantum Communication 2025 (SCQC2025), 27 June 2025.
    \item ``Towards a heterogeneous quantum network: bridging all-photonic and memory-based quantum repeaters,'' Quantum Internet Task Force Workshop, 11 June 2025. (online)
    \item ``Architecture and Protocols for All-Photonic Quantum Repeaters,'' Invited talk at Mahidol University special seminar Celebrating 100 Years of Quantum Physics, 21 March 2025.
    \item ``Architecture and Protocols for All-Photonic Quantum Repeaters,'' IEEE International Conference on Quantum Computing and Engineering (QCE24), 15 September 2024.
    \item ``Lower Bound on Number of QAOA Round Required for Guaranteed Approximation Ratios,'' Keio University Yagami Campus, 29 February 2024.
    \item ``Lower Bound on Number of QAOA Round Required for Guaranteed Approximation Ratios,'' Invited talk at Mahidol University, 8 January 2024.
    \item ``Lower Bound on Number of QAOA Round Required for Guaranteed Approximation Ratios,'' Japanese-French Quantum Information workshop (JFQI), 15 December 2023.
    \item ``Amplitude Amplification for Optimization via Subdivided Phase Oracle,'' IEEE International Conference on Quantum Computing and Engineering (QCE22), 19 September 2022.
    \item ``Distributing Multipartite Entangled State over Quantum Internet,'' International Congress on Science, Technology and Technology-based Innovation (STT47), 5 October 2021.
\end{enumerate}



\end{document}